%% file: strategic-experiments.tex
\documentclass{article}
\pdfoutput=1

\usepackage[utf8]{inputenc} 
\usepackage[T1]{fontenc}    
\PassOptionsToPackage{hyphens}{url} 
\usepackage[hyperfootnotes=false]{hyperref}
\usepackage{booktabs}       
\usepackage{amsfonts}       
\usepackage{nicefrac}       
\usepackage{microtype}      
\usepackage[table]{xcolor}  
\usepackage{breqn}
\usepackage[hang,flushmargin]{footmisc}
\usepackage{algorithm}
\usepackage{algpseudocode}
\usepackage{graphicx}
\usepackage{enumerate}
\usepackage{caption}
\usepackage{array}           
\usepackage{subcaption}
\usepackage{comment}
\usepackage[export]{adjustbox}
\usepackage{mathtools}
\usepackage{amsthm}
\usepackage{authblk}
\usepackage{fullpage}
\usepackage{multirow}

\usepackage{enumitem}

\usepackage[sort&compress, numbers]{natbib}

\usepackage{tikz}
\usetikzlibrary{positioning}
\usetikzlibrary{decorations.pathreplacing}
\usetikzlibrary{calc}
\usetikzlibrary{fit}
\usepackage[HTML]{xcolor}
\usepackage[most]{tcolorbox}
\usepackage{Definitions}
\usepackage{dsfont}

\definecolor{lightbluebox}{RGB}{245,249,255}
\definecolor{darkbluebox}{RGB}{20,40,120}

\newtcolorbox{mybox}[2]{
    enhanced,
    breakable,
    colback=lightbluebox,
    colframe=darkbluebox,
    coltitle=white,
    colbacktitle=darkbluebox,
    fonttitle=\bfseries,
    arc=6pt,
    boxrule=1pt,
    left=6pt,
    right=6pt,
    top=8pt,
    bottom=6pt,
    attach boxed title to top left={
        xshift=6pt,
        yshift=-6pt
    },
    boxed title style={
        arc=6pt,
        boxrule=0pt,
        left=6pt,
        right=6pt,
        top=4pt,
        bottom=4pt
    },
    title=#1,
    label={#2}
}

\definecolor{mycommentcolor}{HTML}{4f9739}
\algrenewcommand\algorithmiccomment[1]{\hfill\textcolor{mycommentcolor}{\(\triangleright\) #1}}

\newlength{\inlineheight}
\algnewcommand{\LineComment}[1]{\State \textcolor{mycommentcolor}{\(\triangleright\) #1}}

\title{Optimizing Social Utility in Sequential Experiments}

\author{Ander Artola Velasco$^{\S}$}
\author{Stratis Tsirtsis$^{\dagger}$}
\author{Manuel~Gomez-Rodriguez$^{\S}$}
\affil{$^{\S}$Max Planck Institute for Software Systems, Kaiserslautern, Germany \\
\{avelasco, manuel\}@mpi-sws.org}

\affil{$^{\dagger}$Hasso Plattner Institute, Potsdam, Germany \\ stratis.tsirtsis@hpi.de}

\date{}

\begin{document}

\maketitle

\begin{abstract}
\input{000abstract.tex}
\end{abstract}

\section{Introduction}
\label{sec:intro}
\input{010introduction.tex}

\section{A Protocol for Subsidized Sequential Experimentation}
\label{sec:model}
\input{020model.tex}

\section{Experimental Design Using Belief Markov Decision Processes}
\label{sec:MDP}
\input{030MDP.tex}

\section{Optimal Experimental Design under Subsidies}
\label{sec:agent}
\input{040agent.tex}

\section{Finding Optimal Subsidies} 
\label{sec:principal}
\input{050regulator.tex}

\section{Experiments: Subsidizing Antibiotic Development}
\label{sec:experiments}
\input{060experiments.tex}

\section{Discussion and Limitations}
\label{sec:discussion}
\input{070discussion.tex}

\section{Conclusions}
\label{sec:conclusions}
\input{080conclusions.tex}

\vspace{2mm}
\xhdr{Acknowledgements} 
Gomez-Rodriguez acknowledges support from the European Research Council (ERC) under the European Union'{}s Horizon 2020 research and innovation programme (grant agreement No. 101169607).
Tsirtsis acknowledges supports from the Alexander von Humboldt Foundation in the framework of the Alexander von Humboldt Professorship (Humboldt Professor of Technology and Regulation awarded to Sandra Wachter) endowed by the Federal Ministry of Education and Research via the Hasso Plattner Institute.
{ 
\small
\bibliographystyle{unsrt}
\bibliography{strategic-experiments}
}

\clearpage
\newpage

\appendix

\input{090appendix.tex}

\end{document}

%% file: 000abstract.tex
Regulatory approval of products in high-stakes domains such as drug development requires statistical evidence of safety and efficacy through large-scale randomized controlled trials.
%
However, the high financial cost of these trials may deter developers who lack absolute certainty in their product's efficacy, ultimately stifling the development of `moonshot' products that could offer high social utility.
%
To address this inefficiency, in this paper, we introduce a statistical protocol for experimentation where the product developer (the agent) conducts a randomized controlled trial sequentially and the regulator (the principal) partially subsidizes its cost.
%
By modeling the protocol using a belief Markov decision process, we show that the agent'{}s optimal strategy can be found efficiently using dynamic programming.
%
Further, we show that the social utility is a piecewise linear and convex function over the subsidy level the principal selects, and thus the socially optimal subsidy can also be found efficiently using divide-and-conquer.
%
Simulation experiments using publicly available data on antibiotic development and approval demonstrate that our statistical protocol can be used to increase social utility by more than $35$$\%$ relative to standard, non-sequential protocols. 

%% file: 010introduction.tex
%
Access to markets in high-stakes domains, such as drug development, is strictly governed by regulatory bodies to ensure that new products meet rigorous safety and efficacy standards~\citep{gieringer1985safety,fda2019effectiveness}.
Randomized controlled trials (RCTs) serve as one of the primary mechanisms for access control, requiring developers to gather sufficient statistical evidence to prove that a product is safe and effective~\citep{o2010building, Janiaud2021-jg, farina2021strength}.

However, the high financial costs of RCTs can stifle the development of `moonshot' products that could offer high social utility, as developers may hesitate to proceed if they lack absolute certainty in their product'{}s efficacy~\citep{detsky1989clinical,frantz2003clinical,martin2017much}.
%
To avoid missing such opportunities, calls have emerged for adaptive RCTs~\citep{brown2009adaptive,mahajan2010adaptive,fda2026use} and targeted subsidies~\citep{nih,fda,dfg,edctp}, particularly in the context of orphan and rare diseases where small patient populations often discourage private investment.

%
Yet, the effectiveness of such interventions depends on how developers strategically respond to them. In this context, a recent line of work~\citep{bates2022principal,shi2024sharp,hossain2025strategic} has argued that the regulatory approval of products is best modeled as a principal-agent game~\citep{grossman1992analysis}, 
where the regulator (the principal) designs an approval protocol and a hypothesis testing rule to incentivize the product developer (the agent) to act in a way that aligns with the principal's interests.
%
In this paper, we extend this line of work to consider a setting in which the RCTs are conducted sequentially and are (partially) subsidized by the principal.

%
\xhdr{Our contributions}
%
We introduce a statistical protocol for experimentation where the product developer (the agent) conducts an RCT sequentially and the regulator (the principal) partially subsidizes its cost.
%
At each step of the protocol, the agent and principal update their beliefs about the product's effectiveness based on the latest experimental outcome. 
%
If there is sufficient statistical evidence to reject the null hypothesis~\citep{ramdas2023gametheoreticstatisticssafeanytimevalid}, the principal approves the product; otherwise, they request that the agent gather more evidence, and the agent may either proceed with the experiment or terminate without approval.
%
By modeling the protocol using a belief Markov decision process~\citep{kaelbling1998planning}, we show that the agent'{}s optimal strategy can be found efficiently using dynamic programming.
%
Further, we show that the social utility is a piecewise linear and convex function over the subsidy level the principal selects, and thus the socially optimal subsidy can also be found efficiently via divide-and-conquer. 

%
To validate our statistical protocol, we conduct simulation experiments using publicly available data on antibiotic development and approval.
The results show that our protocol can be used to increase social utility by more than $35$$\%$ relative to standard, non-sequential protocols.\footnote{The code for our experiments is publicly available at \url{https://github.com/Human-Centric-Machine-Learning/strategic-experiments}.}

%
\xhdr{Further related work} Our work builds upon further related work on the economic aspects of statistical testing, sequential hypothesis testing, and Bayesian experimental design.

%
A recent and closely related line of work studies the economic and strategic incentives arising in regulatory approval processes~\citep{TetenovEcon, bates2022principal, McClellan2022, shi2024sharp, hossain2025strategic}. Therein, Shi et al.~\citep{shi2024sharp} and Hossain et al.~\citep{hossain2025strategic} focus on determining the optimal hypothesis test that the principal can use to control false positives and false negatives, while Tetenov~\cite{TetenovEcon} and Bates et al.~\citep{bates2022principal} study, respectively, the design of approval and payment rules to disincentivize agents who know their product is ineffective from participating in the approval process. Relatedly, McClellan~\citep{McClellan2022} considers a setting in which the principal designs approval rules to encourage agent participation without monetary transfers, \eg, by lowering approval standards if previous experiments were not successful. In contrast to these works, we study how the principal can optimally subsidize the agent to increase social utility in a scenario where both are uncertain about the effectiveness of the product.

%
Within the hypothesis testing literature, our work draws on an active line of research on anytime-valid statistical inference using e-values~\citep{ramdas2023gametheoreticstatisticssafeanytimevalid, GrunwaldSafe, ramdas2025hypothesistestingevalues}, which has been successfully applied to a wide range of statistical problems~\citep{waudbysmith2022estimatingmeansboundedrandom, xu2024Online, shin_ramdas_rinaldo_2024, shekhar2024Two, waudbysmith2025universallogoptimalitygeneralclasses, chugg2026post, bates2022principal, gauthier2025backward, huang2026towards, velasco2026auditing, dhillon2026escores}.
Most closely related to ours is the work of Bates et al.~\citep{bates2022principal}, who also uses e-values in the context of RCTs. However, their focus is on designing contracts that disincentivize agents with ineffective products from participating in the approval process. We instead use e-values to design an approval protocol that provides sequential error guarantees to the principal.

%
Our work also connects to the broad literature on Bayesian experimental design~\citep{RobbinsOptimal, Peskir2006-vv,Powell2012-tq, Ghavamzadeh_2015, cho2022bayesian, buening2023minimaxbayesreinforcementlearning, shen2023bayesian, cheng2025optimal, Shen2025-zc}. Among these, the closest works to ours~\citep{shen2023bayesian,cheng2025optimal,Shen2025-zc} model a Bayesian agent’s experiment selection problem using (partially observable) Markov decision processes. However, their objective is to design policies that maximize information gain while trading off experimental costs. In contrast, in the approval setting we study, the agent has direct economic incentives to conduct experiments, as product approval is financially beneficial.

%% file: 020model.tex
We consider an agent (the product developer) who seeks regulatory approval for a product from a principal (the regulator). The product is characterized by an efficacy parameter $\theta^*\in [0,1]$, unknown both to the agent and the principal, with higher values indicating a more effective product.\footnote{In certain settings, the principal may be interested not only in the efficacy of a product but also in other properties of the product such as its safety. In Appendix~\ref{app:general-model}, we discuss how to extend the approval process to such settings.} Motivated by multi-stage clinical trials~\citep{fda2019effectiveness, Janiaud2021-jg}, in this section, we introduce a sequential approval process in which the agent conducts a sequence of RCTs to provide sufficient evidence that the product meets the principal's standards. In turn, the principal commits to subsidize a fraction of the agent's total experimentation cost and, after each trial, they decide whether to approve the product or require the agent to conduct further experimentation.

The agent begins the (sequential) approval process with a prior belief $B_0$ about the efficacy $\theta^*$ of its product, which we model using a Beta distribution $B_0 = \mathrm{Beta}(\alpha_0, \beta_0)$, where $\alpha_0, \beta_0 > 0$ are given parameters.\footnote{The parameters $\alpha_0$ and $\beta_0$ characterize prior information the agent may have about the product based on, \eg, preliminary tests. In Appendix~\ref{app:general-model}, we show that the approval process can be extended to more general settings with arbitrary prior beliefs.}
At each time step $t \in [T] = \{0, 1, \dots, T\}$ of the approval process, the agent's action is to either (i) conduct a randomized trial with a sample size $n_t \in \{1, \dots, n^{\texttt{max}}\}$ and incur a cost $c(n_t)$, where $c: \mathbb{N} \to \mathbb{R}_{+}$ is a non-decreasing cost function, or (ii) opt out and stop the approval process at no additional cost, which we represent as $n_t=0$ with $c(0) = 0$. If the agent decides to conduct a trial by selecting $n_t>0$, the agent observes a random outcome  $X_t \sim \mathrm{Bin}(n_t, \theta^*)$, which represents the number of \emph{successes} in the trial and depends on the unknown efficacy $\theta^{*}$, and this outcome is then revealed to the principal.\footnote{The specific meaning of the number of successes is application dependent. In the context of clinical trials, it may correspond to the number of patients who recover after receiving a treatment.}
Then, based on the outcome $X_t$, the agent updates their (posterior) belief $B_t$ about the true efficacy $\theta^*$ of the product, \ie,
\begin{equation}\label{eq:bayes-belief}
        B_{t+1} = \mathrm{Beta}(\underbrace{\alpha_t + X_t}_{\alpha_{t+1}}, \underbrace{\beta_t + n_t -X_t}_{\beta_{t+1}}).
\end{equation}
Throughout the process, the agent employs a (possibly randomized) policy $\pi \in \Pi$ to select the sample sizes $n_t \sim \pi(\alpha_t, \beta_t, C_t, t)$ based on their belief $B_t = \mathrm{Beta}(\alpha_t, \beta_t)$ about the efficacy of the product, the total cost $C_t = \sum_{k=0}^{t-1} c(n_k)$ they have incurred so far, and the time step $t$ of the approval process.
Moreover, the principal uses the revealed outcomes $X_t$ as evidence to decide on the approval of the product. Formally, the principal conducts a (sequential) hypothesis test with null and alternative hypotheses given by
\begin{equation}\label{eq:hypothesis-test}
    H_0 = \left\{ \theta^* \, \colon \, \theta^* < \theta^{\texttt{b}} \right\} \,\, \text{and} \,\, H_1 = \left\{ \theta^* \, \colon \, \theta^* \geq \theta^{\texttt{b}} \right\},
\end{equation}
where $\theta^{\texttt{b}}\in(0,1)$ is a baseline efficacy mandated by the principal and known to the agent (\eg, the efficacy of the current standard-of-care treatment in the context of clinical trials). The null hypothesis $H_0$ therefore corresponds to the product failing to meet the principal's standard, and we will equivalently refer to the principal approving the product as \emph{rejecting} $H_0$. 

In general, the principal can implement any decision rule to reject $H_0$. 
However, in what follows, we draw on the literature on sequential hypothesis testing~\citep{waldSeq,ramdas2023gametheoreticstatisticssafeanytimevalid,ramdas2025hypothesistestingevalues}, and consider a principal who aims to control the false positive rate, \ie, the probability of approving a product whose efficacy does not exceed the baseline $\theta^{\texttt{b}}$.
Concretely, based on the sample size $n_t$ and revealed outcome $X_t$, the principal computes---and shares with the agent---a non-negative quantity $E(X_t, n_t) \in \mathbb{R}_{+}$, referred to as an \emph{e-value}, quantifying the observed evidence against $H_0$ at time $t$ (\ie, larger e-values correspond to stronger evidence against $H_0$).\footnote{Computing $E(X_t, n_t)$ plays a role similar to that of a p-value for rejecting $H_0$, while offering stronger guarantees in sequential settings. See Appendix~\ref{app:background-seq} for a brief overview of sequential hypothesis testing with e-values.}
Then, the principal aggregates all available evidence multiplicatively via a (stochastic) process $M$, which we will refer to as the \emph{test process}, and whose value at time $t$ is defined as:
\begin{equation}\label{eq:test-process}
M_{t}=
\begin{dcases}
    1 & t=0\\
    E(X_{t-1}, n_{t-1}) \cdot M_{t-1} & t\geq1.
\end{dcases}
\end{equation}

Finally, based on the test process $M$, the principal rejects the null hypothesis $H_0$ (approves the product) as soon as the accumulated evidence exceeds a fixed threshold $\kappa\in (0,1)$ set in advance, \ie, as soon as $M_{t+1} \geq 1/\kappa$ for some time step $t$. Here, $\kappa$ acts as a tolerance parameter specifying how much evidence must be accumulated before rejecting $H_0$, where smaller values of $\kappa$ correspond to a more conservative approval standard set by the principal, and note that, if the process fails to yield sufficient evidence for approval within $T$ time steps, it concludes without approval. 

In the above sequential test, the false positive rate is bounded by $\kappa$ as long as the e-values are such that $\EE_{H_0}[E(X_t, n_t)] \leq 1$, as shown elsewhere~\citep{ramdas2023gametheoreticstatisticssafeanytimevalid}, \ie,
\begin{equation}
    P_{H_0}\left( \left\{ \exists t \in [T] \colon M_{t+1} \geq 1/\kappa \right\} \right) \leq \kappa.
\end{equation}
Importantly, the above guarantee on the false positive rate holds throughout the entire approval process---a particularly desirable property known as \emph{any-time validity}~\cite{GrunwaldSafe}. To ensure that $\EE_{H_0}[E(X_t, n_t)] \leq 1$ holds, we construct an e-value that exponentiates a sufficient statistic for the unknown efficacy $\theta^*$~\citep{PenaExp, HowardTimeUni} by comparing the empirical success rate $X_t / n_t$ against a monotone transformation of the baseline $\theta^{\texttt{b}}$:\footnote{See Appendix~\ref{app:general-model} for an extension to arbitrary e-values; proofs are deferred to Appendix~\ref{app:proofs}.}
\begin{proposition}\label{prop:binomial-e-value}
    Given the (unknown) efficacy parameter $\theta^*\in [0,1]$, $n_t>0$, and a binomial variable $X_t\sim \mathrm{Bin}(n_t,\theta^*)$, the positive random variable
    \begin{equation}\label{eq:binomial-e-value}
        E(X_t,n_t) = \exp\left( X_t - n_t\cdot \log(1+\theta^{\texttt{b}}(e-1)) \right)
    \end{equation}
    is a valid e-value under $H_0 = \left\{ {\theta^* \, \colon \, \theta^* < \theta^{\texttt{b}}} \right\}$, that is, $\EE_{X_t\sim \mathrm{Bin}(n_t,\theta^*)}[E(X_t,n_t)] \leq 1$ for any $\theta^*\in H_0$.
\end{proposition}

Furthermore, given the specific form of e-value in Eq.~\ref{eq:binomial-e-value}, it is easy to verify that the value of the test process $M_t$ is uniquely determined by the parameters $\alpha_t$ and $\beta_t$ characterizing the agent's belief:

\begin{proposition}\label{prop:test-process}
At any time step $t\in [T]$, the value of the test process $M_t$ satisfies:
\begin{equation}\label{eq:f-function}
    M_t = f(\alpha_t, \beta_t)
    \quad \text{where} \quad
    f(\alpha, \beta)
    =
    \exp\left(
    \alpha - \alpha_0
    - (\alpha + \beta - \alpha_0 - \beta_0)\cdot\log(1+\theta^{\texttt{b}}(e-1))
    \right).
\end{equation}
\end{proposition}

%
%
If the agent's product is approved, the agent and the principal obtain benefits $\rho^{\texttt{A}}, \rho^{\texttt{S}} > 0$, respectively, and we refer to the latter as the social benefit upon approval. In the context of clinical trials, $\rho^{\texttt{A}}$ can be interpreted as the (estimated) economic benefit obtained by the agent---the pharmaceutical company---from drug sales if the drug is approved, while $\rho^{\texttt{S}}$ represents the corresponding benefit to society resulting from bringing an effective treatment to market. 

Since the principal also benefits from a product's approval, it can be in their interest to incentivize the agent to continue experimenting, particularly in situations where the agent's benefit $\rho^{\texttt{A}}$ is not high enough to compensate for the total experimentation cost required to reject $H_0$. 
In the following, we focus our attention on a natural and widely-used mechanism through which the principal can provide such an incentive: subsidizing, conditional on approval, a fraction $\varepsilon \in [0, \varepsilon^{\texttt{max}}]$ of the total cost incurred by the agent, where $\varepsilon^{\texttt{max}} \leq 1$ denotes the maximum fraction the principal is willing to subsidize~\citep{Renwick2016-wj, Wu31122022}.

Given the above benefits, experimentation cost, and subsidies, the agent's and the principal's (expected) utilities $U^{\texttt{A}}$ and $U^{\texttt{S}}$ are given by 
\begin{equation}\label{eq:true-utilities}
\left.
\begin{aligned}
        U^{\texttt{A}}(\pi; \varepsilon) &= \EE_{ \pi} \left[\left( \rho^{\texttt{A}} + \varepsilon \cdot \sum_{t=0}^\tau c(n_t) \right) \cdot \mathds{1}\{ M_{\tau+1} \geq 1/\kappa\} - \sum_{t=0}^\tau c(n_t) \,\middle|\, \theta^* \right], \\
        U^{\texttt{S}}(\varepsilon;\pi) &= \EE_{ \pi} \left[ \left( \rho^{\texttt{S}} - \varepsilon \cdot \sum_{t=0}^\tau c(n_t) \right) \cdot \mathds{1}\{ M_{\tau+1} \geq 1/\kappa\} \,\middle|\, \theta^* \right],
\end{aligned}
\right.
\end{equation}
where $\tau = T \wedge \min\{ t\in[T] \colon n_t = 0 \,\,\text{or}\,\, M_{t+1} \geq 1/\kappa\} $ is the last step of the approval process, $\mathds{1}\{\bullet\}$ is the indicator function, and the expectation $\EE_{\pi}[\bullet|\theta^*]$ is taken over the random outcomes $X_t \sim \mathrm{Bin}(n_t, \theta^*)$.\footnote{Throughout, we adopt the convention $\min \emptyset = +\infty$ and write $x \wedge y = \min(x, y)$ for $x, y \in \mathbb{R} \cup \{+\infty\}$.}

However, since the agent and the principal do not know the true efficacy $\theta^*$ a priori, they cannot find the policy $\pi$ and subsidy $\varepsilon$ that maximize their respective utilities, as defined in Eq.~\ref{eq:true-utilities}.
In the next sections, we investigate how the agent and principal may leverage their evolving beliefs about efficacy to find near-optimal policies and subsidies as the approval process progresses.

%% file: 030MDP.tex
If the agent does not know the true efficacy $\theta^*$, we argue that, as the approval process progresses, they act based on their beliefs regarding $\theta^*$ given the trial outcomes observed~\citep{harsanyi1979bayesian}.
Consequently, we consider an agent who determines their trial sample sizes by planning ahead---not according to the (unknown) actual evolution of the process, but according to how they anticipate their beliefs will evolve depending on their actions and anticipated outcomes.
%

To formalize the agent's planning strategy, we employ the framework of (belief) Markov decision processes~\citep{Sutton1998,kaelbling1998planning}.
Specifically, for a fixed subsidy $\varepsilon$ specified by the principal, we define the process $\Mcal^\varepsilon = (\Scal, \Acal, P, r^\varepsilon, T)$, whose components we describe next.

The state space $\Scal = \Scal^{\texttt{in}} \cup \{S^{\texttt{out}} \}$ includes states $S = (\alpha, \beta, C) \in \Scal^{\texttt{in}}$, where $\alpha$ and $\beta$ are the parameters characterizing the agent's belief and $C$ is the agent's total running cost, 
as well as a special absorbing state $S^{\texttt{out}}$, which indicates that the agent has opted out of the approval process.
The action space $\Acal= \{0,\dots,n^{\texttt{max}}\}$ consists of all possible sample sizes the agent may select.

Further, the transition distribution $P$ characterizes how the agent anticipates their belief and cost will evolve after they conduct a trial of size $n$.
Formally, a transition from a state $S = (\alpha, \beta, C)$ to a state $S'$ follows from the (randomized) assignment
\begin{equation}\label{eq:transition-dynamics}
S' =
    \begin{dcases}
    (\alpha + X,\; \beta + n - X,\; C + c(n)) 
    & \text{if } n > 0 \text{ and } f(S) < 1/\kappa\\
    S 
    & \text{if } n > 0 \text{ and } f(S) \geq 1/\kappa\\
    S^{\texttt{out}} 
    & \text{if } n = 0 
\end{dcases}
\end{equation}
where 
$X \sim \mathrm{Bin}(n, \theta)$ with $\theta \sim \mathrm{Beta}(\alpha,\beta)$,
and the function $f(S) = f(\alpha,\beta)$ as in Eq.~\ref{eq:f-function} if $S \in \Scal^{\texttt{in}}$ and $f(S^{\texttt{out}}) = 0$.
%
In words, the upper case captures the Bayesian belief update and the additional cost incurred by the agent for conducting an experiment of size $n$ and observing an outcome $X$ (see Eq.~\ref{eq:bayes-belief}), 
the middle case captures the successful conclusion of the approval process, 
and the lower case captures the scenario in which the agent decides to opt out.
%
Importantly, note that evaluating the likelihood of such transitions does not depend on the efficacy $\theta^*$ and thus is possible solely based on the agent's belief.

Finally, the reward $r^\varepsilon$ characterizes the agent's anticipated profit or loss due to conducting a trial under subsidy $\varepsilon$.
Formally, the reward for transitioning from state $S$ to state $S'$ via action $n$ is given by
\begin{equation}\label{eq:MDP-reward-def}
r^\varepsilon(S,n,S') =
\begin{dcases}
    -c(n) + (\rho^{\texttt{A}} + \varepsilon\cdot(C+c(n)))\cdot\mathds{1}\{ f(S') \ge 1/\kappa \} & 
    \text{if } S \neq S^{\texttt{out}} \text{ and } f(S)<1/\kappa\\
    0 & \text{if } S = S^{\texttt{out}} \text{ or } f(S) \ge 1/\kappa.
\end{dcases}
\end{equation}
In words, as long as the approval process is in progress, the reward is simply equal to the (negative) cost $c(n)$ the agent anticipates to incur by conducting a trial with their chosen sample size $n$.
If the updated state $S'$ leads to product approval (\ie, $f(S)<1/\kappa$ and $f(S')\geq 1/\kappa$), the agent anticipates receiving their one-off benefit $\rho^{\texttt{A}}>0$ and a subsidized fraction of their total incurred cost throughout the process, as determined by the subsidy $\varepsilon$ selected by the principal.

Now that we have defined all the components of the Markov decision process $\Mcal^\varepsilon$, we can formalize the agent's strategy within the process.
Starting from the initial state $S_0 = (\alpha_0,\beta_0,0)$, the agent aims to select a policy $\pi: \Scal \times [T] \to \Delta(\Acal)$ maximizing their \emph{anticipated utility}
\begin{equation}\label{eq:adaptive-utility}
    \Bar{U}^{\texttt{A}}(\pi;\varepsilon) = \EE_{\pi}\left[ \sum_{t=0}^T r^\varepsilon(S_t,n_t,S_{t+1}) \middle| S_0=(\alpha_0, \beta_0, 0)\right],
\end{equation}
where the expectation is taken over the state transitions that the agent anticipates to occur throughout the approval process, given their initial belief.
Here, it is important to note that both the reward $r^\varepsilon$ and the transition distribution $P$ are known to the agent, hence they can evaluate the anticipated utility of any given policy $\pi$ before the approval process starts.
Moreover, due to Bellman's optimality principle in MDPs~\citep{Sutton1998}, for any $S \in \Scal$ and time step $l \in [T]$, the agent's optimal policy $\pi^\varepsilon \in \argmax_{\pi} \Bar{U}^{\texttt{A}}(\pi;\varepsilon)$ satisfies that
\begin{equation}\label{eq:value-function}
        \pi^\varepsilon (S, l) \in \argmax_{\pi} V^\varepsilon_\pi(S,l) \quad \text{where} \quad V^\varepsilon_\pi(S,l) = \EE_{\pi}\left[ \sum_{t=l}^T r^\varepsilon(S_t,n_t,S_{t+1}) \middle| S_l=S\right].
\end{equation}
In the above equation, the function $V^\varepsilon_\pi$ is often referred as the value function. Moreover, note that $\Bar{U}^{\texttt{A}}(\pi;\varepsilon) = V^\varepsilon_\pi(\alpha_0,\beta_0,0,0)$.

In this context, a natural question is how the anticipated utility $\Bar{U}^{\texttt{A}}(\pi;\varepsilon)$ relates to the agent's true utility $U^{\texttt{A}}(\pi; \varepsilon)$ defined in Eq.~\ref{eq:true-utilities}.
The following proposition shows that the anticipated utility is equal to true agent utility averaged over the agent's initial belief about the true efficacy $\theta^*$ of their product at the start of the approval process. 
\begin{proposition}\label{prop:MDP-equiv}
    Let $\varepsilon \in [0,\varepsilon^{\texttt{max}}]$ be any subsidy set by the principal, and $\pi$ be any agent policy. Then,
        \begin{equation}
                \Bar{U}^{\texttt{A}}(\pi;\varepsilon) = \EE_{\theta^* \sim B_0} \left[U^{\texttt{A}}(\pi; \varepsilon)\right].
        \end{equation}
\end{proposition}
The above proposition reveals that an agent implementing the decision policy $\pi^\varepsilon$ is, in fact, a Bayesian decision maker maximizing their expected utility~\citep{Ghavamzadeh_2015,russo2020tutorialthompsonsampling,buening2023minimaxbayesreinforcementlearning}.

Next, we formalize the principal's strategy within the process.
Let $Q$ be a prior distribution charac\-te\-ri\-zing the principal's prior knowledge about the agent's initial belief $(\alpha_0, \beta_0)$. Then, the principal aims to select a subsidy $\varepsilon$ maximizing the \emph{anticipated social utility}
\begin{equation}\label{eq:adaptive-utility-regulator}
    \Bar{U}^{\texttt{S}}(\varepsilon;\pi)
    = \mathbb{E}_{(\alpha_0,\beta_0)\sim Q}\left[\mathbb{E}_{\pi} \left[
    \sum_{t=0}^T \left( \rho^{\texttt{S}} - \varepsilon \cdot C_{t+1} \right) \cdot
    \mathds{1}\left\{ 0 < f(S_t) < 1/\kappa \leq f(S_{t+1}) \right\}
    \;\middle|\; S_0 = (\alpha_0,\beta_0,0)
    \right]\right],
\end{equation}
where the inner expectation is taken over the state transitions that the principal anticipates to occur throughout the approval process, given the prior knowledge about the agent's initial belief.

In the next sections, we analyze how the agent and the principal maximize their anticipated utilities $\Bar{U}^{\texttt{A}}(\pi;\varepsilon)$ and $\Bar{U}^{\texttt{S}}(\varepsilon;\pi)$ within a natural Stackelberg setting~\citep{Von_Stackelberg2010}. 
In this setting, the principal (the leader) first commits to a subsidy level $\varepsilon^*$ maximizing the social utility $\Bar{U}^{\texttt{S}}$, accounting for the fact that, given any $\varepsilon$, the agent (the follower) will implement the optimal policy $\pi^\varepsilon$ that maximizes their utility $\Bar{U}^{\texttt{A}}$. 
Then, after observing this commitment, the agent implements the optimal policy. 
Formally, the subsidy level $\varepsilon^{*}$ is the solution to the following optimization problem:
\begin{equation}\label{eq:stackelberg-equilibrium}
    \varepsilon^*
    =
    \argmax_{\varepsilon \in [0,\varepsilon^{\texttt{max}}]}
    \;
    \Bar{U}^{\texttt{S}}\big(\varepsilon;\pi^\varepsilon\big)
    \quad \text{subject to} \quad
    \pi^\varepsilon
    \in
    \argmax_{\pi\in\Pi}\, \Bar{U}^{\texttt{A}}(\pi;\varepsilon).
\end{equation}
The above Stackelberg setting fits a variety of real-world applications. 
For example, in clinical trials, funding agencies act as leaders who commit to and announce a funding program; drug developers are the followers who apply to the program and, if successful, run a trial~\citep{usc26_45c,nih2026smallbusiness}.

%% file: 040agent.tex
In this section, we characterize the agent's optimal policy $\pi^\varepsilon$ along with its corresponding value function
\begin{equation}\label{eq:optimal-policy-value}
     V^\varepsilon(S,l) = \max_{\pi\in\Pi}\, V^\varepsilon_\pi(S,l)\quad \forall\, S \in \mathcal{S},\, 
    l \in [T].
\end{equation}
%
To this end, we first note that, once the process $\Mcal^{\varepsilon}$ reaches a state where $f(S)\geq 1/\kappa$, or the state $\Scal^{\texttt{out}}$, 
the action $n=0$ is always optimal. 
This is implied by the transition dynamics and the reward definition in Eqs.~\ref{eq:transition-dynamics} and~\ref{eq:MDP-reward-def}. 
Therefore, to find the policy $\pi^{\varepsilon}$, it suffices to consider states $S$ where $0<f(S)<1/\kappa$. 
Further, we show that, in the process $\Mcal^{\varepsilon}$, the number of (unique) reachable states $S$ where $f(S)\geq 1/\kappa$ is finite:
\begin{proposition}\label{prop:reacheable-set}
    For any realization of the process $\Mcal^\varepsilon$, any reached state $S$ where $0<f(S)<1/\kappa$ belongs to a finite set $S^{\texttt{r}} \subset \Scal$. 
    Moreover, under linear cost $c(n)$, the set $S^{\texttt{r}}$ has size $\Ocal((n^{\texttt{max}})^2 \cdot T^3)$.
\end{proposition}
As an immediate consequence, we can find an optimal (deterministic) policy $\pi^{\varepsilon}$ using standard planning methods.
In particular, in Algorithm~\ref{alg:value_iteration}, we provide an adaptation of the classical value iteration algorithm~\cite{HernndezLerma1996}, which is guaranteed to find an optimal policy $\pi^{\varepsilon}$ in $\Ocal((n^{\texttt{max}})^4 \cdot T^3)$.
%
In the remainder of this section, we derive several key structural insights about the optimal policy $\pi^{\varepsilon}$ and the optimal value function $V^{\varepsilon}$, which will be helpful to efficiently find the principal's optimal subsidy $\varepsilon^{*}$.

Our starting point is the observation that, in light of Proposition~\ref{prop:test-process}, the value of the test process $M$ at a state $S=(\alpha, \beta,C)$ increases with $\alpha$ and decreases with $\beta$.
This suggests that states with larger $\alpha$ and smaller $\beta$ are more favorable.
The following proposition formalizes this intuition by showing that the optimal value function $V^{\varepsilon}$ satisfies a monotonicity property in both the belief parameters and the accumulated cost.
\begin{proposition}\label{prop:value-monotonicity}
    For any time step $l \in [T]$ and pair of states $S = (\alpha, \beta, C)$ and $S' = (\alpha', \beta', C')$ such that $f(S) < 1/\kappa$ and $f(S') < 1/\kappa$, the following holds:
    \begin{enumerate}[label=(\roman*)]
        \item[1.] $V^\varepsilon(\alpha,\beta,C,l) \leq V^\varepsilon(\alpha',\beta,C,l)$ if $\alpha \leq \alpha'$; 
        \item[2.] $V^\varepsilon(\alpha,\beta,C,l) \geq V^\varepsilon(\alpha,\beta',C,l)$ if $\beta  \leq \beta'$;
        \item[3.] $V^\varepsilon(\alpha,\beta,C,l) \leq V^\varepsilon(\alpha,\beta,C',l)$ if $C \leq C'$.
    \end{enumerate}
\end{proposition}
Leveraging the above proposition, we can characterize the conditions under which the optimal policy $\pi^{\varepsilon}$ opts out of the approval process by selecting the action $n=0$. 
In particular, the following proposition shows that $\pi^{\varepsilon}$ exhibits a threshold structure: 
at a state with belief $(\alpha, \beta)$, it selects $n=0$ if and only if $\beta$ exceeds a threshold that depends on $\alpha$ (see Figure~\ref{fig:geometry} in Appendix~\ref{app:geometry} for an illustration). 
%
\begin{proposition}\label{prop:opt-out-mono}
For any time step $t \in [T]$ and total cost $C_t$, there exists a non-decreasing function $\tilde{\beta} \,\colon\, \mathbb{R}_{+} \to \mathbb{R}_{+}$ such that, for any state $S_t = (\alpha_t,\beta_t, C_t) \in \Scal^{\texttt{in}}$, the optimal policy $\pi^{\varepsilon}$ opts out of the approval process if $\beta_t > \tilde{\beta}(\alpha_t)$, and it does not opt out if $\beta_t < \tilde{\beta}(\alpha_t)$.
\end{proposition}

Furthermore, we can also characterize how the agent's anticipated utility $\Bar{U}^{\texttt{A}}(\pi^\varepsilon;\varepsilon)$ under the optimal policy $\pi^\varepsilon$ depends on the subsidy. 
To this end, we first show that, for any policy $\pi$, the value function $V^\varepsilon_\pi$ is linear in the subsidy $\varepsilon$:
\begin{proposition}\label{prop:linear-value}
For any policy $\pi$, state $S \in \Scal$, and time step $l \in [T]$, the value function $V^\varepsilon_\pi$ admits a linear decomposition
\begin{equation}\label{eq:linear-value-decomposition}
    V^{\varepsilon}_{\pi}(S, l) = V^{0}_{\pi}(S, l) + \varepsilon \cdot A_{\pi}(S, l),
\end{equation}
where $V^{0}_{\pi}(S, l)$ is the value function in the unsubsidized process $\Mcal^0$, and $A_{\pi}(S, l)\geq 0$ denotes the expected total cost incurred conditional on approval, starting from state $S$ at time $l$ (see Eq.~\ref{eq:rejection-cost-conditional} in Appendix~\ref{app:proof-linear-values}).
\end{proposition}
Building upon this result, the agent's anticipated utility $\Bar{U}^{\texttt{A}}(\pi^\varepsilon;\varepsilon)$ under the optimal policy $\pi^\varepsilon$ admits a concise structural characterization, as formalized by the following proposition:
\begin{proposition}\label{prop:utility-piecewise-convex}
    %
    The agent’s optimal anticipated utility $\bar{U}^{\texttt{A}}(\pi^\varepsilon; \varepsilon)$, and its expectation $\mathbb{E}_{(\alpha_0, \beta_0) \sim Q}[\bar{U}^{\texttt{A}}(\pi^\varepsilon; \varepsilon)]$ according to the principal's belief $Q$, are piecewise linear, continuous, and convex functions of the subsidy $\varepsilon$ over a partition $\mathcal{P} = \{ \varepsilon_0, \varepsilon_1, \dots, \varepsilon_L \}$ of the interval $[0, \varepsilon^{\max}]$, with $0 = \varepsilon_0 < \varepsilon_1 < \cdots < \varepsilon_L = \varepsilon^{\max}$. Moreover, for each interval of $\Pcal$, the agent's optimal policy is constant, \ie, $\pi^\varepsilon = \pi_i$ for all $\varepsilon \in [\varepsilon_{i}, \varepsilon_{i+1})$.
\end{proposition}
In the next section, we leverage the above characterization of the agent's anticipated optimal utility to develop an algorithm that computes the principal’s optimal subsidy $\varepsilon^*$, as defined in Eq.~\ref{eq:stackelberg-equilibrium}.\footnote{In principle, Eq.~\ref{eq:stackelberg-equilibrium} may allow for multiple solutions. However, our objective is not to characterize the full set of solutions, but rather to provide an algorithmic procedure to compute one such solution.} 

%% file: 050regulator.tex
To find the optimal subsidy $\varepsilon^{*}$, 
our starting point is the observation that, for any policy $\pi$, the anticipated social utility $\Bar{U}^{\texttt{S}}(\varepsilon;\pi)$ is a decreasing function of $\varepsilon$. 
This is because, as the subsidy $\varepsilon$ increases, the principal covers a higher fraction of the agent's cost, but the probability that the product receives approval remains unchanged.
More formally, we have the following proposition:
\begin{proposition}\label{prop:regulator-utility-linear}
For any policy $\pi$ and subsidy $\varepsilon\in[0,\varepsilon^{\texttt{max}}]$, the anticipated social utility $\Bar{U}^{\texttt{S}}(\varepsilon;\pi)$ admits a linear decomposition
    \begin{equation}
        \Bar{U}^{\texttt{S}}(\varepsilon;\pi) = \rho^{\texttt{S}}\cdot  \EE_{(\alpha_0, \beta_0) \sim Q} \left[ P_{\pi}(\left\{\exists t \in [T] \colon f(S_{t+1})\geq 1/\kappa\right\} \,|\, S_0 ) \right] - \varepsilon \cdot \EE_{(\alpha_0, \beta_0) \sim Q} \left[ A_{\pi}(S_0,0)\right],
    \end{equation}
\end{proposition}
In the above expression, the first term corresponds to the principal's anticipated probability of rejecting $H_0$ 
under policy $\pi$, and the second term corresponds to the anticipated total cost borne by the principal through the subsidy.

As a consequence, and in light of Proposition~\ref{prop:utility-piecewise-convex}, in each interval $[\varepsilon_i, \varepsilon_{i+1})$ of the partition $\Pcal$ where a fixed policy $\pi_i$ is optimal, 
the anticipated social utility $\Bar{U}^{\texttt{S}}(\varepsilon;\pi_i)$ is a decreasing (linear) function of the subsidy over $[\varepsilon_i, \varepsilon_{i+1})$, and therefore it is maximized at the left point $\varepsilon_i$:
\begin{equation}
    \max_{\varepsilon \in [\varepsilon_i, \varepsilon_{i+1})} \Bar{U}^{\texttt{S}}(\varepsilon;\pi_i) = \Bar{U}^{\texttt{S}}(\varepsilon_i;\pi_i).
\end{equation}
Crucially, the principal can compute each policy $\pi_i$ \emph{without knowing the agent’s initial belief} $(\alpha_0, \beta_0)$, since the optimal policy for the process $\Mcal^{\varepsilon_i}$ does not depend on the initial state. This stands in contrast to many settings in the literature on (Bayesian) Stackelberg games, where the principal typically must anticipate the agent's best response by averaging over the agent's private information~\cite{Conitzer,paruchuri2008playing}.

\begin{algorithm}[t]
\caption{Finds the Principal's Optimal Subsidy}
\label{alg:epsilon}

\small

\begin{algorithmic}[1]

\State \textbf{Input:} MDP solver $\texttt{SolveMDP}$, maximum subsidy $\varepsilon^{\texttt{max}}$, principal's belief $Q$
\State \textbf{Initialize:} $\mathcal{I} \gets \emptyset, \quad \mathcal{U} \gets\emptyset$
\State $(\pi_L, V^0_L, A_L) \gets \texttt{SolveMDP}(\mathcal{M}^{0}), \quad (\pi_R, V^0_R, A_R) \gets \texttt{SolveMDP} (\mathcal{M}^{\varepsilon^{\texttt{max}}})$ \label{line:initial-solve}
\Comment{Compute optimal policies for $\varepsilon=0$ and $\varepsilon=\varepsilon^{\texttt{max}}$ using the decomposition in Proposition~\ref{prop:linear-value}}
\State $\mathcal{U} \gets \mathcal{U} \cup \{(0, \Bar{U}^{\texttt{S}}(0; \pi_L))\}$\label{line:initial-U}
\State $\Bar{V}^0_{L} \gets \EE_{(\alpha_0, \beta_0)\sim Q} [V^0_{L}(\alpha_0,\beta_0,0,0)]$, $\Bar{A}_{L} \gets \EE_{(\alpha_0, \beta_0)\sim Q} [A_{L}(\alpha_0,\beta_0,0,0)]$
\State $\Bar{V}^0_{R} \gets \EE_{(\alpha_0, \beta_0)\sim Q} [V^0_{R}(\alpha_0,\beta_0,0,0)]$, $\Bar{A}_{R} \gets \EE_{(\alpha_0, \beta_0)\sim Q} [A_{R}(\alpha_0,\beta_0,0,0)]$ \Comment{Compute the average agent value using the principal's belief.}
\State Push $(\varepsilon_L, \pi_L, \Bar{V}^0_L, \Bar{A}_L, \varepsilon_R, \pi_R, \Bar{V}^0_R, \Bar{A}_R)$ into $\mathcal{I}$ \label{line:initial-interval}

\While{$\mathcal{I}$ is not empty}
    \State Pop $(\varepsilon_L, \pi_L, \Bar{V}^0_L, \Bar{A}_L, \varepsilon_R, \pi_R, \Bar{V}^0_R, \Bar{A}_R)$ from $\mathcal{I}$    
    \If{$A_L \neq A_R$}
        \State $\varepsilon_{int} \gets (\Bar{V}^0_L - \Bar{V}^0_R) / (\Bar{A}_R - \Bar{A}_L)$ \label{line:subsidy-intersection}
        \Comment{Compute the candidate subsidy to evaluate}
        \State $(\pi_{int}, V^0_{int}, A_{int}) \gets \texttt{SolveMDP}(\mathcal{M}^{\varepsilon_{int}})$ \Comment{Compute the optimal policy for the candidate subsidy} \label{line:mdp-int}
        \State $\Bar{V}^0_{int} \gets \EE_{(\alpha_0, \beta_0)\sim Q} [V^0_{int}(\alpha_0,\beta_0,0,0)]$, $\Bar{A}_{int} \gets \EE_{(\alpha_0, \beta_0)\sim Q} [A_{int}(\alpha_0,\beta_0,0,0)]$ \Comment{Compute the average agent value using the principal's belief.}

        \If{$\Bar{V}^0_{int} + \varepsilon_{int}\cdot \Bar{A}_{int} \leq \Bar{V}^0_L +\cdot \varepsilon_{int} \Bar{A}_L$} 
            \Comment{Verify if the optimal policy improves over $\pi_R$ and $\pi_L$}
            \State $\mathcal{U} \gets \mathcal{U} \cup \{(\varepsilon_{int}, \Bar{U}^{\texttt{S}}(\varepsilon_{int}; \pi_R))\}$ \label{line:vertex-found}
            \Comment{Found a vertex; compute social utility using Eq.~\ref{eq:adaptive-utility-regulator} and principal's belief $Q$}
        \Else 
            \Comment{Split the subsidy interval}
            \State Push $\{ (\varepsilon_L, \pi_L, \Bar{V}^0_L, \Bar{A}_L, \varepsilon_{int}, \pi_{int}, \Bar{V}^0_{int}, \Bar{A}_{int}), (\varepsilon_{int}, \pi_{int}, \Bar{V}^0_{int}, \Bar{A}_{int}, \varepsilon_R, \pi_R, \Bar{V}^0_R, \Bar{A}_R) \}$ into $\mathcal{I}$ \label{line:new-intervals}
        \EndIf
    \EndIf
\EndWhile
\State \Return $\varepsilon^*$ where $(\varepsilon^*, u^*) = \argmax_{(\varepsilon, u) \in \mathcal{U}} u$ 
\Comment{Return the optimal subsidy}

\end{algorithmic}
\end{algorithm}

Leveraging the above results, we derive an efficient divide-and-conquer 
procedure to find the optimal subsidy $\varepsilon^*$. The procedure (i) explicitly constructs the partition $\Pcal$ and (ii) applies Proposition~\ref{prop:regulator-utility-linear} to determine the optimal anticipated social utility on each interval $[\varepsilon_i, \varepsilon_{i+1})$ by evaluating it at the left endpoint $\varepsilon_i$. 
Algorithm~\ref{alg:epsilon} summarizes the overall procedure, and the following proposition establishes its correctness.

\begin{proposition}\label{prop:algorithm}
    Algorithm~\ref{alg:epsilon} is guaranteed to find an optimal subsidy $\varepsilon^{*}$ in a finite number of iterations.
\end{proposition}

Algorithm~\ref{alg:epsilon} maintains a stack $\mathcal{I}$ of intervals, alongside their corresponding optimal policies at the endpoints and the linear value function decompositions provided by Proposition~\ref{prop:linear-value}. 
We abstract the computation of these policies and value functions into a procedure, \texttt{SolveMDP}, which may be implemented using a value-iteration algorithm (e.g., Algorithm~\ref{alg:value_iteration_decomposition}).
The stack $\Ical$ is initialized to the full interval $[0,\varepsilon^{\texttt{max}}]$ (lines~\ref{line:initial-solve} to~\ref{line:initial-interval}). At each iteration, the algorithm calculates the intersection point $\varepsilon_{int}$ between the endpoint value functions, $\Bar{V}_L^0 + \varepsilon \cdot \Bar{A}_L$ and $\Bar{V}_R^0 + \varepsilon \cdot \Bar{A}_R$ (line~\ref{line:subsidy-intersection}), where the bar denotes evaluation at $S_0$ and averaging over the principal’s belief $Q$.
Then, it computes the optimal value function at $\varepsilon_{int}$ (line~\ref{line:mdp-int}). 
Due to the convexity of the optimal value function, if the optimal value function at $\varepsilon_{int}$ coincides with the value $\Bar{V}_L^0 + \varepsilon_{int}\cdot \Bar{A}_L = \Bar{V}_R^0 + \varepsilon_{int}\cdot \Bar{A}_R$, the policy $\pi_L$ is optimal in $[\varepsilon_L, \varepsilon_{int})$, and the policy $\pi_R$ is optimal in $[\varepsilon_{int}, \pi_R]$. In this case, $\varepsilon_{int}$ is a vertex of $\Pcal$, and the algorithm stores the social utility at $\varepsilon_{int}$ (line~\ref{line:vertex-found}). Conversely, if the optimal policy at $\varepsilon_{int}$ strictly improves over the policy $\pi_L$, then the algorithm has found a new interval of the partition, and the stack $\Ical$ is updated (line~\ref{line:new-intervals}). The algorithm then iterates the same steps over all intervals in the stack $\Ical$.\footnote{We report runtime measurements of Algorithm~\ref{alg:epsilon} in Appendix~\ref{app:experimental-details}.}

%% file: 060experiments.tex
Antimicrobial resistance is a major global threat, projected to cause 10 million deaths annually by 2050~\citep{GBD_2021_Antimicrobial_Resistance_Collaborators2024, Shallcross2015-tg}. Yet, FDA antibiotic approvals have dropped from $13\%$ of all drugs in 1980 to $4\%$ in the 2000s~\citep{Outterson2013-qz}, largely for economic reasons: treatments are short, prices must stay low to ensure availability, use is restricted to limit resistance, and competition from existing or generic drugs is intense~\citep{Shlaes2019-fe, Gargate2025}. Consequently, many major pharmaceutical companies have exited or reduced antibiotic pipelines~\citep{Plackett2020-qi}, while small biotech firms often struggle financially~\citep{Courtemanche2021-up, Piddock2024-dn, Wells2024-fs}. To counter this, public and private efforts have focused on incentivizing and subsidizing development~\citep{Outterson2016-et, Anderson2023-gl, HR7352_119th_Congress_2026}. 
In this section, we conduct a series of experiments to demonstrate the effectiveness of our approval protocol in optimally subsidizing antibiotic development.

\begin{figure*}[t]
    \vspace{0cm}
    \centering
    \subfloat[Social utility vs. subsidy $(\rho^{\texttt{S}}= \$2000\,\text{M})$ ]{
    \includegraphics[width=0.48\linewidth]{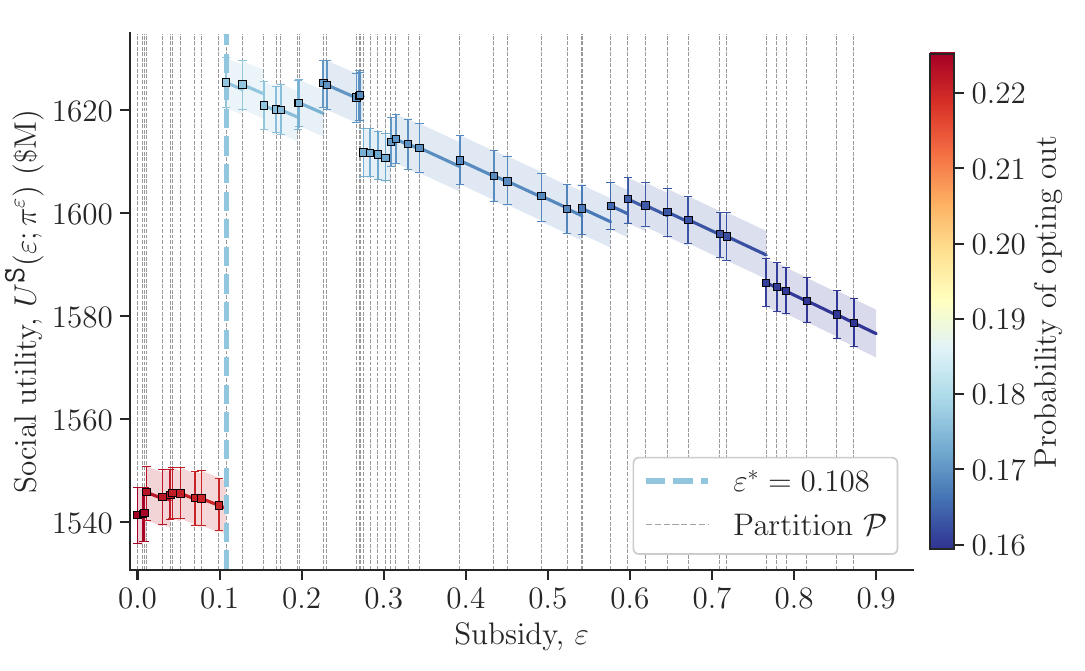}
    }
    \hspace{0mm}
    \subfloat[Optimal subsidy and social utility gain vs. $\rho^{\texttt{S}} / \rho^{\texttt{A}}$]{
    \includegraphics[width=0.48\linewidth]{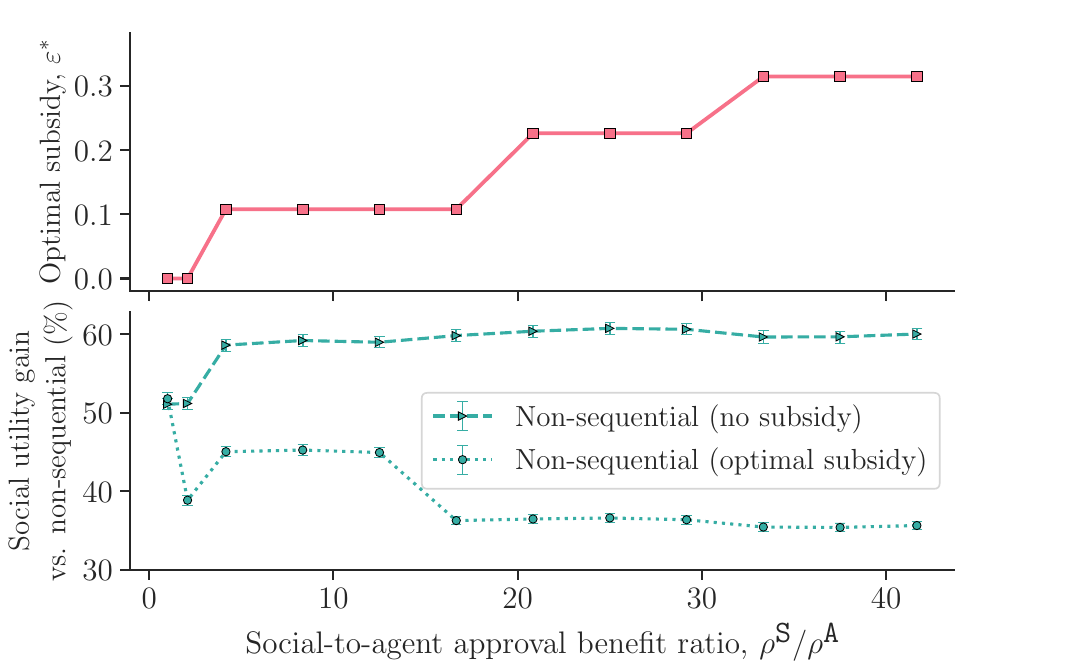}
    }
    \caption{\textbf{Subsidizing antibiotic development.}
    The figure shows the results of the approval process for an antibiotic with true (unknown) efficacy $\theta^* = 0.65$. Panel~(a) shows the result of running Algorithm~\ref{alg:epsilon} to compute the optimal subsidy for the principal $\varepsilon^* = 0.108$ when the social benefit of approval is $\rho^{\texttt{S}}= \$2000\,\text{M}$. The dashed vertical lines correspond to the intervals of the partition $\mathcal{P}$ where the agent's optimal policy is constant (Proposition~\ref{prop:utility-piecewise-convex}), and the colors indicate the probability that the agent opts out by selecting $n=0$ during the approval process (before the drug is approved).
    Panel~(b) shows, as a function of the social-to-agent approval benefit ratio, the optimal subsidy, together with the percentage increase in social utility of the sequential approval protocol relative to a non-sequential approval protocol in which the agent is restricted to a single trial with $n^{\texttt{max}}=800$.
    The error bars represent $95\%$ bootstrapped confidence intervals.
    }
    \label{fig:main}
\end{figure*}

\xhdr{Experimental setup}
The principal subsidizes the development of an antibiotic and conducts a hypothesis test with $\kappa = 0.05$ (\ie, a false positive rate of at most $0.05$) to determine whether the antibiotic’s (unknown) efficacy $\theta^* = 0.65$ exceeds the benchmark $\theta^{\texttt{b}} = 0.5$. 
The agent can conduct up to four trials ($T=3$), each with a maximum sample size of $n^{\texttt{max}} = 200$ patients. 
Although data on the economic cost and sales of antibiotic development are mostly private, recent reports estimate that the average present value of sales is approximately $\$240 \, \text{M}$~\citep{Rahman2021-zk}, the per-patient Phase III cost is $\$66\,\text{k}$~\citep{Stergiopoulos2018-fp}, and the average fixed cost per trial is $\$48.9\,\text{M}$~\citep{Moore2018-ls}. 
Therefore, we set $\rho^{\texttt{A}} = \$240\,\text{M}$, and $c(n) = \$48.9\,\text{M} + \$0.066\,\text{M} \cdot n$ for any $n\in \{1,\dots,n^{\texttt{max}}\}$. 
Further, we assume the agent has a non-informative (uniform) prior with $\alpha_0=1$, $\beta_0=1$, known to the principal, and vary $\rho^{\texttt{S}}$. 
Refer to Appendix~\ref{app:experimental-details} for additional details regarding our experimental setup, and to Appendix~\ref{app:additional-experimental-results} for results under alternative parameter choices and extensive sensitivity analyses.

\xhdr{Results}
%
%
For an antibiotic with a social benefit upon approval of $\rho^{\texttt{S}} = \$2000\,\text{M}$ (a ratio $\rho^{\texttt{S}} / \rho^{\texttt{A}} \approx 8.3$), Panel~(a) of Figure~\ref{fig:main} shows (i) the social utility $U^{\texttt{S}}(\varepsilon;\pi^\varepsilon)$, (ii) the partition $\mathcal{P}$, which consists of 49 intervals, and (iii) the agent's opt-out probability, over the entire range of subsidy levels.
%
We find that, under the optimal subsidy $\varepsilon^{*} = 0.108$, the social utility increases by $\sim$$5.5 \%$ and the agent's opt-out probability decreases by $\sim$$22\%$ compared to a scenario with no subsidies.
Interestingly, we also find that the optimal subsidy $\varepsilon^*$---which maximizes the anticipated  social utility $\bar{U}^{\texttt{S}}(\varepsilon;\pi^\varepsilon)$---also maximizes the true social utility $U^{\texttt{S}}(\varepsilon;\pi^\varepsilon)$. 
Further, Panel~(b) of Figure~\ref{fig:main} shows the optimal subsidy $\varepsilon^{*}$ and the social utility gain compared to two non-sequential baselines for different values of the social-to-agent approval benefit ratio $\rho^{\texttt{S}} / \rho^{\texttt{A}}$.
We find that the optimal subsidy increases with the social-to-agent approval benefit ratio, reflecting that higher societal utility strengthens the principal’s incentive to subsidize experimentation.
We also find that, compared to two non-sequential baselines in which the agent is restricted to conducting at most a single clinical trial (with a larger maximum sample size of $n^{\texttt{max}} = 800$), our approval protocol yields substantial gains in social utility. 
Specifically, relative to a non-sequential protocol without subsidies, our protocol increases social utility by approximately $50\%$–$60\%$, depending on the ratio $\rho^{\texttt{S}} / \rho^{\texttt{A}}$ and, relative to a non-sequential protocol with optimal subsidies, our protocol still achieves gains exceeding $35\%$.

%% file: 070discussion.tex
In this section, we highlight several limitations of our work and discuss avenues for future research.

\xhdr{Methodology}
In our work, we have considered a Beta–Binomial model, which is particularly natural in the context of RCTs. 
However, in other application domains, it may be desirable to consider more general models where the agent's beliefs, the experimental outcomes, and the e-values exhibit greater complexity---for instance, through dependence across trials or through the inclusion of safety characteristics. 
In Appendix~\ref{app:general-model}, we outline how to extend our framework to these more general settings; however, this extension introduces significant computational challenges.
Furthermore, although our results do not require the principal to know the agent’s initial belief (\ie, its private information $B_0$), observing the agent's experimental actions over time may provide information about its prior, suggesting that the principal could, in principle, dynamically update its belief about the agent's prior. 
It would be interesting to incorporate such learning and elicitation into a sequential approval protocol in future work~\citep{shi2025instance}. 
In addition, it would also be valuable to extend our protocol to allow for uncertainty in the agent's benefit $\rho^{\texttt{A}}$, which in practice may vary across product developers and contexts.
Finally, real-world regulatory settings may involve additional non-economic factors that are not explicitly captured in our model. For instance, clinical trial duration, patient follow-up requirements, and other operational constraints can affect both the feasibility and optimality of sequential experimentation protocols, and accounting for these factors remains an important direction for future work.

\xhdr{Implementation and Evaluation}
We have conducted a case study applying our approval protocol to antibiotic development, a well-known setting characterized by underinvestment and market failure, which is particularly suited to studying the effects of subsidies (see Appendix~\ref{app:additional-experimental-results} for further experimental results using different parameter values). However, it would be interesting to extend this analysis to other settings---such as orphan drugs or rare disease treatments---where trial costs and sample sizes may differ substantially from those in standard antibiotic development and may therefore lead to different optimal subsidies.
We also empirically evaluated the computational cost of computing the agent's optimal decision policy and the optimal subsidy. 
In our experiments, both procedures were efficient and typically completed within a few minutes (see Appendix~\ref{app:experimental-details} for more details). However, as shown theoretically in Section~\ref{sec:agent}, the complexity of computing the optimal agent policy scales polynomially with the number of actions and trials. Consequently, extending our methodology to more general settings---such as those described in Appendix~\ref{app:general-model}---might require developing approximate algorithms, as computing (Stackelberg) equilibria in general Bayesian games is known to be computationally intractable in many cases~\citep{Conitzer}.

\xhdr{Broader Impact}
Regulatory agencies and product developers are progressively exploring more flexible and data-driven approval methodologies, including Bayesian approaches, to better balance safety and innovation~\citep{fda2026use}.
Our work contributes to this direction by providing a principled framework for designing subsidy mechanisms that improve social utility in approval processes, which may inform future policy discussions and be of interest to both public and private regulators and decision makers.

%% file: 080conclusions.tex
How can approval protocols be designed to incentivize experimentation? In this work, 
we have addressed this question by introducing a sequential approval protocol that allows the agent to continuously refine its knowledge about the product and the principal to subsidize a fraction of the agent's experimental costs---all while maintaining anytime-valid guarantees on the false positive rate.
Along the way, we have shown that the agent can efficiently compute the optimal experimental policy, and the principal can find the subsidy that maximizes social utility, even when anticipating that the agent selects its policy strategically.
Finally, using real-world data on antibiotic development, we have demonstrated that our sequential, subsidized protocol can substantially improve social utility, yielding gains of up to $60\%$ relative to non-sequential designs without subsidies.
More broadly, we hope our work provides insights for designing approval protocols that better align agents' incentives with social objectives.

%% file: 090appendix.tex

\section{Summary of Notation}
In Table~\ref{tab:notation} we summarize the key symbols used in the main body of the paper.

\begin{table}[H]
\centering
\caption{\textbf{Summary of notation.}}
\label{tab:notation}
\begin{tabular}{cl}
\toprule
\textbf{Symbol} & \textbf{Description} \\
\midrule
    $\kappa$ & Principal’s false positive rate bound \\
    $\theta^*$ & True (unknown) product efficacy \\
    $\theta^{\texttt{b}}$ & Baseline efficacy \\
    $H_0, H_1$ & Null and alternative hypotheses \\
    $t$ & Time index \\
    $l$ & Initial time index (if not $0$) \\
    $T$ & Maximum number of trials (time horizon of the belief MDP) \\
    $\varepsilon$ & Principal’s subsidy level \\
    $\varepsilon^{\texttt{max}}$ & Maximum allowable subsidy \\
    $\varepsilon^*$ & Principal’s optimal subsidy \\
    $B_t$ & Agent’s belief at time $t$ \\
    $\alpha_t, \beta_t$ & Parameters of the agent’s belief at time $t$ \\
    $Q$ & Principal’s belief over the agent’s prior \\
    $n_t$ & Agent’s action (sample size) at time $t$ \\
    $n^{\texttt{max}}$ & Maximum sample size \\
    $\pi$ & Agent’s policy \\
    $\pi^\varepsilon$ & Optimal policy under subsidy $\varepsilon$ \\
    $V^\varepsilon$ & Optimal value function under subsidy $\varepsilon$ \\
    $V^\varepsilon_{\pi}$ & Value function of policy $\pi$ under subsidy $\varepsilon$ \\
    $A_{\pi}$ & Expected total cost conditional on approval under policy $\pi$ \\
    $c$ & Cost function \\
    $X_t$ & Experimental outcome at time $t$ \\
    $E_t$ & e-value at time $t$ \\
    $M_t$ & Principal’s test process at time $t$ \\
    $f$ & Mapping from beliefs to the test process \\
    $C_t$ & Cumulative experimental cost at time $t$ \\
    $\Mcal^\varepsilon$ & Belief MDP under subsidy $\varepsilon$ \\
    $r^\varepsilon$ & Reward function in the belief MDP under subsidy $\varepsilon$ \\
    $S_t$ & State of $\Mcal^\varepsilon$ at time $t$ \\
    $\tau$ & Stopping time (last non-absorbing state of $\Mcal^\varepsilon$) \\
    $\Pcal$ & Partition associated with fixed policies \\
    $\rho^{\texttt{A}}$ & Agent’s approval benefit \\
    $\rho^{\texttt{S}}$ & Social approval benefit \\
    $U^{\texttt{A}}$ & Agent’s utility under $\theta^*$ \\
    $U^{\texttt{S}}$ & Social utility under $\theta^*$ \\
    $\bar{U}^{\texttt{A}}$ & Agent’s anticipated utility computed via the belief MDP \\
    $\bar{U}^{\texttt{S}}$ & Principal's anticipated social utility computed via the belief MDP \\
\bottomrule
\end{tabular}
\end{table}

\clearpage
\newpage

\section{Background on Sequential Hypothesis Testing Using e-values}\label{app:background-seq}

In this section, we provide a brief overview of sequential hypothesis testing using e-values. For a detailed exposition, we refer the reader to Ramdas et al.~\citep{ramdas2025hypothesistestingevalues}.

A central object in classical hypothesis testing is the \emph{p-value}, used to assess a null hypothesis $H_0$ against an alternative $H_1$. Despite their widespread use across scientific disciplines, p-values suffer from important limitations when used without appropriate precautions. A canonical example is \emph{p-hacking} (also called \emph{sampling to a foregone conclusion}): if one repeatedly collects data, computes a p-value, and checks whether it falls below a fixed threshold (\eg, $0.05$), rejection of $H_0$ is eventually guaranteed even when $H_0$ is true~\citep{B2005-xl,Simmons2011-ny}.

Alongside classical methods based on p-values, a rich line of work has developed methods based on a game-theoretic formulation of statistics that is better suited to sequential hypothesis testing~\citep{ramdas2023gametheoreticstatisticssafeanytimevalid} and mitigates issues such as p-hacking. The key objects in this framework are \emph{e-values}, which are nonnegative random variables $E$ that satisfy:
\begin{equation}\label{app:e-values-def}
    \EE_{H_0}[E] \leq 1,
\end{equation}
where the expectation is taken under any distribution in $H_0$. Typically, as with p-values, e-values are computed from some observed data $X$ that is intended to provide evidence to reject $H_0$; in the sequential approval protocol considered in Section~\ref{sec:model}, we make this dependence explicit by writing $E(n, X)$, where $n$ is the sample size selected by the agent and $X$ the number of successes in the control trial. The intuition behind e-values is straightforward: since $E$ has expectation at most $1$ under $H_0$, it can take large values only with small probability. Hence, observing a large e-value can be interpreted as evidence against $H_0$---and the larger the e-value, the stronger the evidence. Consequently, one can construct a statistical test that rejects $H_0$ whenever the e-value takes high values, \ie, whenever $E \geq 1/\kappa$ for a certain threshold $\kappa\in (0,1)$. More precisely, by Markov's inequality, this threshold automatically controls the false positive rate:
\begin{equation}
    P_{H_0}(E \geq 1/\kappa) \leq \kappa \cdot \EE_{H_0}[E] \leq \kappa.
\end{equation}

\begin{proposition}
    If $E$ is an e-value for $H_0$, \ie, $\EE_{H_0}[E] \leq 1$, then the non-sequential
    hypothesis test
    \begin{equation*}
        \phi = \mathds{1}\{ E \geq 1/\kappa\}
    \end{equation*}
    is a level-$\kappa$ test for $H_0$ for any $\kappa \in (0,1)$: its false positive rate is at most $\kappa$.
\end{proposition}
Informally, e-values are as general as p-values in the sense that they exist under essentially the same technical conditions and can be transformed into one another. Their main advantage, however, arises in sequential settings---\ie, when data arrive over time~\citep{shin_ramdas_rinaldo_2024,shekhar2024Two,xu2024Online, velasco2026auditing,gauthier2026bettingequilibriummonitoringstrategic}---as in clinical trials, which is the primary reason our formalism in Section~\ref{sec:model} and~\ref{sec:MDP} is built on e-values.
One way to understand this advantage is to note that the defining condition of e-values (Eq.~\ref{app:e-values-def}) is preserved under a wide range of operations. For example, convex combinations of e-values and multiplications of (independent) e-values remain valid e-values. In contrast, analogous operations do not generally preserve the validity of p-values. 

To formalize the use of e-values in a sequential setting, let $\Fcal = (\Fcal_t)_{t=0}^\infty$ be a filtration on a given sample space, where each $\Fcal_t$ is a $\sigma$-algebra representing the information available at time $t$. For instance, $\Fcal$ may be the filtration generated by the observations $X_0, \dots, X_t$.
Then, given a sequential data stream $X_0, X_1, \dots$, the goal of sequential hypothesis testing is to maintain and update a running measure of evidence against $H_0$ as new observations arrive. In the e-value framework, this can be achieved by constructing a sequence $E_0, E_1, \dots$ adapted to $\Fcal$---so that each $E_t$ depends only on data observed up to time $t$---and satisfying
\begin{equation}\label{app:e-values-sequential}
    \EE_{H_0}[E_t \mid \Fcal_{t-1}] \leq 1.
\end{equation}
The above condition is the sequential analogue of the defining property in Eq.~\ref{app:e-values-def}. For instance, if the observations $X_t$ are independent, one may construct each $E_t$ from $X_t$ alone, in which case the e-values $E_0,\dots,E_t$ are mutually independent. In Section~\ref{sec:agent}, we adopt this construction, but we emphasize that more complex e-values can be defined if the experimental protocol requires it. For example, in Eq.~\ref{eq:e-values-conditional} we consider a data-dependent construction in which new evidence is collected only if a prior experiment satisfies certain conditions.

By interpreting the quantity $E_t$ as the new evidence against $H_0$ obtained at time $t$, a canonical way to construct a sequential test is to define a stochastic test process $M = (M_t)_{t\geq0}$ as
\begin{equation}
M_t =
\begin{dcases}
1 & t = 0 \\
E_t\cdot M_{t-1} & t \geq 1,
\end{dcases}
\end{equation}
which simply corresponds to aggregating the previous e-values multiplicatively.\footnote{In Eq.~\ref{eq:mixture-process} of Appendix~\ref{app:general-model} we discuss a different method to construct a test process $M$ without multiplying e-values.} By contrast, combining p-values in a sequential or dependent setting typically requires specialized corrections to maintain validity~\cite{GrunwaldSafe}.
Whenever the e-values satisfy Eq.~\ref{app:e-values-sequential}, the above process $M$ forms a (non-negative) supermartingale under $H_0$, \ie,
\begin{equation}
    \EE_{H_0}[M_t \mid \Fcal_{t-1}] \leq M_{t-1}.
\end{equation}

Intuitively, this property means that if $H_0$ is true, the value of $M_t$ does not, in expectation, increase over time. Conversely, sustained growth of the process $M_t$ provides evidence against $H_0$. Thus, a natural sequential test to reject $H_0$ is to monitor whether $M_t \geq 1/\kappa$ for a chosen $\kappa \in (0,1)$, and to reject $H_0$ as soon as this condition is met. Ville's inequality~\citep{ville1939étude}, a sequential extension of Markov's inequality, guarantees that this procedure controls the false positive rate uniformly over time:
\begin{theorem}[Ville's inequality]\label{thm:ville}
    If the process $M$ is a non-negative supermartingale, then,
    \begin{equation}
        P_{H_0}\!\left( \exists t \in \mathbb{N} : M_t \geq 1/\kappa \right) = P_{H_0}\!\left( \sup_{t\geq 0} M_t \geq 1/\kappa \right) \leq \kappa.
    \end{equation}
\end{theorem}
Moreover, validity is preserved if the test is stopped at any stopping time
$\tau$ adapted to $\Fcal$ (but possibly depending on the observed data):
\begin{equation}
    P_{H_0}(M_\tau \geq 1/\kappa) \leq \kappa,
\end{equation}
which is a property known as \emph{any-time validity} and provides a principled solution to the problem of p-hacking.

Although the above results hold for any choice of e-values satisfying Eq.~\ref{app:e-values-sequential}, a natural practical question is which specific form to adopt. We discuss this, and how different choices affect the approval process, in Appendix~\ref{app:general-model}.

\clearpage
\newpage

\section{Extension to Arbitrary Belief Functions and e-values}\label{app:general-model}

The approval protocol introduced in Section~\ref{sec:model} can be formulated in a more general setting (at the expense of losing tractability), as we now outline.

\xhdr{Generalization of the approval process}
In a general setting, the agent may begin the approval process with an arbitrary initial belief $B_0 \in \Delta(\Theta)$, where $\Theta$ denotes a general parameter space for the unknown parameter $\theta^*$ characterizing the product. At each step of the approval process, the agent selects an action $n_t \in \Acal^{\texttt{gen}} \cup \{0\}$, where $\Acal^{\texttt{gen}}$ is a (potentially infinite) set describing the design of the next randomized controlled trial (\eg, sample size, participant characteristics, etc.), and $n_t = 0$ again denotes the option to opt out. If $n_t \neq 0$, an outcome $X_t \sim P(X_t \mid n_t, \theta^*)$ is observed and the agent incurs a cost $c(n_t)$, where $P(\cdot \mid n_t, \theta^*)$ is a likelihood function characterizing the data-collection process specific to each experimental setting and $c \,\colon\, \Acal^{\texttt{gen}} \to \mathbb{R}_{+}$ is a cost function. The agent then updates its belief via the Bayesian posterior:
\begin{equation*}
    B_{t+1} (\theta) = \frac{B_t(\theta) \cdot P(X_t \mid n_t, \theta)}{\int_\Theta B_t(\theta) \cdot P(X_t \mid n_t, \theta)\, d\theta}.
\end{equation*}

To decide on approval, the principal partitions the parameter space as $\Theta = \Theta_0 \sqcup \Theta_1$ and defines the following null and alternative hypotheses:
\begin{equation}
    \begin{dcases}
        H_0=\left\{ {\theta^* \in \Theta_0} \right\} & (\text{null})\\
        H_1=\left\{ \theta^* \in \Theta_1 \right\} & (\text{alternative}).
    \end{dcases}
\end{equation}
This formulation allows for richer principal objectives. For instance, the principal may be concerned not only with a drug's efficacy but also with its safety, in which case $H_0$ may correspond to treatments that are either insufficiently effective or unsafe. To conduct the above hypothesis test sequentially, at each time step, the principal may compute any test process value $M_t$, subject only to two conditions: (i) the test process $M$ must be predictable with respect to the filtration $\Fcal = (\Fcal_t)_{t=0}^\infty$ generated by the data and agent actions $X_0, n_0, X_1, n_1, \dots$, and (ii) the test process $M$ must be a supermartingale under $H_0$. The principal can then reject $H_0$ whenever $M_t \geq 1/\kappa$ at any time step, as described in Appendix~\ref{app:background-seq}. A particular construction of the test process that generalizes Eq.~\ref{eq:test-process} proceeds multiplicatively as
\begin{equation}
M_{t}=
\begin{dcases}
    1 & t=0,\\
    E_{t-1}(X_{t-1}, n_{t-1}) \cdot M_{t-1} & t\geq1,
\end{dcases}
\end{equation}
where each $E_{t-1}$ is an e-value adapted to $\Fcal$ that may depend on all previously observed experimental outcomes and actions, \ie, $X_0, n_0, X_1, n_1, \dots, X_{t-1}, n_{t-1}$, and satisfies:
\begin{equation*}
    \EE_{H_0}[E_t(n_t,X_t) \mid \Fcal_{t-1}] \leq 1.
\end{equation*}
This allows modeling adaptive experimental designs in which future trials depend on past outcomes. Such settings arise naturally in multi-stage clinical trials~\cite{fda2019effectiveness}, where progression to a subsequent phase may be contingent on earlier success. This can be represented by defining:
\begin{equation}\label{eq:e-values-conditional}
    E_t(n_t, X_t) =
    \begin{dcases}
        E(n_t, X_t) & \text{if } g(n_0,X_0,\dots,n_{t-1},X_{t-1}) = 1, \\
        0 & \text{if } g(n_0,X_0,\dots,n_{t-1},X_{t-1}) = 0,
    \end{dcases}
\end{equation}
where $E(\cdot,\cdot)$ is a fixed e-value and $g$ is a decision rule determining whether the agent is permitted to conduct the next experiment. If $g(n_0,X_0,\dots,n_{t-1},X_{t-1}) = 0$ at any time $t-1$, the test process stops accumulating evidence and approval becomes impossible after time $t$.

Given a test process $M$, the utilities of both the agent and the principal can be defined analogously to Eq.~\ref{eq:true-utilities}, and a belief Markov decision process $\Mcal^\varepsilon$ can be formulated as in Section~\ref{sec:MDP}, with the difference that the state space now includes any possible belief and any possible value of the test process, \ie,\footnote{If the e-values depend on all past observations, as in Eq.~\ref{eq:e-values-conditional}, the state space may need to be augmented to explicitly track the history of experimental outcomes.}
\begin{equation*}
    \Scal^{\texttt{gen}} = \Delta(\Theta) \times \underbrace{\mathbb{R}_{+}}_{\text{total cost } C} \times \underbrace{\mathbb{R}_{+}}_{\text{test process } M} \cup \quad S^{\texttt{out}},
\end{equation*}
with a state $S_t$ transitioning to $S_{t+1}$ after taking action $n_t$ and observing outcome $X_t$ according to:
\begin{equation}
S_{t+1} =
    \begin{dcases}
    ( B_{t+1},\; C_t + c(n_t),\; M_{t+1}(n_t,X_t))
    & \text{if } n_t \neq 0 \text{ and } 0<f(S_t) < 1/\kappa,\\
    S
    & \text{if } f(S_t) \geq 1/\kappa,\\
    S^{\texttt{out}}
    & \text{if } n_t = 0 \text{ or } S_t = S^{\texttt{out}},
\end{dcases}
\end{equation}
where for $S_t=(B_t, C_t, M_t) \neq S^{\texttt{out}}$, $B_{t+1}(\theta) \propto B_t(\theta) \cdot P(X_t \mid n_t, \theta)$. In the above, $f$ is the function that maps a state to the test process value, namely $f(S_t) = M_t$ for $S_t=(B_t, C_t, M_t) \neq S^{\texttt{out}}$ and $f(S^{\texttt{out}}) = 0$. While this formulation is fully general, it also significantly complicates the subsequent analysis of the optimal agent policy.

\xhdr{Alternative e-values and statistical power}
For a given choice of test process $M$, it is important to note that the statistical power of the test---verifying whether $M_t \geq 1/\kappa$ at any time step---need not be $1$. That is, with nonzero probability it may occur that, even for $\theta^* \geq \theta^{\texttt{b}}$, the approval process is unsuccessful.

To quantify the statistical power of sequential tests based on e-values (or supermartingales), there exists a canonical notion called \emph{e-power}. For a given e-value $E$ for the null $H_0$, its e-power against an alternative $L \in H_1$ is defined as $\EE_{L}[\log E]$. Under this definition, likelihood ratios are the optimal e-values for simple nulls and alternatives, \ie, when $H_0$ and $H_1$ each correspond to a single probability distribution. In contrast, when $H_0$ and $H_1$ are composite, as in the approval process of Section~\ref{sec:agent}, the choice of e-value becomes more intricate~\citep{waudbysmith2025universallogoptimalitygeneralclasses}. For concreteness, in Section~\ref{sec:agent} we therefore adopt a general strategy that constructs e-values by exponentiating random variables~\citep{PenaExp, waudbysmith2022estimatingmeansboundedrandom}. The rationale is as follows: if we wish to test whether an arbitrary random variable $X$ has mean $\EE[X] < \theta^{\texttt{b}}$, we consider the quantity $X - \theta^{\texttt{b}}$, which we expect to be large when $\EE[X] > \theta^{\texttt{b}}$ and small otherwise. To form an e-value, we consider the positive quantity $\exp(X - \theta^{\texttt{b}})$; since this can have expectation exceeding $1$ even when $\EE[X] < \theta^{\texttt{b}}$, we obtain a valid e-value satisfying Eq.~\ref{app:e-values-def} by shifting the argument by an appropriately chosen constant $\lambda$, \ie, $\exp(X - \theta^{\texttt{b}} + \lambda)$, as given by Proposition~\ref{prop:binomial-e-value}. This e-value is particularly simple and admits a closed-form expression that simplifies the exposition in Section~\ref{sec:MDP}; however, in the context of the approval protocol in Section~\ref{sec:model}, other choices of e-values may achieve better statistical power for certain values of the efficacy parameter $\theta^*$. Indeed, for $E$ as defined in Proposition~\ref{prop:binomial-e-value} and using the notation therein, the e-power is:
\begin{equation*}
    \EE_{\theta^*}[\log E_t] = n_t \cdot \theta^* - n_t \cdot \log(1+\theta^{\texttt{b}}(e-1)).
\end{equation*}
Thus, $E_t$ has positive power $\EE_{\theta^*}[\log E_t] \geq 0$ if and only if $\theta^* \geq \log(1+\theta^{\texttt{b}}(e-1)) > \theta^{\texttt{b}}$. This means that if $\theta^{\texttt{b}} < \theta^* < \log(1+\theta^{\texttt{b}}(e-1))$, it may be that $E_t \not\to \infty$ even as $n \to \infty$. In other words, the process $M_t$ might never exceed the threshold $1/\kappa$ if $\theta^*$ is sufficiently close to---yet strictly above---$\theta^{\texttt{b}}$, even for RCTs of arbitrary size. The standard remedy (see Chapter 3 of~\citep{ramdas2025hypothesistestingevalues}) is to form a \emph{mixture} of e-values or supermartingales that generalizes Proposition~\ref{prop:binomial-e-value}. We now outline how our approval protocol can be extended to use such mixtures, and refer the reader to Appendix~\ref{app:mixture-results} for additional experimental results.

Our starting point is the simple observation that the e-value $E(X_t,n_t)$ in Proposition~\ref{prop:binomial-e-value} can be equivalently written as a likelihood ratio between a Binomial distribution evaluated at $\theta^{\texttt{b}}$ and one evaluated at a particular alternative:
\begin{equation*}
    E(X_t,n_t) = \left( \frac{\tilde{\theta}}{\theta^{\texttt{b}}}\right)^{X_t} \cdot \left( \frac{1-\tilde{\theta}}{1-\theta^{\texttt{b}}}\right)^{n_t - X_t},
\end{equation*}
where $\tilde{\theta} \geq \theta^{\texttt{b}}$ is the unique efficacy value satisfying:
\begin{equation*}
    e = \frac{\tilde{\theta}\cdot(1- \theta^{\texttt{b}})}{(1-\tilde{\theta})\cdot \theta^{\texttt{b}}}.
\end{equation*}

To generalize this construction, one can mix over other values of the alternative parameter. This is achieved by choosing any (smooth) distribution $P^{\texttt{mix}}$ over $[\theta^{\texttt{b}}, 1]$ and defining the mixture test process as:\footnote{An argument analogous to that of Proposition~\ref{prop:binomial-e-value} shows that $M^{\texttt{mix}}_{t}$ is a supermartingale under $H_0$.}
\begin{equation}\label{eq:mixture-process}
    M^{\texttt{mix}}_{t+1} = \int_{\theta^{\texttt{b}}}^1 P^{\texttt{mix}}(\theta) \cdot \left( \frac{{\theta}}{\theta^{\texttt{b}}}\right)^{\sum_{i=0}^t X_i} \cdot \left( \frac{1-{\theta}}{1-\theta^{\texttt{b}}}\right)^{\sum_{i=0}^t n_i - \sum_{i=0}^t X_i} d\theta.
\end{equation}
It can be shown that, if the true efficacy $\theta^* > \theta^{\texttt{b}}$ lies in the support of $P^{\texttt{mix}}$, then $M^{\texttt{mix}}_{t+1} \to \infty$ as $\sum_{i=0}^t n_i \to \infty$, \ie, the mixture achieves asymptotic power $1$~\citep{ramdas2023gametheoreticstatisticssafeanytimevalid}. In this context, Algorithm~\ref{alg:value_iteration_decomposition} can be used to compute the optimal agent policy for any subsidy with minimal modifications: it suffices to replace the test process function $f$ with
\begin{equation*}
    f^{\texttt{mix}}(\alpha,\beta) = \int_{\theta^{\texttt{b}}}^1 P^{\texttt{mix}}(\theta) \cdot \left( \frac{{\theta}}{\theta^{\texttt{b}}}\right)^{\alpha-\alpha_0} \cdot \left( \frac{1-{\theta}}{1-\theta^{\texttt{b}}}\right)^{\beta - \beta_0} d\theta,
\end{equation*}
and Algorithm~\ref{alg:epsilon} can then be used to compute the principal's optimal subsidy with no further modification.\footnote{Note that the fact that Algorithm~\ref{alg:epsilon} returns an optimal subsidy only relies on the agent's optimal utility being a piecewise linear and convex function of the subsidy (Proposition~\ref{prop:utility-piecewise-convex}), which holds independently of the specific form of $f$, as can be seen from its proof in Appendix~\ref{app:proof-utility-piecewise-convex}}
\begin{figure}[H]
    \centering
    \includegraphics[width=0.5\linewidth]{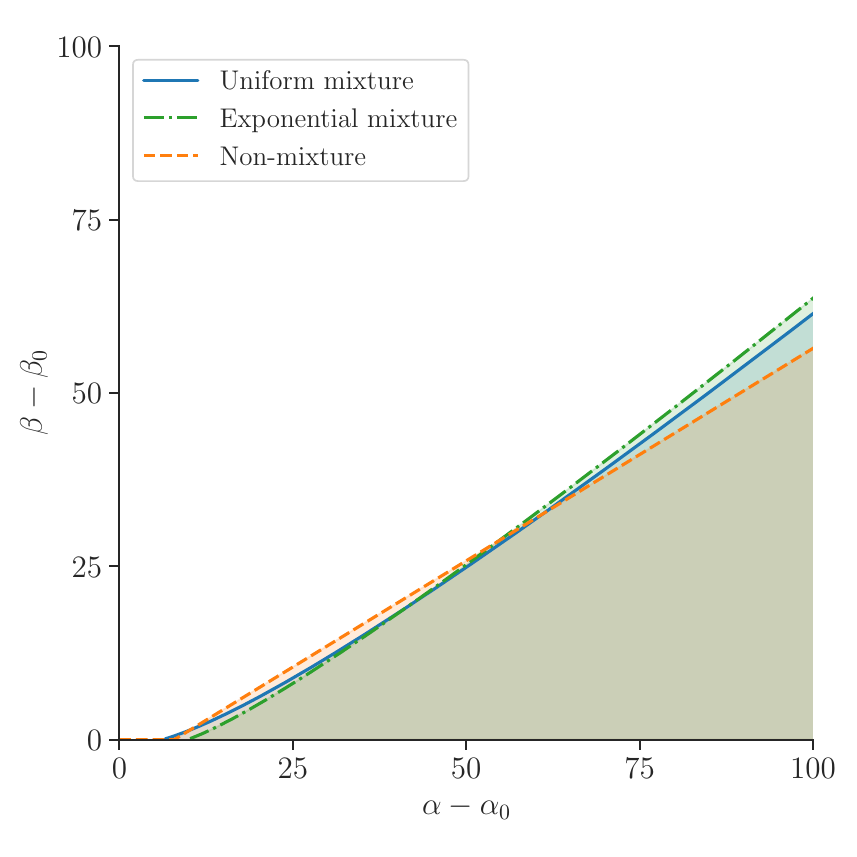}
    \vspace{-0.5cm}
    \caption{\textbf{Rejection regions in belief space.} The figure shows, for each agent belief with parameters $(\alpha, \beta)$, whether the condition $f(\alpha,\beta) \geq 1/\kappa$ is satisfied (\ie, whether $H_0$ is rejected; shaded region), under different test processes: the test process defined using the non-mixed e-values in Proposition~\ref{prop:binomial-e-value} (orange), a test process defined in Eq.~\ref{eq:mixture-process} with a uniform mixture $P^{\texttt{mix}} = \mathrm{U}(\theta^{\texttt{b}},1)$ (blue), and a test process defined in Eq.~\ref{eq:mixture-process} with an exponential mixture $P^{\texttt{mix}} = \mathrm{Exp}(10)$ restricted to $(\theta^{\texttt{b}},1)$ (green).
    Here, we set $\kappa = 0.05$, $\theta^{\texttt{b}} = 0.5$, and $\alpha_0 = \beta_0 = 1$.
    }
    \label{fig:mixture-rejection-regions}
\end{figure}
We have presented the non-mixture process $M$ defined in Section~\ref{sec:model} in the main text for two reasons. First, it yields closed-form expressions which simplify the exposition and provide a more transparent (geometric) intuition for the belief Markov decision process (see Figure~\ref{fig:geometry}). Second, the multiplicative property in Eq.~\ref{eq:test-process} is exploited extensively in the theoretical proofs in Appendix~\ref{app:proofs}, whereas the mixture process $M^{\texttt{mix}}$ does not satisfy this property in non-trivial cases where $P^{\texttt{mix}}$ is not a point mass. However, we conjecture that all our theoretical results in Section~\ref{sec:MDP} and~\ref{sec:agent} carry over to the mixture setting---in particular, because the belief $(\alpha_t,\beta_t)$ remains a sufficient statistic for $M^{\texttt{mix}}_t$ via the monotone function $f^{\texttt{mix}}$, and the Bayesian update of the agent is unchanged---but adapting the proofs would require routing arguments through the Markov property of the belief MDP rather than through the algebraic identity $M_{t+1} = E_t \cdot M_t$. We leave a formal treatment of this extension for future work.

In Appendix~\ref{app:mixture-results}, we present experimental results using the above mixture test process with a uniform mixture $P^{\texttt{mix}} = \mathrm{U}(\theta^{\texttt{b}},1)$. Further, in Figure~\ref{fig:mixture-rejection-regions}, we compare the regions of the belief space where $f(\alpha,\beta) \geq 1/\kappa$, \ie, where $H_0$ is rejected, for different test processes. The test process without mixture (based on Proposition~\ref{prop:binomial-e-value}) yields a linear rejection region, while the mixture processes yield nonlinear regions whose boundary slope increases for larger values of $(\alpha,\beta)$.

\clearpage
\newpage

\section{Geometry of the Belief Markov Decision Process}\label{app:geometry}

In Figure~\ref{fig:geometry} we illustrate the geometry of the belief MDP $\Mcal^\varepsilon$ in the $(\alpha, \beta)$-plane corresponding to all the possible beliefs for the agent. In blue, we represent the beliefs to which the agent may transition after selecting an action $n_t>0$ (the figure shows $n_t=3$ for concreteness). The exact next state depends on the realized outcome $X_t$ of the experiment, as described in Eq.~\ref{eq:transition-dynamics}.
Further, in light of Proposition~\ref{prop:test-process}, note that each pair $(\alpha, \beta)$ is associated with a value $f(\alpha,\beta)$ for the test process $M$, with pairs such that $f(\alpha,\beta) \geq 1/\kappa$ corresponding to a state where $H_0$ as been rejected by the principal. The condition

\begin{equation}
    f(\alpha, \beta) \geq 1/\kappa
    \iff 
    \alpha - \alpha_0
    - (\alpha + \beta - \alpha_0 - \beta_0)\log(1+\theta^{\texttt{b}}(e-1)) \geq \log(1/\kappa)
\end{equation}
corresponds to a linear region in the $(\alpha, \beta)$-plane, represented in green in Figure~\ref{fig:geometry} and labelled ``Reject $H_0$''. Note that the reward function defined in Eq.~\ref{eq:MDP-reward-def} only includes the positive term $\rho^{\texttt{A}}$ for a transition that crosses the boundary of this region. 

Similarly, the red region corresponds to beliefs for which---given a fixed total cost $C\geq 0$ and initial time step $l\in[T]$--- it is optimal for the agent to opt out of the approval process by choosing $n=0$. This region is bounded by the curve $\alpha \mapsto \tilde{\beta}(\alpha)$ in Proposition~\ref{prop:opt-out-mono}. Refer to Appendix~\ref{app:additional-experimental-results} for concrete numerical solutions to the MDP $\Mcal^\varepsilon$.

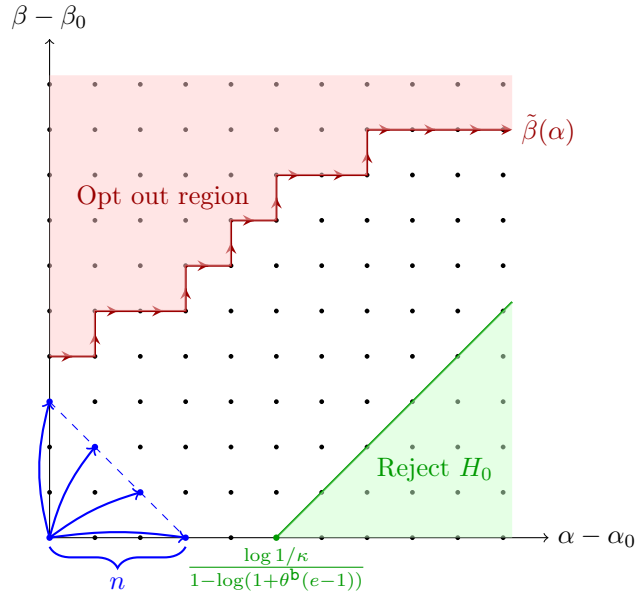
\begin{figure}[H]
    \centering
    \begin{tikzpicture}[scale=0.6]
    \draw[->] (0,0) -- (11,0) node[right] {$\alpha-\alpha_0$};
    \draw[->] (0,0) -- (0,11) node[above] {$\beta-\beta_0$};
    \foreach \i in {0,...,10} {
        \foreach \j in {0,...,10} {
            \fill[black] (\i,\j) circle (1.5pt);
        }
    }
    \fill[blue] (0,0) circle (2pt);
    \fill[blue] (3,0) circle (2pt);
    \fill[blue] (2,1) circle (2pt);
    \fill[blue] (1,2) circle (2pt);
    \fill[blue] (0,3) circle (2pt);
    \draw[decorate, decoration={brace, amplitude=8pt, mirror}, thick, blue] 
        (0,-0.2) -- (3,-0.2) 
        node[midway, below=8pt, blue] {$n$};
    \draw[dashed, blue] (0,3) -- (3,0);
    \foreach \i/\j/\b in {0/3/15,1/2/12,2/1/10,3/0/8} {
        \draw[->, blue, thick] (0,0) to[bend left=\b] (\i,\j);
    }
    \draw[thick, red!60!black, -stealth]
          (0,4) -- (1,4) -- (1,5) -- (2,5) -- (3,5) -- (3,6)
        -- (4,6) -- (4,7) -- (5,7) -- (5,8) -- (6,8)
        -- (7,8) -- (7,9) -- (8,9) -- (9,9) -- (10.2,9);
    \foreach \x/\y in {
        0.5/4,
        1.5/5,
        2.5/5,
        3.5/6,
        4.5/7,
        5.5/8,
        6.5/8,
        7.5/9,
        8.5/9,
        9.5/9
    } {
        \draw[-stealth, thick, red!60!black] (\x-0.001,\y) -- (\x+0.001,\y);
    }
    \foreach \x/\y in {
        1/4.5,
        3/5.5,
        4/6.5,
        5/7.5,
        7/8.5
    } {
        \draw[-stealth, thick, red!60!black] (\x,\y-0.001) -- (\x,\y+0.001);
    }
    \fill[red!20, opacity=0.5]
        (0,4) -- (1,4) -- (1,5) -- (2,5) -- (3,5) -- (3,6)
        -- (4,6) -- (4,7) -- (5,7) -- (5,8) -- (6,8)
        -- (7,8) -- (7,9) -- (8,9) -- (9,9) -- (10.2,9)
        -- (10.2,10.2) -- (0,10.2) -- cycle;
    \node[red!60!black] at (2.5,7.5) {Opt out region};
    \node[red!60!black] at (11,9) {$\tilde{\beta}(\alpha)$};
    \draw[thick, green!60!black] (5,0) -- (10.2,5.2);
    \fill[green!20, opacity=0.5]
        (5,0) -- (10.2,5.2) -- (10.2,0) -- cycle;
    \fill[green!60!black] (5,0) circle (2pt)
        node[below] {$\frac{\log 1/\kappa}{1-\log(1+\theta^{\texttt{b}}(e-1))}$};
    \node[green!60!black] at (8.5,1.5) {Reject $H_0$};
    \end{tikzpicture}
    \caption{\textbf{Illustration of the geometry of the state space in the MDP $\Mcal^\varepsilon$.}
    }
    \label{fig:geometry}
\end{figure}

\clearpage
\newpage

\section{Value Iteration in the Belief Markov Decision Process}\label{app:mdp-algorithms}

In Algorithm~\ref{alg:value_iteration}, we present, for completeness, an adaptation of the value-iteration algorithm~\citep{Sutton1998} to exactly compute an optimal policy $\pi^\varepsilon$ for the MDP $\Mcal^\varepsilon$.\footnote{$\mathrm{BB}(n,\alpha,\beta)$ denotes the Beta–Binomial distribution.}

In the proof of Proposition~\ref{prop:reacheable-set} in Appendix~\ref{app:proof-reacheable-set}, we showed that the set of accessible states in $t$ time steps, which we denoted by $\Scal^{\texttt{r}}(t)$, satisfies:
\begin{equation*}
    |\Scal^{\texttt{r}}(t)| = \Ocal\left( (n^{\texttt{max}})^2\cdot t^2 \right).
\end{equation*}

Building on this bound, it is straightforward to verify that the value-iteration method in Algorithm~\ref{alg:value_iteration} has a time complexity $\Ocal\left( (n^{\texttt{max}})^4 \cdot T^3 \right)$. Indeed, for a given $l$, the set $|\Scal^{\texttt{r}}(l) \setminus \{S^{\texttt{out}}\}|=\Ocal\left( (n^{\texttt{max}})^2\cdot l^2 \right)$, and for each such state, the inner loops (line~\ref{line:value-iteration-state} and~\ref{line:value-iteration-action}) iterate over all actions $n$ and sums over possible outcomes $x$. That is, for each $l$, the algorithm performs $\Ocal\left( (n^{\texttt{max}})^4\cdot l^2 \right)$
iterations. Summing over $l$, we conclude that the total complexity is $\Ocal\left( (n^{\texttt{max}})^4 \cdot 
T^3 \right)$. We refer the reader to Appendix~\ref{app:experimental-details} for further details regarding the implementation of the algorithm.

\begin{algorithm}
\caption{Finds the Optimal Agent Policy $\pi^\varepsilon$ for the belief MDP $\Mcal^\varepsilon$}
\label{alg:value_iteration}
\begin{algorithmic}[1]
\State \textbf{Input} Subsidy $\varepsilon$, horizon $T$, max trials $n^{max}$, approval benefit $\rho^A$, cost function $c(n) = c_0 + c_1 n$, threshold $\kappa$, prior parameters $(\alpha_0, \beta_0)$, test process function $f$

\State \textbf{Initialize} $V^\varepsilon(S, T+1) \gets 0$ \textbf{for all} states $S\in \Scal^{\texttt{r}}$

\For{$l = T$ \textbf{down to} $0$}
    \For{each $S = (\alpha, \beta, C) \in \Scal^{\texttt{r}}(l) \setminus \{S^{\texttt{out}}\}$ such that $f(\alpha, \beta) < 1/\kappa$} \label{line:value-iteration-state}
        
        \For{each action $n \in \{1, \dots, n^{max}\}$}\label{line:value-iteration-action}
            \State $Q^\varepsilon(S, n,l) \gets -(c_0 + c_1 n) + \sum_{x=0}^n \text{BB}(n, \alpha, \beta)(x) \cdot \text{NextValue}(x, n)$
            \State \textbf{where} $\text{NextValue}(x, n) =$
            \State \hspace{\algorithmicindent}$\begin{cases} 
                \rho^A + \varepsilon(C + c_0 + c_1 n) & \text{if } f(\alpha+x, \beta+n-x) \geq 1/\kappa \\
                V^\varepsilon(\alpha+x, \beta+n-x, C+c_0+c_1 n, l+1) & \text{otherwise}
            \end{cases}$
        \EndFor
        
        \State $V^\varepsilon(S, l) \gets \max \left\{ 0, \max_{n \in \{1, \dots, n^{max}\}} Q^\varepsilon(S, n,l) \right\}$
        
        \If{$V^\varepsilon(S, l) > 0$}
            \State $\pi^\varepsilon(S, l) \gets \arg\max_{n \in \{1, \dots, n^{max}\}} Q^\varepsilon(S, n,l)$ \LineComment{Ties broken arbitrarily}
        \Else
            \State $\pi^\varepsilon(S, l) \gets 0$ 
        \EndIf
        
    \EndFor
\EndFor

\State \textbf{return} $V^\varepsilon, \pi^\varepsilon$
\end{algorithmic}
\end{algorithm}
Note that the above algorithm can readily be modified to return the linear decomposition of the optimal value function at a subsidy $\varepsilon$ in Proposition~\ref{prop:linear-value}, namely $V^\varepsilon_{\pi^\varepsilon} = V^0_{\pi^\varepsilon} + \varepsilon \cdot A_{\pi^\varepsilon}$, as shown in Algorithm~\ref{alg:value_iteration_decomposition}.

\begin{algorithm}
\caption{Finds the Optimal Agent Policy $\pi^\varepsilon$ and Value Function Decomposition}
\label{alg:value_iteration_decomposition}
\begin{algorithmic}[1]
\State \textbf{Input} Subsidy $\varepsilon$, horizon $T$, max trials $n^{max}$, approval benefit $\rho^A$,
       cost $c(n) = c_0 + c_1 n$, threshold $\kappa$, prior $(\alpha_0, \beta_0)$, test process function $f$

\State \textbf{Initialize} $V^\varepsilon(S, T+1) \gets 0$,\; $V^0(S, T+1) \gets 0$,\;
       $A(S, T+1) \gets 0$ \textbf{ for all } $S\in \Scal^{\texttt{r}}$

\For{$l = T$ \textbf{down to} $0$}
    \For{each $S = (\alpha, \beta, C) \in \Scal^{\texttt{r}}(l) \setminus \{S^{\texttt{out}}\}$
         s.t.\ $f(\alpha, \beta) < 1/\kappa$}
        
        \For{each action $n \in \{1, \dots, n^{max}\}$}

            \State $c_n \gets c_0 + c_1 n$,\quad
                   $(\alpha'_x, \beta'_x) \gets (\alpha{+}x,\; \beta{+}n{-}x)$

            \State $Q^\varepsilon(S,n,l) \gets -c_n +
                   \displaystyle\sum_{x=0}^{n} \mathrm{BB}(n,\alpha,\beta)(x)
                   \cdot v^\varepsilon(x,n)$
            \State $Q^0(S,n,l) \gets -c_n +
                   \displaystyle\sum_{x=0}^{n} \mathrm{BB}(n,\alpha,\beta)(x)
                   \cdot v^0(x,n)$
            \State $Q^A(S,n,l) \gets
                   \displaystyle\sum_{x=0}^{n} \mathrm{BB}(n,\alpha,\beta)(x)
                   \cdot a(x,n)$

            \State \textbf{where} (writing $\mathbf{approved}$ for $f(\alpha'_x,\beta'_x)\ge 1/\kappa$):
            \[
              v^\varepsilon(x,n) = \begin{cases}
                \rho^A + \varepsilon(C+c_n) & \text{if approved} \\
                V^\varepsilon(\alpha'_x,\beta'_x,C{+}c_n,l{+}1) & \text{otherwise}
              \end{cases}
            \]
            \[
              v^0(x,n) = \begin{cases}
                \rho^A & \text{if approved} \\
                V^0(\alpha'_x,\beta'_x,C{+}c_n,l{+}1) & \text{otherwise}
              \end{cases}
            \]
            \[
              a(x,n) = \begin{cases}
                C + c_n & \text{if approved} \\
                A(\alpha'_x,\beta'_x,C{+}c_n,l{+}1) & \text{otherwise}
              \end{cases}
            \]
            \Comment{Note: $Q^\varepsilon = Q^0 + \varepsilon \cdot Q^A$ by construction}
        \EndFor

        \State $n^* \gets \arg\max_{n \in \{1,\dots,n^{max}\}} Q^\varepsilon(S,n,l)$
               \Comment{Ties broken arbitrarily}

        \If{$Q^\varepsilon(S, n^*, l) > 0$}
            \State $V^\varepsilon(S,l) \gets Q^\varepsilon(S,n^*,l)$,\quad
                   $V^0(S,l) \gets Q^0(S,n^*,l)$,\quad
                   $A(S,l) \gets Q^A(S,n^*,l)$
            \State $\pi^\varepsilon(S,l) \gets n^*$
        \Else
            \State $V^\varepsilon(S,l) \gets 0$,\quad $V^0(S,l) \gets 0$,\quad $A(S,l) \gets 0$
            \State $\pi^\varepsilon(S,l) \gets 0$
        \EndIf

    \EndFor
\EndFor

\State \textbf{return} $V^\varepsilon,\; V^0,\; A,\; \pi^\varepsilon$
\end{algorithmic}
\end{algorithm}

\clearpage
\newpage

\section{Proofs}\label{app:proofs}

\subsection{Proof of Proposition~\ref{prop:binomial-e-value}}

Fix $\theta^*\in [0,1]$ and $n_t>0$. Let $X_t\sim \mathrm{Bin}(n_t,\theta^*)$ and define
\begin{equation*}
    E(X_t, n_t) = \exp\left( X_t - n_t\cdot \log(1+\theta^{\texttt{b}}(e-1)) \right).
\end{equation*}
We explicitly show that the expectation of the above random variable is upper-bounded by $1$ if $ \theta^* \in H_0 = \left\{ {\theta^* \, \colon \, \theta^* < \theta^{\texttt{b}}} \right\}$:

\begin{align*}
            \EE_{X_t \sim \mathrm{Bin}(n_t,\theta^*)}[E(X_t, n_t)]&=\EE_{X_t \sim \mathrm{Bin}(n_t,\theta^*)}[\exp\left( X_t - n_t\cdot \log(1+\theta^{\texttt{b}}(e-1)) \right)] \\
            &= \frac{1}{\left( 1+\theta^{\texttt{b}}(e-1) \right)^{n_t}} \cdot \EE_{X_t \sim \mathrm{Bin}(n_t,\theta^*)}[\exp\left( X_t \right)]\\
            &\overset{(*)}{=} \frac{1}{\left( 1+\theta^{\texttt{b}}(e-1) \right)^{n_t}} \cdot (1-\theta^* + \theta^*\cdot e)^{n_t}\\
            &\overset{(**)}{\leq } \frac{1}{\left( 1+\theta^{\texttt{b}}(e-1) \right)^{n_t}} \cdot (1-\theta^{\texttt{b}} + \theta^{\texttt{b}}\cdot e)^{n_t}\\
            &=1
\end{align*}
where in $(*)$ we have used the formula for the moment-generating function of the Binomial distribution, and in $(**)$ we have used that, by definition, if $\theta^* \in H_0$, then $\theta^* < \theta^{\texttt{b}}$.

\clearpage
\newpage

\subsection{Proof of Proposition~\ref{prop:test-process}}
We begin by noting that the definition of the test process in Eq.~\ref{eq:test-process} together with the e-value in Eq.~\ref{eq:binomial-e-value} implies that, if the agent has continued the approval process up to time $t\leq T$ by selecting non-null $n_0,\dots,n_t$, then:

\begin{align}\label{app:aux-test-process-simplification}
    M_{t+1} &= \prod_{s=0}^t \exp\left( X_s - n_s\cdot \log(1+\theta^{\texttt{b}}(e-1)) \right)\\
    &=\exp\left( \sum_{s=0}^t X_s - \sum_{s=0}^t n_s \cdot \log(1+\theta^{\texttt{b}}(e-1)) \right)
\end{align}

Now, unfolding Eq.~\ref{eq:bayes-belief} for $t$ time steps, we can write the parameters $\alpha_t$ and $\beta_t$ that characterize the belief of the agent at time $t$ as:
\begin{align}\label{eq:bijection-belief-outcomes}
&\begin{dcases}
        \alpha_{t+1} = \alpha_0 + \sum_{s=0}^t X_s\\
        \beta_{t+1} = \beta_0 + \sum_{s=0}^t n_s - \sum_{s=0}^t X_s,
\end{dcases}\\
\iff&
\begin{dcases}
        \sum_{s=0}^t X_s = \alpha_{t+1} - \alpha_0 \\
        \sum_{s=0}^t n_s = \beta_{t+1} - \beta_0 + \sum_{s=0}^t X_s = \alpha_{t+1} - \alpha_0 + \beta_{t+1} - \beta_0,
\end{dcases}
\end{align}
Thus, substituting the above in Eq.~\ref{app:aux-test-process-simplification} we readily obtain:
\begin{align*}
    M_{t+1} &= \exp\left( \sum_{s=0}^t X_s - \sum_{s=0}^t n_s \cdot \log(1+\theta^{\texttt{b}}(e-1)) \right)\\
    &=\exp\left( \alpha_{t+1} - \alpha_0 - ( \alpha_{t+1} - \alpha_0 + \beta_{t+1} - \beta_0) \cdot \log(1+\theta^{\texttt{b}}(e-1)) \right)\\
\end{align*}

\clearpage
\newpage

\subsection{Auxiliary Lemma to Proposition~\ref{prop:MDP-equiv}}
Here, we prove a lemma that will be used later in the proof of Proposition~\ref{prop:MDP-equiv}.

\begin{lemma}\label{lemma:app-total-expectation}
    Let $\pi\in\Pi$ be a policy, $S\in\Scal$ be a state such that $S\neq S^{\texttt{out}}$, denote $S=(\alpha,\beta, C)$ and fix an integrable function $G \,\colon\, \mathbb{R}^3 \to \mathbb{R}$. Further, denote by $\pi(S,t)(\bullet)$ the density over $\Acal$ given by $\pi$ at time $t$, by $\mathrm{Beta}(\alpha, \beta)(\bullet)$ the density of the Beta distribution, by $\mathrm{Bin}(n,\theta)(\bullet)$ the density of a Binomial distribution with parameters $n$ and $\theta$, and by $\mathrm{BB}(n, \alpha,\beta)(\bullet)$ the density of a Beta-Binomial distribution. Consider any random variables $\theta^*$, $n_t$ and $X_t$ with a joint density $P$ such that $n_t$ and $\theta^*$ are independent and:\footnote{We adopt the notation $\theta^*$, $n_t$, and $X_t$ to align with the proof of Proposition~\ref{prop:MDP-equiv}, where this lemma is applied to the efficacy, sample size, and outcomes of the approval process.}
    \begin{equation}
        \begin{dcases}
            P(n_t) = \pi(S,t)(n_t)\\
            P(\theta^*) = \mathrm{Beta}(\alpha, \beta)(\theta^*)\\
            P(X_t \,|\, n_t, \theta^*) = \mathrm{Bin}(n_t,\theta^*)(X_t).
        \end{dcases}
    \end{equation}
    Then, it holds that:
    \begin{multline}
        \EE_{\theta^* \sim \mathrm{Beta}(\alpha,\beta)} \EE_{n_t \sim \pi(S,t)} \EE_{X_t \sim \mathrm{Bin}(n_t, \theta^*)}[G(n_t, \theta^*, X_t)] =\\ \EE_{n_t \sim \pi(S,t)} \EE_{X_t \sim \mathrm{BB}(n_t, \alpha,\beta)} \EE_{\theta^* \sim \mathrm{Beta}(\alpha + X_t,\beta +n_t - X_t)} [G(n_t, \theta^*, X_t)].
    \end{multline}
\end{lemma}
\begin{proof}
    We first note that $X_t |n_t$ follows a Beta-Binomial distribution with parameters $(n_t, \alpha,\beta)$. Indeed, denoting by $B$ the beta function and by $\Gamma$ the gamma function:
    \begin{align*}
        P(X_t = k | n_t) &= \int_0^1 P(X_t = k | n_t, \theta^*) P(\theta^* ) \, d\theta^* \\
        &= \int_0^1 \left[ \binom{n_t}{k} (\theta^*)^k (1-\theta^*)^{n_t-k} \right] \left[ \frac{1}{B(\alpha, \beta)} (\theta^*)^{\alpha-1} (1-\theta^*)^{\beta-1} \right] d\theta^* \\
        &= \frac{\binom{n_t}{k}}{B(\alpha, \beta)} \int_0^1 (\theta^*)^{k+\alpha-1} (1-\theta^*)^{n_t-k+\beta-1} d\theta^* \\
        &= \binom{n_t}{k} \frac{B(k+\alpha, n_t-k+\beta)}{B(\alpha, \beta)} \\
        &= \binom{n_t}{k} \frac{\Gamma(k+\alpha)\Gamma(n_t-k+\beta)}{\Gamma(n_t+\alpha+\beta)} \frac{\Gamma(\alpha+\beta)}{\Gamma(\alpha)\Gamma(\beta)},
    \end{align*}
    which is precisely the density of a Beta-Binomial distribution with parameters $(n_t, \alpha,\beta)$. Similarly, $\theta^* | n_t, X_t$ follows a Beta distribution with parameters $(\alpha + X_t, \beta + n_t - X_t)$:
    \begin{align*}
        P(\theta^* | X_t, n_t) &= \frac{P(X_t | \theta^*, n_t) P(\theta^* )}{P(X_t | n_t)} \\
        &= \frac{\left[ \binom{n_t}{X_t} (\theta^*)^{X_t} (1-\theta^*)^{n_t-X_t} \right] \left[ \frac{1}{B(\alpha, \beta)} (\theta^*)^{\alpha-1} (1-\theta^*)^{\beta-1} \right]}{\binom{n_t}{X_t} \frac{B(\alpha+X_t, \beta+n_t-Xt0)}{B(\alpha, \beta)}} \\
        &= \frac{(\theta^*)^{X_t} (1-\theta^*)^{n_t-X_t} (\theta^*)^{\alpha-1} (1-\theta^*)^{\beta-1}}{B(\alpha+X_t, \beta+n_t-Xt0)} \\
        &= \frac{(\theta^*)^{\alpha + X_t - 1} (1-\theta^*)^{\beta + n_t - X_t - 1}}{B(\alpha + X_t, \beta + n_t - X_t)}.
    \end{align*}
    Since $n_t$ and $\theta^*$ are independent, the joint distribution can be written as:
    \begin{equation}
        P(\theta^*,n_t,X_t) = P(n_t) \cdot P(\theta^*) \cdot P(X_t | \theta^*, n_t)
    \end{equation}
    Alternatively, we can factor the distribution $P$ as:
    \begin{equation}
        P(\theta^*,n_t,X_t) = P(n_t) \cdot P(X_t | n_t) \cdot P(\theta^* | n_t, X_t).
    \end{equation}
    Then, we can conclude using the law of total expectations for the expression $\EE_{(n_t, \theta^*, X_t) \sim P} [G(n_t, \theta^*, X_t)]$ to obtain:
    \begin{multline}
          \EE_{\theta^* \sim \mathrm{Beta}(\alpha,\beta)} \EE_{n_t \sim \pi(S,t)} \EE_{X_t \sim \mathrm{Bin}(n_t, \theta^*)}[G(n_t, \theta^*, X_t)]\\
         =\EE_{n_t \sim \pi(S,t)} \EE_{X_t \sim \mathrm{BB}(n_t, \alpha,\beta)} \EE_{\theta^* \sim \mathrm{Beta}(\alpha + X_t,\beta +n_t - X_t)} [G(n_t, \theta^*, X_t)]
    \end{multline}

\end{proof}

\clearpage
\newpage

\subsection{Proof of Proposition~\ref{prop:MDP-equiv}}

We fix a subsidy $\varepsilon\in[0,\varepsilon^{\texttt{max}}]$ and a policy $\pi\in\Pi$. Since $\varepsilon$ is fixed, for clarity, we omit it from the notation for the remainder of the proof, and we will simply use the notation $r(\bullet)$ to denote the (agent's) reward in $\Mcal^\varepsilon$. Similarly, $V_\pi$ will denote the value function in $\Mcal^\varepsilon$ for policy $\pi$. To show that Proposition~\ref{prop:MDP-equiv} holds, we will proceed by induction over the time steps of the MDP $\Mcal^\varepsilon$. We first introduce some additional notation.

\xhdr{Preliminaries} For any $t\leq T$ (representing a duration) and any $l\leq T - t$ (representing the initial time step), we define the following value function of $\pi$ in the MDP $\Mcal^\varepsilon$:

\begin{equation}\label{eq:value-function-horizon}
    V_{\pi}(S,l,t) =  \EE_{\pi} \left[ \sum_{k=l}^{l+t} r(S_k, n_k, S_{k+1})  \,\middle|\, S_l = S\right],
\end{equation}
which is the total expected reward obtained by policy $\pi$ starting from state $S$ at time $l$ and taking $t$ steps.

Moreover, given a state $S=(\alpha,\beta,C)\in\Scal\setminus \{S^{\texttt{out}}\}$ of the MDP such that $f(S) < 1/\kappa$, $t\leq T$ and $l \leq T - t$, consider the approval process described in Section~\ref{sec:model} when: (i) the initial time index is $l$, (ii) the agent has initial belief $B_0=\mathrm{Beta}(\alpha,\beta)$, (iii) the test process $M$ starts at $M_l = f(\alpha,\beta) < 1/\kappa$ and (iv) the subsidy includes an additional total cost $C$. Then, we define the agent's utility $\Bar{U}_{l,t}^{\texttt{A}}(\pi | S)$ (averaged over its initial belief) for the first $t$ steps as:
\begin{align}\label{eq:app-utility-general-state}
    &\Bar{U}_{l,t}^{\texttt{A}}(\pi | S) = \EE_{\theta^* \sim \mathrm{Beta}(\alpha,\beta)   } \Bigg[ \EE_{ \pi_l} \Bigg[\nonumber\\
    &\left( \rho^{\texttt{A}} + \varepsilon \cdot \left( C + \sum_{j=0}^{\tau(S,l,t)} c(n_j)\right) \right) \cdot \mathds{1}\{\exists j\in[\tau(S,l,t)] \colon f(\alpha,\beta) \cdot M_{j+1} \geq 1/\kappa\} \\
    &- \left( \sum_{j=0}^{\tau(S,l,t)} c(n_j)\right)  \Bigg| \theta^*, B_0 = \mathrm{Beta}(\alpha, \beta), C_0 = C \Bigg]\Bigg]\nonumber
\end{align}
where $\tau(S,l,t) = t \wedge \min\{j \in \{0,\dots,t \} \colon n_j =0\, \text{ or }\, f(\alpha,\beta) \cdot M_{j+1} \geq 1/\kappa \}$ is the last step of the approval process using the shifted policy $\pi_l$ defined as $\pi_l(\bullet,j) = \pi(\bullet,j+l)$. Here, $\EE_{\pi_l}[\bullet|\theta^*, B_0 = \mathrm{Beta}(\alpha, \beta), C_0 = C]$ indicates that: (i) all outcomes $X_t$ are drawn with a fixed efficacy $\theta^*$, and (ii) the agent has initial belief $\mathrm{Beta}(\alpha,\beta)$ and cumulated cost $C$. For instance, in the above expectation, the first action at index $j=0$ taken by the agent is sampled from the distribution $\pi(\alpha, \beta, C, l)$. Alternatively, for a state $S$ such that $S = S^{\texttt{out}}$ or $f(S) \geq 1/\kappa$, we simply define:
\begin{equation}\label{eq:app-utility-terminal-state}
    \Bar{U}^{\texttt{A}}_{l,t}(\pi | S) = 0.
\end{equation}

Importantly, note that, by definition:
\begin{align*}
    &\Bar{U}^{\texttt{A}}_{0,T}(\pi \mid (\alpha_0,\beta_0,0)) \\
    =& \EE_{\theta^* \sim \mathrm{Beta}(\alpha_0,\beta_0)} \Bigg[
        \EE_{\pi_0} \Bigg[
            \left( \rho^{\texttt{A}} + \varepsilon \sum_{j=0}^{\tau(\alpha_0,\beta_0,0,0,0)} c(n_j) \right)
            \cdot \mathds{1}\{\exists j \in [\tau(\alpha_0,\beta_0,0,0,0)] : M_{j+1} \geq 1/\kappa\} \\
    &\qquad
            - \sum_{j=0}^{\tau(\alpha_0,\beta_0,0,0,0)} c(n_j)
         \Bigg| \theta^*, B_0 = \mathrm{Beta}(\alpha_0, \beta_0), C_0 = 0
    \Bigg]\Bigg], \\
    &= \EE_{\theta^* \sim B_0} \big[ U^{\texttt{A}}(\pi,\varepsilon) \big].
\end{align*}
In light of the above, our goal will be to show by induction over $t$ that:

\[
\boxed{
\Bar{U}^{\texttt{A}}_{l,t}(\pi \mid S) = V_{\pi}(S,l,t)
\text{ for any } S \in \Scal,\ t \leq T,\ l \leq T - t.
}
\]

Then, by particularizing to $t=T$, $l=0$ and $S=(\alpha_0,\beta_0,0)$ and using Eq.~\ref{eq:adaptive-utility}, we will obtain that
\begin{equation*}
     \Bar{U}^{\texttt{A}}(\pi;\varepsilon)=V_{\pi}(\alpha_0,\beta_0,0,0,T) = \Bar{U}^{\texttt{A}}_{0,T}(\pi |\, \alpha_0,\beta_0,0 \,) = \EE_{\theta^*\sim B_0} [U^{\texttt{A}}(\pi,\varepsilon)],
\end{equation*}
which is the statement in Proposition~\ref{prop:MDP-equiv}.

\xhdr{Base case $t=0$}
We will show that $\Bar{U}^{\texttt{A}}_{l,0}(\pi |S) = V_{\pi}(S,l,0)$ holds for any $S\in\Scal$ and $l\leq T$. We first consider the case where $S = S^{\texttt{out}}$ or $f(S) \ge 1/\kappa$. In this case, note that the transition dynamics of the MDP (Eq.~\ref{eq:transition-dynamics}) imply that at any time step $k\geq l$, $S_k=S^{\texttt{out}}$ or $S_k=S$. Thus, the expression in Eq.~\ref{eq:value-function-horizon} only contains rewards that are null (see Eq.~\ref{eq:MDP-reward-def}), and hence:
\begin{equation*}
    V_\pi(S,l,0) = 0.
\end{equation*}
Similarly, by definition (Eq.~\ref{eq:app-utility-terminal-state}), $\Bar{U}^{\texttt{A}}_{l,0}(\pi | S) = 0$. Thus, $\Bar{U}^{\texttt{A}}_{l,0}(\pi | S) = V_\pi(S,l,0)$.

Consider now the non-trivial case where $S \neq S^{\texttt{out}}$ and $f(S) < 1/\kappa$, and write $S=(\alpha,\beta,C)$. In this case, the value function starting at time $l$ for $t=0$ can be expanded as (note that $t=0$ corresponds to a single step in the MDP):
\begin{align*}
    V_\pi(S,l,0) &= \EE_{\pi} \left[ r(S_l, n_l, S_{l+1}) \middle| S_l=S\right]\\
    &= \EE_{\pi} \left[ -c(n_l) + (\rho^A +\varepsilon\cdot (C + c(n_l)))\cdot \mathds{1}\{f(S_{l+1}) \geq 1/\kappa) \} \middle| S_l=S\right]\\
    &=\EE_{n_l \sim \pi(S,l)} \EE_{X_l \sim \mathrm{BB}(n_l, \alpha, \beta)} \big[ -c(n_l)\\
    &\quad + (\rho^A +\varepsilon\cdot (C + c(n_l)))\cdot \mathds{1}\{ f(S) \cdot E(X_l, n_l) \geq 1/\kappa) \} \big]
\end{align*}
In the above, we denote by $X_l \sim \mathrm{BB}(n_l, \alpha, \beta)$ a sample from the Beta-Binomial distribution, corresponding to sampling $\theta_l \sim \mathrm{Beta}(\alpha,\beta)$, and then $X_l \sim \mathrm{Bin}(n_l, \theta_l)$, and we have used the definition of the test process in Eq.~\ref{eq:test-process}. On the other hand, let $t=0$ in Eq.~\ref{eq:app-utility-general-state}. Then, $\tau(S,l,0) = 0$, and we obtain:\footnote{Whenever we write nested expectations, such as $\EE_{Y}\EE_{Z}[\bullet]$ for arbitrary random variables $Y$ and $Z$, the inner expectation is understood to be conditional on the outer variable; that is, $\EE_{Y}\big[\EE_{Z}[\bullet \mid Y]\big]$. To simplify notation, we may omit the explicit conditioning when no confusion is likely to arise.}

\begin{align*}
    &\Bar{U}_{l,0}^{\texttt{A}}(\pi | S) \\
    &= \EE_{\theta^* \sim \mathrm{Beta}(\alpha,\beta)   } \Bigg[ \EE_{ \pi_l} \Bigg[ -  c(n_0) + \left( \rho^{\texttt{A}} + \varepsilon \cdot \left( C + c(n_0)\right) \right) \cdot \mathds{1}\{  f(\alpha,\beta) \cdot M_1 \geq 1/\kappa\}\\
    &\quad \Bigg| \theta^*, B_0 = \mathrm{Beta}(\alpha, \beta), C_0 = C\Bigg]\Bigg] \\
    &= \EE_{\theta^* \sim \mathrm{Beta}(\alpha,\beta)   } \Bigg[ \EE_{ \pi_l} \Bigg[ -  c(n_0) + \left( \rho^{\texttt{A}} + \varepsilon \cdot \left( C + c(n_0)\right) \right) \cdot \mathds{1}\{  f(S) \cdot M_1 \geq 1/\kappa\}\\
    &\quad \Bigg| \theta^*, B_0 = \mathrm{Beta}(\alpha, \beta), C_0 = C\Bigg]\Bigg] \\
    &=\EE_{\theta^* \sim \mathrm{Beta}(\alpha,\beta)   }  \EE_{n_0 \sim \pi(S,l)} \EE_{X_0 \sim \mathrm{Bin}(n_0, \theta^*)} \Bigg[ -  c(n_0)\\
    &\quad + \left( \rho^{\texttt{A}} + \varepsilon \cdot \left( C + c(n_0)\right) \right) \cdot \mathds{1}\{  f(S) \cdot M_1 \geq 1/\kappa\}  \Bigg]\\
    &\overset{(*)}{=} \EE_{n_0 \sim \pi(S,l)} \EE_{\theta^* \sim \mathrm{Beta}(\alpha,\beta)   }  \EE_{X_0 \sim \mathrm{Bin}(n_0, \theta^*)} \Bigg[ -  c(n_0)\\
    &\quad + \left( \rho^{\texttt{A}} + \varepsilon \cdot \left( C + c(n_0)\right) \right) \cdot \mathds{1}\{  f(S) \cdot E(n_0, X_0) \geq 1/\kappa\}  \Bigg]\\
    &= \EE_{n_l \sim \pi(S,l)} \EE_{X_l \sim \mathrm{BB}(n_l, \alpha,\beta)} \Bigg[ -  c(n_l) + \left( \rho^{\texttt{A}} + \varepsilon \cdot \left( C + c(n_l)\right) \right) \cdot \mathds{1}\{  f(S) \cdot E(n_l, X_l) \geq 1/\kappa\}  \Bigg]\\
    &=V_\pi(S,l,0),
\end{align*}
where in $(*)$ we have used that $\theta^*$ and $n_0$ are independent (since $n_0$ is sampled independently from $\pi(S,l)$). This concludes the base case.

\xhdr{Inductive step $t \to t+1$} For the inductive step, assume that given a $t < T$, for any $l\leq T -t$ and for any state $S \in \Scal$ it holds that:
\begin{equation*}
    \Bar{U}^{\texttt{A}}_{l,t}(\pi \mid S) = V_{\pi}(S,l,t).
\end{equation*}

We will show that $\Bar{U}^{\texttt{A}}_{l,  t+1}(\pi |S) = V_{\pi}(S,l,t+1)$ for any $l \leq T - (t+1)$ and $S \in \Scal$.  To this end, fix any such $S$ and $l$.

Firstly, we consider the case where $S = S^{\texttt{out}}$ or $f(S) \ge 1/\kappa$, where the transition dynamics of the MDP (Eq.~\ref{eq:transition-dynamics}) imply that at any time step $k \geq l$, $S_k=S^{\texttt{out}}$ or $S_k=S$. Then, similarly to the base case, $V_\pi(S,l,t)=0$ and $\Bar{U}^{\texttt{A}}_{l,t}(\pi |S)=0$ for any $t \geq 0$, so the equality holds.

We focus now on the non-trivial case where $S \neq S^{\texttt{out}}$ and $f(S) < 1/\kappa$, and write $S = (\alpha,\beta,C)$. We begin by expanding the value for the first $t+1$ time steps in Eq.~\ref{eq:value-function-horizon}:
\begin{align}\label{eq:app-aux-iterated-exp}\nonumber
    V_{\pi}(S, l, t+1) &= \EE_{\pi} \left[ \sum_{k=l}^{l+t+1} r(S_k, n_k, S_{k+1})  \,\middle|\, S_l = S\right] \\\nonumber
    & = \EE_{\pi} \left[  r(S_{l}, n_{l}, S_{l+1}) + \sum_{k=l+1}^{l+t+1} r(S_k, n_k, S_{k+1}) \,\middle|\, S_{l} = S\right] \\\nonumber
    & = \EE_{\pi} \left[  r(S_{l}, n_{l}, S_{l+1}) + V_{\pi}(S_{l+1},l+1, t) \,\middle|\, S_{l} = S\right] \\\nonumber
    & \overset{(*)}{=} \EE_{\pi} \left[  r(S_{l}, n_{l}, S_{l+1}) + \Bar{U}^{\texttt{A}}_{l+1,t}(\pi|S_{l+1}) \,\middle|\, S_l = S\right] \\
    & \overset{(**)}{=} \underbrace{\EE_{\pi} \left[  r(S_{l}, n_{l}, S_{l+1}) \,\middle|\, S_l = S\right]}_{\dagger} +  \underbrace{\EE_{\pi} \left[\Bar{U}^{\texttt{A}}_{l+1,t}(\pi|S_{l+1}) \cdot \mathds{1}\{ f(S_{l+1}) < 1/\kappa \} \,\middle|\, S_l = S\right]}_{\ddagger}
\end{align}
where $(*)$ follows from the induction hypothesis and $(**)$ because $\Bar{U}^{\texttt{A}}_t(\pi|S_{l+1}) = 0$ if $f(S_{l+1}) \geq 1/\kappa$.
Analogously to the base case, the term $\dagger$ above can be written as:

\begin{align}\label{eq:app-induction-immediate-reward}\nonumber
    \dagger &= \EE_{\pi} \left[  r(S_l,n_l, S_{l+1})\middle| S_l = S \right]\\\nonumber
            &= \EE_{\pi} \left[  -c(n_l) + (\rho^{\texttt{A}} + \varepsilon\cdot(C + c(n_l)))\cdot\mathds{1}\{ f(S_{l+1}) \ge 1/\kappa \}  \middle| S_l = S \right]\\\nonumber
            &= \EE_{n_l \sim \pi(S,l)} \EE_{\theta^* \sim \mathrm{Beta}(\alpha, \beta)} \EE_{X_l \sim \mathrm{Bin}(n_l,\theta^*)} \Big[  -c(n_l)\\\nonumber
            &\quad + (\rho^{\texttt{A}} + \varepsilon\cdot(C + c(n_l)))\cdot\mathds{1}\{ f(S_{l+1}) \geq 1/\kappa \} \Big]\\\nonumber
            &= \EE_{\theta^* \sim \mathrm{Beta}(\alpha, \beta)} \EE_{n_l \sim \pi(S,l)}  \EE_{X_l \sim \mathrm{Bin}(n_l,\theta^*)} \Big[  -c(n_l)\\\nonumber
            &\quad + (\rho^{\texttt{A}} + \varepsilon\cdot(C + c(n_l)))\cdot\mathds{1}\{ f(\alpha+X_l, \beta + n_l - X_l) \geq 1/\kappa \} \Big]\\\nonumber
            &= \EE_{\theta^* \sim \mathrm{Beta}(\alpha, \beta)} \EE_{n_0 \sim \pi_l(S,0)}  \EE_{X_0 \sim \mathrm{Bin}(n_0,\theta^*)} \Big[  -c(n_0)\\\nonumber
            &\quad + (\rho^{\texttt{A}} + \varepsilon\cdot(C + c(n_0)))\cdot\mathds{1}\{ f(\alpha+X_0, \beta + n_0 - X_0) \geq 1/\kappa \} \Big]\\\nonumber
            &=\EE_{\theta^* \sim \mathrm{Beta}(\alpha, \beta)} \EE_{n_0 \sim \pi_l(S,0)}  \EE_{X_0 \sim \mathrm{Bin}(n_0,\theta^*)} \Big[  -c(n_0)\\\nonumber
            &\quad + (\rho^{\texttt{A}} + \varepsilon\cdot(C + c(n_0)))\cdot\mathds{1}\{ f(\alpha,\beta) \cdot E(n_0, X_0) \geq 1/\kappa \} \Big]\\\nonumber
            &=\EE_{\theta^* \sim \mathrm{Beta}(\alpha,\beta)   } \Big[ \EE_{ \pi_l} \Big[ -c(n_0) + \left( \rho^{\texttt{A}} + \varepsilon \cdot \left( C +  c(n_0)\right) \right) \cdot  \mathds{1}\{ f(\alpha, \beta)\cdot M_1 \geq 1/\kappa\}\\
            &\quad \Big| \theta^*, B_0 = \mathrm{Beta}(\alpha, \beta), C_0 = C\Big]\Big] 
\end{align}

We now focus on the term $\ddagger$ in Eq.~\ref{eq:app-aux-iterated-exp}. Given $S_l = S$, the state $S_{l+1}$ is fully determined by the action $n_l$ and the value $X_l$, namely (see Eq.~\ref{eq:transition-dynamics}),
\begin{equation*}
    S_{l+1}=
    \begin{dcases}
         (\alpha + X_l, \beta + n_l - X_l, C+c(n_l)) & \text{if} \quad n_l >0\\
        S^{\texttt{out}}  & \text{if} \quad n_l = 0.
    \end{dcases}
\end{equation*}
Thus, conditional on the policy $\pi(S,l)$ selecting action $n_l=0$, we have  $S_{l+1} = S^{\texttt{out}}$ and $\Bar{U}^{\texttt{A}}_{l+1,t}(\pi|S_{l+1}) = 0$. That is, conditioning on the the event $\{n_l = 0\}$ inside the expectations $\dagger$ and $\ddagger$ results in:

\begin{equation*}
    \begin{dcases}
        \EE_{\pi} \left[  r(S_{l}, n_{l}, S_{l+1}) \,\middle|\, S_l = S, n_l = 0\right]=0 \\
        \EE_{\pi} \left[\Bar{U}^{\texttt{A}}_{l+1,t}(\pi|S_{l+1}) \cdot \mathds{1}\{ f(S_{l+1}) < 1/\kappa \} \,\middle|\, S_l = S, n_l = 0\right]=0.
    \end{dcases}
\end{equation*}
Similarly, note that conditional on the policy $\pi(\alpha,\beta,C,l)$ selecting action $n_l=0$, we also have:
\begin{multline*}
    \EE_{\theta^* \sim \mathrm{Beta}(\alpha,\beta)   } \Bigg[ \EE_{ \pi_l} \Bigg[ \left( \rho^{\texttt{A}} + \varepsilon \cdot \left( C + \sum_{j=0}^{\tau(S,l,t+1))} c(n_j)\right) \right)\\
    \cdot \mathds{1}\{\exists j\in[\tau(S,l,t+1)] \colon f(\alpha,\beta) \cdot M_{j+1} \geq 1/\kappa\} \\
    - \left( \sum_{j=0}^{\tau(S,l,t+1)} c(n_j)\right)  \Bigg| \theta^*, B_0 = \mathrm{Beta}(\alpha, \beta), C_0 = C , n_0 = 0\Bigg]\Bigg] =0
\end{multline*}
As a consequence, in the following, we assume without loss of generality that the action $0$ is not in the support of the distribution $\pi(S,l)$, and therefore $S_{l+1} \neq S^{\texttt{out}}$.

Then, under the above simplification, we expand the expectation in the term $\ddagger$ to average over the possible values of $n_l$ and $X_l$, and then substitute the expression for the utility \textcolor{blue}{$\Bar{U}^{\texttt{A}}_{l+1,t}(\pi|S_{l+1})$} defined in Eq.~\ref{eq:app-utility-general-state} (emphasized below in blue for clarity):\footnote{We also change notation $n_l \to n_l'$ and $X_l \to X_l'$ to avoid confusion with the actions and outcomes that appear when expanding $\Bar{U}^{\texttt{A}}_{l+1,t}(\pi|S_{l+1})$.}

\begin{align*}
    \ddagger &= \EE_{\pi} \left[ \textcolor{blue}{\Bar{U}^{\texttt{A}}_{l+1,t}(\pi|S_{l+1})} \cdot \mathds{1}\{ f(S_{l+1}) < 1/\kappa \} \,\middle|\, S_l = S\right]\\
    &=\EE_{n'_l \sim \pi(S,l)}\EE_{X_l' \sim \mathrm{BB}(n'_l, \alpha , \beta )} \textcolor{blue}{ \EE_{\theta^* \sim \mathrm{Beta}(\alpha + X'_l, \beta + n'_l -X'_l)   } \Bigg[   \EE_{ \pi_{l+1}} \Bigg[ } \mathds{1}\{ f(\alpha + X'_l, \beta + n'_l - X'_l) < 1/\kappa \} \cdot \\
    &\textcolor{blue}{\left( \rho^{\texttt{A}} + \varepsilon \cdot \left( C + c(n_l') + \sum_{j=0}^{\tilde{\tau}} c(n_j)\right) \right) \cdot \mathds{1}\{\exists j\in[\tilde{\tau}] \colon f(\alpha + X_l', \beta + n_l' -X_l') \cdot M_{j+1} \geq 1/\kappa\} }\\
    & - \mathds{1}\{ f(\alpha + X_l', \beta + n_l' - X_l') < 1/\kappa \} \cdot \textcolor{blue}{\left( \sum_{j=0}^{\tilde{\tau}} c(n_j)\right)}\\
    &\quad \textcolor{blue}{\Bigg| \theta^*, B_0= \mathrm{Beta}(\alpha + X_l', \beta + n_l' -X_l'), C_0 = C + c(n'_l)\Bigg]\Bigg]},
\end{align*}
where we have defined $\tilde{\tau} = \tau(S(n'_l, X'_l), l+1, t) $, and we use $S(n'_l,X'_l)$ to denote the state to which $S$ transitions after selecting action $n'_l$ and observing the outcome $X'_l$, as given by the transition dynamics (Eq.~\ref{eq:transition-dynamics}).

We now leverage Lemma~\ref{lemma:app-total-expectation} to reorder the first three expectations as:
\begin{align}\nonumber
    \ddagger &= \EE_{\theta^* \sim \mathrm{Beta}(\alpha, \beta )   }  \EE_{n'_l \sim \pi(S,l)}\EE_{X_l' \sim \mathrm{Bin}(n'_l, \theta^*)}  \Bigg[   \EE_{ \pi_{l+1}} \Bigg[  \mathds{1}\{ f(\alpha + X'_l, \beta + n'_l - X'_l) < 1/\kappa \} \cdot \\\nonumber
    &\left( \rho^{\texttt{A}} + \varepsilon \cdot \left( C + c(n_l') + \sum_{j=0}^{\tilde{\tau}} c(n_j)\right) \right) \cdot \mathds{1}\{\exists j\in[\tilde{\tau}] \colon f(\alpha + X_l', \beta + n_l' -X_l') \cdot M_{j+1} \geq 1/\kappa\} \\\nonumber
    & - \mathds{1}\{ f(\alpha + X_l', \beta + n_l' - X_l') < 1/\kappa \} \cdot \left( \sum_{j=0}^{\tilde{\tau}} c(n_j)\right) \\\nonumber
    &\quad \Bigg| \theta^*, B_0 = \mathrm{Beta}(\alpha + X_l', \beta + n_l' -X_l'), C_0 = C + c(n'_l)\Bigg]\Bigg]\\\nonumber
    &= \EE_{\theta^* \sim \mathrm{Beta}(\alpha , \beta )} \EE_{n_0 \sim \pi(S,l)}\EE_{X_0 \sim \mathrm{Bin}(n_0, \theta^*)}  \Bigg[  \mathds{1}\{ f(\alpha + X_0, \beta + n_0 - X_0) < 1/\kappa \} \cdot \EE_{ \pi_{l+1}} \Bigg[  \\\nonumber
    &\left( \rho^{\texttt{A}} + \varepsilon \cdot \left( C + c(n_0) + \sum_{j=0}^{\tilde{\tau}} c(\tilde{n}_j)\right) \right) \cdot \mathds{1}\{\exists j\in[\tilde{\tau}] \colon f(\alpha + X_0, \beta + n_0 -X_0) \cdot \tilde{M}_{j+1} \geq 1/\kappa\} \\\nonumber
    & - \left( \sum_{j=0}^{\tilde{\tau}} c(\tilde{n}_j)\right) \Bigg] \Bigg| \theta^*, B_0 = \mathrm{Beta}(\alpha + X_0, \beta + n_0 -X_0), C_0 = C+c(n_0)\Bigg],\\\label{eq:app-ddagger}
\end{align}
where, in the last equality, we have renamed the dummy variables appearing in the inner expectation using tildes, \ie, $\tilde{n}_j$, $\tilde{X}_j$, and $\tilde{M}_j$.\footnote{Here, note that $\tilde{\tau}(S(n_0, X_0), l+1, t) \overset{d}{=} \tilde{\tau}(S(n'_l, X'_l), l+1, t)$, because $n_0 \overset{d}{=} n'_l$ and $X_0 \overset{d}{=} X'_l$ by definition, and so we kept the notation $\tilde{\tau}$ to denote $\tau(S(n_0, X_0), l+1, t)$} This is purely notational at this stage, but will be useful later in the proof when it becomes important to distinguish between different sets of variables.
%

To keep our objective in view, recall that the goal of the inductive step is to show that 
\begin{equation*}
    \Bar{U}^{\texttt{A}}_{l,t+1}(\pi \mid S) = \dagger + \ddagger.
\end{equation*}
Our goal will now be to expand the term $\Bar{U}^{\texttt{A}}_{l,t+1}(\pi \mid S)$ and verify that this equality indeed holds. 
To this end, we will use the identity
\begin{equation*}
    1 = \mathds{1} \{f(\alpha,\beta) \cdot M_1 \geq 1/\kappa \} + \mathds{1} \{f(\alpha,\beta) \cdot M_1 < 1/\kappa \},
\end{equation*}
which holds for any $\alpha, \beta$ and $M_1$ since the two events in the indicator functions are complementary, and substitute it in the definition of $\Bar{U}^{\texttt{A}}_{l,t+1}(\pi \mid S)$ in Eq.~\ref{eq:app-utility-general-state}:

\begin{align*}
    &\Bar{U}_{l,t+1}^{\texttt{A}}(\pi | S)\\
    &=\EE_{\theta^* \sim \mathrm{Beta}(\alpha,\beta)   } \Bigg[ \EE_{ \pi_l} \Bigg[ \left( \rho^{\texttt{A}} + \varepsilon \cdot \left( C + \sum_{j=0}^{\tau(S,l,t+1)} c(n_j)\right) \right) \\
    &\quad \cdot \mathds{1}\{\exists j\in[\tau(S,l,t+1)] \colon f(\alpha,\beta) \cdot M_{j+1} \geq 1/\kappa\} \\
    &\quad- \left( \sum_{j=0}^{\tau(S,l,t+1)} c(n_j)\right) \Bigg| \theta^*, B_0 = \mathrm{Beta}(\alpha, \beta), C_0 = C\Bigg] \Bigg]\\
    &= \EE_{\theta^* \sim \mathrm{Beta}(\alpha,\beta)   } \Bigg[ \EE_{ \pi_l} \Bigg[\left( \rho^{\texttt{A}} + \varepsilon \cdot \left( C + \sum_{j=0}^{\tau(S,l,t+1)} c(n_j)\right) \right) \\
    &\quad\cdot\mathds{1}\{f(\alpha,\beta)\cdot M_1\geq 1/\kappa\}\cdot \mathds{1}\{\exists j\in[\tau(S,l,t+1)] \colon f(\alpha,\beta) \cdot M_{j+1} \geq 1/\kappa\} \\
    &\quad- \mathds{1}\{f(\alpha,\beta)\cdot M_1\geq 1/\kappa\}\cdot \left( \sum_{j=0}^{\tau(S,l,t+1)} c(n_j)\right) \Bigg| \theta^*, B_0 = \mathrm{Beta}(\alpha, \beta), C_0 = C\Bigg]\Bigg] \\
    &\quad+ \EE_{\theta^* \sim \mathrm{Beta}(\alpha,\beta)   } \Bigg[ \EE_{ \pi_l} \Bigg[ \left( \rho^{\texttt{A}} + \varepsilon \cdot \left( C + \sum_{j=0}^{\tau(S,l,t+1)} c(n_j)\right) \right)  \\
    &\quad\cdot \mathds{1}\{f(\alpha,\beta)\cdot M_1 < 1/\kappa\}\cdot \mathds{1}\{\exists j\in[\tau(S,l,t+1)] \colon f(\alpha,\beta) \cdot M_{j+1} \geq 1/\kappa\} \\
    &\quad- \mathds{1}\{f(\alpha,\beta)\cdot M_1< 1/\kappa\}\cdot \left( \sum_{j=0}^{\tau(S,l,t+1)} c(n_j)\right)  \Bigg| \theta^*, B_0 = \mathrm{Beta}(\alpha, \beta), C_0 = C\Bigg]\Bigg]\\
\end{align*}
We now observe that if $f(\alpha,\beta)\cdot M_1 \geq 1/\kappa$, then $\tau(S,l,t+1)=0$, that is, the approval process stops after the first step, which simplifies the first summand and yields:
\begin{align*}
    &\Bar{U}_{l,t+1}^{\texttt{A}}(\pi | S)\\
    &= \EE_{\theta^* \sim \mathrm{Beta}(\alpha,\beta)   } \Bigg[ \EE_{ \pi_l} \Bigg[ -c(n_0) + \left( \rho^{\texttt{A}} + \varepsilon \cdot \left( C +  c(n_0)\right) \right) \cdot  \mathds{1}\{ f(\alpha,\beta)\cdot M_1\geq 1/\kappa\} \\
    &\hphantom{ = \EE_{\theta^* \sim \mathrm{Beta}(\alpha,\beta)   } \Bigg[ \EE_{ \pi_l} \Bigg[ }\quad\Bigg| \theta^*, B_0 = \mathrm{Beta}(\alpha, \beta), C_0 = C\Bigg]\Bigg]\\
    &\quad+ \EE_{\theta^* \sim \mathrm{Beta}(\alpha,\beta)   } \Bigg[ \EE_{ \pi_l} \Bigg[\left( \rho^{\texttt{A}} + \varepsilon \cdot \left( C + \sum_{j=0}^{\tau(S,l,t+1)} c(n_j)\right) \right)\\
    &\quad\cdot \mathds{1}\{f(\alpha,\beta)\cdot M_1 < 1/\kappa\}\cdot \mathds{1}\{\exists j\in[\tau(S,l,t+1)] \colon f(\alpha,\beta) \cdot M_{j+1} \geq 1/\kappa\} \\
    &\quad- \mathds{1}\{f(\alpha,\beta)\cdot M_1< 1/\kappa\}\cdot \left( \sum_{j=1}^{\tau(S,l,t+1)} c(n_j)\right)  \Bigg| \theta^*, B_0 = \mathrm{Beta}(\alpha, \beta), C_0 = C\Bigg]\Bigg]\\
\end{align*}

Further, we can identify that the first summand in the above expression corresponds to the form of the term $\dagger$ derived in Eq.~\ref{eq:app-induction-immediate-reward}, and replacing it, we obtain:

\begin{align*}
    &\Bar{U}_{l,t+1}^{\texttt{A}}(\pi | S)\\
    &=  \dagger\\
    & \quad \textcolor{orange}{+\EE_{\theta^* \sim \mathrm{Beta}(\alpha,\beta)   } \Bigg[ \EE_{ \pi_l} \Bigg[} \textcolor{orange}{\left( \rho^{\texttt{A}} + \varepsilon \cdot \left( C + \sum_{j=0}^{\tau(S,l,t+1)} c(n_j)\right) \right)}\\
    &\quad  \textcolor{orange}{\cdot  \mathds{1}\{f(\alpha,\beta)\cdot M_1 < 1/\kappa\} \cdot \mathds{1}\{\exists j\in[\tau(S,l,t+1)] \colon f(\alpha,\beta) \cdot M_{j+1} \geq 1/\kappa\} }\\
    &\quad \textcolor{orange}{- \mathds{1}\{f(\alpha,\beta)\cdot M_1< 1/\kappa\}\cdot \left( \sum_{j=1}^{\tau(S,l,t+1)} c(n_j)\right)  \Bigg| \theta^*, B_0 = \mathrm{Beta}(\alpha, \beta), C_0 = C\Bigg]\Bigg]}\\
\end{align*}
We now focus on the orange term in the expression above, which we denote by $\textcolor{orange}{\square}$. Our goal will be to show that $\textcolor{orange}{\square} = \ddagger$.
We start by factoring out the term $\mathds{1}\{f(\alpha,\beta)\cdot M_1< 1/\kappa\}$, and using that by definition $M_1 = E(n_0, M_0)$, we obtain:

\begin{align*}
    \textcolor{orange}{\square} &= \EE_{\theta^* \sim \mathrm{Beta}(\alpha,\beta)   } \Bigg[ \EE_{ \pi_l} \Bigg[  \mathds{1}\{f(\alpha,\beta)\cdot E(n_0, X_0)< 1/\kappa\} \cdot \Bigg(\\
    &  \quad\left( \rho^{\texttt{A}} + \varepsilon \cdot \left( C + c(n_0) + \sum_{j=1}^{\tau(S,l,t+1)} c(n_j)\right) \right)   \mathds{1}\{\exists j\in[\tau(S,l,t+1)] \colon f(\alpha,\beta) \cdot M_{j+1} \geq 1/\kappa\} \\
    & \quad - \sum_{j=1}^{\tau(S,l,t+1)} c(n_j) \Bigg)  \Bigg| \theta^*, B_0 = \mathrm{Beta}(\alpha, \beta), C_0 = C\Bigg]\Bigg]\\
\end{align*}
In the above expectation $\EE_{ \pi_l} \left[ \bullet \middle| \theta^*, B_0 = \mathrm{Beta}(\alpha, \beta), C_0 = C\right]$, we can use the law of iterated expectations by conditioning on the first samples $n_0$ and $X_0$. In particular, note that $n_0 \sim \pi(\alpha,\beta,C,l)$ and that $X_0 \sim \mathrm{Bin}(n_0,\theta^*)$, and thus:
\begin{align}\nonumber
    \textcolor{orange}{\square} &= \EE_{\theta^* \sim \mathrm{Beta}(\alpha,\beta)} \EE_{n_0 \sim \pi(S,l )} \EE_{X_0 \sim \mathrm{Bin}(n_0, \theta^*)} \Bigg[ \EE_{ \pi_l} \Bigg[  \mathds{1}\{f(\alpha,\beta)\cdot E(n_0, X_0)< 1/\kappa\} \cdot \Bigg(\\\nonumber
    &  \left( \rho^{\texttt{A}} + \varepsilon \cdot \left( C + c(n_0) + \sum_{j=1}^{\tau(S,l,t+1)} c(n_j)\right) \right)   \mathds{1}\{\exists j\in[\tau(S,l,t+1)] \colon f(\alpha,\beta) \cdot M_{j+1} \geq 1/\kappa\} \\\nonumber
    & - \sum_{j=1}^{\tau(S,l,t+1)} c(n_j) \Bigg)  \Bigg| \theta^*, B_0 = \mathrm{Beta}(\alpha, \beta), C_0 = C, n_0, X_0\Bigg]\Bigg] \\\nonumber
    &= \EE_{\theta^* \sim \mathrm{Beta}(\alpha,\beta)} \EE_{n_0 \sim \pi(S,l )}  \EE_{X_0 \sim \mathrm{Bin}(n_0, \theta^*)} \Bigg[
    \mathds{1}\{f(\alpha + X_0,\beta + n_0 + X_0) < 1/\kappa\} \cdot \EE_{ \pi_l} \Bigg[   \\\nonumber
    &  \left( \rho^{\texttt{A}} + \varepsilon \cdot \left( C + c(n_0) + \sum_{j=1}^{\tau(S,l,t+1)} c(n_j)\right) \right)   \mathds{1}\{\exists j\in[\tau(S,l,t+1)] \colon f(\alpha,\beta) \cdot M_{j+1} \geq 1/\kappa\} \\\nonumber
    & - \sum_{j=1}^{\tau(S,l,t+1)} c(n_j) \Bigg]  \Bigg| \theta^*, B_0 = \mathrm{Beta}(\alpha, \beta), C_0 = C, n_0, X_0\Bigg]\\\label{eq:app:orange-square}
\end{align}
Here, note that the inner expectation $\EE_{ \pi_l} \left[ \bullet \middle| \theta^*, B_0 = \mathrm{Beta}(\alpha, \beta), C_0 = C, n_0, X_0\right]$ is conditioned on $n_0$ and $X_0$. Therefore, the first action $n_1$ is distributed according to $\pi(\alpha + X_0, \beta + n_0 - X_0,  C + c(n_0), l + 1)$, and any subsequent action $n_j$ is distributed according to:

\begin{equation*}
    \begin{dcases}
        n_1 | n_0, X_0\sim \pi \left( \alpha + X_0, \beta + n_0 - X_0, C + c(n_0), l + 1 \right) \\
        X_1 | n_0, X_0\sim \operatorname{Bin}(n_1, \theta^*) \\[1ex]
        n_2 | n_0, X_0\sim \pi \left( \alpha + X_0 + X_1, \beta + n_0 + n_1 - X_0 - X_1, C + c(n_0) + c(n_1), l + 2 \right) \\
        X_2 | n_0, X_0\sim \operatorname{Bin}(n_2, \theta^*) \\
        \quad \vdots \\
        \begin{aligned}[t]
            &n_j | n_0, X_0\sim \pi \big( \alpha + X_0 + \dots + X_{j-1}, \beta + n_0 + \dots + n_{j-1} - X_0 - \dots - X_{j-1},\\
            &\hphantom{n_j | n_0, X_0\sim \pi \big(\,} C + c(n_0) + \dots + c(n_{j-1}), l + j \big)
        \end{aligned} \\
        X_j | n_0, X_0\sim \operatorname{Bin}(n_j, \theta^*) \\
        \quad \vdots
    \end{dcases}
\end{equation*}
These are equal in distribution to the sequence of actions and experimental outcomes $\tilde{n}_1, \tilde{X}_1,\dots$ in the inner expectation of Eq.~\ref{eq:app-ddagger}, \ie, we have the following equalities in distribution
\begin{equation}\label{eq:app-trajectory-shift}
    \begin{dcases}
        n_1  | n_0, X_0\overset{d}{=} \tilde{n}_0 \\
        X_1 | n_0, X_0\overset{d}{=} \tilde{X}_0\\
        E(n_1,X_1) | n_0, X_0\overset{d}{=} E(\tilde{n}_0, \tilde{X}_0)\\
        \quad \vdots \\
        n_j | n_0, X_0\overset{d}{=} \tilde{n}_{j-1} \\
        X_j | n_0, X_0\overset{d}{=} \tilde{X}_{j-1}\\
        E(n_j,X_j) | n_0, X_0\overset{d}{=} E(\tilde{n}_{j-1},\tilde{X}_{j-1})\\
        \quad \vdots
    \end{dcases}
    \implies
        \begin{dcases}
        M_1 | n_0, X_0\overset{d}{=} E(n_0, X_0) \cdot \tilde{M}_0\\
        M_2 | n_0, X_0\overset{d}{=} E(n_0, X_0) \cdot \tilde{M}_1\\        \quad \vdots \\
        M_j | n_0, X_0\overset{d}{=} E(n_0, X_0) \cdot \tilde{M}_{j-1}\\
        \quad \vdots
    \end{dcases}. 
\end{equation}
Based on the above, we can also conclude the following about the stopping times under the event $\{f(\alpha,\beta)\cdot E(n_0 ,X_0)< 1/\kappa \}$:
\begin{align}\nonumber
    \tau(S,l,t+1) | n_0, X_0 & \overset{d}{=} (t+1) \wedge \min\{ j \in \{1,\dots,t+1\} \colon n_j=0 \, \text{or}\, f(\alpha,\beta)\cdot M_{j+1} \geq 1/\kappa \}\\\nonumber
    & \overset{d}{=} (t+1) \wedge \Big( 1 + \min\{ j \in \{0,\dots,t\} \colon \tilde{n}_j = 0\, \text{or}\,\\\nonumber
    &\qquad f(\alpha,\beta)\cdot E(n_0, X_0)\cdot \tilde{M}_{j+1} \geq 1/\kappa \}\Big)\\
    & \overset{d}{=} 1 + \underbrace{\tau(S(n_0, X_0), l+1, t)}_{\tilde{\tau}}\label{eq:app-stopping-shift}
\end{align}

As a consequence, in Eq.~\ref{eq:app:orange-square}, we can change the summation index $j \to j-1$ and use the equalities in distribution in Eq.~\ref{eq:app-trajectory-shift} and Eq.~\ref{eq:app-stopping-shift} to finally conclude:

\begin{align*}
    \textcolor{orange}{\square}&= \EE_{\theta^* \sim \mathrm{Beta}(\alpha,\beta)} \EE_{n_0 \sim \pi(S,l )}  \EE_{X_0 \sim \mathrm{Bin}(n_0, \theta^*)} \Bigg[
    \mathds{1}\{f(\alpha + X_0,\beta + n_0 + X_0) < 1/\kappa\} \cdot  \EE_{ \pi_{l+1}} \Bigg[  \\\nonumber
    &\left( \rho^{\texttt{A}} + \varepsilon \cdot \left( C + c(n_0) + \sum_{j=0}^{\tilde{\tau}} c(\tilde{n}_j)\right) \right) \cdot \mathds{1}\{\exists j\in[\tilde{\tau}] \colon f(\alpha + X_0, \beta + n_0 -X_0) \cdot \tilde{M}_{j+1} \geq 1/\kappa\} \\\nonumber
    & - \left( \sum_{j=0}^{\tilde{\tau}} c(\tilde{n}_j)\right) \Bigg| \theta^*, B_0 = \mathrm{Beta}(\alpha + X_0, \beta + n_0 -X_0), C_0 = C+c(n_0)\Bigg]\Bigg] \\
    &= \ddagger
\end{align*}
That is, we have shown that $\Bar{U}_{l,t+1}^{\texttt{A}}(\pi | S) = \dagger + \ddagger = V_{\pi}(S, l, t+1) $. This concludes the induction step and the proof.

\clearpage
\newpage

\clearpage
\newpage

\subsection{Proof of Proposition~\ref{prop:reacheable-set}}\label{app:proof-reacheable-set}

Consider the state space $\Scal$ of the MDP $\Mcal^\varepsilon$. We define the set $\Scal^{\texttt{r}}$ of \emph{reachable} states from the initial state $(\alpha_0,\beta_0,0)$ to be any state that can be reached with non-negative probability by a policy $\pi$. More formally, $S \in \Scal^{\texttt{r}}$ if an only if $S = S^{\texttt{out}}$ or $S=(\alpha,\beta,C)$ and there exist a $0\leq t \leq T$ and a sequence of actions and outcomes $n_0,X_0,\dots n_t,X_t$ such that 
\begin{equation}\label{eq:app-accessible-state-condition}
    \begin{dcases}
            0 \leq X_k \leq n_k, \quad \text{for } k = 0, \dots, t \\
            \alpha = \alpha_0 + \sum_{k=0}^{t} X_k \\
            \beta = \beta_0 + \sum_{k=0}^{t} (n_k - X_k)\\
        C = \sum_{k=0}^t c(n_t).
     \end{dcases}
\end{equation}
In words, the sequence $(n_0,X_0,\dots n_t,X_t)$ allows the initial state to eventually transition to the state $S$ according to the transition dynamics (Eq.~\ref{eq:transition-dynamics}). Observe that then, any state visited under any realization of the MDP for any policy $\pi$ is contained in the set $\Scal^{\texttt{r}}$. Then, $\Scal^{\texttt{r}}$ is finite because, for a fixed $t$, any sequence $(n_0,X_0,\dots n_t,X_t)$ can only take finitely many values (since $|\Acal|=n^{\texttt{max}} + 1$), and the MDP has only finitely many steps, namely, $T+1$.

We focus now on the case where the cost function $c$ is linear. Let $c(n) = c_0 + c_1\cdot n$, and consider any state that can be reached at time step $t$ with a total number $N=\sum_{k=0}^t n_t$, $X = \sum_{k=0}^t X_t$. Then, 
\begin{equation*}
    \sum_{k=0}^t c(n_t) = (t+1)\cdot c_0 + N\cdot c_1.
\end{equation*}
In particular, in light of Eq.~\ref{eq:app-accessible-state-condition}, the states in $\Scal^{\texttt{r}}$ (except $S^{\texttt{out}}$) are in a bijection with the triplets $(t,N,X)$. This is because $\alpha$ and $\beta$ are uniquely determined by $X$ and $N$, and if $c_0,c_1 \neq 0$, then $t$ and $N$ uniquely determine $C$.

Next, observe that if the policy never opts out, given a time step $0\leq t \leq T$, the minimum value that $N$ can take is $N=t+1$, which corresponds to a sequence of actions $n_0=1,\dots,n_t=1$. On the other hand, the maximum value that $N$ can take is $N=(t+1)\cdot n^{\texttt{max}}$, which corresponds to a sequence of actions $n_0=n^{\texttt{max}},\dots,n_t=n^{\texttt{max}}$. Further, for a fixed value of $N$, the total positive outcomes $X$ can take exactly $N+1$ values, \ie, $X\in \{0,\dots,N \}$. Thus,
\begin{align*}
    |\Scal^{\texttt{r}}| &= \underbrace{1}_{S^{\texttt{out}}} + \sum_{t=0}^T \sum_{N=t+1}^{(t+1)\cdot n^{\texttt{max}}} (N+1)\\
    &= 1 + \sum_{t=0}^T \frac{(t+1)\cdot n^{\texttt{max}} - t}{2} \left((t+2) + ((t+1)\cdot n^{\texttt{max}} + 1)\right)\\
    &= 1 + \sum_{t=0}^T \left( \frac{(n^{\texttt{max}})^2 - 1}{2}\cdot (t+1)^2 + \frac{3 n^{\texttt{max}} - 1}{2}\cdot (t+1) + 1\right)\\
    &= 1 + \frac{(n^{\texttt{max}})^2 - 1}{12}\cdot (T+1)(T+2)(2T+3) + \frac{3 n^{\texttt{max}} - 1}{4}\cdot (T+1)(T+2) + (T+1)\\
    &= \Ocal\left( (n^{\texttt{max}})^2\cdot T^3 \right)
\end{align*}
where we have used standard summation formulas for the arithmetic progression $\sum_{t=0}^T (t+1)$ and the quadratic progression $\sum_{t=0}^T (t+1)^2$. In particular, the above also shows that the set of reachable states in $t$ time steps, which we denote by $\Scal^{\texttt{r}}(t)$, satisfies $|\Scal^{\texttt{r}}(t)| = \Ocal\left( (n^{\texttt{max}})^2 \cdot t^2 \right)$.

\clearpage
\newpage

\subsection{Auxiliary Lemma for Proposition~\ref{prop:value-monotonicity}}

\begin{lemma}\label{lemma:bound-value-function}
    Consider a state $S \in \Scal$ such that $S\neq S^{\texttt{out}}$ and $f(S) < 1/\kappa$, any policy $\pi$ and $0\leq l \leq T$. Write $S=(\alpha,\beta,C)$. Then,
    \begin{equation}
        \EE_{\pi} \left[ \sum_{t=l}^T r^\varepsilon(S_t,n_t,S_{t+1}) \middle| S_l = S\right] \leq \rho^{A} +\varepsilon\cdot C.
    \end{equation}
\end{lemma}
Note that as an immediate consequence of the above, it also holds for the optimal value function (Eq.~\ref{eq:app-optimal-value}):

\begin{equation*}
    V^\varepsilon(S,l) \leq \rho^{A} +\varepsilon\cdot C
\end{equation*}

\begin{proof}
    Fix any such $S,l$ and policy $\pi$. Denote by $\tau$ the corresponding stopping time when $S_l=S$, \ie,
    \begin{equation*}
        \tau = T \wedge \min\{ t \in \{l,\dots,T\} \colon S_t = S^{\texttt{out}} \,\text{or}\, f(S_t) \geq 1/\kappa \mid S_l = S  \}.
    \end{equation*}
    Then, $r^\varepsilon(S_t,n_t,S_{t+1})=0$ for $t\geq \tau$ by definition of the rewards (Eq.~\ref{eq:MDP-reward-def}). Thus,
    \begin{align*}
        \sum_{t=l}^T r^\varepsilon(S_t,n_t,S_{t+1}) &= \sum_{t=l}^{\tau-1}r^\varepsilon(S_t,n_t,S_{t+1})\\
        &=\sum_{t=l}^{\tau-1} \left(  -c(n_t) + (\rho^{\texttt{A}} + \varepsilon\cdot(C_t+c(n_t)))\cdot\mathds{1}\{ f(S_{t+1}) \ge 1/\kappa \} \right)\\
        &= -\sum_{t=l}^{\tau-1} c(n_t) + \sum_{t=l}^{\tau-1} (\rho^{\texttt{A}}+\varepsilon\cdot(C_t+c(n_t)))\cdot\mathds{1}\{ f(S_{t+1}) \ge 1/\kappa \}.
    \end{align*}

Note that by definition of $\tau$, the second term is at most non-null for the summand with $t=\tau-1$. Thus:
\begin{align*}
        \sum_{t=l}^T r^\varepsilon(S_t,n_t,S_{t+1}) &\leq -\sum_{t=l}^{\tau-1} c(n_t) +  (\rho^{\texttt{A}}+\varepsilon\cdot(C_{\tau-1}+c(n_{\tau-1})))\\
        &= -\sum_{t=l}^{\tau-1} c(n_t) +  \rho^{\texttt{A}} + \varepsilon\cdot(C_{\tau-1}+c(n_{\tau-1}))\\
        &= -\sum_{t=l}^{\tau-1} c(n_t) +  \rho^{\texttt{A}} + \varepsilon\cdot C_\tau\\
        &=\rho^{\texttt{A}}  -\sum_{t=l}^{\tau-1} c(n_t) + \varepsilon \cdot \left( C + \sum_{t=l}^{\tau-1}c(n_t) \right)\\
        &\leq \rho^{\texttt{A}} +\varepsilon\cdot C,
\end{align*}
where in the last step we have used $\varepsilon \leq \varepsilon^{\texttt{max}}\leq 1$. This concludes the proof.

\end{proof}

\clearpage
\newpage

\subsection{Proof of Proposition~\ref{prop:value-monotonicity}}\label{app:proof-prop-value-monotonicity}
We fix any subsidy $\varepsilon\in [0,\varepsilon^{\texttt{max}}]$. We will use the Bellman optimality equation~\cite{Sutton1998}, which states that the optimal value function $V^\varepsilon$ for $\Mcal^\varepsilon$, defined by:

\begin{equation}\label{eq:app-optimal-value}
V^\varepsilon(S,l) =
\begin{dcases}
     \sup_{\pi \in \Pi} \EE_{\pi}\left[ \sum_{t=l}^T r^\varepsilon(S_t,n_t,S_{t+1}) \middle| S_l = S \right] & \text{if}\quad 0\leq l \leq T\\
    0 & \text{if} \quad l = T+1,
\end{dcases}
\end{equation}
satisfies the following recursive condition:
\begin{equation*}
    V^\varepsilon(S,l) = \max_{n \in \Acal} \left\{ \sum_{S' \in \Scal} P(S'|S,n) \cdot  \left(r^\varepsilon(S,n,S') + V^\varepsilon(S',l+1) \right) \right\},
\end{equation*}
where $P$ are the transition dynamics of the MDP, defined implicitly in Eq.~\ref{eq:transition-dynamics}.
We can now particularize to the action space $\Acal$ in $\Mcal^\varepsilon$. Since the action $n=0$ results in a null reward (see Eq.~\ref{eq:MDP-reward-def}), we obtain for $l\leq T$ and any state $S=(\alpha, \beta,C)$ such that $f(S) < 1/\kappa$ (and $S\neq S^{\texttt{out}}$):

\begin{equation}\label{eq:app-bellman-value-optimality}
    V^\varepsilon(S,l) = 
        \max \left\{0, \max_{n\in \{1,\dots,n^{\texttt{max}}\}}  \EE_{X\sim \mathrm{BB}(n,\alpha,\beta)}\left[ r^\varepsilon(S,n,S(n,X)) + V^\varepsilon(S(n,X),l+1)  \right] \right\} 
\end{equation}
while $V^\varepsilon(S,l) = 0$ if $f(S)\geq 1/\kappa$  or $S=S^{\texttt{out}}$. In the above, to simplify notation, we denote by $S(n,X)$ the state to which the MDP transitions from  $S$ after selecting $n>0$ and observing outcome $X$, \ie,
\begin{equation*}
   S(n,X) = (\alpha + X,\beta + n - X,C+c(n)).
\end{equation*}

\subsubsection{Monotonicity on the belief}

We focus on proving the monotonicity property in Proposition~\ref{prop:value-monotonicity} of $V^\varepsilon$ in the parameter $\alpha$.\footnote{The monotonicity in the parameter $\beta$ follows using a completely symmetric argument.} More precisely, fix a state $S=(\alpha,\beta,C)$ such that $f(\alpha,\beta) < 1/\kappa$ and $\alpha'\geq \alpha$ such that $f(\alpha',\beta) < 1/\kappa$. Denote $S'=(\alpha',\beta,C)$. We want to show that
\[\boxed{
    V^\varepsilon(\alpha,\beta,C,l) \leq V^\varepsilon(\alpha',\beta,C,l)
}
\]
for any $0 \leq l \leq T$. We proceed by induction over $l$.

\xhdr{Base case $l=T$}
Using Eq.~\ref{eq:app-bellman-value-optimality} we obtain:

\begin{align*}
    &V^\varepsilon(\alpha',\beta,C,T)\\
    &= \max \left\{0, \max_{n\in \{1,\dots,n^{\texttt{max}}\}}  \EE_{X\sim \mathrm{BB}(n,\alpha',\beta)}\left[ r^\varepsilon(S',n,S'(n,X))  \right] \right\} \\
    &=\max \bigg\{0, \max_{n\in \{1,\dots,n^{\texttt{max}}\}}  \EE_{X\sim \mathrm{BB}(n,\alpha',\beta)}\Big[ 
    -c(n)\\
    &\quad + (\rho^{\texttt{A}} + \varepsilon\cdot(C+c(n)))\cdot\mathds{1}\{ f(\alpha',\beta)\cdot E(n,X) \ge 1/\kappa \}
    \Big] \bigg\}\\
    &\overset{(*)}{\geq}\max \bigg\{0, \max_{n\in \{1,\dots,n^{\texttt{max}}\}}  \EE_{X\sim \mathrm{BB}(n,\alpha',\beta)}\Big[ 
    -c(n)\\
    &\quad + (\rho^{\texttt{A}} + \varepsilon\cdot(C+c(n)))\cdot\mathds{1}\{ f(\alpha,\beta)\cdot E(n,X) \ge 1/\kappa \}
    \Big] \bigg\}\\
    &\overset{(**)}{\geq}\max \bigg\{0, \max_{n\in \{1,\dots,n^{\texttt{max}}\}}  \EE_{X\sim \mathrm{BB}(n,\alpha,\beta)}\Big[ 
    -c(n)\\
    &\quad + (\rho^{\texttt{A}} + \varepsilon\cdot(C+c(n)))\cdot\mathds{1}\{ f(\alpha,\beta)\cdot E(n,X) \ge 1/\kappa \}
    \Big] \bigg\}\\
    &= V^\varepsilon(\alpha,\beta,C,T)
\end{align*}
where in $(*)$ we have used that the function $f(\bullet,\bullet)$ defined in Eq.~\ref{eq:test-process} is non-decreasing in its first component\footnote{This is immediate to verify since $\theta^{\texttt{b}}\in (0,1) \implies \log(1+\theta^{\texttt{b}}(e-1)) \in (0,1)$.}. In step $(**)$ we have used that the distribution $\mathrm{BB}(\alpha',\beta,n)$ stochastically dominates (in the first-order sense) $\mathrm{BB}(\alpha,\beta,n)$ if $\alpha' \geq \alpha$.

\xhdr{Induction step $l+1 \to l$} Suppose now that  $V^\varepsilon(\alpha,\beta,C,l+1) \leq V^\varepsilon(\alpha',\beta,C,l+1)$ holds for a certain $l+1 \leq T$ and for all $\alpha,\beta,\alpha'$ such that $\alpha'\geq \alpha$, $f(\alpha,\beta)< 1/\kappa$ and $f(\alpha',\beta)< 1/\kappa$. Using Eq.~\ref{eq:app-bellman-value-optimality} again, note that:

\begin{align*}
    V^\varepsilon(\alpha',\beta,C,l) &=\max \left\{0, \max_{n\in \{1,\dots,n^{\texttt{max}}\}}  \underbrace{\EE_{X\sim \mathrm{BB}(n,\alpha',\beta)}\left[ r^\varepsilon(S',n,S'(n,X)) + V^\varepsilon(S'(n,X),l+1)  \right] }_{\dagger}\right\}
\end{align*}
Consider a given $n>0$ and expand the term $\dagger$ above:

\begin{align*}
    \dagger &= \EE_{X\sim \mathrm{BB}(n,\alpha',\beta)}\left[ r^\varepsilon(S',n,S'(n,X)) + V^\varepsilon(S'(n,X),l+1)  \right]\\
    &=\EE_{X\sim \mathrm{BB}(n,\alpha',\beta)}\Big[ 
    \mathds{1}\{ f(S'(n,X))\geq 1/\kappa\}
    \left( r^\varepsilon(S',n,S'(n,X)) + V^\varepsilon(S'(n,X),l+1) \right)\\
    &+
    \mathds{1}\{ f(S'(n,X))< 1/\kappa\}
    \left( r^\varepsilon(S',n,S'(n,X)) + V^\varepsilon(S'(n,X),l+1) \right)
    \Big]\\
    &\overset{(*)}{\geq}\EE_{X\sim \mathrm{BB}(n,\alpha',\beta)}\Big[ 
    \mathds{1}\{ f(S'(n,X))\geq 1/\kappa\}
    \left( r^\varepsilon(S',n,S'(n,X)) + V^\varepsilon(S'(n,X),l+1) \right)\\
    &+
    \mathds{1}\{ f(S'(n,X))< 1/\kappa\}
    \left( r^\varepsilon(S',n,S'(n,X)) + V^\varepsilon(S(n,X),l+1) \right)
    \Big]\\
    &\overset{(**)}{=}\EE_{X\sim \mathrm{BB}(n,\alpha',\beta)}\Big[ 
    \mathds{1}\{ f(S'(n,X))\geq 1/\kappa, f(S(n,X))\geq 1/\kappa\}\\
    &\quad \cdot \left( r^\varepsilon(S',n,S'(n,X)) + V^\varepsilon(S(n,X),l+1) \right)\\
    &+ \mathds{1}\{ f(S'(n,X))\geq 1/\kappa, f(S(n,X))< 1/\kappa\}\\
    &\quad \cdot \left( r^\varepsilon(S',n,S'(n,X)) + V^\varepsilon(S'(n,X),l+1) \right)\\
    &+
    \mathds{1}\{ f(S'(n,X))< 1/\kappa\}
    \left( r^\varepsilon(S',n,S'(n,X)) + V^\varepsilon(S(n,X),l+1) \right)
    \Big]\\
\end{align*}
where in $(*)$ we have used that if $f(S'(n,X))< 1/\kappa$, then $f(S(n,X))< 1/\kappa$ and hence the induction hypothesis applies $V^\varepsilon(S'(n,X),l+1) \geq V^\varepsilon(S(n,X),l+1)$. In $(**)$ we have used that if $f(S'(n,X)) \geq 1/\kappa$ and $f(S(n,X)) \geq 1/\kappa$, then $V^\varepsilon(S(n,X),l+1)=V^\varepsilon(S'(n,X),l+1)=0$.

Lastly, in $\dagger$, consider the case $f(S'(n,X)) \geq 1/\kappa$ and $f(S(n,X)) < 1/\kappa$. Then, 
\begin{equation*}
    r^\varepsilon(S',n,S'(n,X)) + V^\varepsilon(S'(n,X),l+1)= -c(n) + (\rho^{\texttt{A}} + \varepsilon\cdot(C+c(n)))
\end{equation*}
and,
\begin{equation*}
    r^\varepsilon(S,n,S(n,X)) + V^\varepsilon(S(n,X),l+1)= -c(n) + V^\varepsilon(S(n,X),l+1)\overset{\diamond}{\leq} -c(n) +\rho^{\texttt{A}}+\varepsilon\cdot(C+c(n))),
\end{equation*}
where $\diamond$ follows from Lemma~\ref{lemma:bound-value-function}.

Finally, noting that $r^\varepsilon(S',n,S'(n,X)) \geq r^\varepsilon(S,n,S(n,X))$ and using the first-order stochastic dominance for the Beta-Binomial again, we conclude that:
\begin{align*}
    \dagger &\geq \mathbb{E}_{X\sim\mathrm{BB}(n,\alpha',\beta)}\Big[ \mathds{1}\{f(S'(n,X))\geq 1/\kappa,\, f(S(n,X))\geq 1/\kappa\} \\
    &\hphantom{{}\geq \mathbb{E}_{X\sim\mathrm{BB}(n,\alpha',\beta)}\Big[\,}{}\cdot \big(r^\varepsilon(S,n,S(n,X)) + V^\varepsilon(S(n,X),l+1)\big) \\
    &\hphantom{{}\geq \mathbb{E}_{X\sim\mathrm{BB}(n,\alpha',\beta)}\Big[}{}+ \mathds{1}\{f(S'(n,X))\geq 1/\kappa,\, f(S(n,X))< 1/\kappa\} \\
    &\hphantom{{}\geq \mathbb{E}_{X\sim\mathrm{BB}(n,\alpha',\beta)}\Big[\,}{}\cdot \big(r^\varepsilon(S,n,S(n,X)) + V^\varepsilon(S(n,X),l+1)\big) \\
    &\hphantom{{}\geq \mathbb{E}_{X\sim\mathrm{BB}(n,\alpha',\beta)}\Big[}{}+ \mathds{1}\{f(S'(n,X))< 1/\kappa\} \\
    &\hphantom{{}\geq \mathbb{E}_{X\sim\mathrm{BB}(n,\alpha',\beta)}\Big[\,}{}\cdot \big(r^\varepsilon(S,n,S(n,X)) + V^\varepsilon(S(n,X),l+1)\big) \Big] \\
    &\geq \mathbb{E}_{X\sim\mathrm{BB}(n,\alpha,\beta)}\big[r^\varepsilon(S,n,S(n,X)) + V^\varepsilon(S(n,X),l+1)\big].
\end{align*}
That is,
\begin{multline}\label{eq:app-monotonicity-action-value}
    \EE_{X\sim \mathrm{BB}(n,\alpha',\beta)}\left[ r^\varepsilon(S',n,S'(n,X)) + V^\varepsilon(S'(n,X),l+1)  \right] \geq\\ \EE_{X\sim \mathrm{BB}(n,\alpha,\beta)}\left[ r^\varepsilon(S,n,S(n,X)) + V^\varepsilon(S(n,X),l+1)  \right],
\end{multline}
and taking maximum over the action $n$ and using Eq.~\ref{eq:app-bellman-value-optimality}:
\begin{equation*}
    V^\varepsilon(\alpha',\beta,C,l)  \geq V^\varepsilon(\alpha,\beta,C,l) .
\end{equation*}
This concludes the induction step and thus the proof of the monotonicity in the belief.

\subsubsection{Monotonicity on the cost}
We now focus on proving the monotonicity property in Proposition~\ref{prop:value-monotonicity} of $V^\varepsilon$ in the cumulated cost $C$. More precisely, fix a state $S=(\alpha,\beta,C)$ such that $f(\alpha,\beta) < 1/\kappa$ and consider $C' \geq C$. Denote $S'=(\alpha,\beta,C')$. We want to show that
\[\boxed{
    V^\varepsilon(\alpha,\beta,C,l) \leq V^\varepsilon(\alpha,\beta,C',l)
}
\]
for any $0 \leq l \leq T$. Again, we proceed by induction over $l$.

\xhdr{Base $l=T$}
Using Eq.~\ref{eq:app-bellman-value-optimality} we obtain:

\begin{align*}
    V^\varepsilon(\alpha,\beta,C',T) &= \max\!\left\{0,\, \max_{n\in\{1,\ldots,n^{\texttt{max}}\}} \mathbb{E}_{X\sim\mathrm{BB}(n,\alpha,\beta)}\!\left[ r^\varepsilon(S',n,S'(n,X)) \right] \right\} \\
    &= \max\bigg\{0,\, \max_{n\in\{1,\ldots,n^{\texttt{max}}\}} \mathbb{E}_{X\sim\mathrm{BB}(n,\alpha,\beta)}\Big[ -c(n) \\
    &\qquad {}+ \big(\rho^{\texttt{A}} + \varepsilon(C'+c(n))\big)\cdot\mathds{1}\{ f(\alpha,\beta)\cdot E(n,X) \ge 1/\kappa \} \Big]\bigg\} \\
    &\geq \max\bigg\{0,\, \max_{n\in\{1,\ldots,n^{\texttt{max}}\}} \mathbb{E}_{X\sim\mathrm{BB}(n,\alpha,\beta)}\Big[ -c(n) \\
    &\qquad {}+ \big(\rho^{\texttt{A}} + \varepsilon(C+c(n))\big)\cdot\mathds{1}\{ f(\alpha,\beta)\cdot E(n,X) \ge 1/\kappa \} \Big]\bigg\} \\
    &= V^\varepsilon(\alpha,\beta,C,T).
\end{align*}

\xhdr{Induction step $l+1 \to l$} Suppose now that  $V^\varepsilon(\alpha,\beta,C,l+1) \leq V^\varepsilon(\alpha,\beta,C',l+1)$ holds for a certain $l+1 \leq T$ and for all $\alpha,\beta$ such that $f(\alpha,\beta)< 1/\kappa$ and $0\leq C \leq C'$. Using Eq.~\ref{eq:app-bellman-value-optimality} again, note that:

\begin{align*}
    V^\varepsilon(\alpha,\beta,C',l) &=\max \left\{0, \max_{n\in \{1,\dots,n^{\texttt{max}}\}}  \underbrace{\EE_{X\sim \mathrm{BB}(n,\alpha,\beta)}\left[ r^\varepsilon(S',n,S'(n,X)) + V^\varepsilon(S'(n,X),l+1)  \right] }_{\dagger}\right\}
\end{align*}
Consider a given $n>0$ and expand the term $\dagger$ above:

\begin{align*}
    \dagger &= \EE_{X\sim \mathrm{BB}(n,\alpha,\beta)}\left[ r^\varepsilon(S',n,S'(n,X)) + V^\varepsilon(S'(n,X),l+1)  \right]\\
    &=\EE_{X\sim \mathrm{BB}(n,\alpha,\beta)}\Big[ 
    \mathds{1}\{ f(S'(n,X))\geq 1/\kappa\}
    \left( r^\varepsilon(S',n,S'(n,X)) + V^\varepsilon(S'(n,X),l+1) \right)\\
    &+
    \mathds{1}\{ f(S'(n,X))< 1/\kappa\}
    \left( r^\varepsilon(S',n,S'(n,X)) + V^\varepsilon(S'(n,X),l+1) \right)
    \Big]\\
    &\overset{(*)}{\geq}\EE_{X\sim \mathrm{BB}(n,\alpha,\beta)}\Big[ 
    \mathds{1}\{ f(S'(n,X))\geq 1/\kappa\}
    \left( r^\varepsilon(S',n,S'(n,X)) + V^\varepsilon(S'(n,X),l+1) \right)\\
    &+
    \mathds{1}\{ f(S'(n,X))< 1/\kappa\}
    \left( r^\varepsilon(S',n,S'(n,X)) + V^\varepsilon(S(n,X),l+1) \right)
    \Big]\\
    &\overset{(**)}{=}\EE_{X\sim \mathrm{BB}(n,\alpha,\beta)}\Big[ 
    \mathds{1}\{ f(S'(n,X))\geq 1/\kappa\}
    \left( r^\varepsilon(S',n,S'(n,X)) + V^\varepsilon(S(n,X),l+1) \right)\\
    &+
    \mathds{1}\{ f(S'(n,X))< 1/\kappa\}
    \left( r^\varepsilon(S',n,S'(n,X)) + V^\varepsilon(S(n,X),l+1) \right)
    \Big]\\
\end{align*}
where in $(*)$ we have used that if $f(S'(n,X))< 1/\kappa$, then the induction hypothesis applies $V^\varepsilon(S'(n,X),l+1) \geq V^\varepsilon(S(n,X),l+1)$ because the total cost in state $S'(n,X)$ is $C' +c(n) \geq C + c(n)$, which equals the total cost in $S(n,X)$. In $(**)$ we have used that if $f(S'(n,X)) \geq 1/\kappa$, then $f(S(n,X)) \geq 1/\kappa$ and hence $V^\varepsilon(S'(n,X),l+1)=V^\varepsilon(S(n,X),l+1)=0$.

Finally, noting that $r^\varepsilon(S',n,S'(n,X)) \geq r^\varepsilon(S,n,S(n,X))$, we conclude that:
\begin{align*}
    \dagger &\geq \EE_{X\sim \mathrm{BB}(n,\alpha,\beta)}\Big[ 
    \mathds{1}\{ f(S'(n,X))\geq 1/\kappa\}
    \left( r^\varepsilon(S,n,S(n,X)) + V^\varepsilon(S(n,X),l+1) \right)\\
    &+
    \mathds{1}\{ f(S'(n,X))< 1/\kappa\}
    \left( r^\varepsilon(S,n,S(n,X)) + V^\varepsilon(S(n,X),l+1) \right)
    \Big]\\
    &=\EE_{X\sim \mathrm{BB}(n,\alpha,\beta)}\left[ r^\varepsilon(S,n,S(n,X)) + V^\varepsilon(S(n,X),l+1)  \right]
\end{align*}
That is,
\begin{multline*}
    \EE_{X\sim \mathrm{BB}(n,\alpha,\beta)}\left[ r^\varepsilon(S',n,S'(n,X)) + V^\varepsilon(S'(n,X),l+1)  \right] \geq\\ \EE_{X\sim \mathrm{BB}(n,\alpha,\beta)}\left[ r^\varepsilon(S,n,S(n,X)) + V^\varepsilon(S(n,X),l+1)  \right],
\end{multline*}
and taking maximum over the action $n$ and using Eq.~\ref{eq:app-bellman-value-optimality}:
\begin{equation*}
    V^\varepsilon(\alpha,\beta,C,l)  \geq V^\varepsilon(\alpha,\beta,C',l) .
\end{equation*}
This concludes the induction step and thus the proof of the monotonicity in the cost.

\clearpage
\newpage

\subsection{Proof of Proposition~\ref{prop:opt-out-mono}}\label{app:proof-opt-out-mono}

Fix any subsidy level $\varepsilon\in [0,\varepsilon^{\texttt{max}}]$, total cost $C\geq0$, and an initial time step $0\leq l \leq T$.
Consider any state $S\neq S^{\texttt{out}}$ and write $S=(\alpha,\beta,C)$. Then, the Bellman optimality condition~\citep{Sutton1998} establishes that the action $n=0$ is optimal if and only if the value function coincides with the value of action $n=0$ (which is $0$ by definition of the reward function in Eq.~\ref{eq:MDP-reward-def}), that is,

\begin{equation*}
    V^\varepsilon(S,l) = 0,
\end{equation*}
where $V^\varepsilon(S,l)$ is the optimal value function defined in Eq.~\ref{eq:app-optimal-value}. Using Eq.~\ref{eq:app-bellman-value-optimality}, this is equivalent to the condition:

\begin{align*}
    &\max \left\{0, \underbrace{\max_{n\in \{1,\dots,n^{\texttt{max}}\}}  \EE_{X\sim \mathrm{BB}(n,\alpha,\beta)}\left[ r^\varepsilon(S,n,S(n,X)) + V^\varepsilon(S(n,X),l+1)  \right]}_{= H(\alpha,\beta,C)} \right\} = 0\\
    \iff& H(\alpha,\beta,C) \leq 0.
\end{align*}
Now, we observe that in proving Proposition~\ref{prop:value-monotonicity} in Appendix~\ref{app:proof-prop-value-monotonicity} (see Eq.~\ref{eq:app-monotonicity-action-value}), we precisely showed by induction that the function $H(\alpha,\beta, C)$ is non-decreasing in $\alpha$, and an analogous argument shows that it is non-increasing in $\beta$. We leverage this property in what follows.

For any $\alpha>0$, define the quantity
    \begin{equation*}
        \tilde{\beta}(\alpha) = \sup \{ \beta \,\colon\, H(\alpha,\beta,C)>0 \text{ and } \beta>0\}.
    \end{equation*}
As a consequence of the monotonicity of $H(\alpha,\beta,C)$, if $\alpha'\geq\alpha$, then
\begin{equation*}
        H(\alpha,\beta,C)>0 \implies H(\alpha',\beta,C)>0,
\end{equation*}
and thus the following set inclusion holds:
\begin{equation*}
        \{ \beta \,\colon\, H(\alpha,\beta,C)>0 \text{ and } \beta>0\} \subseteq \{ \beta \,\colon\, H(\alpha',\beta,C)>0 \text{ and } \beta>0\}.
\end{equation*}
Taking the supremum, the above implies:

\begin{equation*}
    \underbrace{\sup \{ \beta \,\colon\, H(\alpha,\beta,C)>0 \text{ and } \beta>0\} }_{\tilde{\beta}(\alpha)} \leq \underbrace{\sup \{ \beta \,\colon\, H(\alpha',\beta,C)>0 \text{ and } \beta>0\}}_{\tilde{\beta}(\alpha')},
\end{equation*}
which shows that $\tilde{\beta}(\alpha)$ is non-decreasing in $\alpha$.

Next, observe that:
\begin{itemize}
        \item If $\beta>\tilde{\beta}(\alpha)$, then by definition of the supremum we have that $H(\alpha,\beta,C)\leq 0$, and thus the action $n=0$ is optimal.
        \item If $\beta<\tilde{\beta}(\alpha)$, then again by the definition of the supremum, there exists $\beta'$ such that $\beta < \beta' \leq \tilde{\beta}(\alpha)$ with $H(\alpha,\beta',C)>0$. Then, using the monotonicity, $H(\alpha,\beta, C) \geq H(\alpha,\beta', C)>0$, which means that opting out is strictly sub-optimal: there exists an action $n>0$ that leads to a strictly higher expected value if taken at state $S$ and time $l$.
\end{itemize}
This proves the claim in Proposition~\ref{prop:opt-out-mono}.

\clearpage
\newpage

\subsection{Proof of Proposition~\ref{prop:linear-value}}\label{app:proof-linear-values}
Recall that the value function in the MDP $\Mcal^\varepsilon$ is defined by:
\begin{equation}\label{eq:app-value-function}
        V^\varepsilon_\pi(S,l)  = \EE_{\pi}\left[ \sum_{t=l}^T r^\varepsilon(S_t,n_t,S_{t+1}) \middle| S_l=S\right],
\end{equation}
with $r^\varepsilon$ defined by:
\begin{equation}\label{eq:app-reward-def}
r^\varepsilon(S,n,S') =
\begin{dcases}
    -c(n) + (\rho^{\texttt{A}} + \varepsilon\cdot(C+c(n)))\cdot\mathds{1}\{ f(S') \ge 1/\kappa \} & 
    \text{if } 0<f(S)<1/\kappa\\
    0 & \text{if } S = S^{\texttt{out}} \text{ or } f(S) \ge 1/\kappa.
\end{dcases}
\end{equation}

Firstly, note that all rewards in Eq.~\ref{eq:app-value-function} become $0$ as soon as $S_t = S^{\texttt{out}}$ or $f(S_t) \ge 1/\kappa$. This is because, from the transition dynamics in Eq.~\ref{eq:transition-dynamics}, any such state $S_t$ is absorbing, \ie, $S_{t+1}=S_t$ if $S_t = S^{\texttt{out}}$ or $f(S_t) \geq 1/\kappa$, and satisfies $r^\varepsilon(S_t,n_t,S_{t+1})=0$. Then, the terms in Eq.~\ref{eq:app-value-function} that are non-zero correspond to time steps $t \leq \tau$, where $\tau$ is the stopping time defined by:

\begin{equation*}
    \tau = T \wedge \min \{t\in\{l,\dots,T \} \colon n_t = 0 \,\text{or}\, f(S_{t+1}) \geq 1/\kappa\}.
\end{equation*}
For $l\leq t \leq \tau $, the first case of the reward function in Eq.~\ref{eq:app-reward-def} applies:
\begin{align*}
V_\pi^\varepsilon(S,l)
&= \mathbb{E}_\pi\!\left[ \sum_{t=l}^{\tau} \left( -c(n_t) + (\rho^{\texttt{A}} + \varepsilon C_{t+1}) \mathds{1}\{f(S_{t+1}) \ge 1/\kappa\} \right) \;\middle|\; S_l = S \right] \\
&= \mathbb{E}_\pi\Bigg[ \sum_{t=l}^{\tau}\big({-}c(n_t) + \rho^{\texttt{A}} \mathds{1}\{f(S_{t+1}) \ge 1/\kappa\}\big) \\
&\hphantom{{}= \mathbb{E}_\pi\Bigg[}{} + \varepsilon \sum_{t=l}^{\tau} C_{t+1} \mathds{1}\{f(S_{t+1}) \ge 1/\kappa\} \;\Bigg|\; S_l = S \Bigg]
\end{align*}
Note that the first part of the expectation corresponds exactly to the value function under no subsidy, $V_\pi^0(S,l)$, and thus,

\begin{equation*}
    V_\pi^\varepsilon(S,l)  = V_\pi^0(S,l) + \varepsilon\cdot \mathbb{E}_\pi \left[ \sum_{t=l}^{\tau} C_{t+1} \cdot \mathds{1}\{f(S_{t+1}) \ge 1/\kappa\} \;\middle|\; S_l = S \right]
\end{equation*}
Now, define:
\begin{align}\label{eq:rejection-cost-conditional}\nonumber
        A_{\pi}(S,l) &= \mathbb{E}_\pi \left[ \sum_{t=l}^{\tau} C_{t+1} \cdot \mathds{1}\{f(S_{t+1}) \ge 1/\kappa\} \;\middle|\; S_l = S \right]\\
        &= \mathbb{E}_\pi \left[ C_{\tau+1} \cdot \mathds{1}\{f(S_{\tau+1}) \ge 1/\kappa\} \;\middle|\; S_l = S \right]
\end{align}
This term corresponds to the (expected) total cost incurred by the agent conditional on the product being approved, when the MDP starts from state $S$ at time $l$.\footnote{We use the term ``conditional on approval'' informally and for didactic purposes, since $A_\pi(S,l)$ is the quantity that naturally appears in the value decomposition $V_\pi^\varepsilon(S,l) = V_\pi^0(S,l) + \varepsilon \cdot A_\pi(S,l)$, representing the expected subsidy paid by the principal. Strictly speaking, however, $A_\pi(S,l)$ is the expected cost weighted by the indicator of approval, rather than a conditional expectation in the measure-theoretic sense.}
Indeed, by the definition of the stopping time $\tau$, approval (\ie, $f(S_{t+1}) \ge 1/\kappa$) can only happen at exactly $t = \tau$. If the agent opts out or the horizon $T$ is reached without approval, the indicator $\mathds{1}\{f(S_{t+1}) \ge 1/\kappa\}$ is $0$ for all $t$. 
We conclude that 
\begin{equation*}
    V_\pi^\varepsilon(S,l)  = V_\pi^0(S,l) + \varepsilon \cdot A_{\pi}(S,l).
\end{equation*}

\clearpage
\newpage

\subsection{Proof of Proposition~\ref{prop:utility-piecewise-convex}}\label{app:proof-utility-piecewise-convex}

Denote by $\Pi^{\texttt{r}} \subset \Pi$ the set of all deterministic policies that select action $n=0$ at $S^{\texttt{out}}$, any state such that $f(S)\geq 1/\kappa$, or any state such that $S \notin S^{\texttt{r}}$, and observe that $\Pi^{\texttt{r}}$ is finite. Recall from Proposition~\ref{prop:linear-value} that, for any policy $\pi \in \Pi^{\texttt{r}}$,
\begin{equation}
    V^{\varepsilon}_{\pi}(S,l) = V^{0}_{\pi}(S,l) + \varepsilon \cdot A_{\pi}(S,l).
\end{equation}
In particular, at the initial state $S_0 = (\alpha_0, \beta_0,0)$ and time $l=0$, we have
\begin{equation}
    \Bar{U}^{\texttt{A}}(\pi;\varepsilon)
    = V^{\varepsilon}_{\pi}(S_0,0)
    = V^{0}_{\pi}(S_0,0) + \varepsilon \cdot A_{\pi}(S_0,0).
\end{equation}

Since $\Pi^{\texttt{r}}$ is finite, for every $(S,l)\in \Scal^{\texttt{r}}\times [T]$ the optimal value function
\begin{equation}
    V^{\varepsilon}(S,l)
    = \max_{\pi \in \Pi^{\texttt{r}}} \big\{ V^0_{\pi}(S,l) + \varepsilon \cdot A_\pi(S,l) \big\}
\end{equation}
is the point-wise maximum of finitely many affine functions in $\varepsilon$. It is well-known that such a point-wise maximum is convex, continuous, and piecewise linear~\cite{rockafellar1970convex}. Here, we particularize to our problem, with the goal of constructing a single partition of $[0,\varepsilon^{\texttt{max}}]$ on which a single optimal policy (for $\Mcal^\varepsilon$ at every state $S \in \Scal^{\texttt{r}}$ and time $l \in [T]$ simultaneously) is optimal in each interval of the partition.

For any two distinct policies $\pi, \pi' \in \Pi^{\texttt{r}}$ and any $(S,l) \in \Scal^{\texttt{r}}\times [T]$, consider the difference
\begin{equation*}
    (V^0_{\pi}(S,l) - V^0_{\pi'}(S,l)) + \varepsilon (A_\pi(S,l) - A_{\pi'}(S,l)).
\end{equation*}
If $(V^0_{\pi}(S,l), A_\pi(S,l)) = (V^0_{\pi'}(S,l), A_{\pi'}(S,l))$ for all $(S,l)$, then the two policies yield identical values for all $\varepsilon$ and all $(S,l)$, and we may break ties arbitrarily and retain only one of them. 
Similarly, if $V^0_{\pi}(S,l) + \varepsilon \cdot A_\pi(S,l) \ge V^0_{\pi'}(S,l) + \varepsilon \cdot A_{\pi'}(S,l)$ for all $\varepsilon \in [0,\varepsilon^{\texttt{max}}]$ and all $(S,l)$, then policy $\pi'$ can be removed without changing the optimal value function at any $(S,l)$. Thus, without loss of generality, we restrict our attention to a subset $\tilde{\Pi}^{\texttt{r}} \subset \Pi^{\texttt{r}}$ such that for any distinct $\pi, \pi' \in \tilde{\Pi}^{\texttt{r}}$ and any $(S,l)$, the corresponding affine functions intersect exactly once in $[0,\varepsilon^{\texttt{max}}]$, and each policy is optimal for some value of $\varepsilon$.

Then, let
\begin{multline*}
    \Kcal = \big\{ \varepsilon \in [0,\varepsilon^{\texttt{max}}] : \exists\, \pi \neq \pi' \in \tilde{\Pi}^{\texttt{r}},\, \exists\, (S,l) \in \Scal^{\texttt{r}}\times [T] \\ \text{ such that } V^0_{\pi}(S,l) + \varepsilon \cdot A_\pi(S,l) = V^0_{\pi'}(S,l) + \varepsilon \cdot A_{\pi'}(S,l) \big\}.
\end{multline*}
Since $\tilde{\Pi}^{\texttt{r}}$, $\Scal^{\texttt{r}}$, and $[T]$ are all finite, and each quadruple $(\pi,\pi',S,l)$ contributes at most one point to $\Kcal$, the set $\Kcal$ is finite. Ordering its elements and adding the endpoints if necessary, we obtain a partition
\begin{equation*}
    0 = \varepsilon_0 < \varepsilon_1 < \cdots < \varepsilon_L = \varepsilon^{\texttt{max}}.
\end{equation*}

By construction, no two affine functions $\varepsilon \mapsto V^0_{\pi}(S,l) + \varepsilon \cdot A_\pi(S,l)$ intersect in any open interval $(\varepsilon_i, \varepsilon_{i+1})$ at any $(S,l)$. Hence, for every $(S,l)$, the ordering of $\{V^0_{\pi}(S,l) + \varepsilon \cdot A_\pi(S,l)\}_{\pi \in \tilde{\Pi}^{\texttt{r}}}$ is constant on each such interval. It follows that there exists a single policy $\pi_i \in \tilde{\Pi}^{\texttt{r}}$ that is optimal for $\Mcal^\varepsilon$ at every $(S,l)$ simultaneously, and such that
\begin{equation*}
    V^{\varepsilon}(S,l)
    = V^0_{\pi_i}(S,l) + \varepsilon \cdot A_{\pi_i}(S,l)
    \quad \text{for all } \varepsilon \in [\varepsilon_i, \varepsilon_{i+1}) \text{ and all } (S,l).
\end{equation*}
In particular, specializing to $(S,l)=(S_0,0)$,
\begin{equation*}
    \Bar{U}^{\texttt{A}}(\pi^\varepsilon;\varepsilon)
    = V^0_{\pi_i}(S_0,0) + \varepsilon \cdot A_{\pi_i}(S_0,0)
    \quad \text{for all } \varepsilon \in [\varepsilon_i, \varepsilon_{i+1}).
\end{equation*}
This establishes the result for $\Bar{U}^{\texttt{A}}(\pi^\varepsilon;\varepsilon)$. Note that $\EE_{(\alpha_0,\beta_0)\sim Q}[\Bar{U}^{\texttt{A}}(\pi^\varepsilon;\varepsilon)]$ is then also convex and continuous because it is an average of convex continuous functions that are uniformly bounded (observe that by Lemma~\ref{lemma:bound-value-function}, we have $0 \leq \Bar{U}^{\texttt{A}}(\pi^\varepsilon;\varepsilon)\leq \rho^{\texttt{A}}$ uniformly for any initial belief parameters $(\alpha_0,\beta_0)$, where the lower bound follows from the fact that the agent can always opt out at no cost); continuity of the expectation then follows from the dominated convergence theorem. Lastly, in each interval $[\varepsilon_i, \varepsilon_{i+1})$, $\Bar{U}^{\texttt{A}}(\pi^\varepsilon;\varepsilon)$ is linear, which implies that $\EE_{(\alpha_0,\beta_0)\sim Q}[\Bar{U}^{\texttt{A}}(\pi^\varepsilon;\varepsilon)]$ is also linear:

\begin{equation*}
    \EE_{(\alpha_0,\beta_0)\sim Q}[\Bar{U}^{\texttt{A}}(\pi^\varepsilon;\varepsilon)]
    = \EE_{(\alpha_0,\beta_0)\sim Q}[V^0_{\pi_i}(S_0,0)] + \varepsilon \cdot \EE_{(\alpha_0,\beta_0)\sim Q}[A_{\pi_i}(S_0,0)]
    \quad \text{for all } \varepsilon \in [\varepsilon_i, \varepsilon_{i+1}).
\end{equation*}
This concludes the proof.

\clearpage
\newpage

\subsection{Proof of Proposition~\ref{prop:regulator-utility-linear}}

Recall that the social utility $\Bar{U}^{\texttt{S}}(\varepsilon;\pi)$ is defined as (Eq.~\ref{eq:adaptive-utility-regulator}):
\begin{multline*}
    \Bar{U}^{\texttt{S}}(\varepsilon;\pi)
    = \mathbb{E}_{\pi, (\alpha_0,\beta_0)\sim Q} \Bigg[
    \sum_{t=0}^T \left( \rho^{\texttt{S}} - \varepsilon \cdot C_{t+1} \right) \cdot
    \mathds{1}\left\{ 0 < f(S_t) < 1/\kappa \leq f(S_{t+1}) \right\}
    \;\Bigg|\; S_0 = (\alpha_0,\beta_0,0)
    \Bigg].
\end{multline*}
Let $\tau$ be the stopping time defined by:

\begin{equation*}
    \tau = T \wedge \min \{t\in\{0,\dots,T \} \colon n_t = 0 \,\text{ or }\, f(S_{t+1}) \geq 1/\kappa\},
\end{equation*}
that is, the last time step before reaching an absorbing state---either $S^{\texttt{out}}$, for which $f(S^{\texttt{out}})=0$, or any state $S$ such that $f(S) \geq 1/\kappa$. Then, using the linearity of the expectation:

\begin{align*}
    &\Bar{U}^{\texttt{S}}(\varepsilon;\pi) \\
    &= \mathbb{E}_{\pi, (\alpha_0,\beta_0)\sim Q} \left[ \sum_{t=0}^T \left( \rho^{\texttt{S}} - \varepsilon \cdot C_{t+1} \right) \cdot\mathds{1}\left\{ 0 < f(S_t) < 1/\kappa \leq f(S_{t+1}) \right\}\;\middle|\; S_0 = (\alpha_0,\beta_0,0) \right]\\
    &=\rho^{\texttt{S}}\cdot \mathbb{E}_{\pi, (\alpha_0,\beta_0)\sim Q} \left[ \sum_{t=0}^T \mathds{1}\left\{ 0 < f(S_t) < 1/\kappa \leq f(S_{t+1}) \right\}\;\middle|\; S_0 = (\alpha_0,\beta_0,0) \right]\\
    &\quad -\varepsilon \cdot \mathbb{E}_{\pi, (\alpha_0,\beta_0)\sim Q} \left[ \sum_{t=0}^\tau  C_{t+1} \cdot\mathds{1}\left\{ 0 < f(S_t) < 1/\kappa \leq f(S_{t+1}) \right\}\;\middle|\; S_0 = (\alpha_0,\beta_0,0) \right]\\
    &=\rho^{\texttt{S}}\cdot \mathbb{E}_{ (\alpha_0,\beta_0)\sim Q} \mathbb{E}_{\pi} \left[  \mathds{1}\left\{ \exists t \in [T] \colon 0 < f(S_t) < 1/\kappa \leq f(S_{t+1}) \right\}\;\middle|\; S_0 = (\alpha_0,\beta_0,0) \right]\\
    &\quad- \varepsilon \cdot \mathbb{E}_{ (\alpha_0,\beta_0)\sim Q} \mathbb{E}_{\pi}  \left[ C_{\tau+1} \cdot\mathds{1}\left\{ 1/\kappa \leq f(S_{\tau+1}) \right\}\;\middle|\; S_0 = (\alpha_0,\beta_0,0) \right]\\
    &=\rho^{\texttt{S}}\cdot \mathbb{E}_{ (\alpha_0,\beta_0)\sim Q} \left[P_{\pi} \left( \exists t \in [T] \colon 0 < f(S_t) < 1/\kappa \leq f(S_{t+1}) \;\middle|\; S_0 = (\alpha_0,\beta_0,0) \right) \right]\\
    &\quad- \varepsilon \cdot \mathbb{E}_{ (\alpha_0,\beta_0)\sim Q} \mathbb{E}_{\pi}  \left[ C_{\tau+1} \cdot\mathds{1}\left\{ 1/\kappa \leq f(S_{\tau+1}) \right\}\;\middle|\; S_0 = (\alpha_0,\beta_0,0) \right]\\
    &=\rho^{\texttt{S}}\cdot \mathbb{E}_{ (\alpha_0,\beta_0)\sim Q} \left[P_{\pi} \left( \exists t \in [T] \colon 1/\kappa \leq f(S_{t+1}) \;\middle|\; S_0 = (\alpha_0,\beta_0,0) \right) \right]\\
    &\quad- \varepsilon \cdot \mathbb{E}_{ (\alpha_0,\beta_0)\sim Q}  \left[ A_{\pi}(\alpha_0,\beta_0,0,0) \right]
\end{align*}
where we have used the definition of $A_{\pi}$ in Eq.~\ref{eq:rejection-cost-conditional}, and the fact that the condition $0 < f(S_t) < 1/\kappa \leq f(S_{t+1})$ can occur at most once at time step $t=\tau$.

\clearpage
\newpage

\subsection{Proof of Proposition~\ref{prop:algorithm}}
In this section, we will show that Algorithm~\ref{alg:epsilon} recovers (in a finite number of steps) the partition of the interval $[0, \varepsilon^{\texttt{max}}]$ given by Proposition~\ref{prop:utility-piecewise-convex}, which we denote by $\mathcal{P}$:\footnote{We denote by $\mathcal{P}(\varepsilon)$ the interval of $\mathcal{P}$ containing $\varepsilon$. }

\begin{equation*}
    \mathcal{P} = \{ 0=\varepsilon_0 < \varepsilon_1 < \dots < \varepsilon_L = \varepsilon^{\texttt{max}} \},
\end{equation*}
where for each interval $[\varepsilon_i, \varepsilon_{i+1})$, there exists a (deterministic) policy $\pi_i$ that is optimal for any $\varepsilon \in [\varepsilon_i, \varepsilon_{i+1})$, that is (writing $S_0 = (\alpha_0, \beta_0,0)$),

\begin{equation*}
\begin{dcases}
    \Bar{U}^{\texttt{A}}(\pi^\varepsilon;\varepsilon)
    = V^0_{\pi_i}(S_0,0) + \varepsilon \cdot A_{\pi_i}(S_0,0)
    &\, \forall \varepsilon \in [\varepsilon_i, \varepsilon_{i+1})\\
    \EE_{(\alpha_0,\beta_0)\sim Q}[\Bar{U}^{\texttt{A}}(\pi^\varepsilon;\varepsilon)]
    = \EE_{(\alpha_0,\beta_0)\sim Q}[V^0_{\pi_i}(S_0,0)] + \varepsilon \cdot \EE_{(\alpha_0,\beta_0)\sim Q}[A_{\pi_i}(S_0,0)]
    &\, \forall \varepsilon \in [\varepsilon_i, \varepsilon_{i+1}).
\end{dcases}
\end{equation*}
Since the above partition does not depend on the belief $Q$ of the principal (Proposition~\ref{prop:utility-piecewise-convex}), we assume without loss of generality that the principal knows the agent's initial belief $\mathrm{Beta}(\alpha_0, \beta_0)$. Consequently, the expectations $\EE_{(\alpha_0,\beta_0)\sim Q}[\bullet]$ simply correspond to evaluating the integrand at the true belief parameters $(\alpha_0, \beta_0)$. We begin by proving the following lemma.\footnote{For convenience, we use $L$, $R$ and $int$ as subscripts instead of using $\pi_L$, $\pi_R$ and $\pi_{int}$.}

\begin{lemma}\label{lemma:convex}
Let $\varepsilon_L < \varepsilon_R$ be two subsidy levels with, respectively, optimal policies $\pi_L$ and $\pi_R$, and value functions (evaluated at the initial state $(\alpha_0, \beta_0,0)$ and initial time step $l=0$, which we omit for notational convenience) $V^0_L + \varepsilon \cdot A_L$ and $V^0_R + \varepsilon \cdot A_R$ (see Proposition~\ref{prop:linear-value}). The following holds:
\begin{enumerate}
    \item If $A_L = A_R$, then $\Bar{U}^{\texttt{A}}(\pi^\varepsilon;\varepsilon) = V^0_L + \varepsilon \cdot A_L$ for all $\varepsilon \in [\varepsilon_L, \varepsilon_R]$.
    \item If $A_L < A_R$, let $\varepsilon_{int} = \frac{V^0_L - V_R^0}{A_R - A_L}$. Then, $\Bar{U}^{\texttt{A}}(\pi^\varepsilon;\varepsilon) = \max(V^0_L + \varepsilon \cdot A_L, V^0_R + \varepsilon \cdot A_R)$ for all $\varepsilon \in [\varepsilon_L, \varepsilon_R]$ if and only if $\Bar{U}^{\texttt{A}}(\pi^{\varepsilon_{int}};\varepsilon_{int}) = V^0_L + \varepsilon_{int} \cdot A_L$.
\end{enumerate}
\end{lemma}

\begin{proof}
For part 1., if $A_L = A_R$, then since $\Bar{U}^{\texttt{A}}(\pi^\varepsilon;\varepsilon)$ is convex, its subgradient must be non-decreasing in $\varepsilon$. Thus, for any $\varepsilon \in (\varepsilon_L, \varepsilon_R)$, we must have $A_L \le A_{\pi^\varepsilon} \le A_R$, which implies $A_{\pi^\varepsilon} = A_L$. By continuity and the fact that $\Bar{U}^{\texttt{A}}(\pi^\varepsilon;\varepsilon)$ is the point-wise maximum of affine functions, it follows that $V^0_{\pi^\varepsilon} = V^0_L = V^0_R$, and the value function is a single affine segment on this interval.

For part 2., the direct implication follows because if $\Bar{U}^{\texttt{A}}(\pi^\varepsilon;\varepsilon) = \max(V^0_L + \varepsilon \cdot A_L, V^0_R + \varepsilon \cdot A_R)$ for all $\varepsilon \in [\varepsilon_L, \varepsilon_R]$, then evaluating at $\varepsilon_{int}$, and since $V^0_L + \varepsilon_{int}\cdot A_L =V^0_R + \varepsilon_{int}\cdot A_R$ by definition, we obtain $\Bar{U}^{\texttt{A}}(\pi^{\varepsilon_{int}};\varepsilon_{int}) = V^0_L + \varepsilon_{int} \cdot A_L$.
For the backward direction, we argue by contradiction and suppose $\Bar{U}^{\texttt{A}}(\pi^{\varepsilon_{int}};\varepsilon_{int}) = V^0_L + \varepsilon_{int} \cdot A_L$ but there exists some $\varepsilon' \in (\varepsilon_L, \varepsilon_R)$ and a policy $\pi'$ such that $V^0_{\pi'} + \varepsilon'\cdot A_{\pi'} > \max(V^0_L + \varepsilon' \cdot A_L, V^0_R + \varepsilon' \cdot A_R)$. 
Assume without loss of generality $\varepsilon' \leq \varepsilon_{int}$ (the argument for the case $\varepsilon' \geq \varepsilon_{int}$ is symmetric, with $\pi_R$ in place of $\pi_L$). Since $\pi_L$ is optimal at $\varepsilon_L$, we have $V^0_L + \varepsilon_L \cdot A_L \geq V^0_{\pi'} + \varepsilon_L \cdot A_{\pi'}$. Combined with the assumption $V^0_{\pi'} + \varepsilon' \cdot A_{\pi'} > V^0_L + \varepsilon' \cdot A_L$ and subtracting, we obtain $(\varepsilon' - \varepsilon_L)(A_{\pi'} - A_L) > 0$, which implies $A_{\pi'} > A_L$ since $\varepsilon' > \varepsilon_L$. Then, at $\varepsilon_{int} \geq \varepsilon'$,
\begin{align*}
V^0_{\pi'} + \varepsilon_{int} \cdot A_{\pi'} &= (V^0_{\pi'} + \varepsilon' \cdot A_{\pi'}) + (\varepsilon_{int} - \varepsilon') \cdot A_{\pi'}\\
&> (V^0_L + \varepsilon' \cdot A_L) + (\varepsilon_{int} - \varepsilon') \cdot A_L\\
&= V^0_L + \varepsilon_{int} \cdot A_L,
\end{align*}
where the strict inequality uses $V^0_{\pi'} + \varepsilon' \cdot A_{\pi'} > V^0_L + \varepsilon' \cdot A_L$ together with $A_{\pi'} > A_L$ and $\varepsilon_{int} \geq \varepsilon'$. This contradicts the assumption that $\bar{U}^{\texttt{A}}(\pi^{\varepsilon_{int}}; \varepsilon_{int}) = V^0_L + \varepsilon_{int} \cdot A_L$.
\end{proof}

To prove Proposition~\ref{prop:algorithm}, we begin by showing that any point added to the set $\Ucal$ in Algorithm~\ref{alg:epsilon} corresponds to a point of $\Pcal$. Consider any iteration of the algorithm where $V^0_{int} +\varepsilon_{int} \cdot A_{int} \leq V^0_L + \varepsilon_{int}\cdot A_L$ and $A_L \neq A_R$. Then, from Lemma~\ref{lemma:convex}, it follows that in this case, for any possible subsidy $\varepsilon\in [\varepsilon_L, \varepsilon_R]$, the optimal value function is given by $\Bar{U}^{\texttt{A}}(\pi^\varepsilon;\varepsilon) = \max( V^0_L + \varepsilon\cdot A_L, V^0_R + \varepsilon\cdot A_R)$, with a change in slope at their intersection, \ie, at $\varepsilon_{int}$. That is, $\varepsilon_{int}$ is a point of the partition $\mathcal{P}$, and the value $\Bar{U}^{\texttt{S}}(\varepsilon_{int};\pi_R) = \max_{\varepsilon\in \mathcal{P}(\varepsilon_{int})} \Bar{U}^{\texttt{S}}(\varepsilon;\pi^{\varepsilon})$ is stored in the set $\Ucal$.  

Reciprocally, consider any iteration of the algorithm where $V^0_{int} +\varepsilon_{int}\cdot A_{int} > V^0_L + \varepsilon_{int}\cdot A_L$. Then, Lemma~\ref{lemma:convex} implies that $\mathcal{P}(\varepsilon_L) \neq \mathcal{P}(\varepsilon_{int})$ and $\mathcal{P}(\varepsilon_R) \neq \mathcal{P}(\varepsilon_{int})$. That is, $\mathcal{P}(\varepsilon_{int})$ is a new interval in the partition $\mathcal{P}$ where the optimal value function is given by the linear component $V^0_{\pi_{int}} + \varepsilon \cdot A_{\pi_{int}}$.
The algorithm pushes the two sub-intervals $[\varepsilon_L, \varepsilon_{int}]$ and $[\varepsilon_{int}, \varepsilon_R]$ onto the stack (line~\ref{line:new-intervals}), along with their respective linear components $V^0_L + \varepsilon \cdot A_L$, $V^0_{\pi_{int}} + \varepsilon \cdot A_{\pi_{int}}$ and $V^0_R + \varepsilon \cdot A_R$. Since by Proposition~\ref{prop:utility-piecewise-convex}, there are a finite number of such linear components (or equivalently, $\mathcal{P}$ is finite), Algorithm~\ref{alg:epsilon} terminates in a finite number of steps.

To conclude, we argue that, upon termination, the set $\Ucal$ obtained from Algorithm~\ref{alg:epsilon} contains every point of $\mathcal{P}$. 
%
Let $N$ denote the total number of pop operations performed by the algorithm before $\mathcal{I}$ becomes empty (finite by the argument above). For $k = 0, 1, \dots, N$, denote by $\mathcal{I}_k$ the state of the stack after $k$ pops, and by $\mathcal{D}_k \subseteq \mathcal{P}$ the set of points of $\mathcal{P}$ found by the algorithm during the first $k$ iterations, with $\mathcal{D}_0=\{0,  \varepsilon^{\texttt{max}}\}$. Observe first that whenever Algorithm~\ref{alg:epsilon} pushes two new intervals onto $\Ical$ , the intersection point $\varepsilon_{int}$ lies strictly interior to the interval $\mathcal{P}(\varepsilon_{int})$ (since by Lemma~\ref{lemma:convex}, $\mathcal{P}(\varepsilon_L) \neq \mathcal{P}(\varepsilon_{int}) \neq \mathcal{P}(\varepsilon_R)$, and $\varepsilon_L < \varepsilon_{int} < \varepsilon_R$), so $\varepsilon_{int} \notin \mathcal{P}$. Consequently, every endpoint of an interval ever pushed onto $\mathcal{I}_k$ is either in $\mathcal{D}_k$ or not in $\mathcal{P}$. In light of this, we prove by induction on $k$ the following condition:\footnote{With a slight abuse of notation, we identify the elements of $\Ical$ with intervals, meaning that if $(\varepsilon_L, \pi_L, \Bar{V}^0_L, \Bar{A}_L, \varepsilon_R, \pi_R, \Bar{V}^0_R, \Bar{A}_R)\in \Ical$, we write that $[\varepsilon_L, \varepsilon_R] \in \Ical$.}
\begin{equation}\label{eq:stack-invariant}
    \mathcal{P} \;\subseteq\; \mathcal{D}_k \;\cup\; \bigcup_{[\varepsilon_L, \varepsilon_R] \in \mathcal{I}_k} (\varepsilon_L, \varepsilon_R).
\end{equation}

\xhdr{Base case $k=0$} $\mathcal{I}_0 = \{[0, \varepsilon^{\texttt{max}}]\}$ and $\mathcal{D}_0 = \{0, \varepsilon^{\texttt{max}}\}$. Every $\varepsilon \in \mathcal{P}$ is either in $\{0, \varepsilon^{\texttt{max}}\} = \mathcal{D}_0$ or in $(0, \varepsilon^{\texttt{max}})$, so  the condition in Eq.~\ref{eq:stack-invariant} holds for $k=0$.

\xhdr{Inductive step $k \to k+1$} Suppose the condition in Eq.~\ref{eq:stack-invariant} holds and consider the $(k+1)$-th pop of some $[\varepsilon_L, \varepsilon_R] \in \mathcal{I}_k$.

\begin{itemize}
    \item  If $V^0_{int} + \varepsilon_{int} \cdot A_{int} \leq V^0_L + \varepsilon_{int} \cdot A_L$, then $\varepsilon_{int}$ is added to $\mathcal{D}_{k+1}$, and the popped interval is not replaced. By Lemma~\ref{lemma:convex}, $(\varepsilon_L, \varepsilon_R) \cap \mathcal{P} \subseteq \{\varepsilon_{int}\} \subseteq \mathcal{D}_{k+1}$. Hence any $\varepsilon \in \mathcal{P}$ previously covered by $(\varepsilon_L, \varepsilon_R)$ is now in $\mathcal{D}_{k+1}$. The endpoints $\varepsilon_L$ and $\varepsilon_R$, if they belong to $\mathcal{P}$, were endpoints of the popped interval and hence are already contained in $\mathcal{D}_k \subseteq \mathcal{D}_{k+1}$. Thus, the condition in Eq.~\ref{eq:stack-invariant} holds for $k+1$.

    \item If $V^0_{int} +\varepsilon_{int}\cdot A_{int} > V^0_L + \varepsilon_{int}\cdot A_L$, then $[\varepsilon_L, \varepsilon_R]$ is replaced by $[\varepsilon_L, \varepsilon_{int}]$ and $[\varepsilon_{int}, \varepsilon_R]$, while $\mathcal{D}_{k+1} = \mathcal{D}_k$. If $\varepsilon_{int} \in \Dcal_k$, Eq.~\ref{eq:stack-invariant} holds for $k+1$, and if $\varepsilon_{int} \notin \mathcal{P}$, we have that $(\varepsilon_L, \varepsilon_R) \cap \mathcal{P} = \big[(\varepsilon_L, \varepsilon_{int}) \cup (\varepsilon_{int}, \varepsilon_R)\big] \cap \mathcal{P}$, and the condition in Eq.~\ref{eq:stack-invariant} also holds for $k+1$.
\end{itemize}

Lastly, at the final iteration, $\mathcal{I}_N = \emptyset$ implies $\mathcal{P} \subseteq \mathcal{D}_N$; combined with $\mathcal{D}_N \subseteq \mathcal{P}$ by construction, $\mathcal{D}_N = \mathcal{P}$. Since every $\varepsilon \in \mathcal{D}_N$ has its social utility saved in $\mathcal{U}$ (at line~\ref{line:initial-interval} for the initial endpoints and line~\ref{line:vertex-found} for each new point of $\Pcal$ found), $\mathcal{U}$ contains $\{(\varepsilon, \bar{U}^{\texttt{S}}(\varepsilon; \pi^\varepsilon)) : \varepsilon \in \mathcal{P}\}$. The maximizer of $\bar{U}^{\texttt{S}}$ over $[0, \varepsilon^{\texttt{max}}]$ is attained at some left endpoint $\varepsilon_i \in \mathcal{P}$, and $\arg\max_{(\varepsilon, u) \in \mathcal{U}} u$ returns it.

\clearpage
\newpage

\section{Additional Experimental Details}\label{app:experimental-details}
The complete code used for our experiments, including the implementation of Algorithm~\ref{alg:value_iteration_decomposition} and Algorithm~\ref{alg:epsilon}, is available as supplementary material. We will publicly release it with the final version of the paper.

\xhdr{Hardware setup} Our experiments are executed on a compute server equipped with 2 $\times$ Intel Xeon Gold 5317 CPU, $1{,}024$ GB main memory, and $2$ $\times$ H100 NVIDIA GPU ($80$ GB, Hopper Architecture). In each experiment, a single Nvidia H100 GPU is used.

\xhdr{Software setup}
All experiments are implemented in \texttt{Python} 3.13.5 using \texttt{PyTorch} 2.1.1 and \texttt{NumPy} 2.4.4. Computations are performed on an NVIDIA GPU with \texttt{CUDA} 13.0 support.

\xhdr{Runtime} For the setting used in Section~\ref{sec:experiments}, Algorithm~\ref{alg:value_iteration_decomposition} computes the optimal policy of the agent (for any given subsidy) in $\sim 2.4$s, and Algorithm~\ref{alg:epsilon} computes the optimal subsidy in $\sim 271$s, which involves solving $114$ times a different belief MDP $\Mcal^\varepsilon$. In Figure~\ref{fig:runtime}, we evaluate the runtime of Algorithm~\ref{alg:epsilon} across multiple configurations for the parameters $T$ and $n^{\texttt{max}}$, which determine the size of the state and action space (Proposition~\ref{prop:reacheable-set}).

\begin{figure}[h]
    \vspace{0cm}
    \centering
    \includegraphics[width=0.95\linewidth]{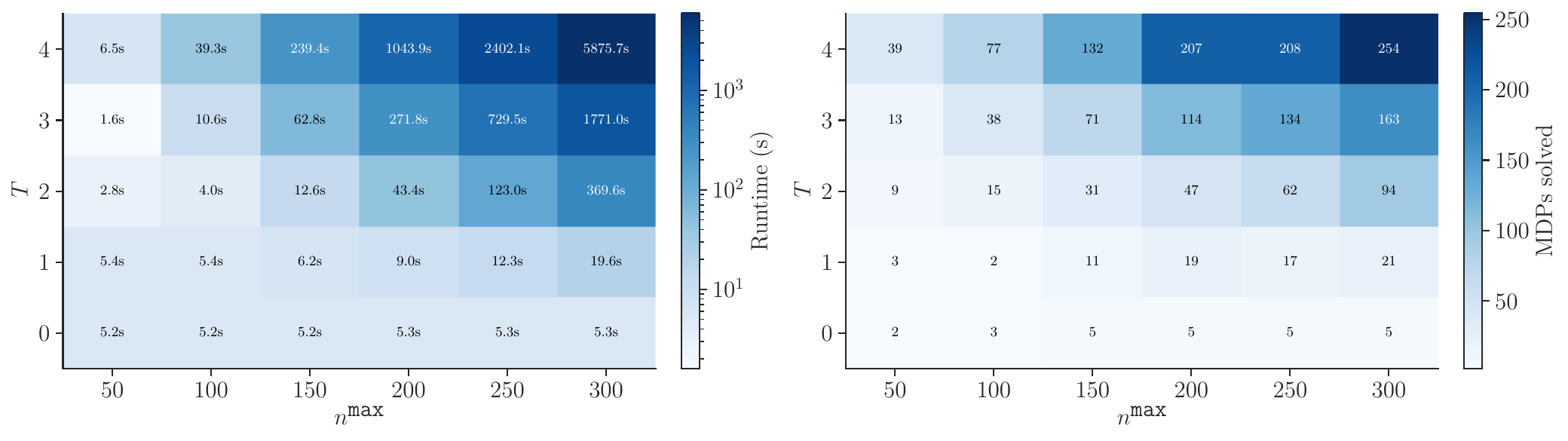}
    \caption{\textbf{Runtime of Algorithm~\ref{alg:epsilon}.} The figure shows, across multiple values of the maximum number of trials $T$ and the maximum sample size per trial $n^{\texttt{max}}$, the runtime of Algorithm~\ref{alg:epsilon}  (left panel) and the number of belief MDPs solved by the algorithm (right panel). All other parameters are fixed as specified in Tables~\ref{tab:non-econ-parameters} and~\ref{tab:econ-parameters}. The experiments are run on an NVIDIA H100 GPU.
    }
    \label{fig:runtime}
\end{figure}

\xhdr{Parameter details} Tables~\ref{tab:non-econ-parameters} and~\ref{tab:econ-parameters} report the values of all parameters required to specify the sequential approval protocol used in our fiducial setting (Section~\ref{sec:experiments}). Unless otherwise stated, all results are obtained using these values. When any parameter is varied (\eg, in Panel~(b) of Figure~\ref{fig:main} or in Appendix~\ref{app:additional-experimental-results}), we explicitly indicate it.

\begin{table}[H]
\centering
\caption{\textbf{Non-economic parameters}}
\label{tab:non-econ-parameters}
\begin{tabular}{cccccccc}
\toprule
  $T$  & $n^{\texttt{max}}$ & $\theta^{\texttt{b}}$ & $\kappa$  &$\theta^*$ & $\varepsilon^{\texttt{max}}$ & $(\alpha_0, \beta_0)$ & $Q$ \\
\midrule
 3     &   200              &              0.5       & 0.05      & 0.65      &     0.9                       & (1,1) & $\delta_{(1, 1)}$\\
\bottomrule
\end{tabular}
\end{table}

\begin{table}[H]
\centering
\caption{\textbf{Economic parameters}}\label{tab:econ-parameters}
\begin{tabular}{cccc}
\toprule
\multirow{2}{*}{$\rho^{\texttt{S}}$} & 
\multirow{2}{*}{$\rho^{\texttt{A}}$} & 
\multicolumn{2}{c}{$c(n) = c_0 + c_1 n$, $n\neq 0$} \\
\cmidrule(lr){3-4}
 &  & $c_0$ & $c_1$ \\
\midrule
\$2000\,\text{M} & \$240\,\text{M} & \$48.9\,\text{M} & \$0.066\,\text{M} \\
\bottomrule
\end{tabular}
\end{table}

\xhdr{Implementation details} Our Python implementation of Algorithm~\ref{alg:value_iteration_decomposition} leverages the bijection between pairs $(\alpha,\beta)$ and pairs $(X,N)$, where $X$ is the total number of successes and $N$ the total number of patients (the total sample size), as we detail in Eq.~\ref{eq:bijection-belief-outcomes}. Furthermore, since the cost function is linear, we also use the fact that the total cumulated cost $C$ at a given state $(X,N)$ at time $l$ can be written as $C = l \cdot c_0 + N \cdot c_1$. As a result, our implementation does not explicitly keep track of the cumulated cost $C$, which significantly reduces the computational overhead. A similar idea is used in the proof of Proposition~\ref{prop:reacheable-set}.

Our implementation of Algorithm~\ref{alg:epsilon} computes the optimal subsidy $\varepsilon^*$ (and solves each belief MDP $\Mcal^\varepsilon$) exactly. Consequently, there is no associated uncertainty in $\varepsilon^*$ or in any of the quantities that we report computed using the belief MDP, including $\Bar{U}^{\texttt{A}}$ and $\Bar{U}^{\texttt{S}}$. In contrast, the true quantities under the approval process using $\theta^*$ described in Section~\ref{sec:model}---such as the true utilities $U^{\texttt{A}}$ and $U^{\texttt{S}}$ (Eq.~\ref{eq:true-utilities}), as well as the probability that the agent opts out during the approval process---are estimated using $100{,}000-200{,}000$ Monte Carlo rollouts, and we report $95\%$ confidence error bars computed via bootstrapping with $1000$ resamples.

\clearpage
\newpage

\section{Additional Experimental Results}\label{app:additional-experimental-results}

\subsection{Additional results complementing Section~\ref{sec:experiments}}
In this section, we provide complementary results to the approval process considered in Section~\ref{sec:experiments}, whose parameters are given in Table~\ref{tab:non-econ-parameters} and Table~\ref{tab:econ-parameters}.

In Figure~\ref{fig:value_action_optimal_fiducial}, we show the optimal value function and policy for the MDP $\Mcal^{\varepsilon^*}$, with $\varepsilon^* = 1.08$, and in Figure~\ref{fig:trajectories_fiducial} we show how the belief of the agent evolves in 300 realizations of the approval process, for different values of the true efficacy $\theta^*$.

\begin{figure}[h]
    \vspace{0cm}
    \centering
    \includegraphics[width=0.99\linewidth]{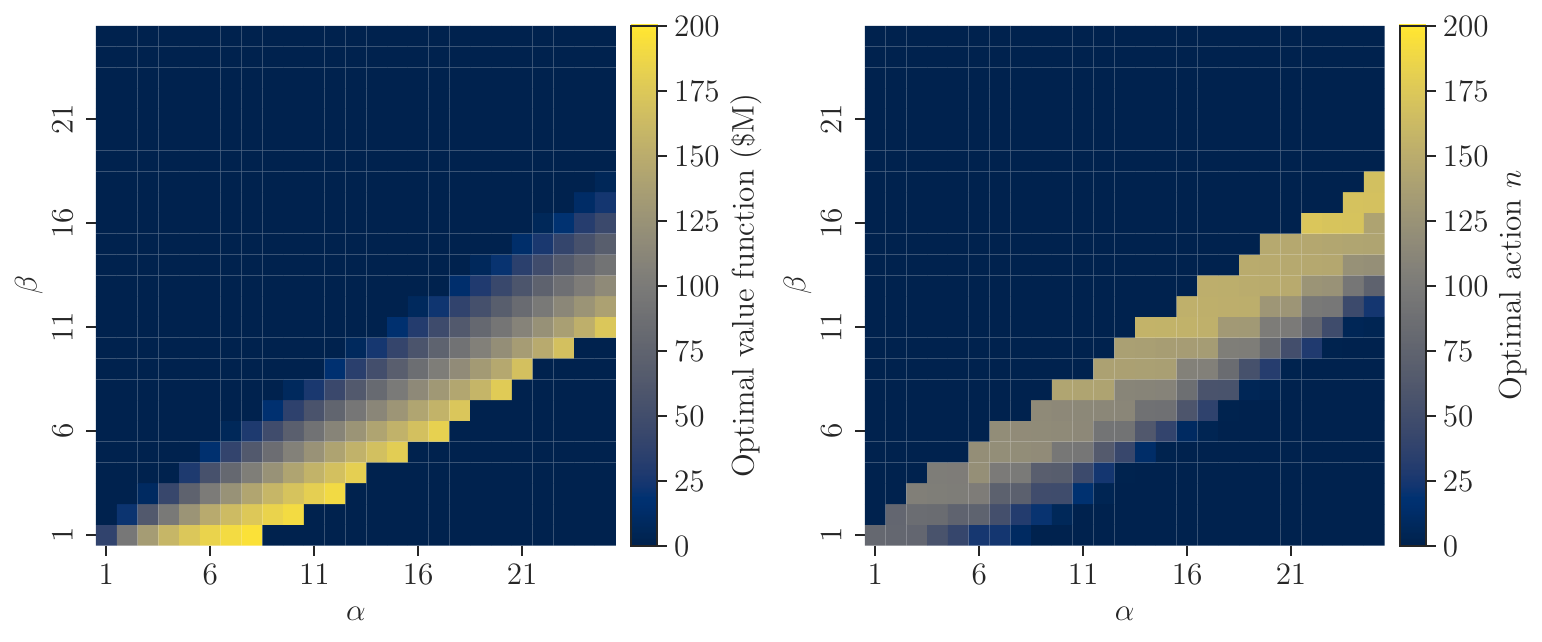}
    \caption{\textbf{Optimal value function and policy in the belief MDP $\Mcal^{\varepsilon^*}$ for the optimal subsidy $\varepsilon^* = 0.108$.}
    The left panel shows the optimal value function in the belief MDP, $V^{\varepsilon^*}(\alpha,\beta,C(\alpha,\beta),1)$, at time step $l=1$, where the cost of each state is given by $C(\alpha,\beta) = 1\cdot c_0 + (\alpha + \beta - \alpha_0 - \beta_0)\cdot c_1$ (see Eq.~\ref{eq:bijection-belief-outcomes}).
    The right panel shows the optimal action $n$ taken by the optimal policy at time step $l=1$ for each belief, \ie, $\pi^{\varepsilon^*}(\alpha,\beta,C(\alpha,\beta),1)$. The optimal action at time step $l=0$ (not shown here) is $n=79$.
    }
    \label{fig:value_action_optimal_fiducial}
\end{figure}

\begin{figure}[h]
    \vspace{0cm}
    \centering
    \subfloat[ $\theta^* = 0.3$ ]{
    \includegraphics[width=0.3\linewidth]{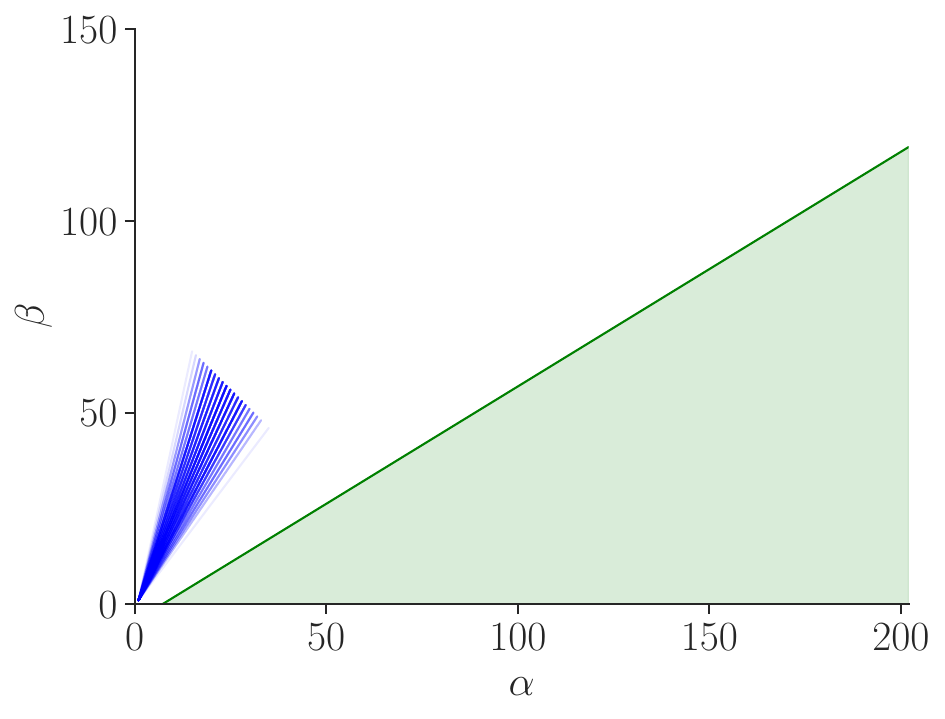}
    }
    \hspace{0mm}
    \subfloat[$\theta^* = 0.65$]{
    \includegraphics[width=0.3\linewidth]{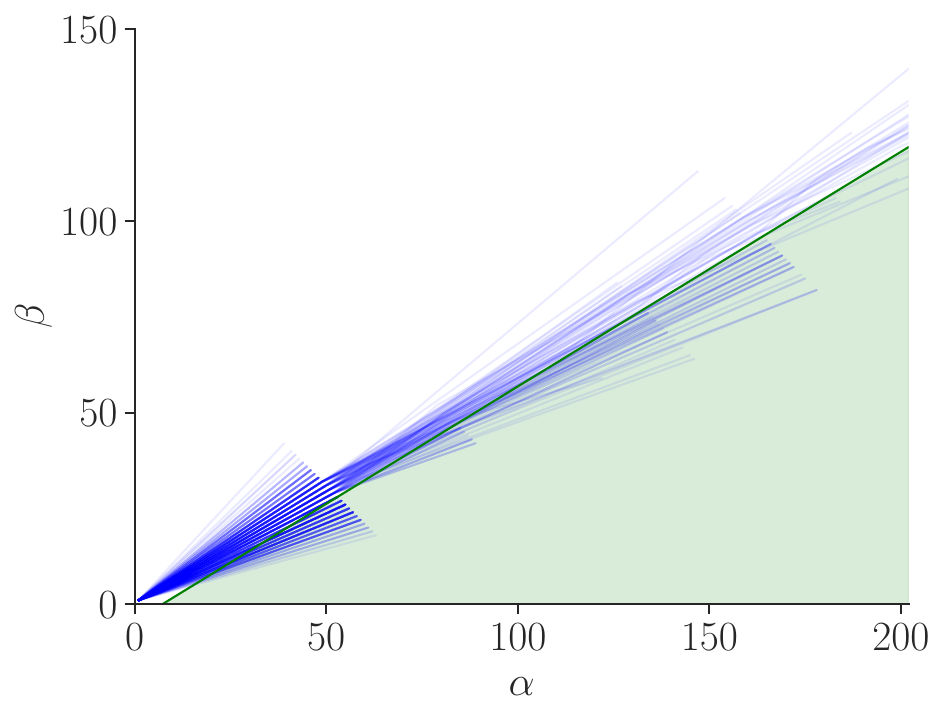}
    }
    \hspace{0mm}
    \subfloat[$\theta^* = 0.8$]{
    \includegraphics[width=0.3\linewidth]{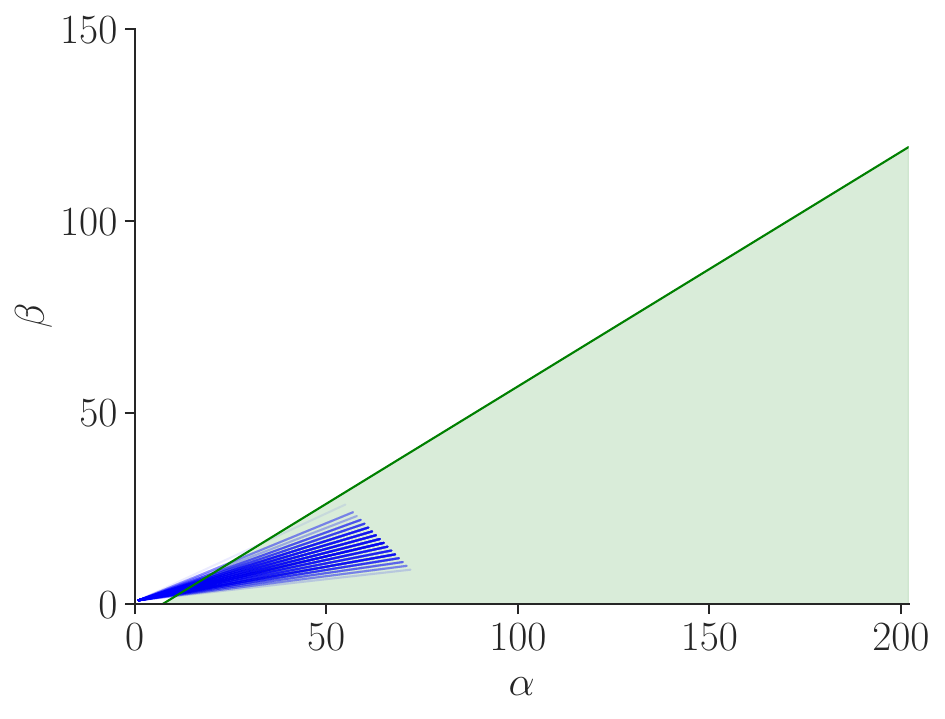}
    }
    \caption{\textbf{Trajectories of the approval process.}
    Each panel shows 300 realizations of the approval process for different true efficacies $\theta^*$ of the antibiotic. Each blue segment corresponds to the agent conducting a new trial and updating its belief (see Figure~\ref{fig:geometry} for an illustration of the geometry). In the left panel, the agent conducts a first trial and then always opts out; in the middle panel, the agent can conduct multiple trials, and in the right panel, the antibiotic is always approved after the first trial.
    }
    \label{fig:trajectories_fiducial}
\end{figure}

In Figure~\ref{fig:action_value_fiducial}, we show, for the initial action taken by the agent at time step $l=0$ (and belief $(\alpha_0,\beta_0)=(1,1)$), the expected cumulative future reward in $\Mcal^{\varepsilon^*}$ for the optimal subsidy $\varepsilon^* = 0.108$, defined for each $n>0$ as \mbox{$\EE_{\pi^{\varepsilon^*}}[\sum_{t=0}^T r^{\varepsilon^*}(S_t,n_t,S_{t+1})| S_0, n_0=n]$}. The sample size maximizing this curve is $n=79$, which is the size of the first trial conducted by the agent. Notably, although the expected reward is unimodal, it exhibits small-scale oscillations. These oscillations are not numerical artifacts. In fact, their approximate period is given by $1 / \log(1+\theta^{\texttt{b}}(e-1))$. The reason is that the agent’s actions are discrete, whereas the function $f$ defined in Proposition~\ref{prop:test-process}, which determines the approval condition, decreases by exactly $\log(1+\theta^{\texttt{b}}(e-1))$ whenever the agent selects a new action $n$ (recall that $\alpha_{t+1} + \beta_{t+1} - \alpha_t - \beta_t = n_t$). Since this change is not an integer quantity, the expected reward only exhibits a small decrease after $n$ increases by approximately $1 / \log(1+\theta^{\texttt{b}}(e-1))$, which gives rise to the observed oscillations.

\begin{figure}[h]
    \vspace{0cm}
    \centering
    \includegraphics[width=0.65\linewidth]{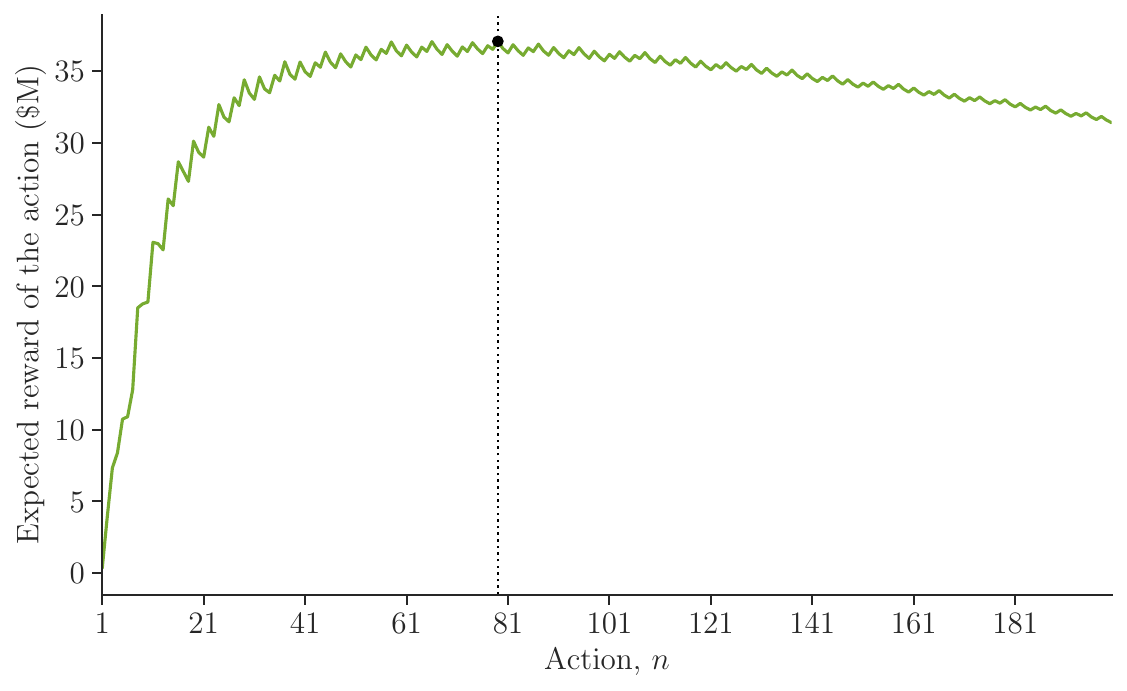}
    \caption{\textbf{Expected reward for each sample size.} The figure show, for the initial action taken by the agent at time step $l=0$ and state $(\alpha_0=1, \beta_0=1, 0)$, the total expected reward in the MDP $\Mcal^{\varepsilon^*}$ under the optimal subsidy $\varepsilon^* = 0.108$ when the agent takes action $n$ and then follows the optimal policy $\pi^{\varepsilon^*}$.
    }
    \label{fig:action_value_fiducial}
\end{figure}

\begin{figure}[h]
    \vspace{0cm}
    \centering
    \subfloat[ Opting out before approval vs. subsidy]{
    \includegraphics[width=0.47\linewidth]{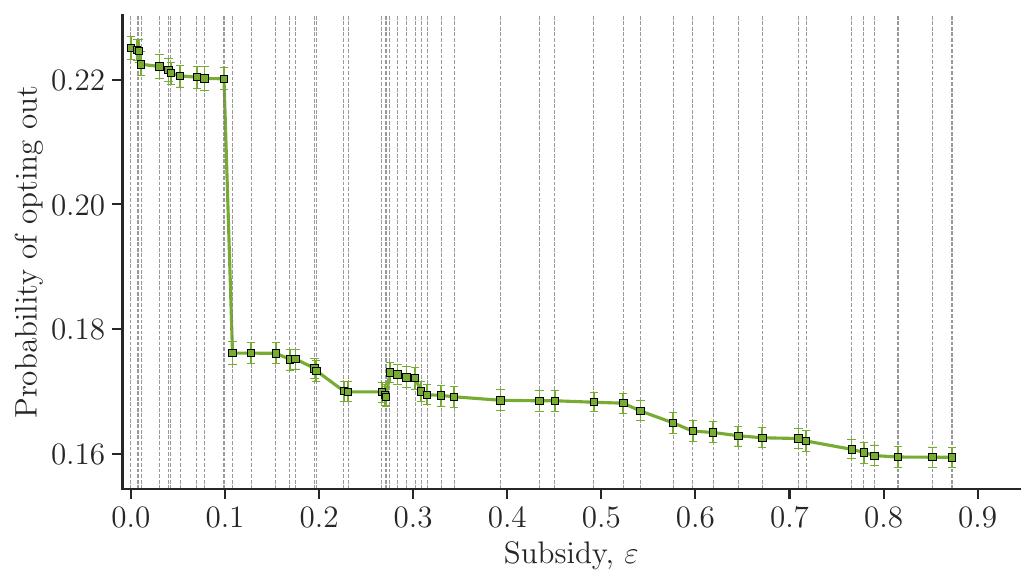}
    }
    \subfloat[ Approval probability vs subsidy]{
    \includegraphics[width=0.47\linewidth]{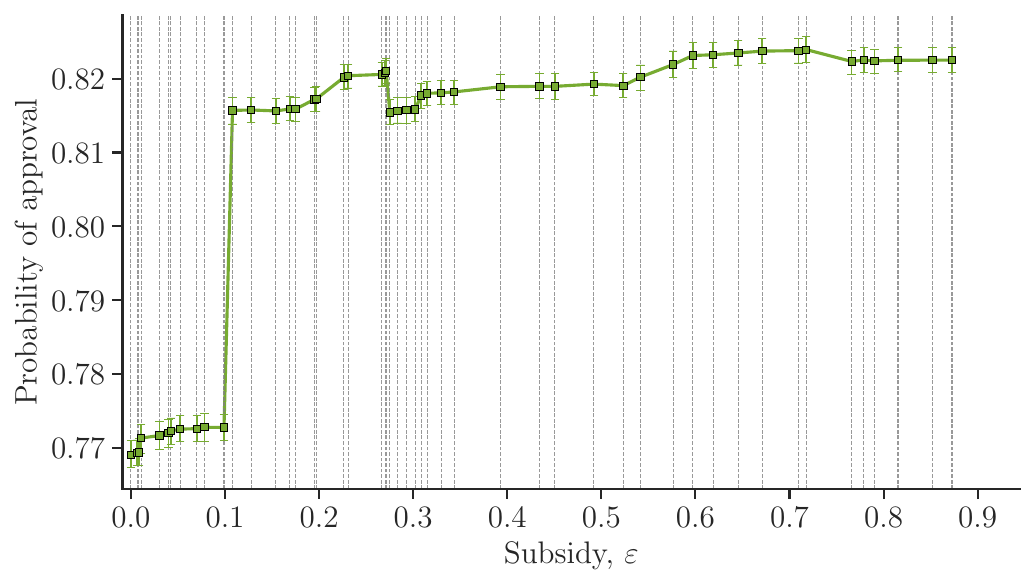}
    }
    \caption{\textbf{Opt out and approval probabilities.} The figure shows, for an antibiotic with $\theta^* = 0.65$, the probability that the agent opts out of the approval process by selecting $n=0$ before approval, as well as the probability that the antibiotic is ultimately approved. For each subsidy level, the agent follows the optimal policy. Note that, in principle, the agent may never opt out during the approval process; however, the antibiotic may still fail to be approved if the maximum number of trials is reached.
    }
    \label{fig:opt_out_prob_fiducial}
\end{figure}

In Figure~\ref{fig:agent_utilities_fiducial} and Figure~\ref{fig:social_utilities_fiducial}, we show, respectively, the utility of the agent and the social utility when the agent selects its optimal policy for each possible subsidy, both computed using the belief MDP, and the true realized utilities when the antibiotic has efficacy $\theta^*=0.65$ (Eq.~\ref{eq:true-utilities}).

\begin{figure}[h]
    \vspace{0cm}
    \centering
    \subfloat[ Agent utility $\Bar{U}^{\texttt{A}}(\pi^\varepsilon;\varepsilon)$ computed using $\Mcal^\varepsilon$ ]{
    \includegraphics[width=0.45\linewidth]{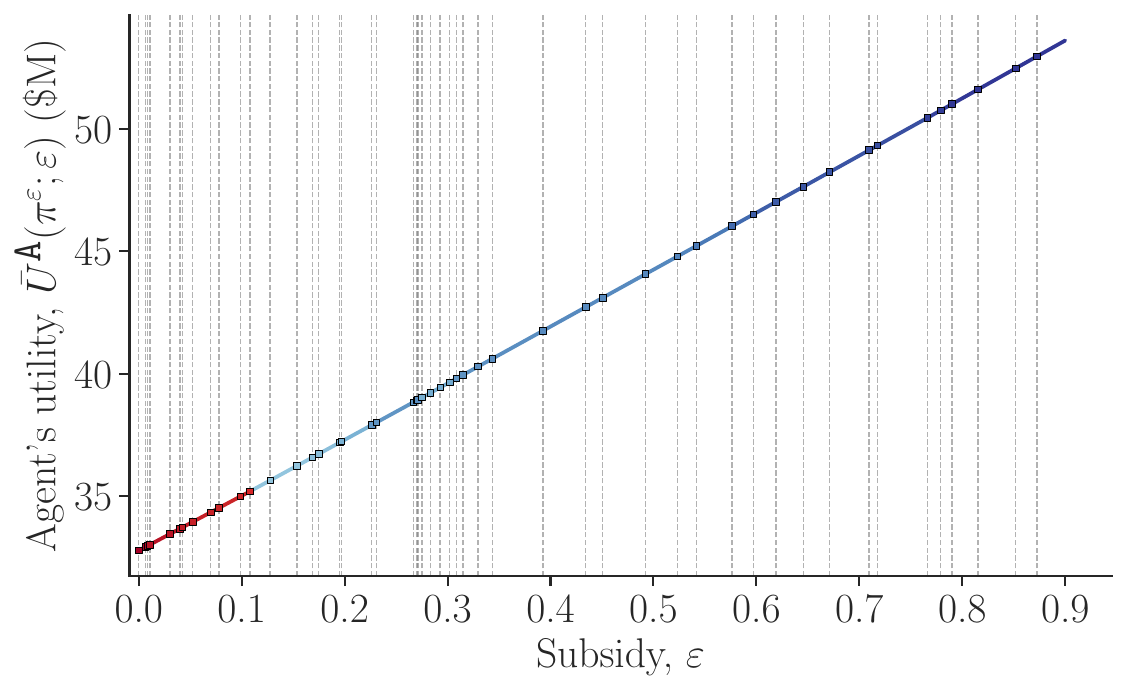}
    }
    \hspace{0mm}
    \subfloat[ Agent utility $U^{\texttt{A}}(\pi^\varepsilon;\varepsilon)$ in the approval process]{
    \includegraphics[width=0.45\linewidth]{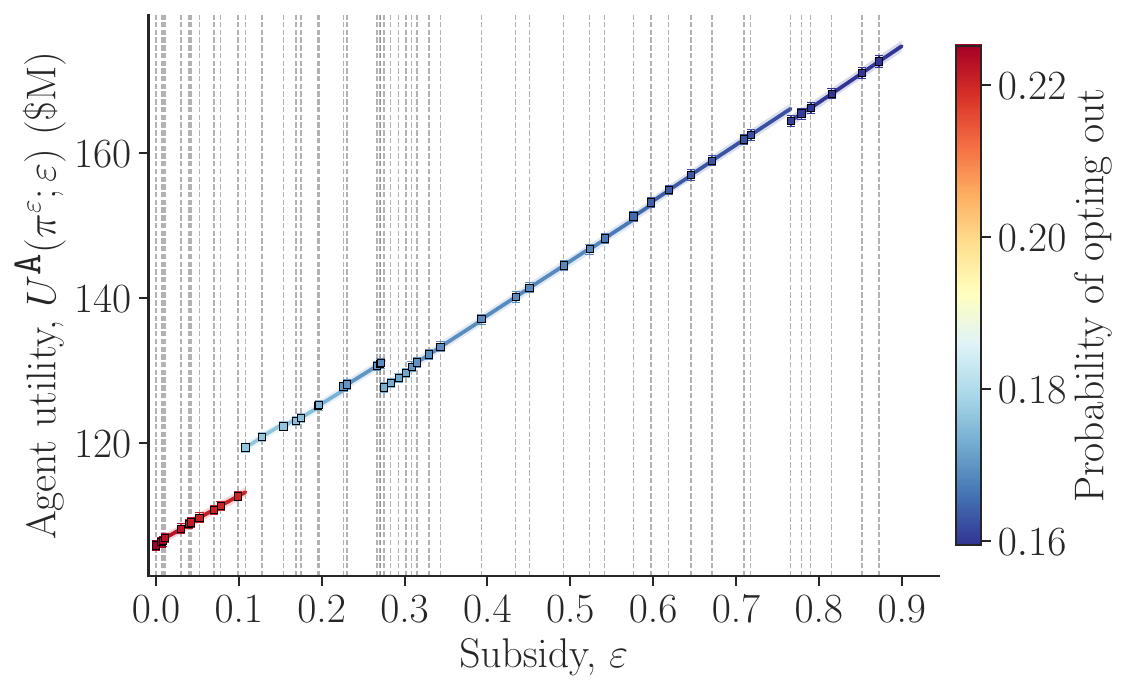}
    }
    \caption{\textbf{Agent utilities.}
    The left panel shows the agent's utility (Eq.~\ref{eq:adaptive-utility}) computed using the belief MDP $\Mcal^\varepsilon$ when the agent uses the optimal policy for each subsidy, which is a piece-wise linear, convex and continuous function in accordance with Proposition~\ref{prop:utility-piecewise-convex}. The right panel shows the true utility of the agent (Eq.~\ref{eq:true-utilities}) in the approval process when using the optimal policy $\pi^\varepsilon$ for each subsidy and $\theta^* = 0.65$.
    The dashed vertical lines correspond to the intervals of the partition $\mathcal{P}$ where the agent's optimal policy is constant (Proposition~\ref{prop:utility-piecewise-convex}).
    }
    \label{fig:agent_utilities_fiducial}
\end{figure}

\begin{figure}[h]
    \vspace{0cm}
    \centering
    \subfloat[ Social utility $\Bar{U}^{\texttt{S}}(\varepsilon;\pi^\varepsilon)$ computed using $\Mcal^\varepsilon$ ]{
    \includegraphics[width=0.45\linewidth]{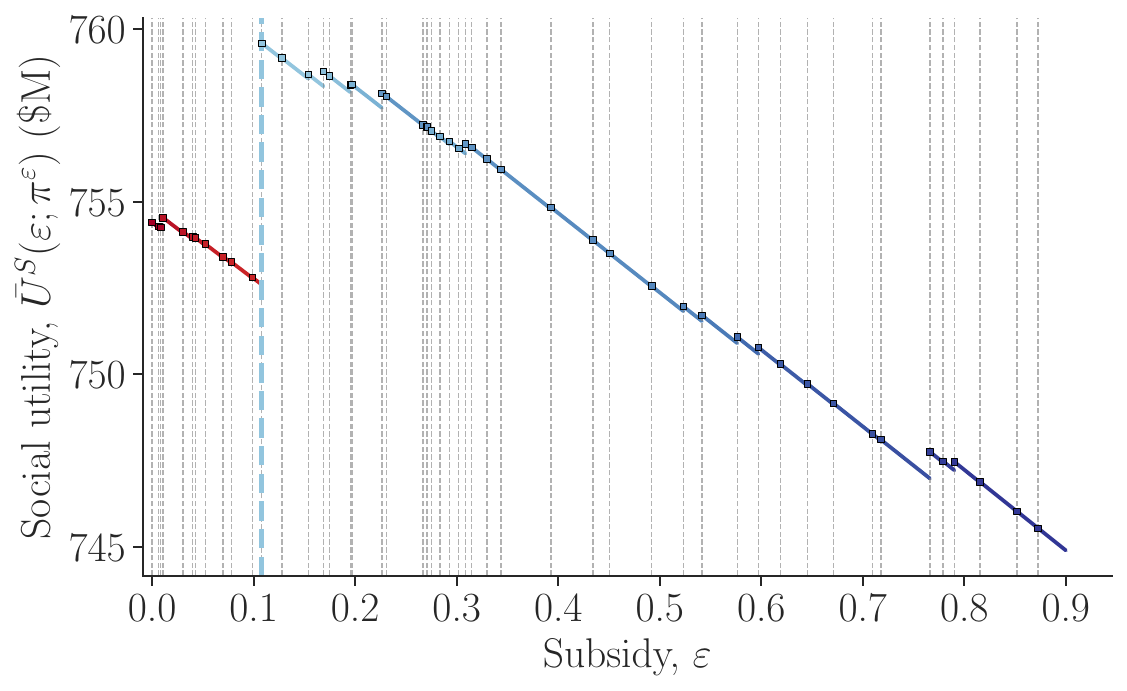}
    }
    \hspace{0mm}
    \subfloat[ Social utility $U^{\texttt{S}}(\varepsilon;\pi^\varepsilon)$ in the approval process]{
    \includegraphics[width=0.45\linewidth]{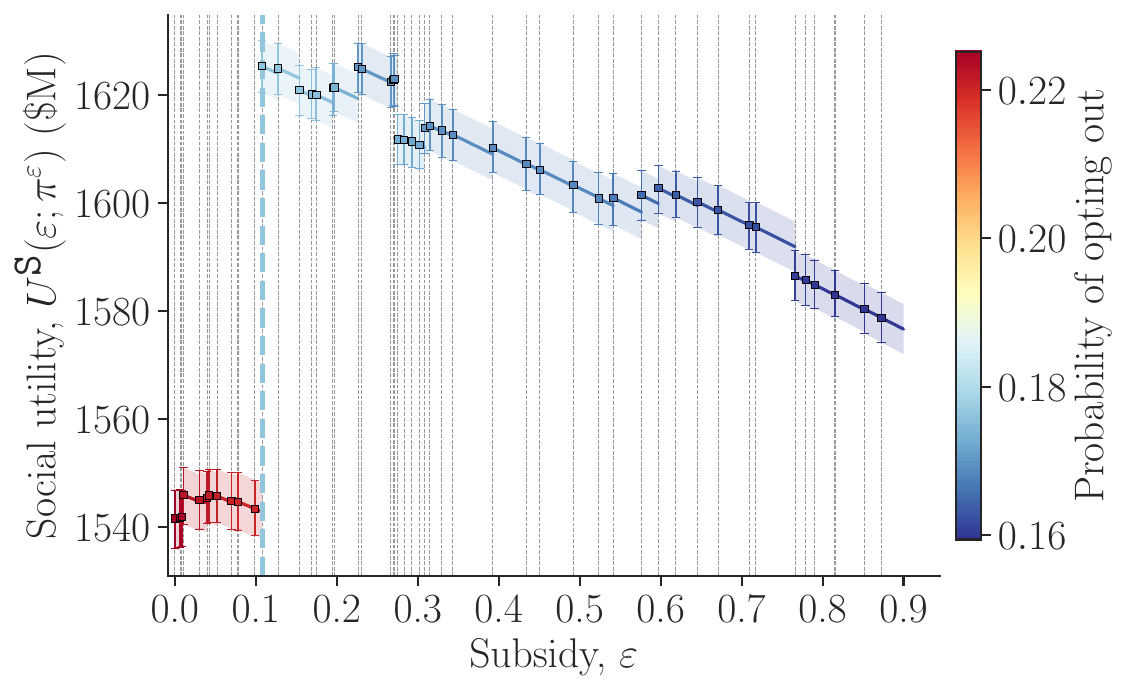}
    }
    \caption{\textbf{Social utilities.}
    The left panel shows the social utility (Eq.~\ref{eq:adaptive-utility-regulator}) computed using the belief MDP $\Mcal^\varepsilon$ when the agent uses the optimal policy for each subsidy. The right panel shows the true social utility (Eq.~\ref{eq:true-utilities}) in the approval process when the agent uses the optimal policy $\pi^\varepsilon$ for each subsidy and $\theta^* = 0.65$.
    The dashed vertical lines correspond to the intervals of the partition $\mathcal{P}$ where the agent's optimal policy is constant (Proposition~\ref{prop:utility-piecewise-convex}).
    }
    \label{fig:social_utilities_fiducial}
\end{figure}

Lastly, in Figure~\ref{fig:social_utility_efficacies} we show that the optimal social utility (that is, the social utility under the optimal subsidy) increases monotonically as the true efficacy $\theta^*$ increases.

\begin{figure}[h]
    \vspace{0cm}
    \centering
    \includegraphics[width=0.65\linewidth]{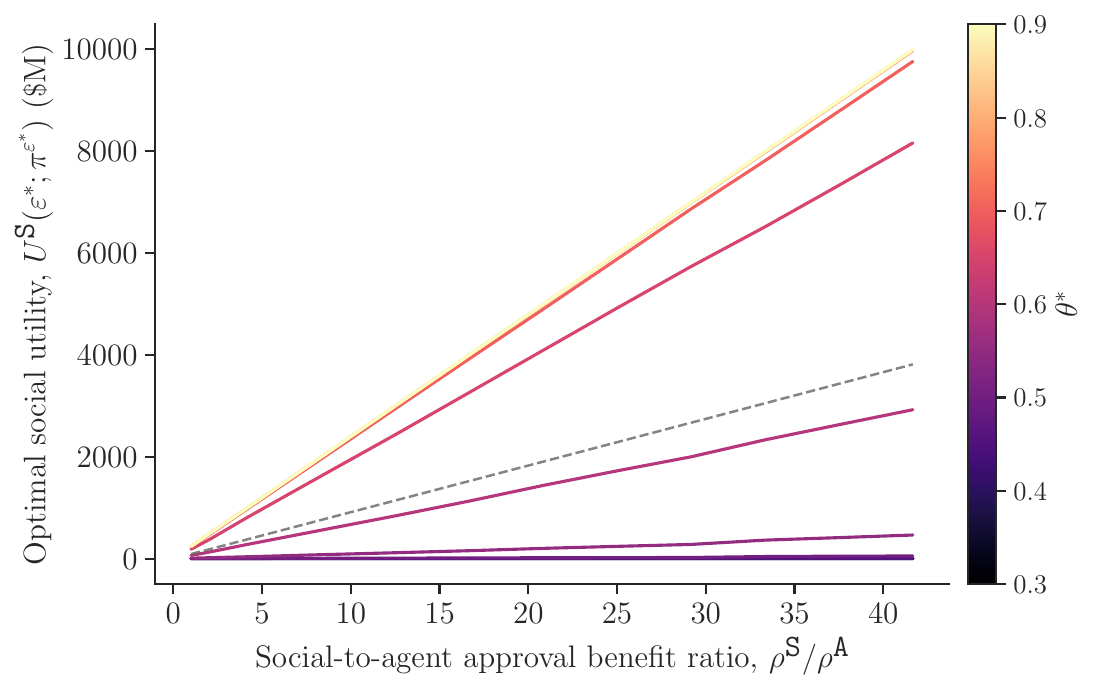}
    \caption{\textbf{Optimal social utility for different antibiotic efficacies.} The figure shows how social utility $U^{\texttt{S}}(\varepsilon^*;\pi^{\varepsilon^*})$---when the principal chooses the optimal subsidy $\varepsilon^*$ and the agent adopts the corresponding optimal policy $\pi^{\varepsilon^*}$---varies as a function of the ratio $\rho^{\texttt{S}} / \rho^{\texttt{A}}$ across different levels of efficacy $\theta^*$. The dashed line corresponds to the social utility $\Bar{U}^{\texttt{S}}(\varepsilon^*;\pi^{\varepsilon^*})$ computed using the belief MDP, which does not depend on the true efficacy $\theta^*$.
    }
    \label{fig:social_utility_efficacies}
\end{figure}

\clearpage
\newpage

\subsection{Additional results using different parameters}

In this section, we present further experimental results for the antibiotic approval process described in Section~\ref{sec:experiments}, where we vary selected parameters (see Tables~\ref{tab:non-econ-parameters} and~\ref{tab:econ-parameters}).

\subsubsection{Approval under increased experimental costs}

Here, we show the result of an antibiotic approval process with an increased experimental cost. In particular, we take the parameters in Table~\ref{tab:non-econ-parameters} and Table~\ref{tab:econ-parameters} but increase the fixed cost $c_0$ of a trial to $\$100\,\text{M}$ and the per-patient cost $c_1$ to $\$0.1\,\text{M}$. We find that the agent opts out at the beginning of the process unless the principal subsidizes a fraction higher than the optimal subsidy, $\varepsilon^* = 0.551$.

\begin{figure}[H]
    \vspace{0cm}
    \centering
    \subfloat[ Agent utility $\Bar{U}^{\texttt{A}}(\pi^\varepsilon;\varepsilon)$ computed using $\Mcal^\varepsilon$ ]{
    \includegraphics[width=0.45\linewidth]{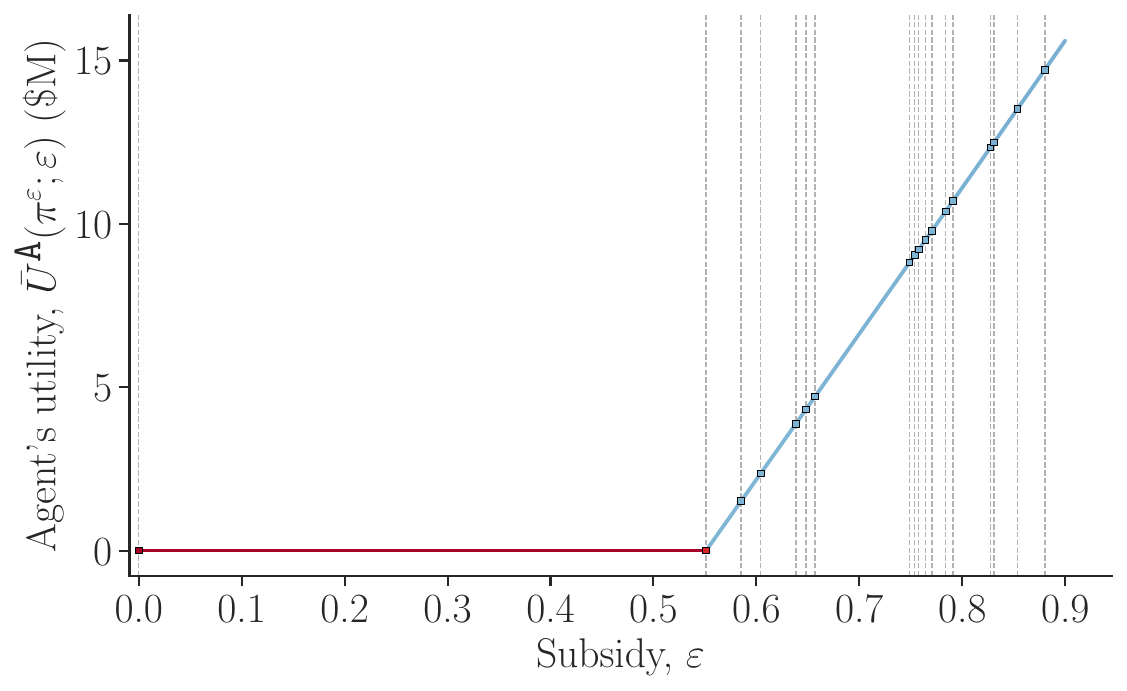}
    }
    \hspace{0mm}
    \subfloat[ Agent utility $U^{\texttt{A}}(\pi^\varepsilon;\varepsilon)$ in the approval process]{
    \includegraphics[width=0.45\linewidth]{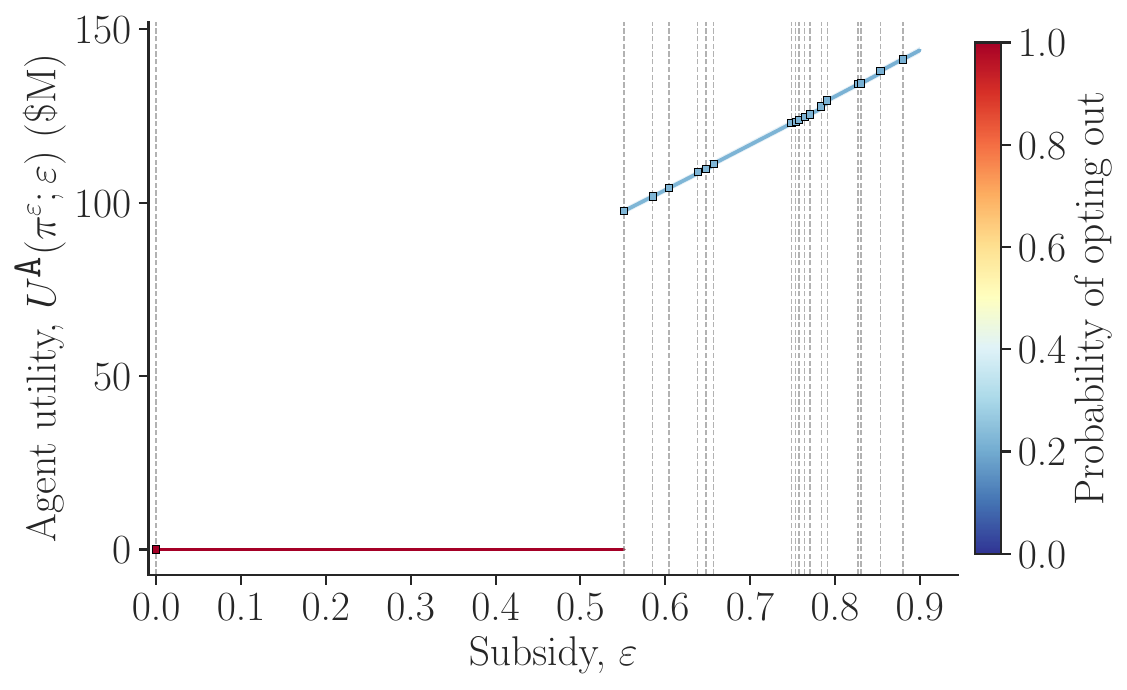}
    }
    \caption{\textbf{Agent utilities under increased experimental costs.}
    The left panel shows the agent's utility (Eq.~\ref{eq:adaptive-utility}) computed using the belief MDP $\Mcal^\varepsilon$ when the agent uses the optimal policy for each subsidy, which is a piece-wise linear, convex, and continuous function in accordance with Proposition~\ref{prop:utility-piecewise-convex}. The right panel shows the true utility of the agent (Eq.~\ref{eq:true-utilities}) in the approval process when using the optimal policy $\pi^\varepsilon$ for each subsidy and $\theta^* = 0.65$.
    The dashed vertical lines correspond to the intervals of the partition $\mathcal{P}$ where the agent's optimal policy is constant (Proposition~\ref{prop:utility-piecewise-convex}).
    }
    \label{fig:agent_utilities_costly}
\end{figure}

\begin{figure}[H]
    \vspace{0cm}
    \centering
    \subfloat[ Social utility $\Bar{U}^{\texttt{S}}(\varepsilon;\pi^\varepsilon)$ computed using $\Mcal^\varepsilon$ ]{
    \includegraphics[width=0.45\linewidth]{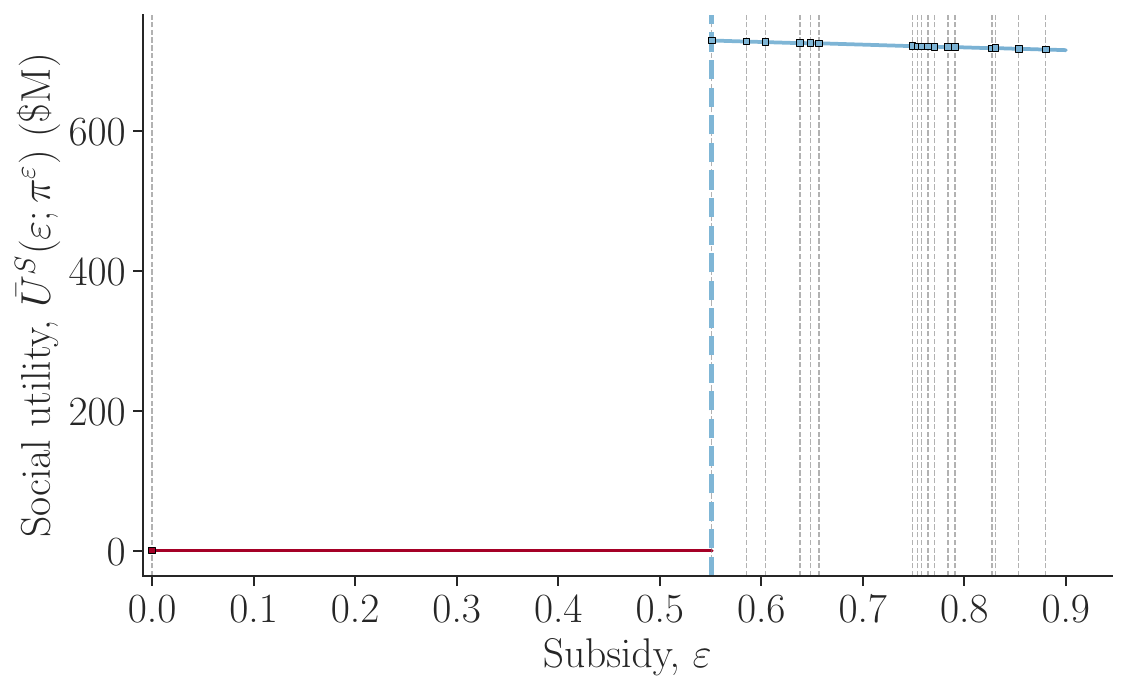}
    }
    \hspace{0mm}
    \subfloat[ Social utility $U^{\texttt{S}}(\varepsilon;\pi^\varepsilon)$ in the approval process]{
    \includegraphics[width=0.45\linewidth]{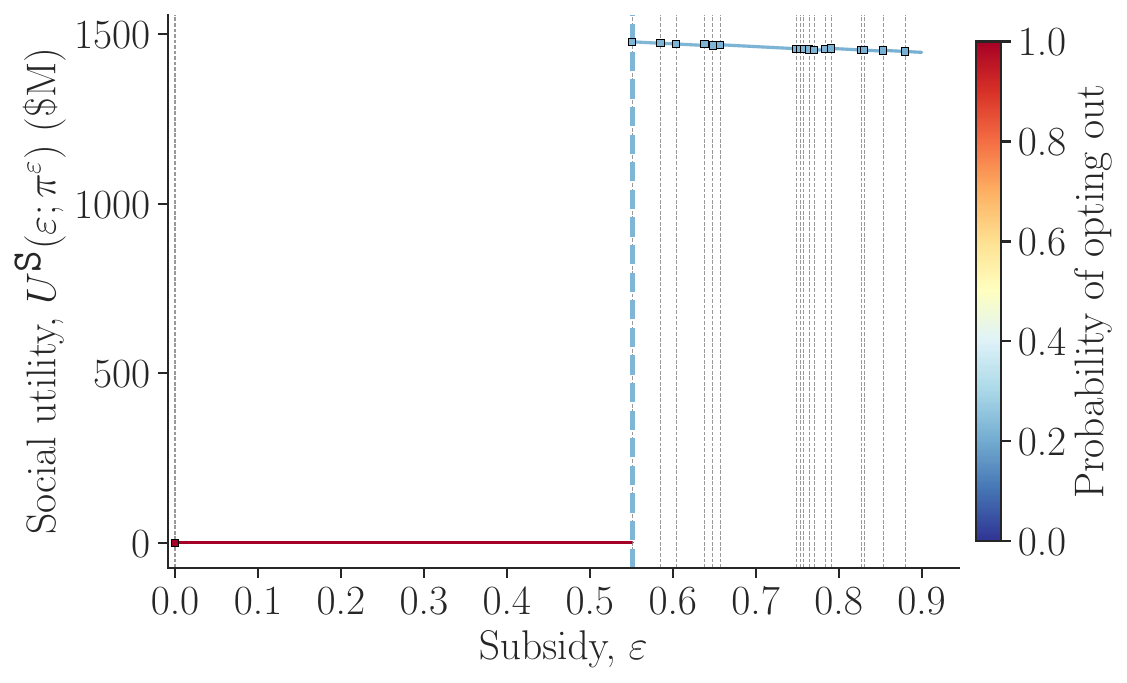}
    }
    \caption{\textbf{Social utilities under increased experimental costs.}
    The left panel shows the social utility (Eq.~\ref{eq:adaptive-utility-regulator}) computed using the belief MDP $\Mcal^\varepsilon$ when the agent uses the optimal policy for each subsidy. The right panel shows the true social utility (Eq.~\ref{eq:true-utilities}) in the approval process when the agent uses the optimal policy $\pi^\varepsilon$ for each subsidy and $\theta^* = 0.65$.
    The dashed vertical lines correspond to the intervals of the partition $\mathcal{P}$ where the agent's optimal policy is constant (Proposition~\ref{prop:utility-piecewise-convex}).
    }
\label{fig:social_utilities_costly}
\end{figure}

In Figure~\ref{fig:costly_subsidy} we show that this optimal subsidy is constant as $\rho^{\texttt{S}}$ increases, and in Figure~\ref{fig:costly_gain} that our sequential protocol yields social utility gains $>20\%$ relative to a non-sequential protocol.

\begin{figure}[H]
    \vspace{0cm}
    \centering
    \includegraphics[width=0.60\linewidth]{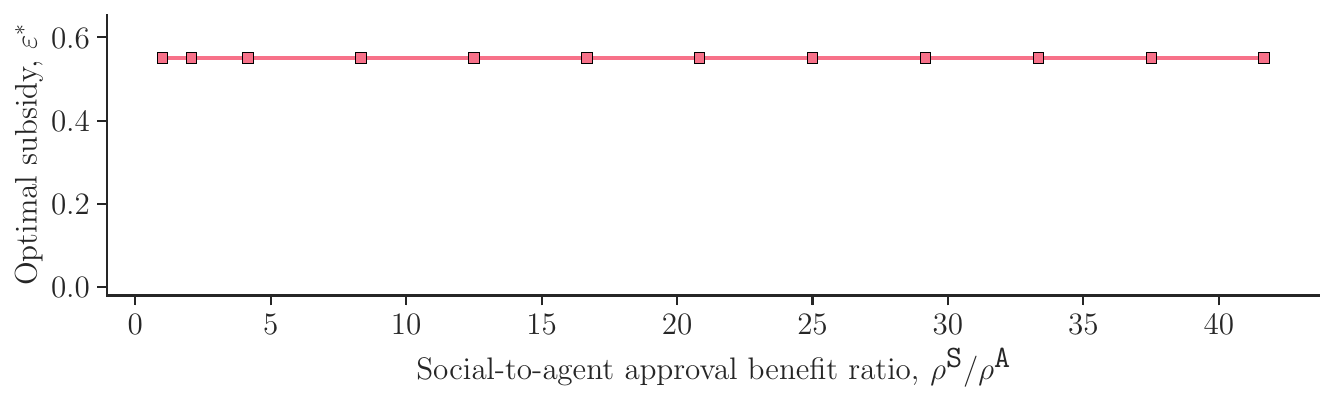}
    \caption{\textbf{Optimal subsidy vs. $\rho^{\texttt{S}} / \rho^{\texttt{A}}$.} The figure shows, as a function of the social-to-agent approval benefit ratio, the optimal subsidy obtained using Algorithm~\ref{alg:epsilon}.
    }
    \label{fig:costly_subsidy}
\end{figure}

\begin{figure}[H]
    \vspace{0cm}
    \centering
    \includegraphics[width=0.60\linewidth]{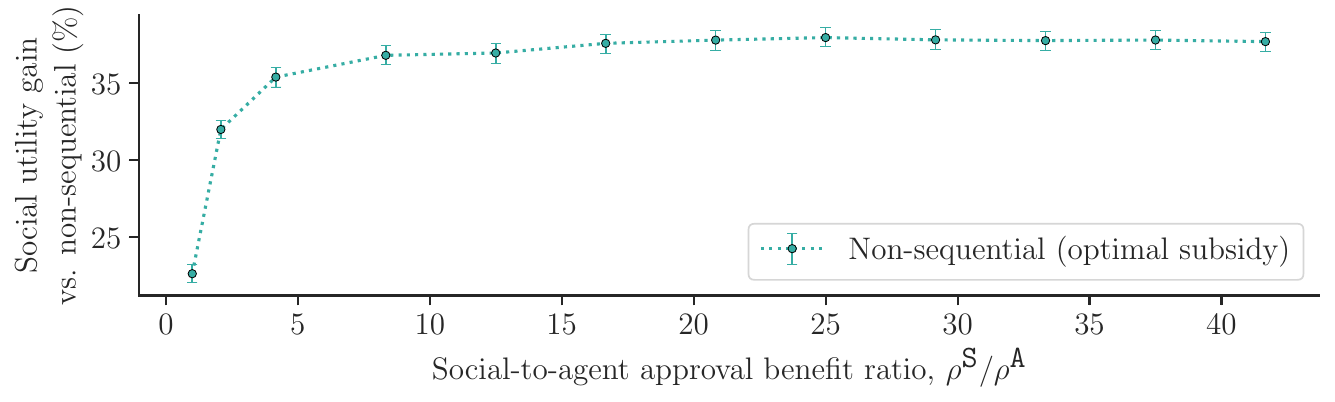}
    \caption{\textbf{Social utility gain vs. $\rho^{\texttt{S}} / \rho^{\texttt{A}}$.} The figure shows, as a function of the social-to-agent approval benefit ratio, the percentage increase in social utility of the sequential approval protocol relative to a non-sequential approval protocol in which the agent is restricted to a single trial with $n^{\texttt{max}}=800$, under the optimal subsidy computed using Algorithm~\ref{alg:epsilon} (in the non-sequential protocol without subsidy the agent always opts out, yielding zero social utility).
    }
    \label{fig:costly_gain}
\end{figure}

\clearpage
\newpage

\subsubsection{Approval under increased agent approval benefit}

Here, we show the result of an antibiotic approval process with an increased approval utility for the agent, $\rho^{\texttt{A}} = \$5000\,\text{M}$ (and the rest of the parameters in Table~\ref{tab:non-econ-parameters} and Table~\ref{tab:econ-parameters} fixed). In this case, the approval utility for the agent covers the expected cost of approval by a large margin, and we find that no subsidy is needed, namely, $\varepsilon^* =0$.

\begin{figure}[H]
    \vspace{0cm}
    \centering
    \subfloat[ Agent utility $\Bar{U}^{\texttt{A}}(\pi^\varepsilon;\varepsilon)$ computed using $\Mcal^\varepsilon$ ]{
    \includegraphics[width=0.45\linewidth]{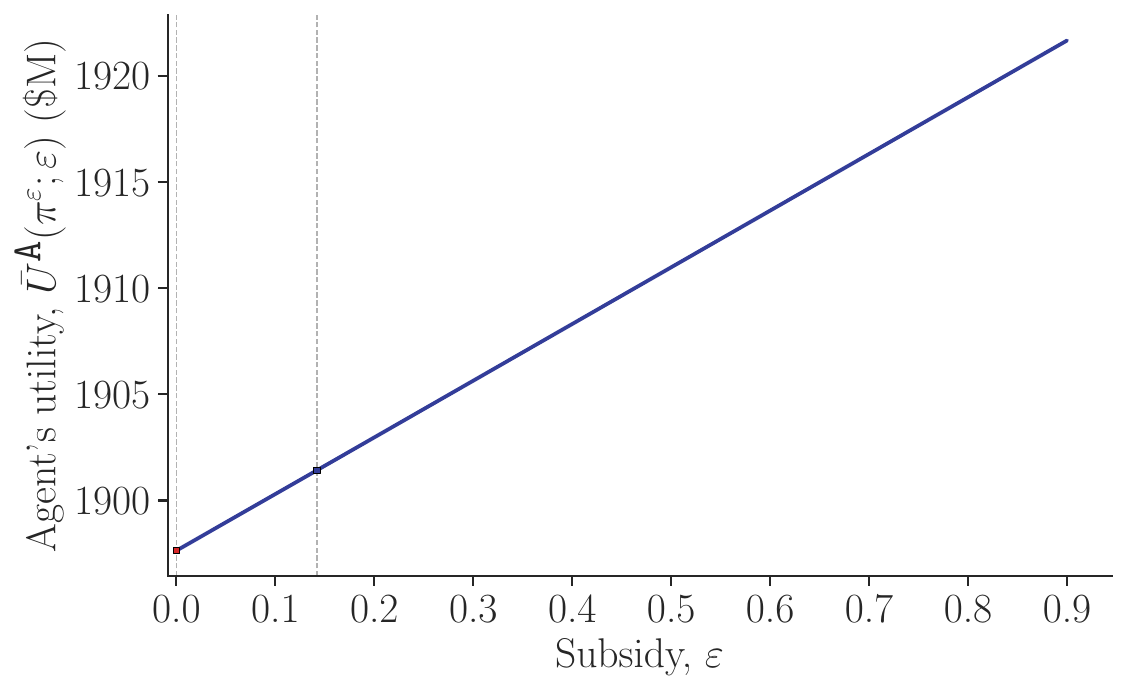}
    }
    \hspace{0mm}
    \subfloat[ Agent utility $U^{\texttt{A}}(\pi^\varepsilon;\varepsilon)$ in the approval process]{
    \includegraphics[width=0.45\linewidth]{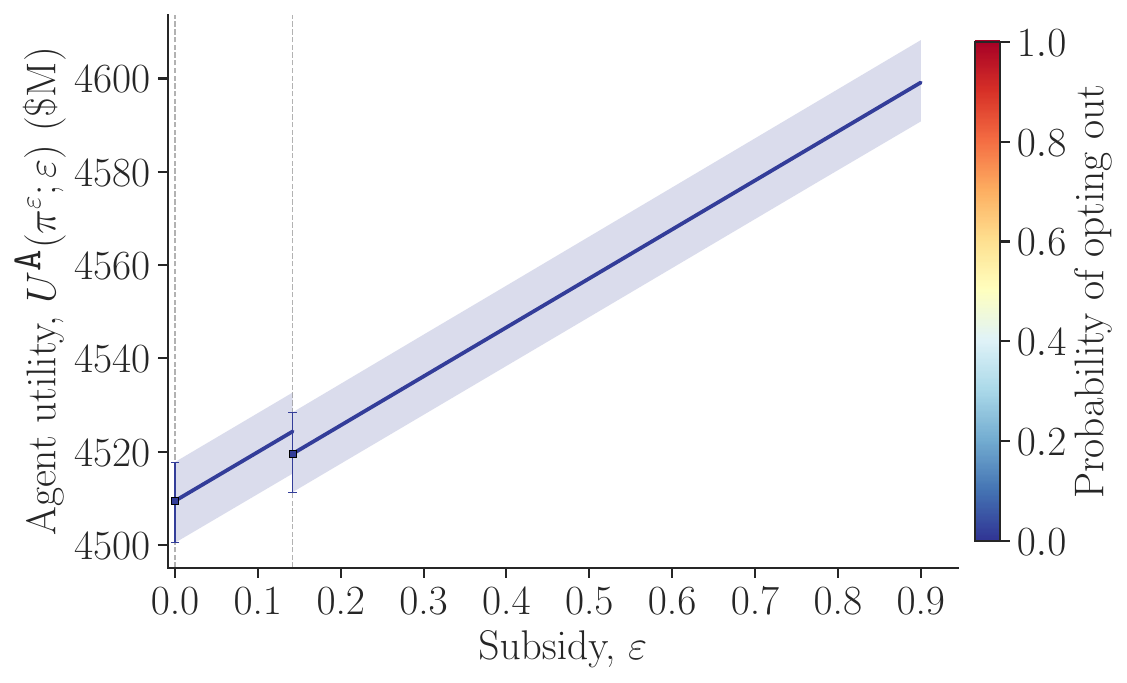}
    }
    \caption{\textbf{Agent utilities under increased approval utility.}
    The left panel shows the agent's utility (Eq.~\ref{eq:adaptive-utility}) computed using the belief MDP $\Mcal^\varepsilon$ when the agent uses the optimal policy for each subsidy, which is a piece-wise linear, convex and continuous function in accordance with Proposition~\ref{prop:utility-piecewise-convex}. The right panel shows the true utility of the agent (Eq.~\ref{eq:true-utilities}) in the approval process when using the optimal policy $\pi^\varepsilon$ for each subsidy and $\theta^* = 0.65$.
    The dashed vertical lines correspond to the intervals of the partition $\mathcal{P}$ where the agent's optimal policy is constant (Proposition~\ref{prop:utility-piecewise-convex}).
    }
    \label{fig:agent_utilities_greedy}
\end{figure}

\begin{figure}[H]
    \vspace{0cm}
    \centering
    \subfloat[ Social utility $\Bar{U}^{\texttt{S}}(\varepsilon;\pi^\varepsilon)$ computed using $\Mcal^\varepsilon$ ]{
    \includegraphics[width=0.45\linewidth]{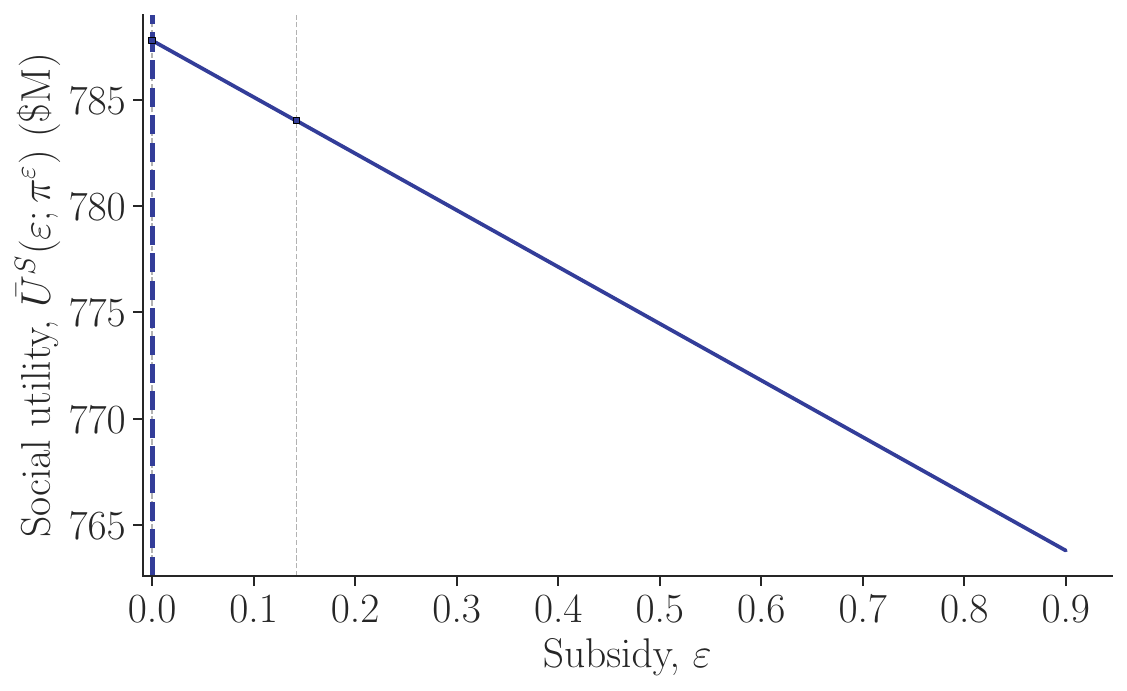}
    }
    \hspace{0mm}
    \subfloat[ Social utility $U^{\texttt{S}}(\varepsilon;\pi^\varepsilon)$ in the approval process]{
    \includegraphics[width=0.45\linewidth]{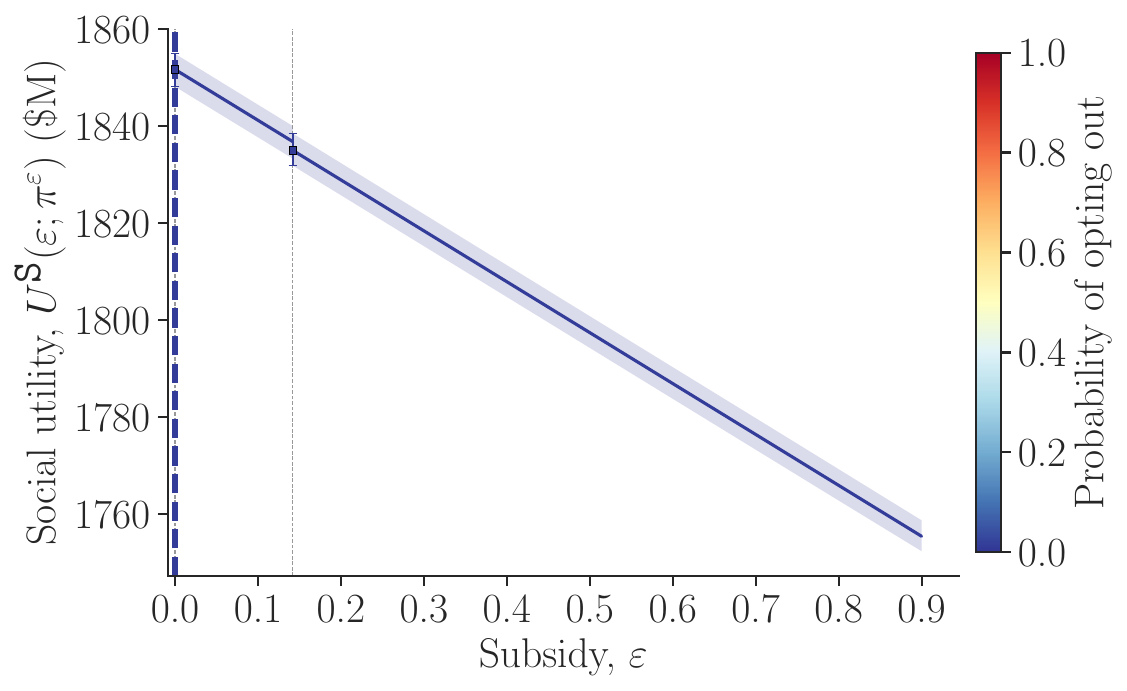}
    }
    \caption{\textbf{Social utilities under increased approval agent utility.}
    The left panel shows the social utility (Eq.~\ref{eq:adaptive-utility-regulator}) computed using the belief MDP $\Mcal^\varepsilon$ when the agent uses the optimal policy for each subsidy. The right panel shows the true social utility (Eq.~\ref{eq:true-utilities}) in the approval process when the agent uses the optimal policy $\pi^\varepsilon$ for each subsidy and $\theta^* = 0.65$.
    The dashed vertical lines correspond to the intervals of the partition $\mathcal{P}$ where the agent's optimal policy is constant (Proposition~\ref{prop:utility-piecewise-convex}).
    }
\label{fig:social_utilities_greedy}
\end{figure}

\begin{figure}[H]
    \vspace{0cm}
    \centering
    \includegraphics[width=0.60\linewidth]{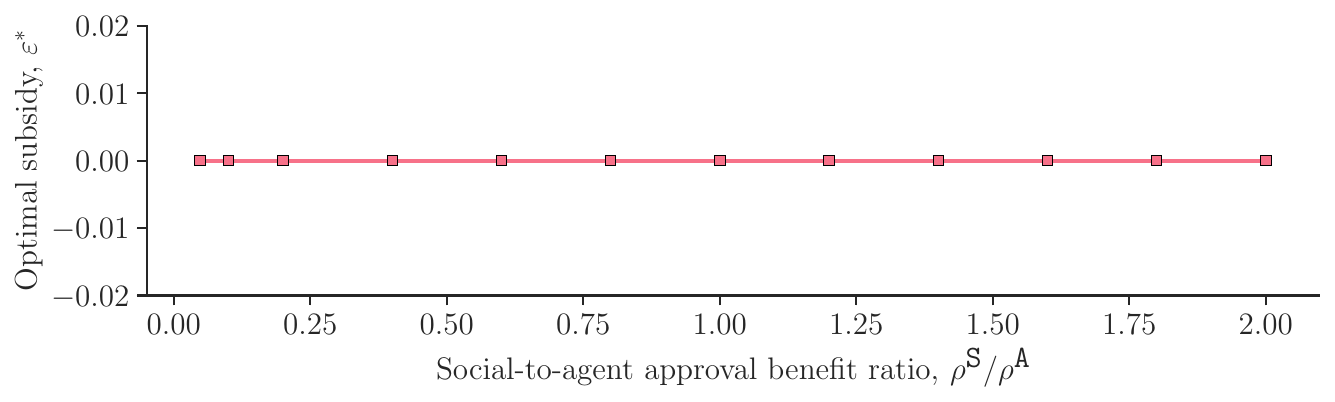}
    \caption{\textbf{Optimal subsidy vs. $\rho^{\texttt{S}} / \rho^{\texttt{A}}$.} The figure shows, as a function of the social-to-agent approval benefit ratio, the optimal subsidy obtained using Algorithm~\ref{alg:epsilon}.
    }
    \label{fig:greedy_subsidy}
\end{figure}

\begin{figure}[H]
    \vspace{0cm}
    \centering
    \includegraphics[width=0.60\linewidth]{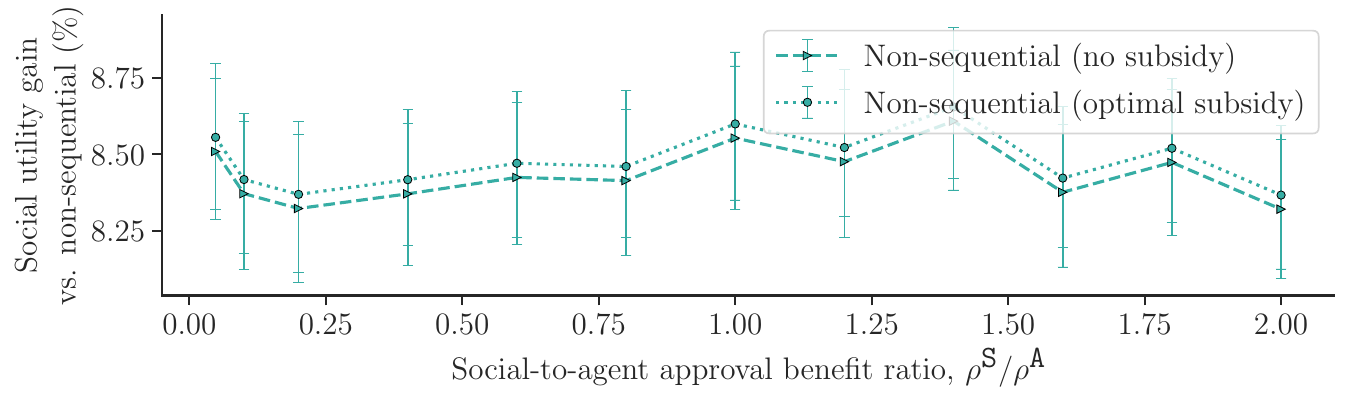}
    \caption{\textbf{Social utility gain vs. $\rho^{\texttt{S}} / \rho^{\texttt{A}}$.} The figure shows, as a function of the social-to-agent approval benefit ratio, the percentage increase in social utility of the sequential approval protocol relative to a non-sequential approval protocol in which the agent is restricted to a single trial with $n^{\texttt{max}}=800$, under (i) the optimal subsidy computed using Algorithm~\ref{alg:epsilon} and (ii) no subsidy ($\varepsilon=0$). In this case, the optimal non-sequential policy is the same for all subsidies, and we apply a small vertical jitter to improve visibility.
    }
    \label{fig:greedy_gain}
\end{figure}

\clearpage
\newpage

\subsubsection{Approval under pessimistic prior}

Here, we show the results of an antibiotic approval process where the agent's prior is $(\alpha_0,\beta_0)=(1,1.5)$, that is, the agent is slightly pessimistic about its product.\footnote{We select $\beta_0 = 1.5$ as higher values lead the agent to opt out at the beginning of the process, regardless of the subsidy.} Note that the mean efficacy of the drug according to its prior is then $0.4 < \theta^* = 0.65$. The principal knows such prior, and the rest of the parameters in Table~\ref{tab:non-econ-parameters} and Table~\ref{tab:econ-parameters} are fixed. In this case, we find that the agent opts out at the beginning of the approval process as long as the subsidy is $\varepsilon \lesssim 0.1$. The optimal subsidy is $\varepsilon^* \approx 0.4$.

\begin{figure}[H]
    \vspace{0cm}
    \centering
    \subfloat[ Agent utility $\Bar{U}^{\texttt{A}}(\pi^\varepsilon;\varepsilon)$ computed using $\Mcal^\varepsilon$ ]{
    \includegraphics[width=0.45\linewidth]{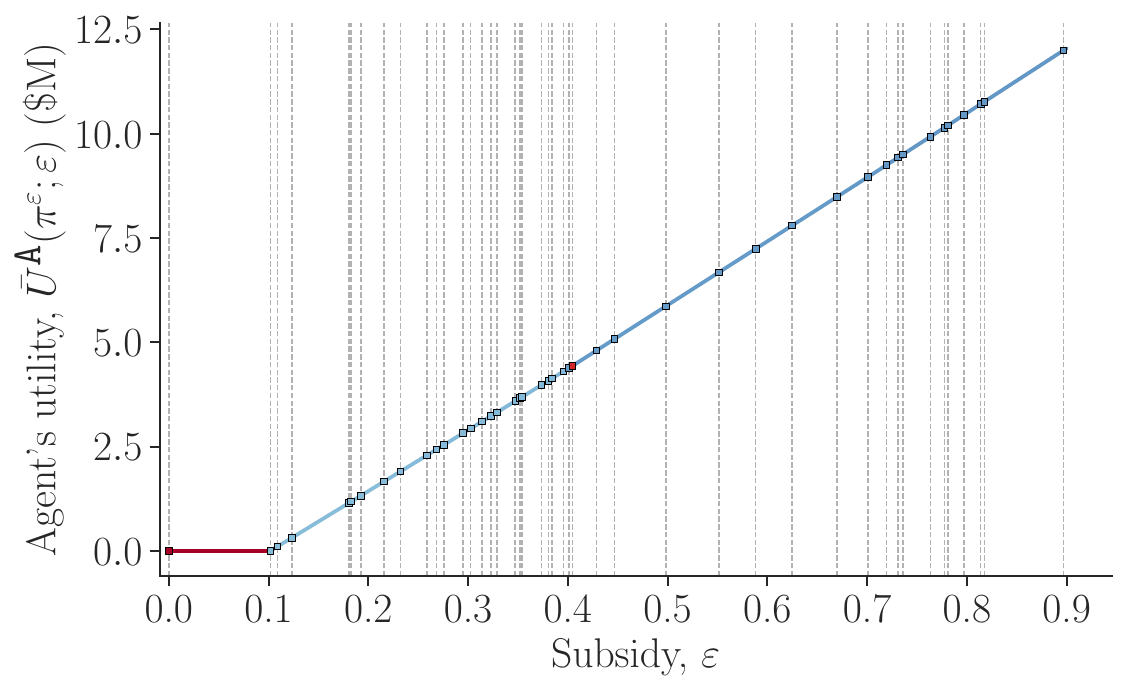}
    }
    \hspace{0mm}
    \subfloat[ Agent utility $U^{\texttt{A}}(\pi^\varepsilon;\varepsilon)$ in the approval process]{
    \includegraphics[width=0.45\linewidth]{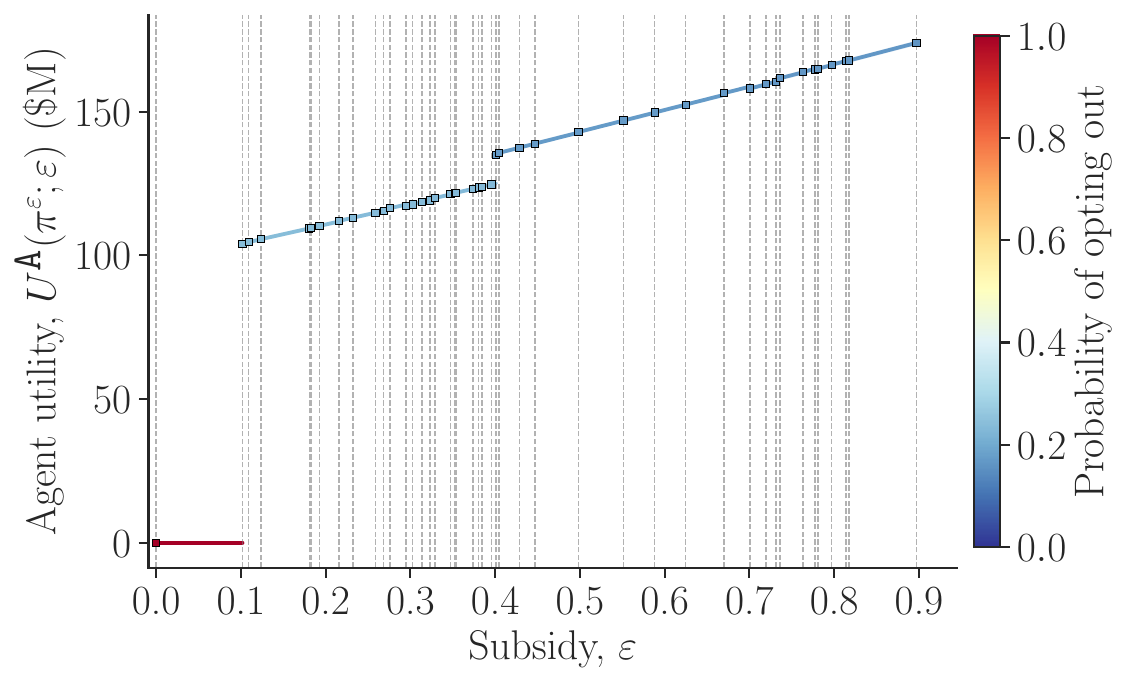}
    }
    \caption{\textbf{Agent utilities under a pessimistic prior.}
    The left panel shows the agent's utility (Eq.~\ref{eq:adaptive-utility}) computed using the belief MDP $\Mcal^\varepsilon$ when the agent uses the optimal policy for each subsidy, which is a piece-wise linear, convex and continuous function in accordance with Proposition~\ref{prop:utility-piecewise-convex}. The right panel shows the true utility of the agent (Eq.~\ref{eq:true-utilities}) in the approval process when using the optimal policy $\pi^\varepsilon$ for each subsidy and $\theta^* = 0.65$.
    The dashed vertical lines correspond to the intervals of the partition $\mathcal{P}$ where the agent's optimal policy is constant (Proposition~\ref{prop:utility-piecewise-convex}).
    }
    \label{fig:agent_utilities_pessimist}
\end{figure}

\begin{figure}[H]
    \vspace{0cm}
    \centering
    \subfloat[ Social utility $\Bar{U}^{\texttt{S}}(\varepsilon;\pi^\varepsilon)$ computed using $\Mcal^\varepsilon$ ]{
    \includegraphics[width=0.45\linewidth]{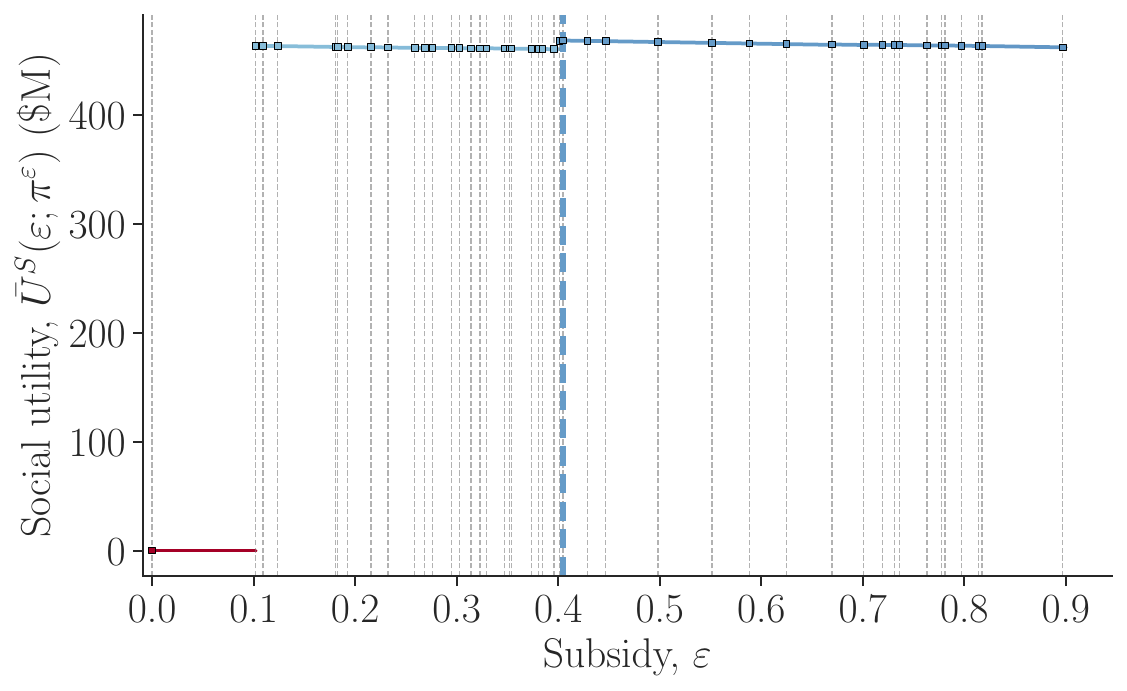}
    }
    \hspace{0mm}
    \subfloat[ Social utility $U^{\texttt{S}}(\varepsilon;\pi^\varepsilon)$ in the approval process]{
    \includegraphics[width=0.45\linewidth]{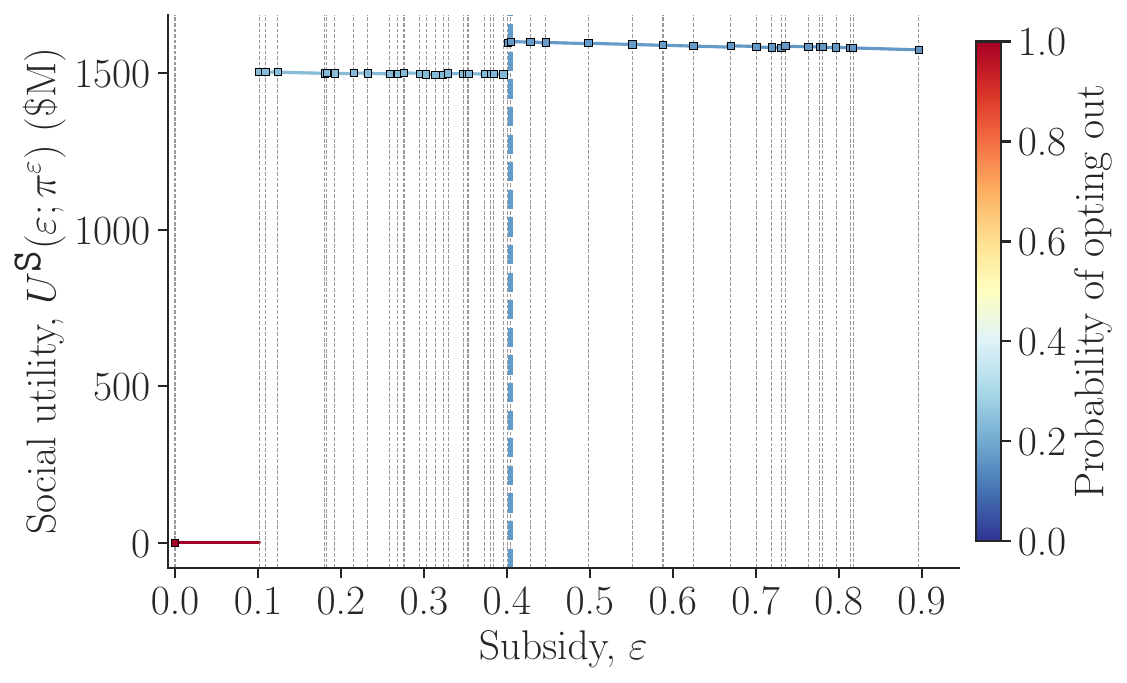}
    }
    \caption{\textbf{Social utilities under a pessimistic prior.}
    The left panel shows the social utility (Eq.~\ref{eq:adaptive-utility-regulator}) computed using the belief MDP $\Mcal^\varepsilon$ when the agent uses the optimal policy for each subsidy. The right panel shows the true social utility (Eq.~\ref{eq:true-utilities}) in the approval process when the agent uses the optimal policy $\pi^\varepsilon$ for each subsidy and $\theta^* = 0.65$.
    The dashed vertical lines correspond to the intervals of the partition $\mathcal{P}$ where the agent's optimal policy is constant (Proposition~\ref{prop:utility-piecewise-convex}).
    }
\label{fig:social_utilities_pessimist}
\end{figure}

In Figure~\ref{fig:pessimist_subsidy} we show that the optimal subsidy increases with $\rho^{\texttt{S}} / \rho^{\texttt{A}}$, and in Figure~\ref{fig:pessimist_gain} that our sequential protocol yields gains $>40\%$ in social utility relative to a non-sequential protocol.

\begin{figure}[H]
    \vspace{0cm}
    \centering
    \includegraphics[width=0.60\linewidth]{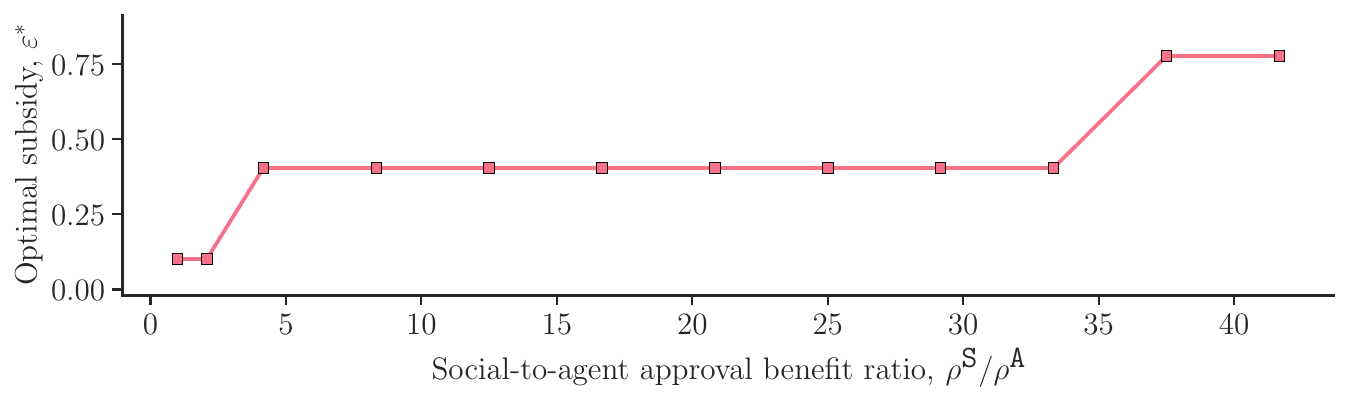}
    \caption{\textbf{Optimal subsidy vs. $\rho^{\texttt{S}} / \rho^{\texttt{A}}$.} The figure shows, as a function of the social-to-agent approval benefit ratio, the optimal subsidy obtained using Algorithm~\ref{alg:epsilon}.
    }
    \label{fig:pessimist_subsidy}
\end{figure}

\begin{figure}[H]
    \vspace{0cm}
    \centering
    \includegraphics[width=0.60\linewidth]{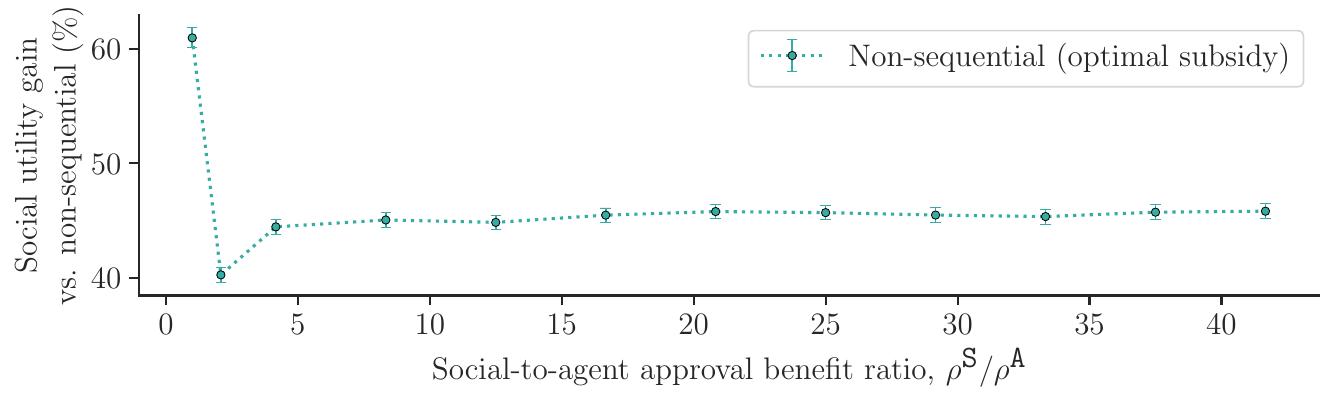}
    \caption{\textbf{Social utility gain vs. $\rho^{\texttt{S}} / \rho^{\texttt{A}}$.} The figure shows, as a function of the social-to-agent approval benefit ratio, the percentage increase in social utility of the sequential approval protocol relative to a non-sequential approval protocol in which the agent is restricted to a single trial with $n^{\texttt{max}}=800$, under (i) the optimal subsidy computed using Algorithm~\ref{alg:epsilon} (in the non-sequential protocol without subsidy the agent always opts out, yielding zero social utility).
    }
    \label{fig:pessimist_gain}
\end{figure}

\clearpage
\newpage

\subsubsection{Approval under optimistic prior}

Here, we show the results of an antibiotic approval process where the agent's prior is $(\alpha_0,\beta_0)=(4,1)$, that is, the agent is slightly optimistic about the antibiotic. Note that the mean efficacy of the drug according to its prior is then $0.8 > \theta^* = 0.65$. The principal knows such prior, and the rest of the parameters in Table~\ref{tab:non-econ-parameters} and Table~\ref{tab:econ-parameters} are fixed. In this case, the optimal subsidy is $\varepsilon^* = 0$, and as can be seen in Figure~\ref{fig:social_utilities_optimist}, the optimal subsidy maximizing $\Bar{U}^{\texttt{S}}(\varepsilon;\pi^\varepsilon)$ does not necessarily maximize the true (unknown) utility $U^{\texttt{S}}(\varepsilon;\pi^\varepsilon)$. Nevertheless, Figure~\ref{fig:optimist_gain} shows that the proposed sequential subsidized protocol can still substantially improve social utility relative to a non-sequential protocol, with gains exceeding approximately $18\%$.

\begin{figure}[H]
    \vspace{0cm}
    \centering
    \subfloat[ Agent utility $\Bar{U}^{\texttt{A}}(\pi^\varepsilon;\varepsilon)$ computed using $\Mcal^\varepsilon$ ]{
    \includegraphics[width=0.45\linewidth]{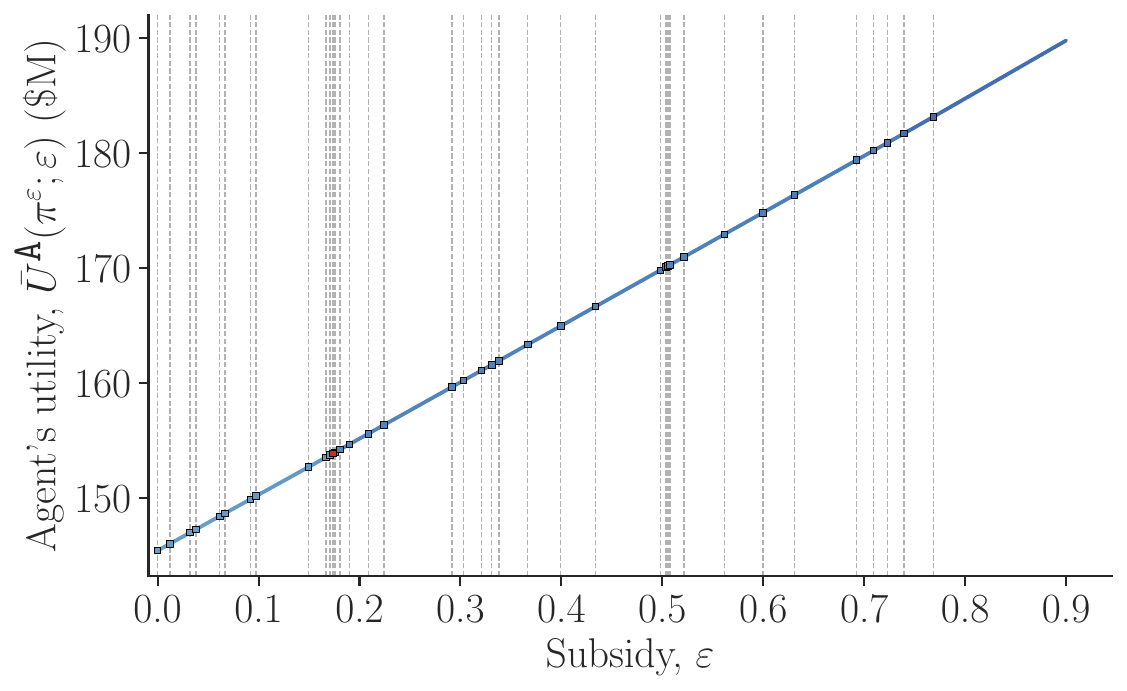}
    }
    \hspace{0mm}
    \subfloat[ Agent utility $U^{\texttt{A}}(\pi^\varepsilon;\varepsilon)$ in the approval process]{
    \includegraphics[width=0.45\linewidth]{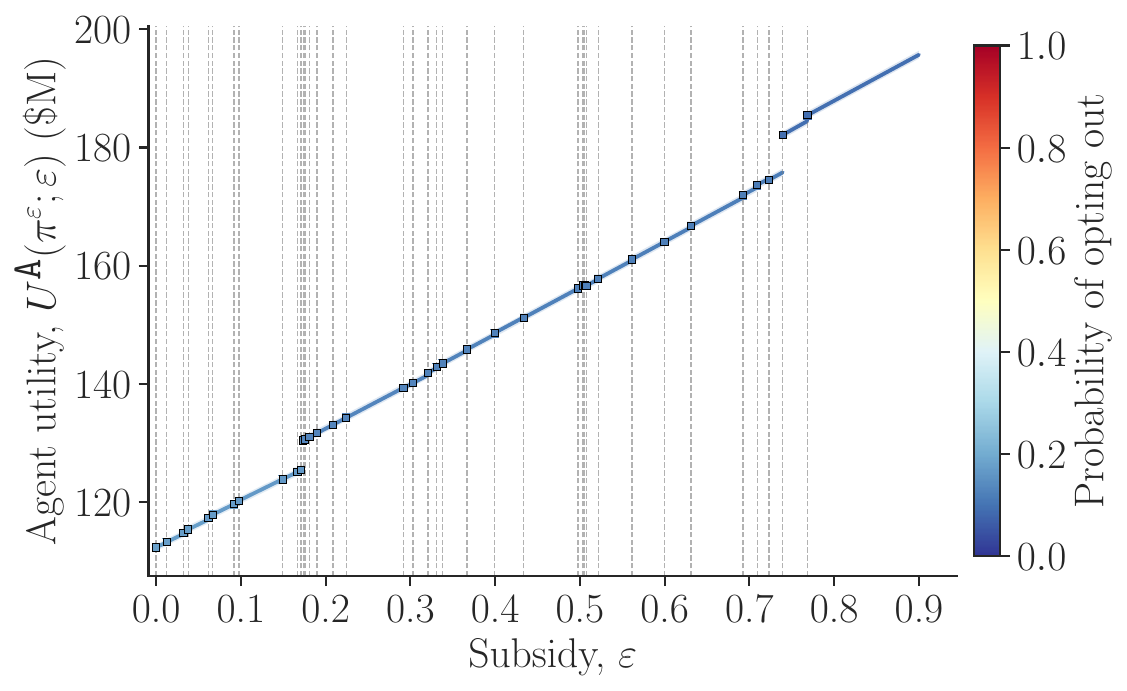}
    }
    \caption{\textbf{Agent utilities under an optimist prior.}
    The left panel shows the agent's utility (Eq.~\ref{eq:adaptive-utility}) computed using the belief MDP $\Mcal^\varepsilon$ when the agent uses the optimal policy for each subsidy, which is a piece-wise linear, convex and continuous function in accordance with Proposition~\ref{prop:utility-piecewise-convex}. The right panel shows the true utility of the agent (Eq.~\ref{eq:true-utilities}) in the approval process when using the optimal policy $\pi^\varepsilon$ for each subsidy and $\theta^* = 0.65$.
    The dashed vertical lines correspond to the intervals of the partition $\mathcal{P}$ where the agent's optimal policy is constant (Proposition~\ref{prop:utility-piecewise-convex}).
    }
    \label{fig:agent_utilities_optimist}
\end{figure}

\begin{figure}[H]
    \vspace{0cm}
    \centering
    \subfloat[ Social utility $\Bar{U}^{\texttt{S}}(\varepsilon;\pi^\varepsilon)$ computed using $\Mcal^\varepsilon$ ]{
    \includegraphics[width=0.45\linewidth]{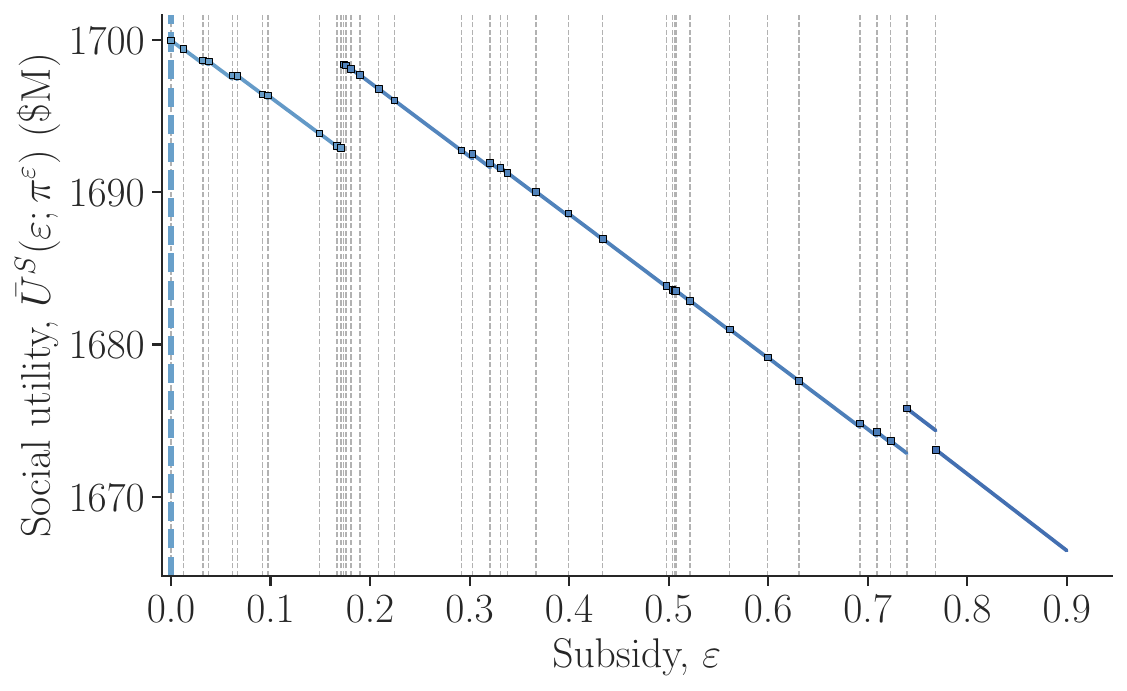}
    }
    \hspace{0mm}
    \subfloat[ Social utility $U^{\texttt{S}}(\varepsilon;\pi^\varepsilon)$ in the approval process]{
    \includegraphics[width=0.45\linewidth]{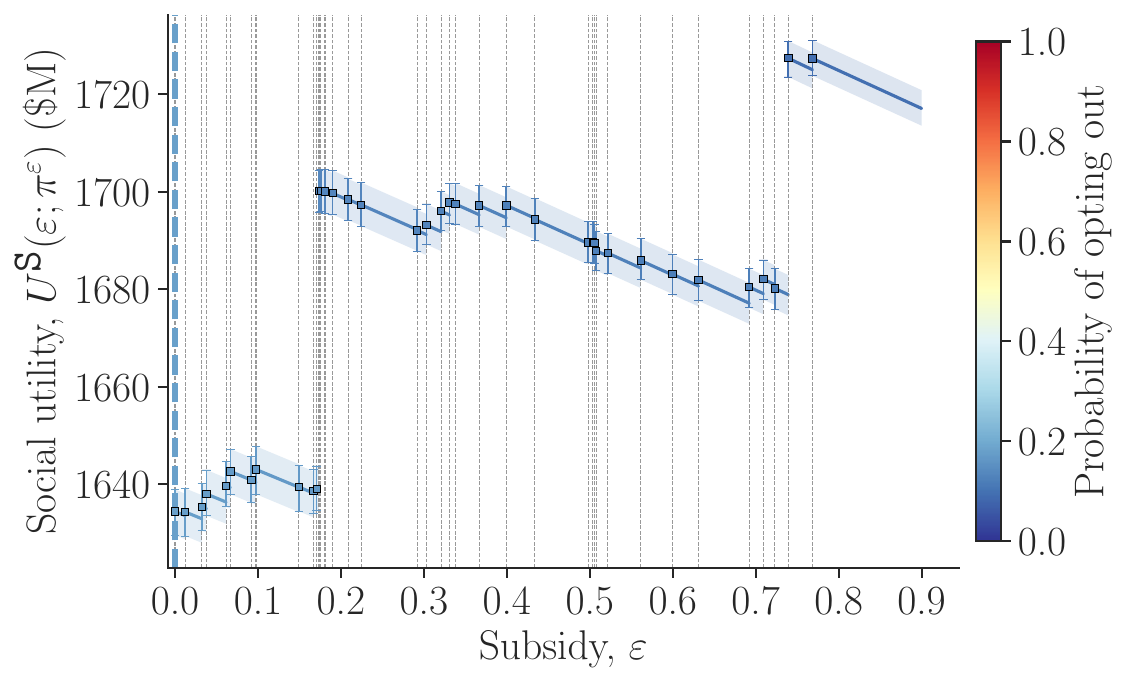}
    }
    \caption{\textbf{Social utilities under an optimist prior.}
    The left panel shows the social utility (Eq.~\ref{eq:adaptive-utility-regulator}) computed using the belief MDP $\Mcal^\varepsilon$ when the agent uses the optimal policy for each subsidy. The right panel shows the true social utility (Eq.~\ref{eq:true-utilities}) in the approval process when the agent uses the optimal policy $\pi^\varepsilon$ for each subsidy and $\theta^* = 0.65$.
    The dashed vertical lines correspond to the intervals of the partition $\mathcal{P}$ where the agent's optimal policy is constant (Proposition~\ref{prop:utility-piecewise-convex}).
    }
\label{fig:social_utilities_optimist}
\end{figure}

\begin{figure}[H]
    \vspace{0cm}
    \centering
    \includegraphics[width=0.60\linewidth]{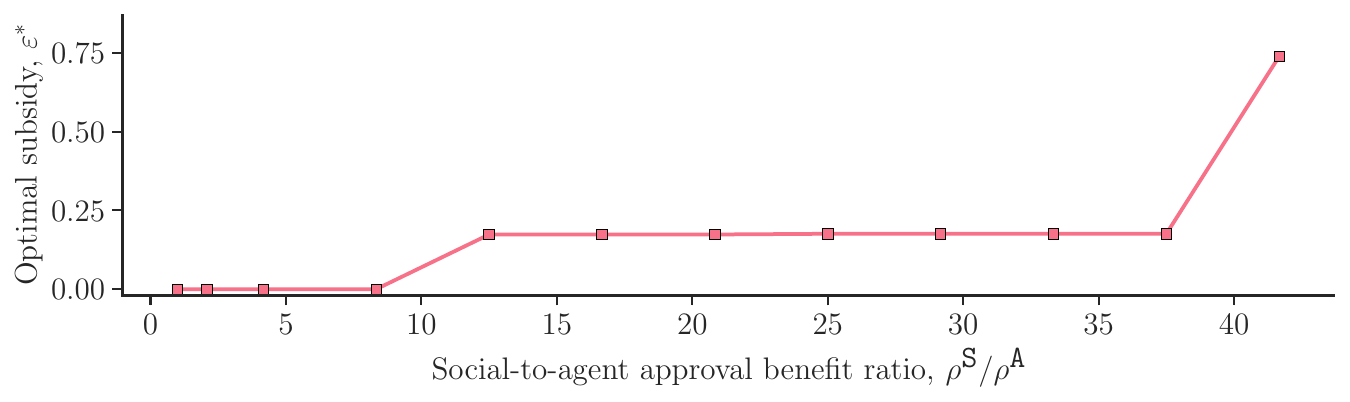}
    \caption{\textbf{Optimal subsidy vs. $\rho^{\texttt{S}} / \rho^{\texttt{A}}$.} The figure shows, as a function of the social-to-agent approval benefit ratio, the optimal subsidy obtained using Algorithm~\ref{alg:epsilon}.
    }
    \label{fig:optimist_subsidy}
\end{figure}

\begin{figure}[H]
    \vspace{0cm}
    \centering
    \includegraphics[width=0.60\linewidth]{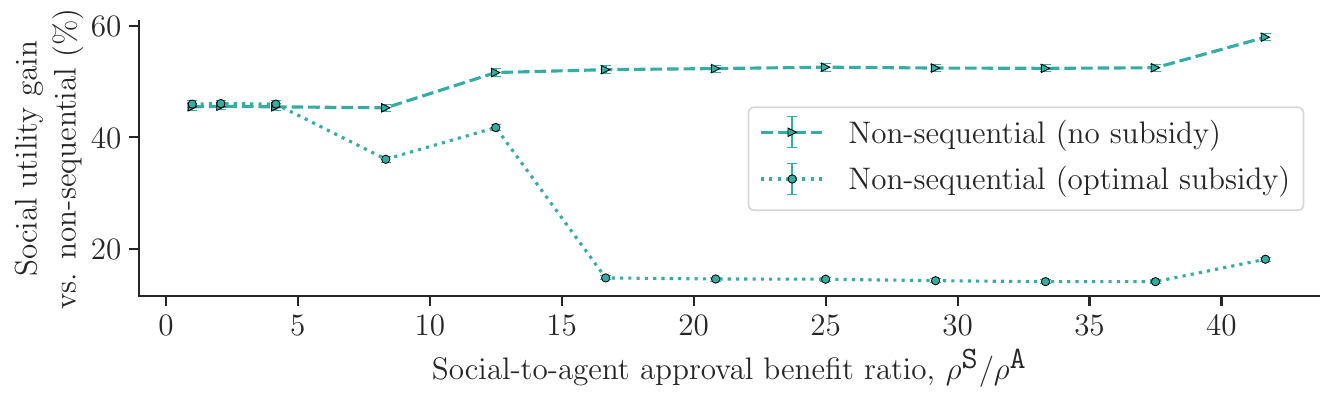}
    \caption{\textbf{Social utility gain vs. $\rho^{\texttt{S}} / \rho^{\texttt{A}}$.} The figure shows, as a function of the social-to-agent approval benefit ratio, the percentage increase in social utility of the sequential approval protocol relative to a non-sequential approval protocol in which the agent is restricted to a single trial with $n^{\texttt{max}}=800$, under (i) the optimal subsidy computed using Algorithm~\ref{alg:epsilon} and (ii) no subsidy ($\varepsilon=0$).
    }
    \label{fig:optimist_gain}
\end{figure}

\clearpage
\newpage

\subsubsection{Approval under calibrated prior}

Here, we show the results of an antibiotic approval process where the agent's prior is $(\alpha_0,\beta_0)=(130,70)$. This corresponds to a very informative prior that is calibrated to the true efficacy $\theta^* = 0.65$, since the mean of the prior is precisely $0.65$. The principal knows such prior, and the rest of the parameters in Table~\ref{tab:non-econ-parameters} and Table~\ref{tab:econ-parameters} are fixed. In this case, we find that the calibrated prior allows the agent to increase its utility $U^{\texttt{A}}(\pi^\varepsilon;\varepsilon)$ for any possible subsidy, as can be seen by comparing Figure~\ref{fig:agent_utilities_calibrated} to the fiducial setting in Figure~\ref{fig:agent_utilities_fiducial}. However, and perhaps surprisingly, we also find that it is still optimal for the principal to subsidize a non-negligible fraction $\varepsilon^* \approx 0.234$ of the agent's cost, despite the prior belief supporting that the drug should be approved.

\begin{figure}[H]
    \vspace{0cm}
    \centering
    \subfloat[ Agent utility $\Bar{U}^{\texttt{A}}(\pi^\varepsilon;\varepsilon)$ computed using $\Mcal^\varepsilon$ ]{
    \includegraphics[width=0.45\linewidth]{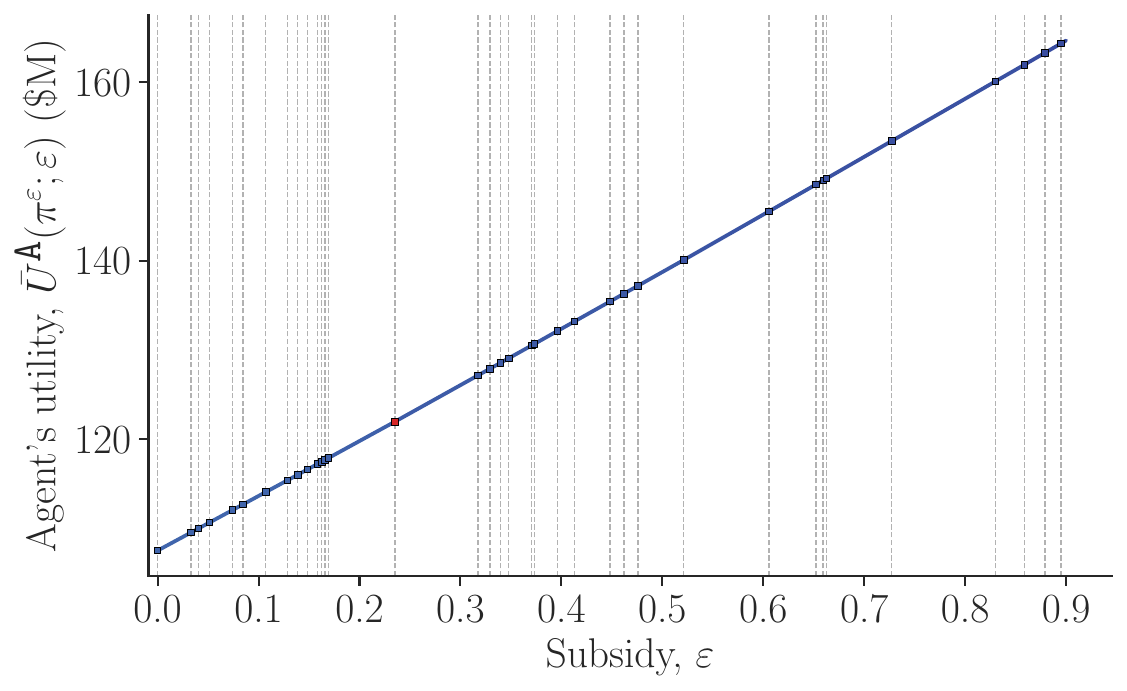}
    }
    \hspace{0mm}
    \subfloat[ Agent utility $U^{\texttt{A}}(\pi^\varepsilon;\varepsilon)$ in the approval process]{
    \includegraphics[width=0.45\linewidth]{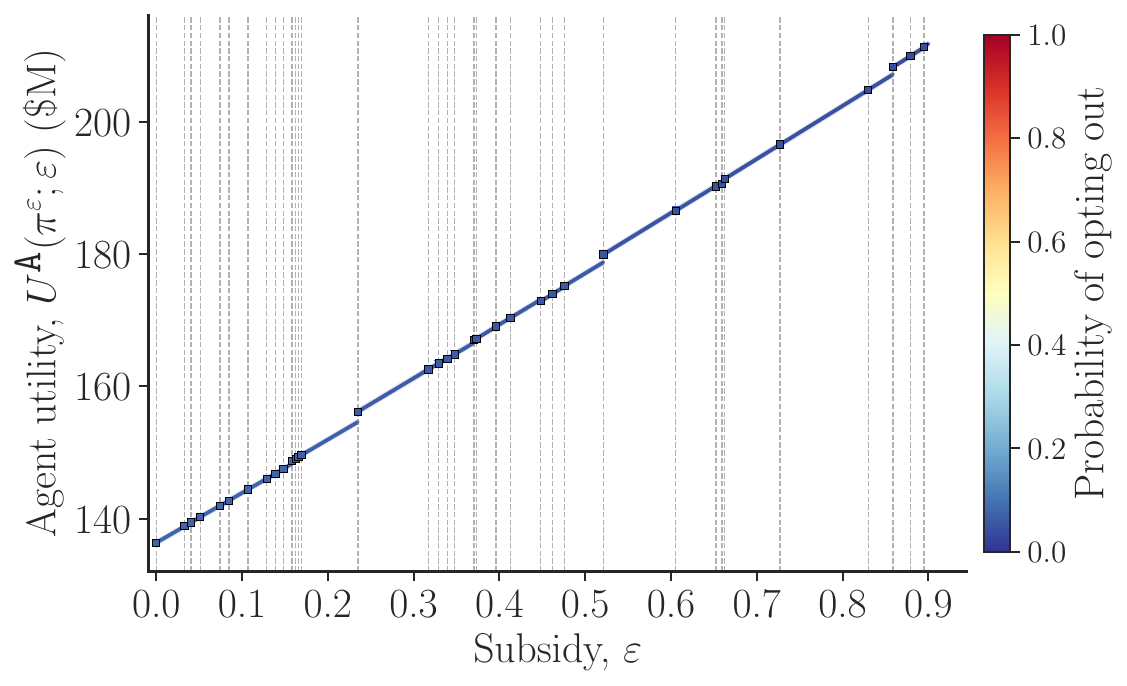}
    }
    \caption{\textbf{Agent utilities under a calibrated informative prior.}
    The left panel shows the agent's utility (Eq.~\ref{eq:adaptive-utility}) computed using the belief MDP $\Mcal^\varepsilon$ when the agent uses the optimal policy for each subsidy, which is a piece-wise linear, convex and continuous function in accordance with Proposition~\ref{prop:utility-piecewise-convex}. The right panel shows the true utility of the agent (Eq.~\ref{eq:true-utilities}) in the approval process when using the optimal policy $\pi^\varepsilon$ for each subsidy and $\theta^* = 0.65$.
    The dashed vertical lines correspond to the intervals of the partition $\mathcal{P}$ where the agent's optimal policy is constant (Proposition~\ref{prop:utility-piecewise-convex}).
    }
    \label{fig:agent_utilities_calibrated}
\end{figure}

\begin{figure}[H]
    \vspace{0cm}
    \centering
    \subfloat[ Social utility $\Bar{U}^{\texttt{S}}(\varepsilon;\pi^\varepsilon)$ computed using $\Mcal^\varepsilon$ ]{
    \includegraphics[width=0.45\linewidth]{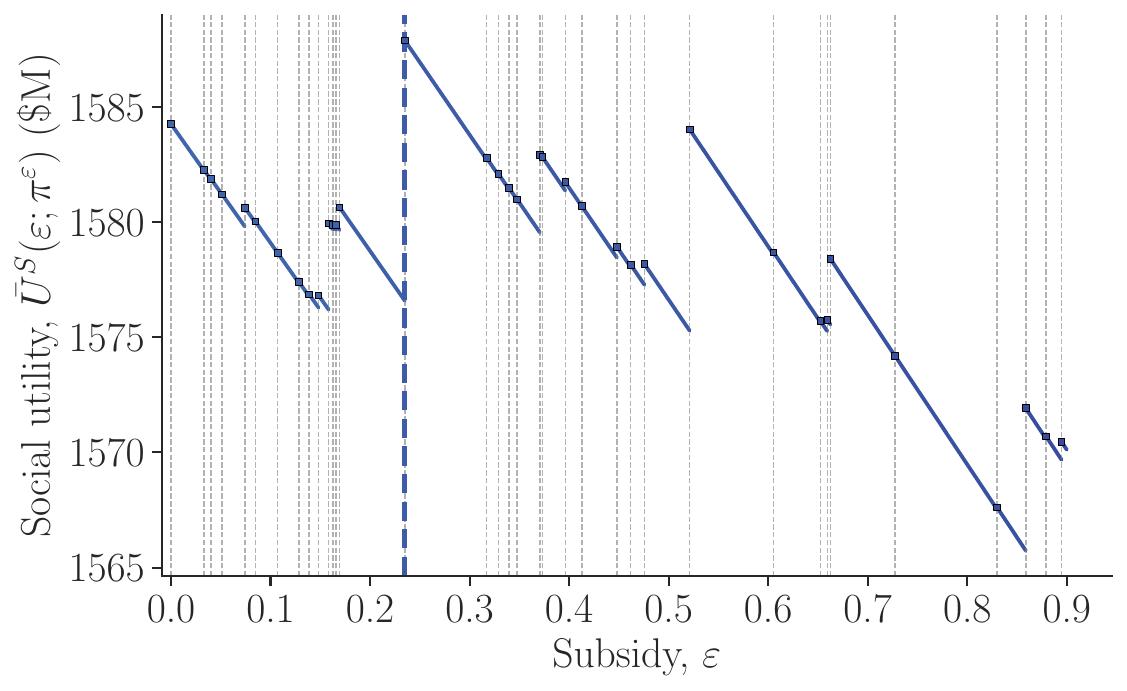}
    }
    \hspace{0mm}
    \subfloat[ Social utility $U^{\texttt{S}}(\varepsilon;\pi^\varepsilon)$ in the approval process]{
    \includegraphics[width=0.45\linewidth]{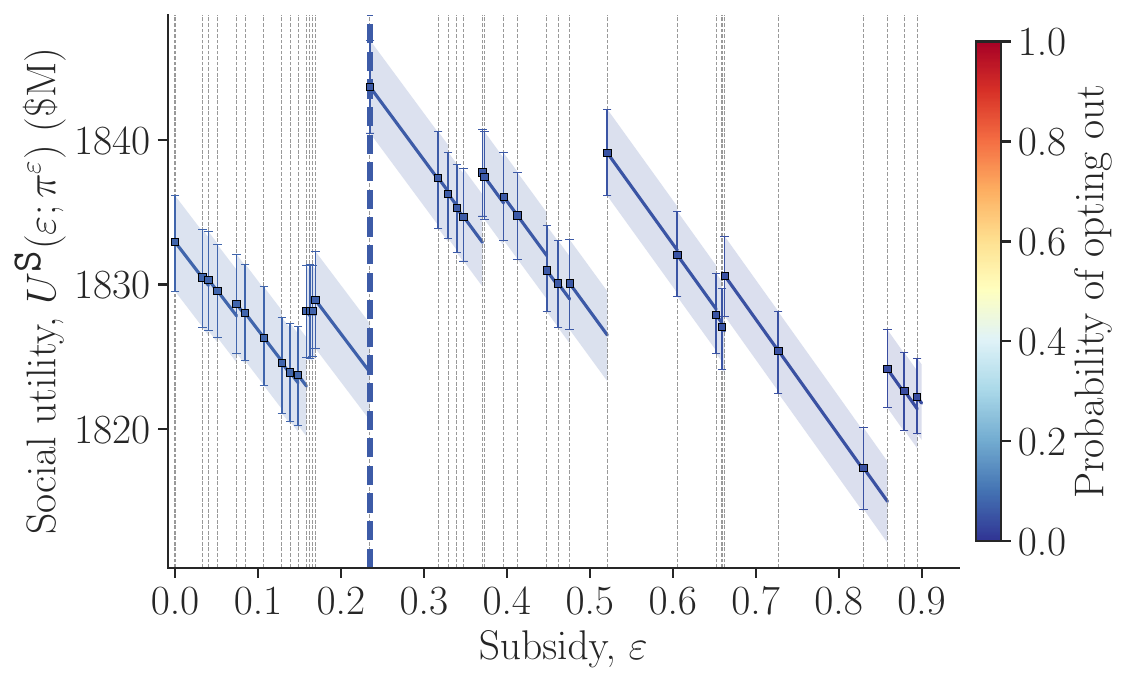}
    }
    \caption{\textbf{Social utilities under a calibrated informative prior.}
    The left panel shows the social utility (Eq.~\ref{eq:adaptive-utility-regulator}) computed using the belief MDP $\Mcal^\varepsilon$ when the agent uses the optimal policy for each subsidy. The right panel shows the true social utility (Eq.~\ref{eq:true-utilities}) in the approval process when the agent uses the optimal policy $\pi^\varepsilon$ for each subsidy and $\theta^* = 0.65$.
    The dashed vertical lines correspond to the intervals of the partition $\mathcal{P}$ where the agent's optimal policy is constant (Proposition~\ref{prop:utility-piecewise-convex}).
    }
\label{fig:social_utilities_calibrated}
\end{figure}

\begin{figure}[H]
    \vspace{0cm}
    \centering
    \includegraphics[width=0.60\linewidth]{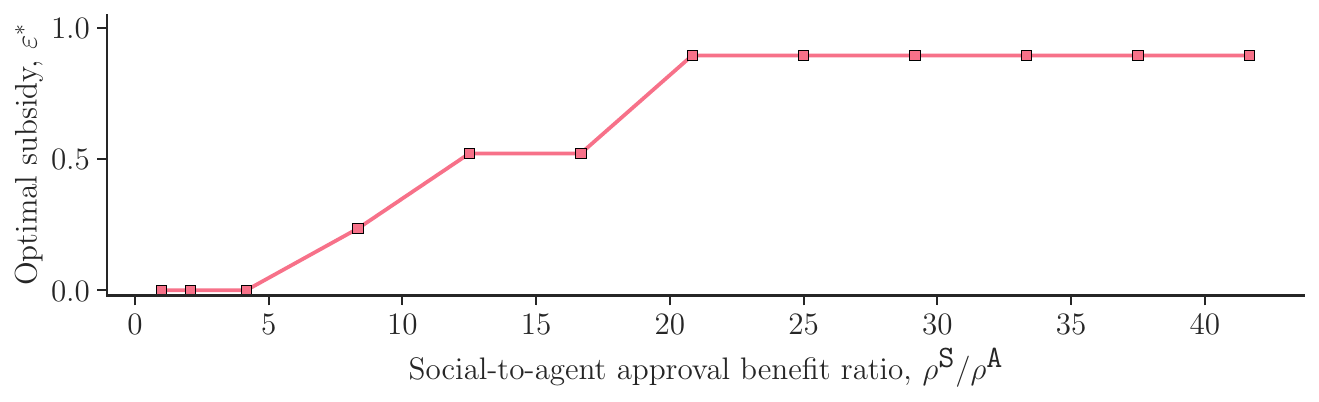}
    \caption{\textbf{Optimal subsidy vs. $\rho^{\texttt{S}} / \rho^{\texttt{A}}$.} The figure shows, as a function of the social-to-agent approval benefit ratio, the optimal subsidy obtained using Algorithm~\ref{alg:epsilon}.
    }
    \label{fig:calibrated_subsidy}
\end{figure}

Lastly, Figure~\ref{fig:calibrated_gain} shows that the social utility gain relative to a non-sequential protocol with optimal subsidies is small for high values of $\rho^{\texttt{S}} / \rho^{\texttt{A}}$. This is expected because under a calibrated prior, the agent already has an accurate estimate of the product’s efficacy and therefore gathering new information brings little benefit. As a result, the agent can effectively select an optimal sample size and complete the process in a single trial.

\begin{figure}[H]
    \vspace{0cm}
    \centering
    \includegraphics[width=0.60\linewidth]{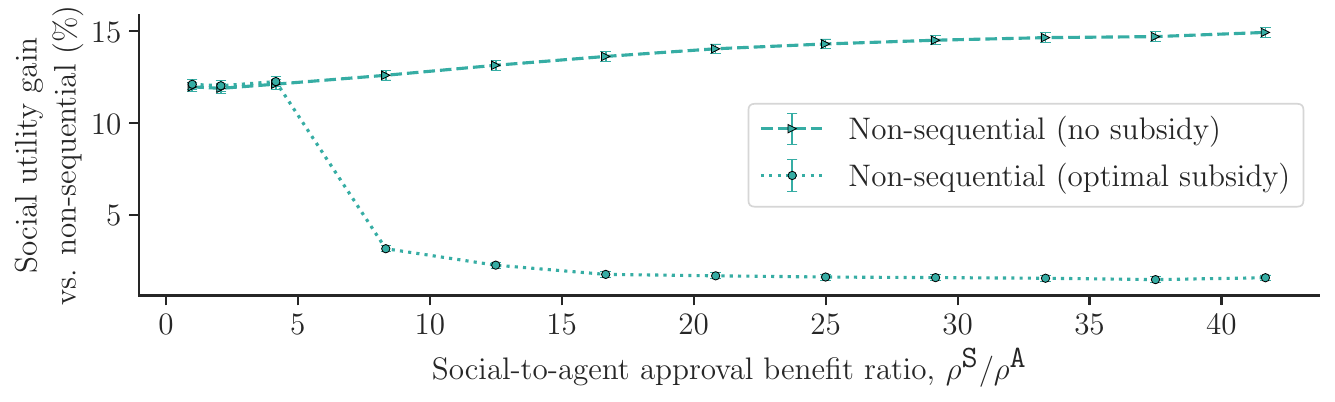}
    \caption{\textbf{Social utility gain vs. $\rho^{\texttt{S}} / \rho^{\texttt{A}}$.} The figure shows, as a function of the social-to-agent approval benefit ratio, the percentage increase in social utility of the sequential approval protocol relative to a non-sequential approval protocol in which the agent is restricted to a single trial with $n^{\texttt{max}}=800$, under (i) the optimal subsidy computed using Algorithm~\ref{alg:epsilon} and (ii) no subsidy ($\varepsilon=0$).
    }
    \label{fig:calibrated_gain}
\end{figure}

\clearpage
\newpage

\subsubsection{Approval under uncalibrated prior}

Here, we show the results of an antibiotic approval process where the agent's prior is $(\alpha_0,\beta_0)=(130,30)$. This corresponds to a very informative and optimistic prior that is uncalibrated to the true efficacy $\theta^* = 0.65$, since the mean of the prior is precisely $0.8125$. The principal knows such prior, and the rest of the parameters in Table~\ref{tab:non-econ-parameters} and Table~\ref{tab:econ-parameters} are fixed. We find that, across subsidies, the utility $U^{\texttt{A}}(\pi^\varepsilon;\varepsilon)$ achieved by the agent slightly decreased compared to an agent with a calibrated prior (Figure~\ref{fig:agent_utilities_uncalibrated} vs. Figure~\ref{fig:agent_utilities_calibrated}). However, surprisingly, we also find that the social utility is greater in this case compared to the case where the agent has a calibrated prior (Figure~\ref{fig:social_utilities_uncalibrated} vs. Figure~\ref{fig:social_utilities_calibrated}).

\begin{figure}[H]
    \vspace{0cm}
    \centering
    \subfloat[ Agent utility $\Bar{U}^{\texttt{A}}(\pi^\varepsilon;\varepsilon)$ computed using $\Mcal^\varepsilon$ ]{
    \includegraphics[width=0.45\linewidth]{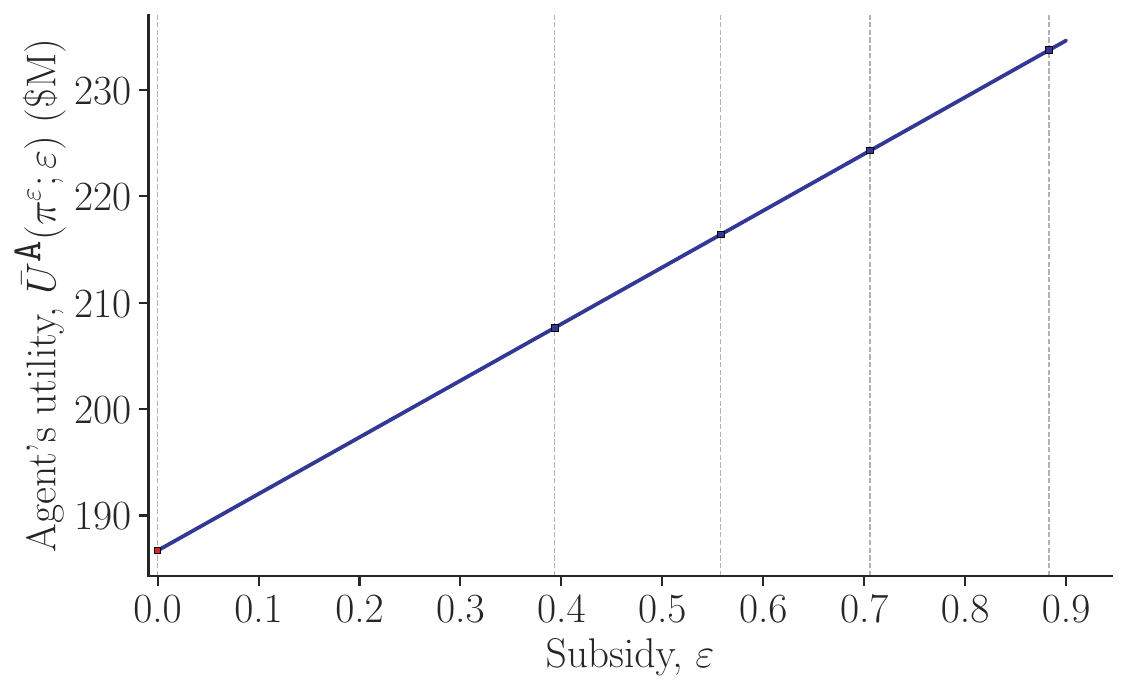}
    }
    \hspace{0mm}
    \subfloat[ Agent utility $U^{\texttt{A}}(\pi^\varepsilon;\varepsilon)$ in the approval process]{
    \includegraphics[width=0.45\linewidth]{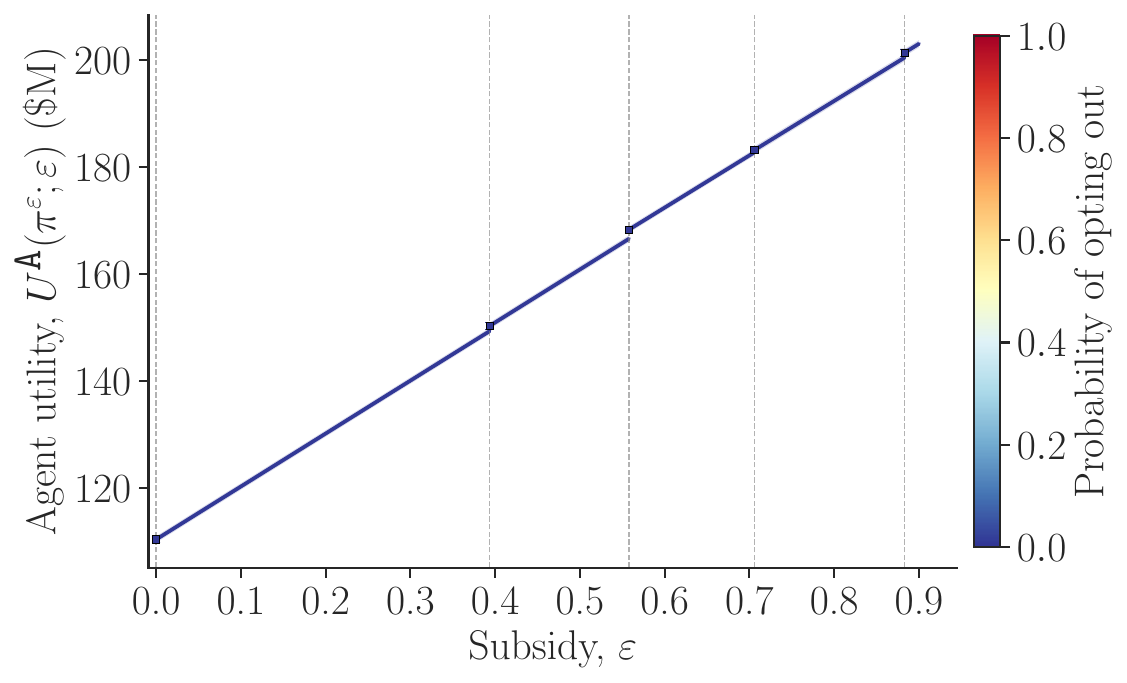}
    }
    \caption{\textbf{Agent utilities under an uncalibrated informative prior.}
    The left panel shows the agent's utility (Eq.~\ref{eq:adaptive-utility}) computed using the belief MDP $\Mcal^\varepsilon$ when the agent uses the optimal policy for each subsidy, which is a piece-wise linear, convex and continuous function in accordance with Proposition~\ref{prop:utility-piecewise-convex}. The right panel shows the true utility of the agent (Eq.~\ref{eq:true-utilities}) in the approval process when using the optimal policy $\pi^\varepsilon$ for each subsidy and $\theta^* = 0.65$.
    The dashed vertical lines correspond to the intervals of the partition $\mathcal{P}$ where the agent's optimal policy is constant (Proposition~\ref{prop:utility-piecewise-convex}).
    }
    \label{fig:agent_utilities_uncalibrated}
\end{figure}

\begin{figure}[H]
    \vspace{0cm}
    \centering
    \subfloat[ Social utility $\Bar{U}^{\texttt{S}}(\varepsilon;\pi^\varepsilon)$ computed using $\Mcal^\varepsilon$ ]{
    \includegraphics[width=0.45\linewidth]{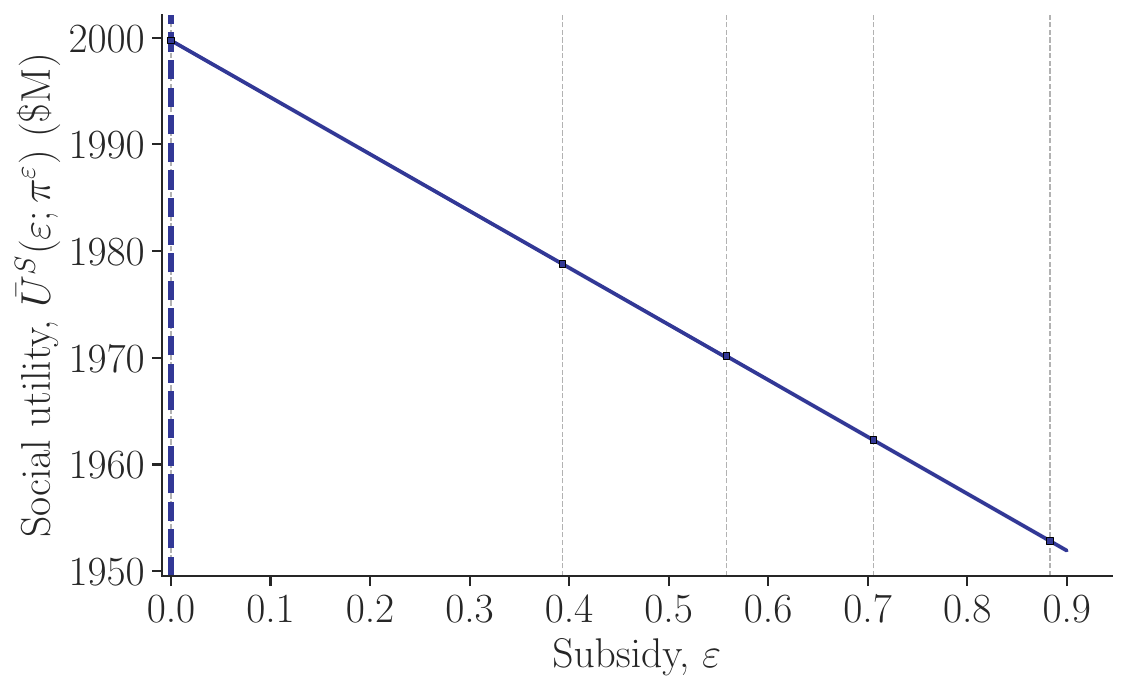}
    }
    \hspace{0mm}
    \subfloat[ Social utility $U^{\texttt{S}}(\varepsilon;\pi^\varepsilon)$ in the approval process]{
    \includegraphics[width=0.45\linewidth]{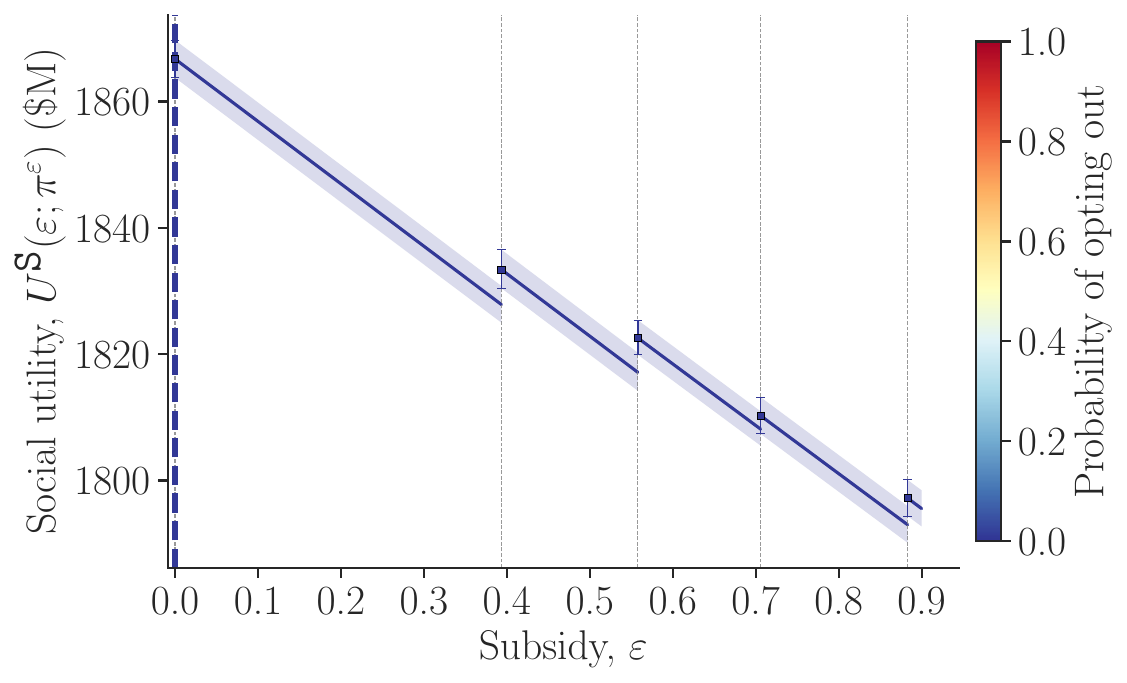}
    }
    \caption{\textbf{Social utilities under an uncalibrated informative prior.}
    The left panel shows the social utility (Eq.~\ref{eq:adaptive-utility-regulator}) computed using the belief MDP $\Mcal^\varepsilon$ when the agent uses the optimal policy for each subsidy. The right panel shows the true social utility (Eq.~\ref{eq:true-utilities}) in the approval process when the agent uses the optimal policy $\pi^\varepsilon$ for each subsidy and $\theta^* = 0.65$.
    The dashed vertical lines correspond to the intervals of the partition $\mathcal{P}$ where the agent's optimal policy is constant (Proposition~\ref{prop:utility-piecewise-convex}).
    }
\label{fig:social_utilities_uncalibrated}
\end{figure}

Interestingly, Figure~\ref{fig:uncalibrated_subsidy} shows that the optimal subsidy remains zero as $\rho^{\texttt{S}} / \rho^{\texttt{A}}$ increases. At the same time, Figure~\ref{fig:uncalibrated_gain} indicates that our sequential protocol achieves gains exceeding $90\%$ in social utility relative to a non-sequential protocol. In other words, the protocol can substantially improve social utility even in the absence of subsidies.

\begin{figure}[H]
    \vspace{0cm}
    \centering
    \includegraphics[width=0.60\linewidth]{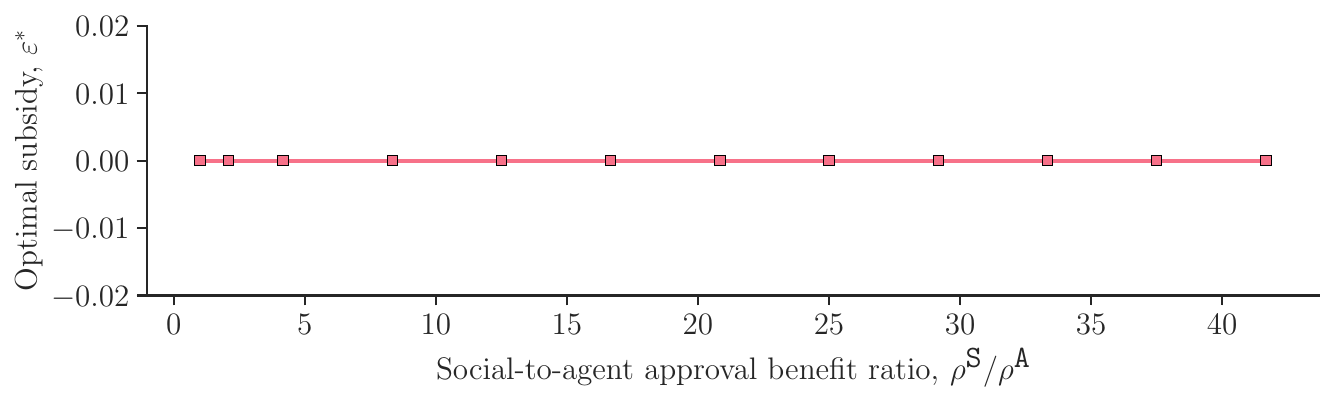}
    \caption{\textbf{Optimal subsidy vs. $\rho^{\texttt{S}} / \rho^{\texttt{A}}$.} The figure shows, as a function of the social-to-agent approval benefit ratio, the optimal subsidy obtained using Algorithm~\ref{alg:epsilon}.
    }
    \label{fig:uncalibrated_subsidy}
\end{figure}

\begin{figure}[H]
    \vspace{0cm}
    \centering
    \includegraphics[width=0.60\linewidth]{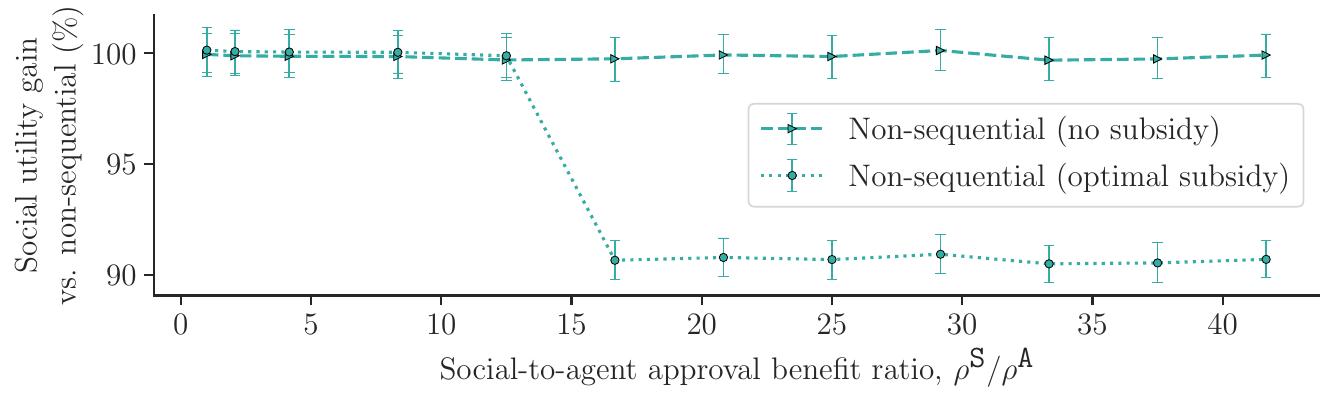}
    \caption{\textbf{Social utility gain vs. $\rho^{\texttt{S}} / \rho^{\texttt{A}}$.} The figure shows, as a function of the social-to-agent approval benefit ratio, the percentage increase in social utility of the sequential approval protocol relative to a non-sequential approval protocol in which the agent is restricted to a single trial with $n^{\texttt{max}}=800$, under (i) the optimal subsidy computed using Algorithm~\ref{alg:epsilon} and (ii) no subsidy ($\varepsilon=0$).
    }
    \label{fig:uncalibrated_gain}
\end{figure}

\clearpage
\newpage

\subsubsection{Approval under a different test process}\label{app:mixture-results}

In this section, we present additional experimental results for the antibiotic approval process described in Section~\ref{sec:experiments}, using an alternative test process to the process $M$ defined in Proposition~\ref{prop:binomial-e-value} and Eq.~\ref{eq:test-process}, while keeping all other parameters fixed as in Tables~\ref{tab:non-econ-parameters} and~\ref{tab:econ-parameters}. Here, our goal is to illustrate that Algorithm~\ref{alg:epsilon} can provide insight into how to optimally subsidize agents under different statistical tests.
More concretely, following Appendix~\ref{app:general-model}, we consider the mixed process $M^{\texttt{mix}}$ defined in Eq.~\ref{eq:mixture-process} with a uniform mixture $P^{\texttt{mix}} = \mathrm{U}(\theta^{\texttt{b}},1)$.
We find that the sequential subsidized protocol yields social utility gains of up to $15\%$ compared to a non-sequential baseline without subsidies. Relative to a non-sequential but optimally subsidized baseline, the gains can also reach up to $15\%$ when the ratio $\rho^{\texttt{S}} / \rho^{\texttt{A}}$ is low, but diminish rapidly as this ratio increases. The initial optimal action taken by the agent is $n\leq 114$ for all subsidies. Lastly, in Figure~\ref{fig:mixture_power}, we confirm that the mixed process $M^{\texttt{mix}}$ has higher power than $M$ for values of $\theta^*$ close to $\theta^{\texttt{b}}$, as discussed in Appendix~\ref{app:general-model}.

\begin{figure}[H]
    \vspace{0cm}
    \centering
    \subfloat[ Agent utility $\Bar{U}^{\texttt{A}}(\pi^\varepsilon;\varepsilon)$ computed using $\Mcal^\varepsilon$ ]{
    \includegraphics[width=0.45\linewidth]{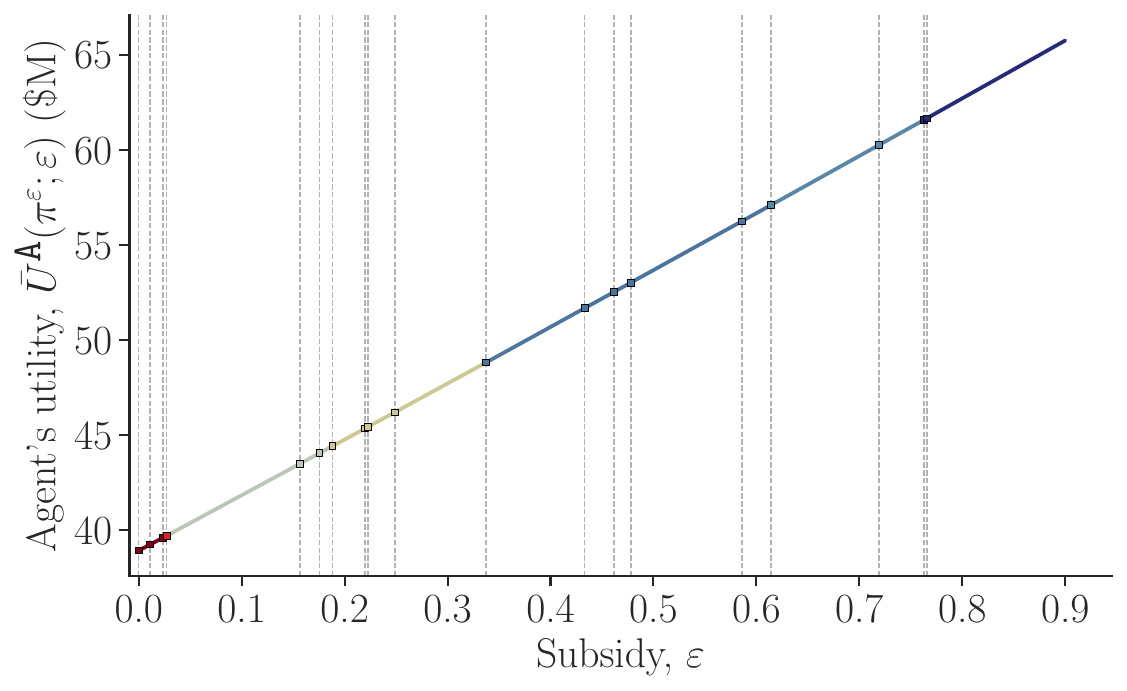}
    }
    \hspace{0mm}
    \subfloat[ Agent utility $U^{\texttt{A}}(\pi^\varepsilon;\varepsilon)$ in the approval process]{
    \includegraphics[width=0.45\linewidth]{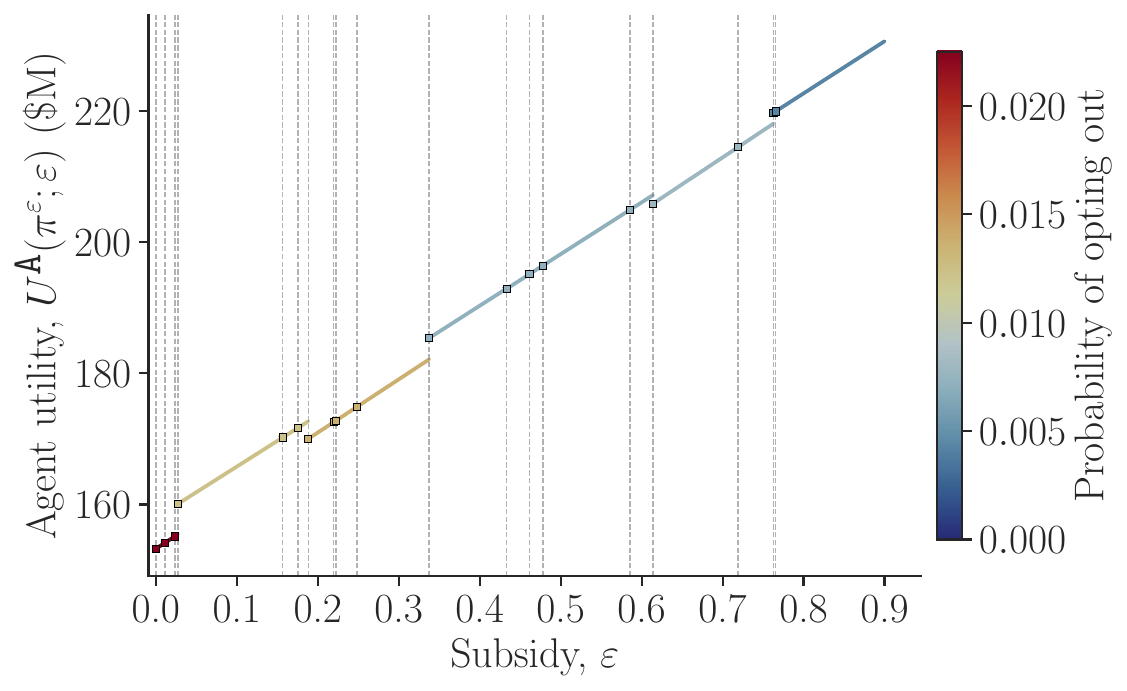}
    }
    \caption{\textbf{Agent utilities under a mixed test process.}
    The left panel shows the agent's utility (Eq.~\ref{eq:adaptive-utility}) computed using the belief MDP $\Mcal^\varepsilon$ when the agent uses the optimal policy for each subsidy, which is a piece-wise linear, convex, and continuous function in accordance with Proposition~\ref{prop:utility-piecewise-convex}. The right panel shows the true utility of the agent (Eq.~\ref{eq:true-utilities}) in the approval process when using the optimal policy $\pi^\varepsilon$ for each subsidy and $\theta^* = 0.65$.
    The dashed vertical lines correspond to the intervals of the partition $\mathcal{P}$ where the agent's optimal policy is constant (Proposition~\ref{prop:utility-piecewise-convex}).
    }
    \label{fig:agent_utilities_mixture}
\end{figure}

\begin{figure}[H]
    \vspace{0cm}
    \centering
    \subfloat[ Social utility $\Bar{U}^{\texttt{S}}(\varepsilon;\pi^\varepsilon)$ computed using $\Mcal^\varepsilon$ ]{
    \includegraphics[width=0.45\linewidth]{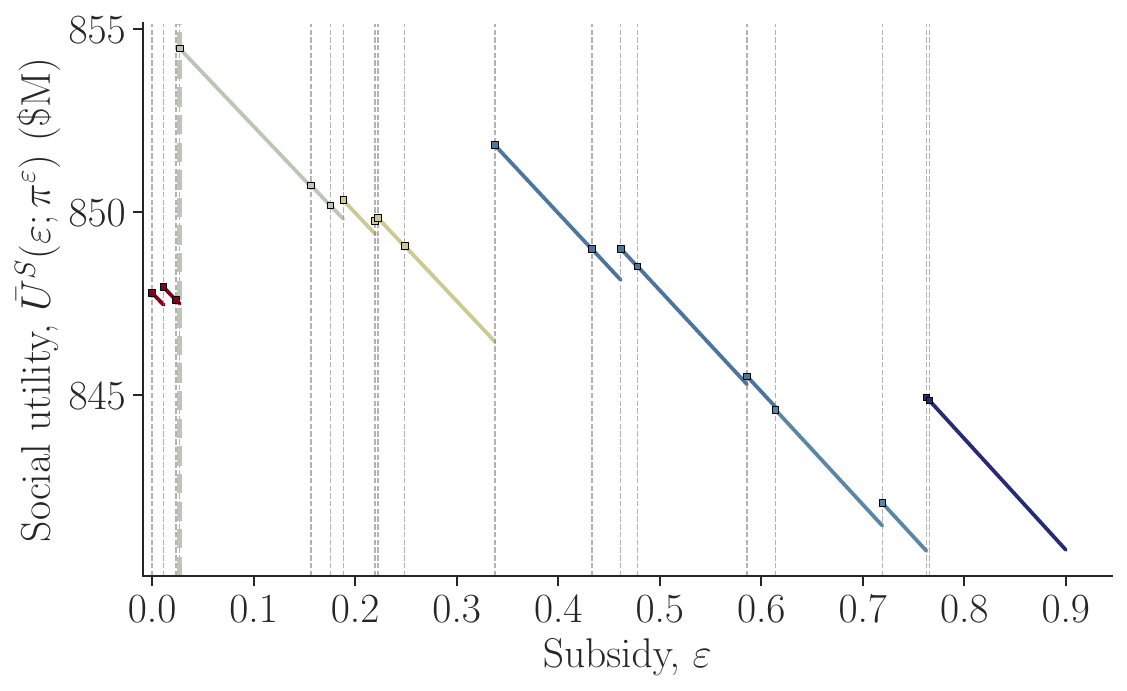}
    }
    \hspace{0mm}
    \subfloat[ Social utility $U^{\texttt{S}}(\varepsilon;\pi^\varepsilon)$ in the approval process]{
    \includegraphics[width=0.45\linewidth]{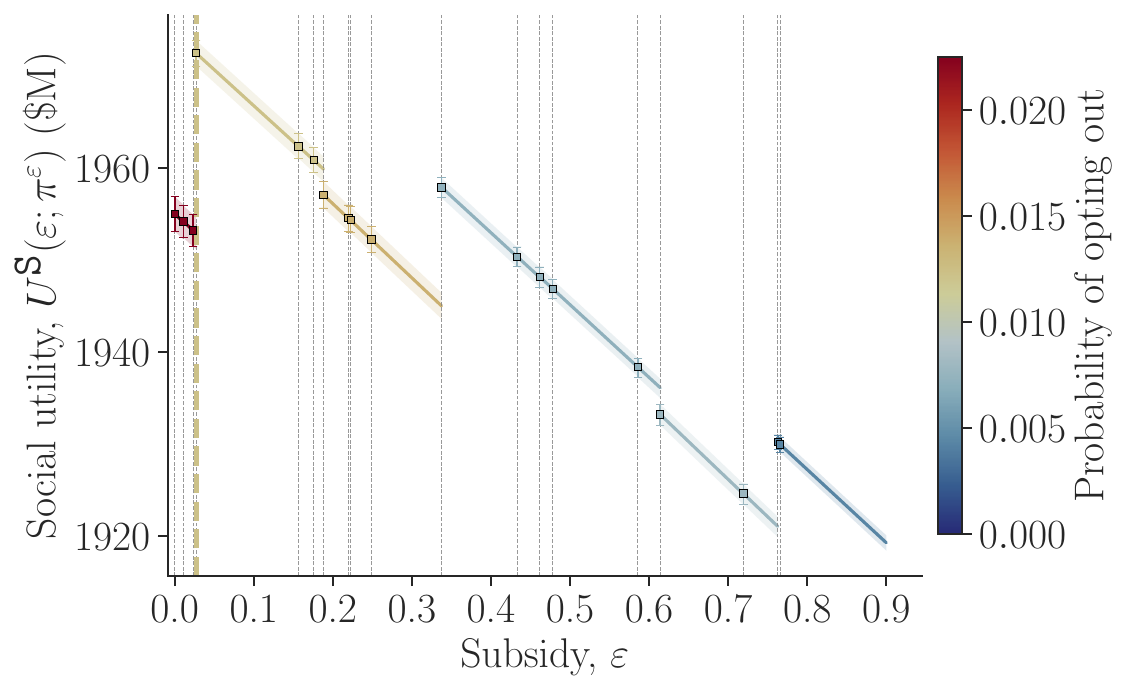}
    }
    \caption{\textbf{Social utilities under a mixed test process.}
    The left panel shows the social utility (Eq.~\ref{eq:adaptive-utility-regulator}) computed using the belief MDP $\Mcal^\varepsilon$ when the agent uses the optimal policy for each subsidy. The right panel shows the true social utility (Eq.~\ref{eq:true-utilities}) in the approval process when the agent uses the optimal policy $\pi^\varepsilon$ for each subsidy and $\theta^* = 0.65$.
    The dashed vertical lines correspond to the intervals of the partition $\mathcal{P}$ where the agent's optimal policy is constant (Proposition~\ref{prop:utility-piecewise-convex}).
    }
\label{fig:social_utilities_mixture}
\end{figure}

\begin{figure}[H]
    \vspace{0cm}
    \centering
    \includegraphics[width=0.60\linewidth]{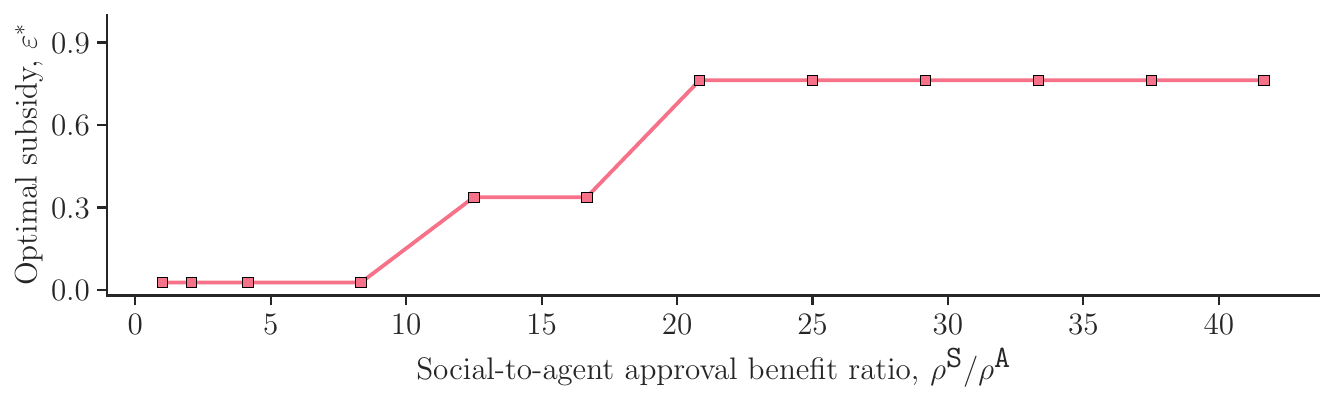}
    \caption{\textbf{Optimal subsidy vs. $\rho^{\texttt{S}} / \rho^{\texttt{A}}$ using a mixed test process.} The figure shows, as a function of the social-to-agent approval benefit ratio, the optimal subsidy obtained using Algorithm~\ref{alg:epsilon}.
    }
    \label{fig:mixture_subsidy}
\end{figure}

\begin{figure}[H]
    \vspace{0cm}
    \centering
    \includegraphics[width=0.60\linewidth]{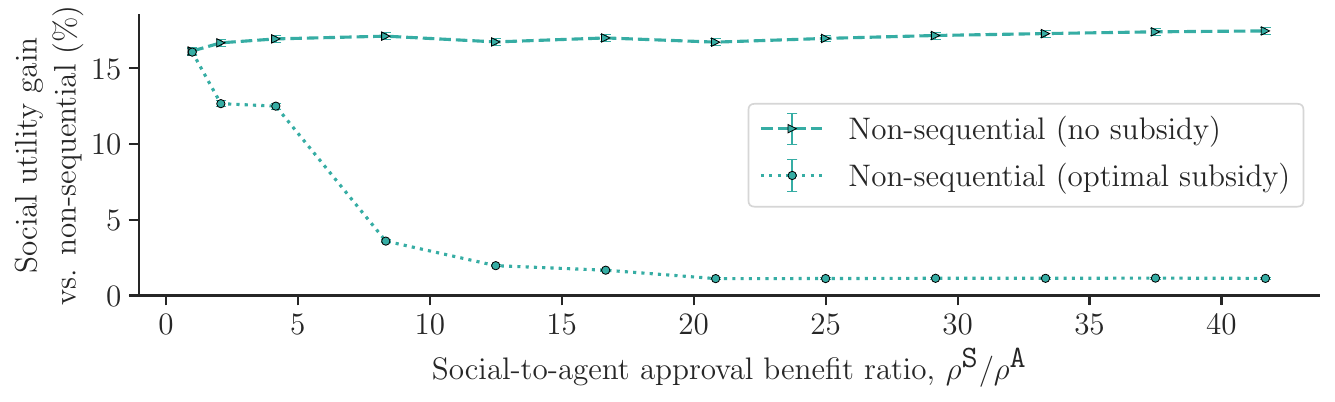}
    \caption{\textbf{Social utility gain vs. $\rho^{\texttt{S}} / \rho^{\texttt{A}}$ using a mixed test process.} The figure shows, as a function of the social-to-agent approval benefit ratio, the percentage increase in social utility of the sequential approval protocol relative to a non-sequential approval protocol in which the agent is restricted to a single trial with $n^{\texttt{max}}=800$, under (i) the optimal subsidy computed using Algorithm~\ref{alg:epsilon} and (ii) no subsidy ($\varepsilon=0$).
    }
    \label{fig:mixture_gain}
\end{figure}

\begin{figure}[h]
    \vspace{0cm}
    \centering
    \includegraphics[width=0.6\linewidth]{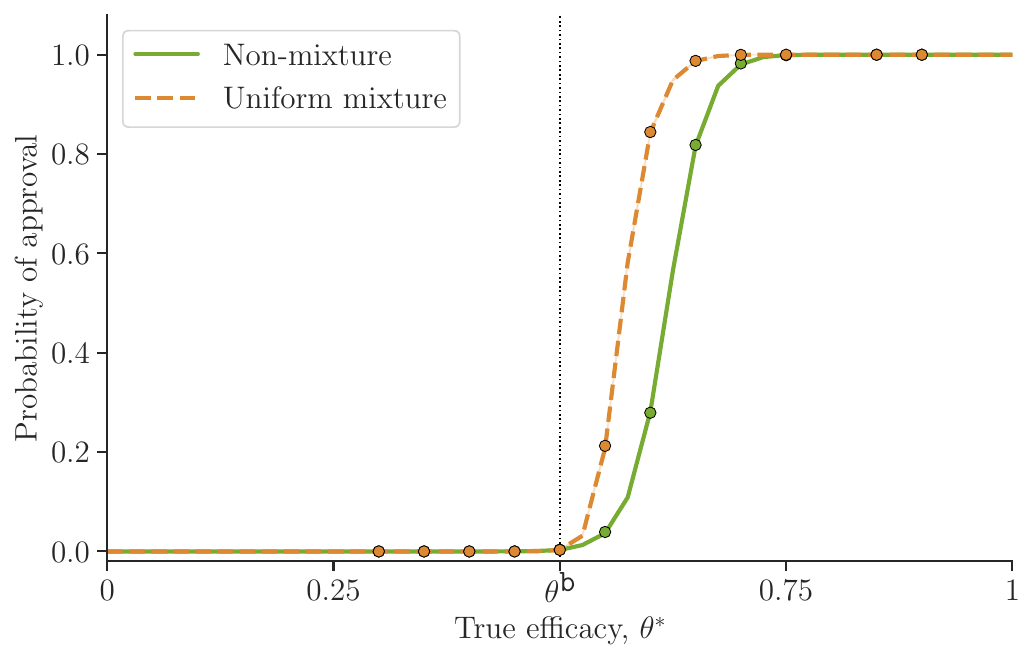}
    \caption{\textbf{Probability of approval under the optimal policy and subsidy.} The figure shows, across various efficacies $\theta^*$ of the antibiotic, the probability of approval (that is, of rejecting $H_0$) when the principal selects the optimal subsidy and the agent its optimal policy. The dashed (orange) curve corresponds to the process $M^{\texttt{mix}}$ defined in Eq.~\ref{eq:mixture-process} for a uniform mixture (optimal subsidy $\varepsilon^* = 0.027$), while the solid (green) curve corresponds to the process $M$ defined in Proposition~\ref{prop:binomial-e-value} and Eq.~\ref{eq:test-process} (optimal subsidy $\varepsilon^* = 0.108$).
    }
    \label{fig:mixture_power}
\end{figure}

\subsubsection{Sensitivity to misspecification of the parameters}
In this section, we empirically evaluate the robustness of the optimal subsidy $\varepsilon^*$ computed by the principal using Algorithm~\ref{alg:epsilon} to misspecifications in the agent's approval benefit $\rho^{\mathcal{A}}$ and the costs $c_0$ and $c_1$. To this end, we compute $\varepsilon^*$ and the maximum anticipated social utility under the nominal parameters in Table~\ref{tab:non-econ-parameters} and Table~\ref{tab:econ-parameters}. We then separately perturb the agent's costs $c_0$ and $c_1$, the agent's approval benefit $\rho^{\mathcal{A}}$, and the belief parameters $(\alpha_0,\beta_0)$ used by the principal to compute the optimal subsidy, each by $\pm 5 \%$, corresponding to a setting in which the principal has only approximate information about the agent's parameters. For each perturbed set of parameters, we compute the resulting subsidy $\varepsilon^*_{\mathrm{pert}}$, its deviation $\Delta \varepsilon^* = \varepsilon^*_{\mathrm{pert}} - \varepsilon^*$ from the well-specified case, and the anticipated social utility the principal loses by deploying $\varepsilon^*_{\mathrm{pert}}$ rather than the perfect-information $\varepsilon^*$. As shown in Table~\ref{tab:pert}, both the social utility and the subsidy change only marginally, confirming that the outcome of our protocol is robust to misspecification.

\begin{table}[H]
\centering
\caption{\textbf{Sensitivity to parameter misspecifications}}
\label{tab:pert}
\begin{tabular}{ccc}
\toprule
  Perturbed parameter  & Subsidy change $\Delta \varepsilon^*$ & Social utility loss (\%)  \\
\midrule
   $\rho^{\mathrm{A}},  +5\%$    & $-0.091$ & $0.69$ \\
   $\rho^{\mathrm{A}},  -5\%$    & $+0.111$ & $0.23$ \\
   $c_0,  +5\%$    & $+0.070$ & $0.13$ \\
   $c_0,  -5\%$    & $-0.033$ & $0.83$ \\
   $c_1,  +5\%$    & $+0.079$ & $0.16$ \\
   $c_1,  -5\%$    & $-0.018$ & $0.87$ \\
   $\alpha_0,  +5\%$    & $-0.009$ & $0.90$ \\
   $\alpha_0,  -5\%$    & $+0.080$ & $0.16$ \\
   $\beta_0,  +5\%$    & $+0.013$ & $0.04$ \\
   $\beta_0,  -5\%$    & $-0.013$ & $0.88$ \\
\bottomrule
\end{tabular}
\end{table}

%% file: strategic-experiments.bib
@book{Sutton1998,
  author = {Sutton, Richard S. and Barto, Andrew G.},
  edition = {Second},
  publisher = {The MIT Press},
  title = {Reinforcement Learning: An Introduction},
  url = {http://incompleteideas.net/book/the-book-2nd.html},
  year = {2018 }
}

@incollection{grossman1992analysis,
  title={An analysis of the principal-agent problem},
  author={Grossman, Sanford J and Hart, Oliver D},
  booktitle={Foundations of insurance economics: Readings in economics and finance},
  pages={302--340},
  year={1992},
  publisher={Springer}
}

@misc{waudbysmith2025universallogoptimalitygeneralclasses,
      title={Universal Log-Optimality for General Classes of e-processes and Sequential Hypothesis Tests}, 
      author={Ian Waudby-Smith and Ricardo Sandoval and Michael I. Jordan},
      year={2025},
      eprint={2504.02818},
      archivePrefix={arXiv},
      primaryClass={math.ST},
      url={https://arxiv.org/abs/2504.02818}, 
}

@article{ramdas2025hypothesistestingevalues,
    author = {Ramdas, Aaditya and Wang, Ruodu},
    title = {Hypothesis Testing with E-values},
    journal = {Foundations and Trends in Statistics},
    volume = {1},
    number = {1-2},
    pages = {1-390},
    year = {2025},
    month = {07},
    issn = {2978-4212},
    doi = {10.1561/3600000002},
    url = {https://doi.org/10.1561/3600000002},
    eprint = {https://www.emerald.com/ftstat/article-pdf/1/1-2/1/11146063/3600000002en.pdf},
}

@book{ville1939étude,
  title={{\'E}tude Critique de la Notion de Collectif},
  author={Ville, J.},
  lccn={ac40000526},
  series={Collection des monographies des probabilit{\'e}s},
  url={https://books.google.de/books?id=ETY7AQAAIAAJ},
  year={1939},
  publisher={Gauthier-Villars}
}

@ARTICLE{ramdas2023gametheoreticstatisticssafeanytimevalid,
  title     = "Game-theoretic statistics and safe anytime-valid inference",
  author    = "Ramdas, Aaditya and Gr{\"u}nwald, Peter and Vovk, Vladimir and
               Shafer, Glenn",
  journal   = "Stat. Sci.",
  publisher = "Institute of Mathematical Statistics",
  volume    =  38,
  number    =  4,
  pages     = "576--601",
  month     =  nov,
  year      =  2023
}

@article{shekhar2024Two,
author = {Shekhar, Shubhanshu and Ramdas, Aaditya},
title = {Nonparametric Two-Sample Testing by Betting},
year = {2024},
issue_date = {Feb. 2024},
publisher = {IEEE Press},
volume = {70},
number = {2},
issn = {0018-9448},
url = {https://doi.org/10.1109/TIT.2023.3305867},
doi = {10.1109/TIT.2023.3305867},
journal = {IEEE Trans. Inf. Theor.},
month = feb,
pages = {1178–1203},
numpages = {26}
}

@InProceedings{xu2024Online,
  title = 	 {Online multiple testing with e-values},
  author =       {Xu, Ziyu and Ramdas, Aaditya},
  booktitle = 	 {Proceedings of The 27th International Conference on Artificial Intelligence and Statistics},
  pages = 	 {3997--4005},
  year = 	 {2024},
  editor = 	 {Dasgupta, Sanjoy and Mandt, Stephan and Li, Yingzhen},
  volume = 	 {238},
  series = 	 {Proceedings of Machine Learning Research},
  month = 	 {02--04 May},
  publisher =    {PMLR},
  url = 	 {https://proceedings.mlr.press/v238/xu24a.html}
}

@article{shin_ramdas_rinaldo_2024,
    author = {Jaehyeok Shin and Aaditya Ramdas and Alessandro Rinaldo},
    title = {E-detectors: A Nonparametric Framework for Sequential Change Detection},
    journal = {The New England Journal of Statistics in Data Science},
    volume = {2},
    number = {2},
    year = {2024},
    pages = {229--260},
    doi = {10.51387/23-NEJSDS51},
    issn = {2693-7166},
    publisher = {New England Statistical Society}
}

@article{waldSeq,
author = {A. Wald},
title = {{Sequential Tests of Statistical Hypotheses}},
volume = {16},
journal = {The Annals of Mathematical Statistics},
number = {2},
publisher = {Institute of Mathematical Statistics},
pages = {117 -- 186},
year = {1945},
doi = {10.1214/aoms/1177731118},
URL = {https://doi.org/10.1214/aoms/1177731118}
}

@ARTICLE{Janiaud2021-jg,
  title     = "{U.S}. food and Drug Administration reasoning in approval
               decisions when efficacy evidence is borderline, 2013-2018",
  author    = "Janiaud, Perrine and Irony, Telba and Russek-Cohen, Estelle and
               Goodman, Steven N",
  journal   = "Ann. Intern. Med.",
  publisher = "American College of Physicians",
  volume    =  174,
  number    =  11,
  pages     = "1603--1611",
  month     =  nov,
  year      =  2021,
  language  = "en"
}

@BOOK{Von_Stackelberg2010,
  title     = "Market Structure and Equilibrium",
  author    = "von Stackelberg, Heinrich",
  publisher = "Springer",
  edition   =  2011,
  month     =  aug,
  year      =  2010,
  address   = "Berlin, Germany",
  copyright = "https://www.springernature.com/gp/researchers/text-and-data-mining",
  language  = "en"
}

@article{Wu31122022,
author = {Wan-Shu Wu and Kai Zhao},
title = {Government R\&D subsidies and enterprise R\&D activities: theory and evidence},
journal = {Economic Research-Ekonomska Istraživanja},
volume = {35},
number = {1},
pages = {391--408},
year = {2022},
publisher = {Routledge},
doi = {10.1080/1331677X.2021.1893204},


URL = { 
    
        https://doi.org/10.1080/1331677X.2021.1893204
    
    

},
eprint = { 
    
        https://doi.org/10.1080/1331677X.2021.1893204
    
    

}

}

@ARTICLE{Renwick2016-wj,
  title     = "A systematic review and critical assessment of incentive
               strategies for discovery and development of novel antibiotics",
  author    = "Renwick, Matthew J and Brogan, David M and Mossialos, Elias",
  journal   = "J. Antibiot. (Tokyo)",
  publisher = "Springer Science and Business Media LLC",
  volume    =  69,
  number    =  2,
  pages     = "73--88",
  month     =  feb,
  year      =  2016,
  language  = "en"
}

@techreport{fda2019effectiveness,
  author      = {{U.S. Food and Drug Administration}},
  title       = {Demonstrating Substantial Evidence of Effectiveness for Human Drug and Biological Products: Guidance for Industry},
  institution = {U.S. Department of Health and Human Services, Food and Drug Administration, Center for Drug Evaluation and Research (CDER), Center for Biologics Evaluation and Research (CBER)},
  year        = {2019},
  month       = {December},
  type        = {Draft Guidance},
  url         = {https://www.fda.gov/drugs/guidance-compliance-regulatory-information/guidances-drugs},
}

@misc{nih2026smallbusiness,
  author       = {{National Institutes of Health}},
  title        = {Small Business Funding},
  organization = {U.S. Department of Health and Human Services, National Institutes of Health, Small Business Education and Entrepreneurial Development (SEED) Office},
  year         = {2026},
  howpublished = {\url{https://seed.nih.gov/small-business-funding}},
  note         = {Accessed: 2026-05-03}
}

@misc{usc26_45c,
  author       = {{U.S. Congress}},
  title        = {26 {U.S.C.} 45{C} --- Clinical Testing Expenses for Certain Drugs for Rare Diseases or Conditions},
  howpublished = {Internal Revenue Code},
  year         = {2024},
  url          = {https://www.govinfo.gov/app/details/USCODE-2024-title26/USCODE-2024-title26-subtitleA-chap1-subchapA-partIV-subpartD-sec45C},
  note         = {Accessed: 2026-05-03}
}

@techreport{fda2026use,
  author ={{U.S. Food and Drug Administration}},
  title = {Use of Bayesian Methodology in Clinical Trials of Drug and Biological Products},
  institution = {Center for Biologics Evaluation and Research and Center for Drug Evaluation and Research, Food and Drug Administration},
  year = {2026},
  month       = {March},
  type        = {Draft Guidance}
}

@article{PenaExp,
author = {Victor H. de la Pe{\~n}a},
title = {{A General Class of Exponential Inequalities for Martingales and Ratios}},
volume = {27},
journal = {The Annals of Probability},
number = {1},
publisher = {Institute of Mathematical Statistics},
pages = {537 -- 564},
year = {1999},
doi = {10.1214/aop/1022677271},
URL = {https://doi.org/10.1214/aop/1022677271}
}

@article{HowardTimeUni,
author = {Steven R. Howard and Aaditya Ramdas and Jon McAuliffe and Jasjeet Sekhon},
title = {{Time-uniform Chernoff bounds via nonnegative supermartingales}},
volume = {17},
journal = {Probability Surveys},
number = {none},
publisher = {Institute of Mathematical Statistics and Bernoulli Society},
pages = {257 -- 317},
year = {2020},
doi = {10.1214/18-PS321},
URL = {https://doi.org/10.1214/18-PS321}
}

@book{HernndezLerma1996,
  title = {Discrete-Time Markov Control Processes},
  ISBN = {9781461207290},
  url = {http://dx.doi.org/10.1007/978-1-4612-0729-0},
  DOI = {10.1007/978-1-4612-0729-0},
  publisher = {Springer New York},
  author = {Hernández-Lerma,  Onésimo and Lasserre,  Jean Bernard},
  year = {1996}
}

@ARTICLE{Simmons2011-ny,
  title     = "False-positive psychology: undisclosed flexibility in data
               collection and analysis allows presenting anything as
               significant",
  author    = "Simmons, Joseph P and Nelson, Leif D and Simonsohn, Uri",
  journal   = "Psychol. Sci.",
  publisher = "SAGE Publications",
  volume    =  22,
  number    =  11,
  pages     = "1359--1366",
  month     =  nov,
  year      =  2011,
  language  = "en"
}

@article{GrunwaldSafe,
    author = {Grünwald, Peter and de Heide, Rianne and Koolen, Wouter},
    title = {Safe testing},
    journal = {Journal of the Royal Statistical Society Series B: Statistical Methodology},
    volume = {86},
    number = {5},
    pages = {1091-1128},
    year = {2024},
    month = {03},
    issn = {1369-7412},
    doi = {10.1093/jrsssb/qkae011},
    url = {https://doi.org/10.1093/jrsssb/qkae011},
    eprint = {https://academic.oup.com/jrsssb/article-pdf/86/5/1091/60648648/qkae011.pdf},
}

@BOOK{B2005-xl,
  title     = "Encyclopedia of statistical sciences",
  author    = "Campbell, B. and Balakrishnan, N and Vidakovic, Brani",
  editor    = "Kotz, Samuel and Read, Campbell B and Balakrishnan, N and
               Vidakovic, Brani and Johnson, Norman L",
  publisher = "John Wiley \& Sons",
  series    = "Methods and Applications of Statistics",
  edition   =  2,
  month     =  dec,
  year      =  2005,
  address   = "Nashville, TN",
  language  = "en"
}

@inproceedings{
velasco2026auditing,
title={Auditing Pay-Per-Token in Large Language Models},
author={Ander Artola Velasco and Stratis Tsirtsis and Manuel Gomez Rodriguez},
booktitle={The 29th International Conference on Artificial Intelligence and Statistics},
year={2026},
url={https://openreview.net/forum?id=6lboj007YA}
}

@inproceedings{
gauthier2026bettingequilibriummonitoringstrategic,
title={Anytime Detection of Strategic Deviations in Multi-Agent Systems},
author={Etienne Gauthier and Francis Bach and Michael I. Jordan},
booktitle={ICML 2026 Statistical Frameworks for Uncertainty in Agentic Systems},
year={2026},
url={https://openreview.net/forum?id=LePUv1OL2p}
}

@ARTICLE{waudbysmith2022estimatingmeansboundedrandom,
  title     = "Estimating means of bounded random variables by betting",
  author    = "Waudby-Smith, Ian and Ramdas, Aaditya",
  journal   = "J. R. Stat. Soc. Series B Stat. Methodol.",
  publisher = "Oxford University Press (OUP)",
  volume    =  86,
  number    =  1,
  pages     = "1--27",
  month     =  feb,
  year      =  2024,
  copyright = "https://academic.oup.com/pages/standard-publication-reuse-rights",
  language  = "en"
}

@article{cho2022bayesian,
  title={Bayesian learning approach to model predictive control},
  author={Cho, Namhoon and Lee, Seokwon and Shin, Hyo-Sang and Tsourdos, Antonios},
  journal={arXiv preprint arXiv:2203.02720},
  year={2022}
}

@ARTICLE{russo2020tutorialthompsonsampling,
  title     = "{A tutorial on Thompson sampling}",
  author    = "Russo, Daniel J and Van Roy, Benjamin and Kazerouni, Abbas and
               Osband, Ian and Wen, Zheng",
  journal   = "Found. Trends\textregistered{} Mach. Learn.",
  publisher = "Emerald",
  volume    =  11,
  number    =  1,
  pages     = "1--96",
  month     =  jul,
  year      =  2018,
  language  = "en"
}

@article{Ghavamzadeh_2015,
   title={Bayesian Reinforcement Learning: A Survey},
   volume={8},
   ISSN={1935-8245},
   url={http://dx.doi.org/10.1561/2200000049},
   DOI={10.1561/2200000049},
   number={5–6},
   journal={Foundations and Trends® in Machine Learning},
   publisher={Emerald},
   author={Ghavamzadeh, Mohammad and Mannor, Shie and Pineau, Joelle and Tamar, Aviv},
   year={2015},
   month=nov, pages={359–483} }

@InProceedings{buening2023minimaxbayesreinforcementlearning,
  title = 	 {Minimax-Bayes Reinforcement Learning},
  author =       {Buening, Thomas Kleine and Dimitrakakis, Christos and Eriksson, Hannes and Grover, Divya and Jorge, Emilio},
  booktitle = 	 {Proceedings of The 26th International Conference on Artificial Intelligence and Statistics},
  pages = 	 {7511--7527},
  year = 	 {2023},
  editor = 	 {Ruiz, Francisco and Dy, Jennifer and van de Meent, Jan-Willem},
  volume = 	 {206},
  series = 	 {Proceedings of Machine Learning Research},
  month = 	 {25--27 Apr},
  publisher =    {PMLR},
  url = 	 {https://proceedings.mlr.press/v206/buening23a.html}
}

@article{kaelbling1998planning,
  title={Planning and acting in partially observable stochastic domains},
  author={Kaelbling, Leslie Pack and Littman, Michael L and Cassandra, Anthony R},
  journal={Artificial intelligence},
  volume={101},
  number={1-2},
  pages={99--134},
  year={1998},
  publisher={Elsevier}
}

@inproceedings{hossain2025strategic,
  title={Strategic Hypothesis Testing},
  author={Hossain, Safwan and Chen, Yatong and Chen, Yiling},
  booktitle={The Thirty-Ninth Annual Conference on Neural Information Processing Systems (NeurIPS)},
  year={2025}
}

@article{bates2022principal,
  title     = "Principal-agent hypothesis testing",
  author    = "Bates, Stephen and Jordan, Michael I and Sklar, Michael and
               Soloff, Jake A",
  journal   = "J. Am. Stat. Assoc.",
  publisher = "Informa UK Limited",
  pages     = "1--38",
  month     =  aug,
  year      =  2026,
  language  = "en"
}

@inproceedings{
huang2026towards,
title={Towards Optimal Statistical Watermarking},
author={Baihe Huang and Banghua Zhu and Hanlin Zhu and Jason Lee and Jiantao Jiao and Michael Jordan},
booktitle={Socially Responsible Language Modelling Research},
year={2023},
url={https://openreview.net/forum?id=Fc2FaS9mYJ}
}

@article{chugg2026post,
  title={Post-Hoc Large-Sample Statistical Inference},
  author={Chugg, Ben and Gauthier, Etienne and Jordan, Michael I and Ramdas, Aaditya and Waudby-Smith, Ian},
  journal={arXiv preprint arXiv:2603.08002},
  year={2026}
}

@inproceedings{
dhillon2026escores,
title={E-Scores for (In)Correctness Assessment of Generative Model Outputs},
author={Guneet S. Dhillon and Javier Gonzalez and Teodora Pandeva and Alicia Curth},
booktitle={The 29th International Conference on Artificial Intelligence and Statistics},
year={2026},
url={https://openreview.net/forum?id=PCRCLYgiVK}
}

@inproceedings{
gauthier2025backward,
title={Backward Conformal Prediction},
author={Etienne Gauthier and Francis Bach and Michael I. Jordan},
booktitle={The Thirty-ninth Annual Conference on Neural Information Processing Systems},
year={2025},
url={https://openreview.net/forum?id=23ichdd74N}
}

@article{McClellan2022,
author = {McClellan, Andrew},
title = {Experimentation and Approval Mechanisms},
journal = {Econometrica},
volume = {90},
number = {5},
pages = {2215-2247},
doi = {https://doi.org/10.3982/ECTA17021},
url = {https://onlinelibrary.wiley.com/doi/abs/10.3982/ECTA17021},
eprint = {https://onlinelibrary.wiley.com/doi/pdf/10.3982/ECTA17021},
year = {2022}
}

@article{shi2024sharp,
  title={Sharp Results for Hypothesis Testing with Risk-Sensitive Agents},
  author={Shi, Flora C and Bates, Stephen and Wainwright, Martin J},
  journal={arXiv preprint arXiv:2412.16452},
  year={2024}
}

@TechReport{TetenovEcon,
type={CeMMAP working papers},
institution={Institute for Fiscal Studies},
author={Aleksey Tetenov},
title={An economic theory of statistical testing},
year={2016},
month={Sep},
number={50/16},
doi={10.1920/wp.cem.2016.5016},
url={https://ideas.repec.org/p/azt/cemmap/50-16.html},
}

@BOOK{Peskir2006-vv,
  title     = "Optimal stopping and free-boundary problems",
  author    = "Peskir, Goran and Shiryaev, Albert N",
  publisher = "Birkhauser Verlag AG",
  series    = "Lectures in Mathematics. ETH Z{\"u}rich",
  edition   =  2006,
  month     =  aug,
  year      =  2006,
  address   = "Basel, Switzerland",
  language  = "en"
}

@BOOK{Powell2012-tq,
  title     = "Optimal Learning",
  author    = "Powell, Warren B and Ryzhov, Ilya O",
  publisher = "Wiley-Blackwell",
  series    = "Wiley Series in Probability and Statistics",
  month     =  mar,
  year      =  2012,
  address   = "Hoboken, NJ",
  language  = "en"
}

@article{RobbinsOptimal,
 ISSN = {00029890, 19300972},
 URL = {http://www.jstor.org/stable/2316139},
 author = {Herbert Robbins},
 journal = {The American Mathematical Monthly},
 number = {4},
 pages = {333--343},
 publisher = {[Taylor & Francis, Ltd., Mathematical Association of America]},
 title = {Optimal Stopping},
 urldate = {2026-03-27},
 volume = {77},
 year = {1970}
}

@ARTICLE{Shen2025-zc,
  title     = "Variational sequential optimal experimental design using
               reinforcement learning",
  author    = "Shen, Wanggang and Dong, Jiayuan and Huan, Xun",
  journal   = "Comput. Methods Appl. Mech. Eng.",
  publisher = "Elsevier BV",
  volume    =  444,
  number    =  118068,
  pages     = "118068",
  month     =  sep,
  year      =  2025,
  language  = "en"
}

@article{cheng2025optimal,
  title={Optimal Stopping for Sequential Bayesian Experimental Design},
  author={Cheng, Chen and Huan, Xun},
  journal={arXiv preprint arXiv:2509.21734},
  year={2025}
}

@article{shen2023bayesian,
  title={Bayesian sequential optimal experimental design for nonlinear models using policy gradient reinforcement learning},
  author={Shen, Wanggang and Huan, Xun},
  journal={Computer Methods in Applied Mechanics and Engineering},
  volume={416},
  pages={116304},
  year={2023},
  publisher={Elsevier}
}

@book{rockafellar1970convex,
  title={Convex Analysis},
  author={Rockafellar, R.T.},
  isbn={9780691080697},
  lccn={68056318},
  series={Princeton landmarks in mathematics and physics},
  url={https://books.google.de/books?id=OI4Ph2dXXhsC},
  year={1970},
  publisher={Princeton University Press}
}

@inproceedings{Conitzer,
author = {Conitzer, Vincent and Sandholm, Tuomas},
title = {Computing the optimal strategy to commit to},
year = {2006},
isbn = {1595932364},
publisher = {Association for Computing Machinery},
address = {New York, NY, USA},
url = {https://doi.org/10.1145/1134707.1134717},
doi = {10.1145/1134707.1134717},
booktitle = {Proceedings of the 7th ACM Conference on Electronic Commerce},
pages = {82–90},
numpages = {9},
location = {Ann Arbor, Michigan, USA},
series = {EC '06}
}

@article{farina2021strength,
  title={Strength of clinical evidence leading to approval of novel cancer medicines in Europe: A systematic review and data synthesis},
  author={Farina, Alberto and Moro, Federico and Fasslrinner, Frederick and Sedghi, Annahita and Bromley, Miluska and Siepmann, Timo},
  journal={Pharmacology Research \& Perspectives},
  volume={9},
  number={4},
  pages={e00816},
  year={2021},
  publisher={Wiley Online Library}
}

@article{gieringer1985safety,
  title={The safety and efficacy of new drug approval},
  author={Gieringer, Dale H},
  journal={Cato J.},
  volume={5},
  pages={177},
  year={1985},
  publisher={HeinOnline}
}

@article{o2010building,
  title={Building comparative efficacy and tolerability into the FDA approval process},
  author={O’Connor, Alec B},
  journal={Jama},
  volume={303},
  number={10},
  pages={979--980},
  year={2010},
  publisher={American Medical Association}
}

@article{frantz2003clinical,
  title={Why are clinical costs so high?},
  author={Frantz, Simon},
  journal={Nature Reviews Drug Discovery},
  volume={2},
  number={11},
  year={2003}
}

@article{detsky1989clinical,
  title={Are clinical trials a cost-effective investment?},
  author={Detsky, Allan S},
  journal={Jama},
  volume={262},
  number={13},
  pages={1795--1800},
  year={1989},
  publisher={American Medical Association}
}

@article{martin2017much,
  title={How much do clinical trials cost?},
  author={Martin, Linda and Hutchens, Melissa and Hawkins, Conrad and Radnov, Alaina},
  journal={Nature Reviews Drug Discovery},
  volume={16},
  number={6},
  pages={381--382},
  year={2017},
  publisher={Nature Publishing Group UK London}
}

@article{mahajan2010adaptive,
  title={Adaptive design clinical trials: Methodology, challenges and prospect},
  author={Mahajan, Rajiv and Gupta, Kapil},
  journal={Indian journal of pharmacology},
  volume={42},
  number={4},
  pages={201--207},
  year={2010},
  publisher={Medknow}
}

@misc{nih,
  title = {Clinical Trial-Specific Funding Opportunities},
  howpublished = {\url{https://grants.nih.gov/policy-and-compliance/policy-topics/clinical-trials/specific-funding-opportunities}},
  note = {Accessed: 2026-03-31},
  year = {2026},
  author       = {{National Institutes of Health}},

}

@misc{fda,
  title = {Clinical Trials Grants Program},
  howpublished = {\url{https://www.fda.gov/industry/orphan-products-grants-program/clinical-trials-grants-program}},
  note = {Accessed: 2026-03-31},
  author       = {{U.S. Food and Drug Administration}},
  year = {2024},


}

@misc{dfg,
  title = {Clinical Trials},
  howpublished = {\url{https://www.dfg.de/en/research-funding/funding-opportunities/programmes/individual/clinical-trials}},
  note = {Accessed: 2026-03-31},
  year = {2024},
  author       = {{Deutsche Forschungsgemeinschaft}}
}

@misc{edctp,
  title = {The European and Developing Countries Clinical Trials Partnership},
  howpublished = {\url{https://www.edctp.org/}},
  note = {Accessed: 2026-03-31},
  year = {2025},
  author       = {{European and Developing Countries Clinical Trials Partnership}},
}

@article{brown2009adaptive,
  title={Adaptive designs for randomized trials in public health},
  author={Brown, C Hendricks and Ten Have, Thomas R and Jo, Booil and Dagne, Getachew and Wyman, Peter A and Muth{\'e}n, Bengt and Gibbons, Robert D},
  journal={Annual review of public health},
  volume={30},
  number={1},
  pages={1--25},
  year={2009},
  publisher={Annual Reviews}
}

@ARTICLE{GBD_2021_Antimicrobial_Resistance_Collaborators2024,
  title     = "Global burden of bacterial antimicrobial resistance 1990-2021: a
               systematic analysis with forecasts to 2050",
  author    = "{GBD 2021 Antimicrobial Resistance Collaborators}",
  journal   = "Lancet",
  publisher = "Elsevier BV",
  volume    =  404,
  number    =  10459,
  pages     = "1199--1226",
  month     =  sep,
  year      =  2024,
  copyright = "http://creativecommons.org/licenses/by/4.0/",
  language  = "en"
}

@ARTICLE{Shallcross2015-tg,
  title     = "Tackling the threat of antimicrobial resistance: from policy to
               sustainable action",
  author    = "Shallcross, Laura J and Howard, Simon J and Fowler, Tom and
               Davies, Sally C",
  journal   = "Philos. Trans. R. Soc. Lond. B Biol. Sci.",
  publisher = "The Royal Society",
  volume    =  370,
  number    =  1670,
  pages     = "20140082",
  month     =  jun,
  year      =  2015,
  language  = "en"
}

@ARTICLE{Gargate2025,
  title    = "Current economic and regulatory challenges in developing
              antibiotics for Gram-negative bacteria",
  author   = "Gargate, Nupur and Laws, Mark and Rahman, Khondaker Miraz",
  journal  = "NPJ Antimicrob. Resist.",
  volume   =  3,
  number   =  1,
  pages    = "50",
  month    =  jun,
  year     =  2025,
  language = "en"
}

@ARTICLE{Outterson2013-qz,
  title     = "Approval and withdrawal of new antibiotics and other
               antiinfectives in the {U.S.}, 1980-2009",
  author    = "Outterson, Kevin and Powers, John H and Seoane-Vazquez, Enrique
               and Rodriguez-Monguio, Rosa and Kesselheim, Aaron S",
  journal   = "J. Law Med. Ethics",
  publisher = "Cambridge University Press (CUP)",
  volume    =  41,
  number    =  3,
  pages     = "688--696",
  year      =  2013,
  language  = "en"
}

@ARTICLE{Piddock2024-dn,
  title     = "Advancing global antibiotic research, development and access",
  author    = "Piddock, Laura J V and Alimi, Yewande and Anderson, James and de
               Felice, Damiano and Moore, Catrin E and R{\o}ttingen, John-Arne
               and Skinner, Henry and Beyer, Peter",
  journal   = "Nat. Med.",
  publisher = "Springer Science and Business Media LLC",
  volume    =  30,
  number    =  9,
  pages     = "2432--2443",
  month     =  sep,
  year      =  2024,
  copyright = "https://www.springernature.com/gp/researchers/text-and-data-mining",
  language  = "en"
}

@ARTICLE{Courtemanche2021-up,
  title     = "Looking for solutions to the pitfalls of developing novel
               antibacterials in an economically challenging system",
  author    = "Courtemanche, Gilles and Wadanamby, Rohini and Kiran,
               Amritanjali and Toro-Alzate, Luisa Fernanda and Diggle, Mathew
               and Chakraborty, Dipanjan and Blocker, Ariel and van Dongen,
               Maarten",
  journal   = "Microbiol. Res. (Pavia)",
  publisher = "MDPI AG",
  volume    =  12,
  number    =  1,
  pages     = "173--185",
  month     =  mar,
  year      =  2021,
  copyright = "https://creativecommons.org/licenses/by/4.0/",
  language  = "en"
}

@ARTICLE{Plackett2020-qi,
  title     = "Why big pharma has abandoned antibiotics",
  author    = "Plackett, Benjamin",
  journal   = "Nature",
  publisher = "Springer Science and Business Media LLC",
  volume    =  586,
  number    =  7830,
  pages     = "S50--S52",
  month     =  oct,
  year      =  2020,
  copyright = "https://www.springernature.com/gp/researchers/text-and-data-mining",
  language  = "en"
}

@ARTICLE{Wells2024-fs,
  title     = "Novel insights from financial analysis of the failure to
               commercialise plazomicin: Implications for the antibiotic
               investment ecosystem",
  author    = "Wells, Nadya and Nguyen, Vinh-Kim and Harbarth, Stephan",
  journal   = "Humanit. Soc. Sci. Commun.",
  publisher = "Springer Science and Business Media LLC",
  volume    =  11,
  number    =  1,
  month     =  jul,
  year      =  2024,
  copyright = "https://creativecommons.org/licenses/by/4.0",
  language  = "en"
}

@ARTICLE{Anderson2023-gl,
  title     = "Challenges and opportunities for incentivising antibiotic
               research and development in Europe",
  author    = "Anderson, Michael and Panteli, Dimitra and van Kessel, Robin and
               Ljungqvist, Gunnar and Colombo, Francesca and Mossialos, Elias",
  journal   = "Lancet Reg. Health Eur.",
  publisher = "Elsevier BV",
  volume    =  33,
  number    =  100705,
  pages     = "100705",
  month     =  oct,
  year      =  2023,
  copyright = "http://creativecommons.org/licenses/by/4.0/",
  language  = "en"
}

@ARTICLE{Shlaes2019-fe,
  title     = "The economic conundrum for antibacterial drugs",
  author    = "Shlaes, David M",
  journal   = "Antimicrob. Agents Chemother.",
  publisher = "American Society for Microbiology",
  volume    =  64,
  number    =  1,
  month     =  dec,
  year      =  2019,
  language  = "en"
}

@misc{HR7352_119th_Congress_2026,
  author       = "{United States Congress}",
  title        = "{H.R. 7352: PASTEUR Act of 2026}",
  year         = "2026",
  note         = "To amend the Public Health Service Act to establish a program to develop innovative antimicrobial drugs",
  url          = "https://www.congress.gov/bill/119th-congress/house-bill/7352",
  date         = "2026-02-04"
}

@ARTICLE{Outterson2016-et,
  title     = "Accelerating global innovation to address antibacterial
               resistance: introducing {CARB-X}",
  author    = "Outterson, Kevin and Rex, John H and Jinks, Tim and Jackson,
               Peter and Hallinan, John and Karp, Steve and Hung, Deborah T and
               Franceschi, Francois and Merkeley, Tyler and Houchens,
               Christopher and Dixon, Dennis M and Kurilla, Michael G and
               Aurigemma, Rosemarie and Larsen, Joseph",
  journal   = "Nat. Rev. Drug Discov.",
  publisher = "Springer Science and Business Media LLC",
  volume    =  15,
  number    =  9,
  pages     = "589--590",
  month     =  sep,
  year      =  2016,
  language  = "en"
}

@ARTICLE{Rahman2021-zk,
  title     = "Market concentration of new antibiotic sales",
  author    = "Rahman, Sakib and Lindahl, Olof and Morel, Chantal M and Hollis,
               Aidan",
  journal   = "J. Antibiot. (Tokyo)",
  publisher = "Springer Science and Business Media LLC",
  volume    =  74,
  number    =  6,
  pages     = "421--423",
  month     =  jun,
  year      =  2021,
  language  = "en"
}

@ARTICLE{Moore2018-ls,
  title     = "Estimated costs of pivotal trials for novel therapeutic agents
               approved by the {US} Food and Drug Administration, 2015-2016",
  author    = "Moore, Thomas J and Zhang, Hanzhe and Anderson, Gerard and
               Alexander, G Caleb",
  journal   = "JAMA Intern. Med.",
  publisher = "American Medical Association (AMA)",
  volume    =  178,
  number    =  11,
  pages     = "1451--1457",
  month     =  nov,
  year      =  2018,
  language  = "en"
}

@ARTICLE{Stergiopoulos2018-fp,
  title     = "Cost drivers of a hospital-acquired bacterial pneumonia and
               ventilator-associated bacterial pneumonia phase 3 clinical trial",
  author    = "Stergiopoulos, Stella and Calvert, Sara B and Brown, Carrie A
               and Awatin, Josephine and Tenaerts, Pamela and Holland, Thomas L
               and DiMasi, Joseph A and Getz, Kenneth A",
  journal   = "Clin. Infect. Dis.",
  publisher = "Oxford University Press (OUP)",
  volume    =  66,
  number    =  1,
  pages     = "72--80",
  month     =  jan,
  year      =  2018,
  language  = "en"
}

@article{shi2025instance,
  title={Instance-Adaptive Hypothesis Tests with Heterogeneous Agents},
  author={Shi, Flora C and Wainwright, Martin J and Bates, Stephen},
  journal={arXiv preprint arXiv:2510.21178},
  year={2025}
}

@inproceedings{paruchuri2008playing,
  title={Playing games for security: An efficient exact algorithm for solving Bayesian Stackelberg games},
  author={Paruchuri, Praveen and Pearce, Jonathan P and Marecki, Janusz and Tambe, Milind and Ordonez, Fernando and Kraus, Sarit},
  booktitle={Proceedings of the 7th international joint conference on Autonomous agents and multiagent systems-Volume 2},
  pages={895--902},
  year={2008}
}

@article{harsanyi1979bayesian,
  title={Bayesian decision theory, rule utilitarianism, and Arrow's impossibility theorem},
  author={Harsanyi, John C},
  journal={Theory and Decision},
  volume={11},
  number={3},
  pages={289--317},
  year={1979},
  publisher={Springer}
}
